\AtBeginDocument{%
	\providecommand\BibTeX{{%
			\normalfont B\kern-0.5em{\scshape i\kern-0.25em b}\kern-0.8em\TeX}}}

\documentclass[10pt,letterpaper]{article}
\usepackage[margin=1in]{geometry}
\usepackage{times}

\usepackage[round]{natbib}
\usepackage{graphicx}
\usepackage{float}

\usepackage{amsmath,amssymb,amsfonts}%
\usepackage{mathrsfs}%
\usepackage[title]{appendix}%
\usepackage{textcomp}%
\usepackage{manyfoot}%
\usepackage{algorithm}%
\usepackage{algorithmicx}%
\usepackage{algpseudocode}%
\usepackage{listings}%

\usepackage{booktabs}
\usepackage{enumitem}
\usepackage{dsfont}
\usepackage[labelformat=simple]{subcaption}

\usepackage{xspace}
\usepackage{amsthm}
\usepackage{multirow}
\theoremstyle{definition}
\newtheorem{definition}{Definition}[]
\usepackage{mathtools}
\usepackage{array}
\usepackage{lipsum}
\usepackage{xcolor}
\usepackage[compact]{titlesec}
\usepackage[T1]{fontenc}
\usepackage[utf8]{inputenc}
\usepackage{authblk}
\usepackage{url}
\usepackage{placeins}

\usepackage{flexisym} 
\usepackage{breqn}

\usepackage{adjustbox}
\usepackage{tabularx}
\usepackage[vlines]{tabularht}
\usepackage[figuresright]{rotating}
\newcommand\red[1]{{#1}}
\newcommand\blue[1]{\textcolor{blue}{#1}}

\usepackage{contour}
\usepackage[normalem]{ulem}
\contourlength{0.8pt}

\newcommand{\ceil}[1]{\lceil #1 \rceil}

\newcommand{\myuline}[1]{%
	\uline{\phantom{#1}}%
	\llap{\contour{white}{#1}}%
}

\newcommand{\smallsection}[1]{{\noindent {\bolden{\myuline{#1}}}}}

\newcommand{\abs}[1]{\lvert#1\rvert}

\definecolor{peace}{RGB}{228, 26, 28}
\definecolor{love}{RGB}{55, 126, 184}
\definecolor{joy}{RGB}{77, 175, 74}
\definecolor{kindness}{RGB}{152, 78, 163}

\newcommand{\setbr}[1]{\{#1\}}

\makeatletter
\newcommand{\AlignFootnote}[1]{%
	\ifmeasuring@
	\else
	\iffirstchoice@
	\footnote{#1}%
	\fi
	\fi}
\makeatother

\newcommand\Algphase[1]{%
	\vspace*{-.7\baselineskip}\Statex\hspace*{\dimexpr-\algorithmicindent-2pt\relax}\rule{\textwidth}{0.4pt}%
	\Statex\hspace*{-\algorithmicindent}\textbf{#1}%
	\vspace*{-.7\baselineskip}\Statex\hspace*{\dimexpr-\algorithmicindent-2pt\relax}\rule{\textwidth}{0.4pt}%
}

\def\mydefbb#1{\expandafter\def\csname bb#1\endcsname{\ensuremath{\mathbb{#1}}}}
\def\mydefallbb#1{\ifx#1\mydefallbb\else\mydefbb#1\expandafter\mydefallbb\fi}
\mydefallbb ABCDEFGHIJKLMNOPQRSTUVWXYZ\mydefallbb

\def\mydefcal#1{\expandafter\def\csname cal#1\endcsname{\ensuremath{\mathcal{#1}}}}
\def\mydefallcal#1{\ifx#1\mydefallcal\else\mydefcal#1\expandafter\mydefallcal\fi}
\mydefallcal ABCDEFGHIJKLMNOPQRSTUVWXYZ\mydefallcal

\usepackage{bm}
\newcommand\bolden[1]{{\boldmath\bfseries#1}}

\newcommand{\Nq}{coauth-DBLP}
\newcommand{\Nw}{coauth-Geology}
\newcommand{\Ne}{NDC-classes}
\newcommand{\Nr}{NDC-substances}
\newcommand{\Nt}{contact-high}
\newcommand{\Ny}{contact-primary}
\newcommand{\Nu}{email-Enron}
\newcommand{\Ni}{email-Eu}
\newcommand{\No}{tags-ubuntu}
\newcommand{\Np}{tags-math}
\newcommand{\Na}{tags-SO}
\newcommand{\Ns}{threads-ubuntu}
\newcommand{\Nd}{threads-math}
\newcommand{\Nf}{threads-SO}

\newtheorem{theorem}{Theorem}
\newtheorem{lemma}{Lemma}
\newtheorem{proposition}{Proposition}%
\newtheorem{remark}{Remark}%
\newtheorem{observation}{Observation}%

\newtheorem{problem}{Problem}%

\title{Hypercore Decomposition for Non-Fragile Hyperedges:\\ Concepts, Algorithms, Observations, and Applications}

\author{Fanchen Bu\thanks{School of Electrical Engineering, KAIST, Daejeon, South Korea, boqvezen97@kaist.ac.kr}, \
Geon Lee\thanks{Kim Jaechul Graduate School of AI, KAIST, Seoul, South Korea, geonlee0325@kaist.ac.kr}, \ and 
Kijung Shin\thanks{Kim Jaechul Graduate School of AI and School of Electrical Engineering, KAIST, Seoul, South Korea, kijungs@kaist.ac.kr}}

\date{}

\begin{document}

\maketitle

\begin{abstract} 
{Hypergraphs are a powerful abstraction for modeling high-order relations, which are ubiquitous in many fields. 
A hypergraph consists of nodes and hyperedges (i.e., subsets of nodes); and there have been a number of attempts to extend the notion of $k$-cores, which proved useful with numerous applications for pairwise graphs, to hypergraphs.
However, the previous extensions are based on an unrealistic assumption that hyperedges are \textit{fragile}, i.e., a high-order relation becomes obsolete as soon as a single member leaves it.

In this work, we propose a new substructure model, called $(k,t)$-hypercore, based on the assumption that high-order relations remain as long as at least $t$ fraction of the members remains.
Specifically, it is defined as the maximal subhypergraph where (1) every node \red{is contained in at least $k$ hyperedges} in it and (2) at least $t$ fraction of the nodes remain in every hyperedge.
We first prove that, given $t$ (or $k$), finding the $(k,t)$-hypercore for every possible $k$ (or $t$) can be computed in time linear w.r.t the sum of the sizes of hyperedges.
Then, we demonstrate that real-world hypergraphs from the same domain share similar $(k,t)$-hypercore structures, which capture different perspectives depending on $t$.
Lastly, we show the successful applications of our model in identifying influential nodes, dense substructures, and vulnerability in hypergraphs.}
\end{abstract}



\section{Introduction}\label{sec:intro}
Graphs are a powerful model for representing pairwise relations, and they have been used for 
recommendation systems \citep{silva2010graph, debnath2008feature},
information retrieval \citep{blanco2012graph, mihalcea2011graph},
knowledge representation \citep{chein2008graph}, and many more.
However, graphs are limited to pairwise relations and thus fail to precisely describe high-order (i.e., group-wise) relations among more than two nodes.

Hypergraphs, where each hyperedge consists of an arbitrary number of nodes, break the limitation by describing high-order relations precisely \citep{benson2018simplicial, yin2017local} and contain graphs as special cases. 
Hypergraphs have been successful in modeling real-life processes in diverse fields, including chemical reactions \citep{konstantinova2001application}, epidemic spread \citep{bodo2016sis}, and blockchain economy~\citep{qu2018hypergraph}.

For a given pairwise graph, the \textit{$k$-core} \citep{seidman1983network} is a cohesive substructure that is defined as the maximal subgraph where each node has degree at least $k$ \red{(i.e., each node is incident to at least $k$ edges)} within it.
Extensive research has been conducted to show its linear-time computability \citep{batagelj2003m} and successful applications to $k$-cores, including graph visualization \citep{alvarez2006large}, community detection \citep{corominas2014detection}, anomaly detection \citep{shin2016corescope}, and biological process modeling \citep{luo2009core}.

\begin{figure}[t!]
	\centering
	\begin{subfigure}[b]{0.9\textwidth}
		\centering
		\includegraphics[scale=0.375]{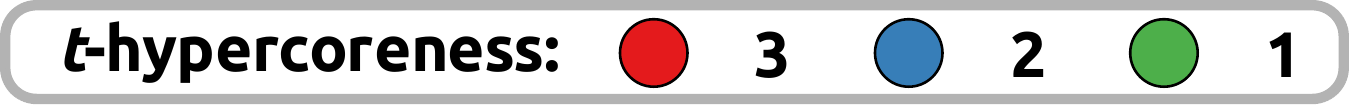}\\
	\end{subfigure}    
	\begin{subfigure}[b]{0.2\textwidth}
		\centering
		\includegraphics[width=\textwidth]{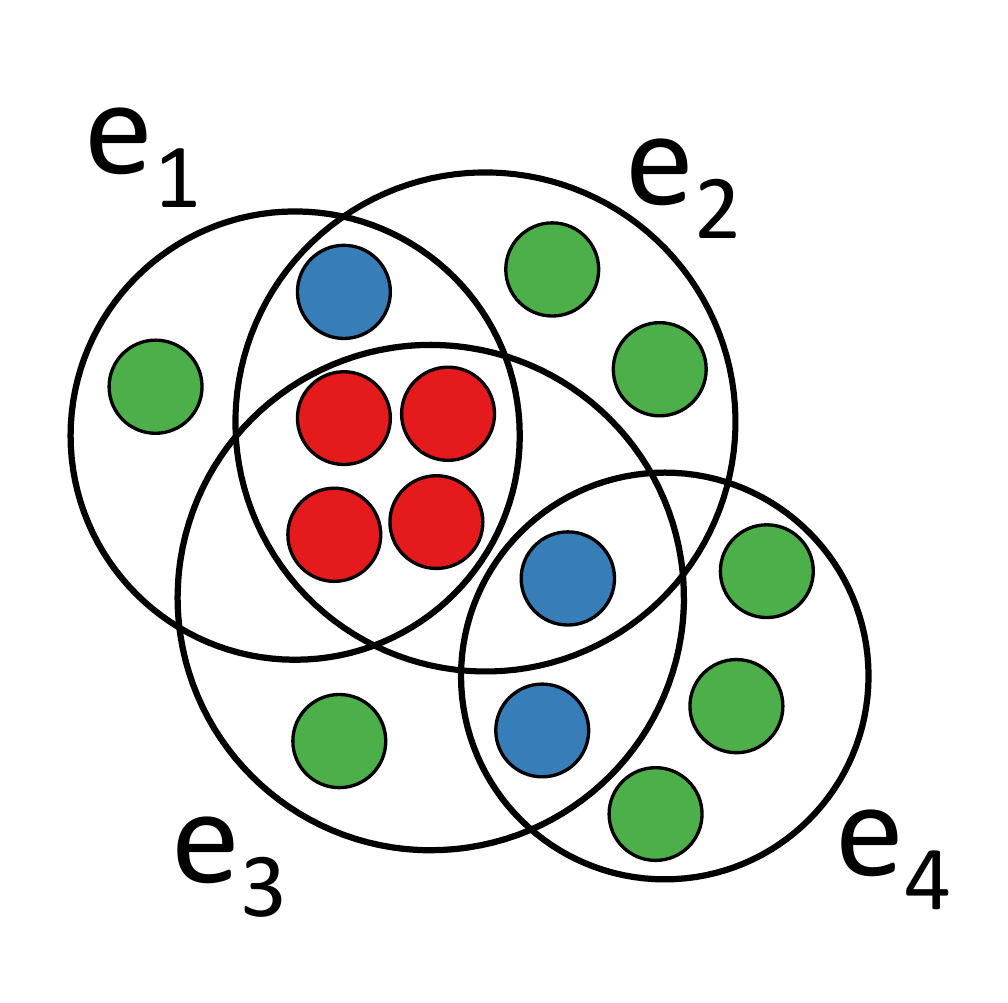}
		\caption{$t \leq \frac{2}{5}$}
		\label{fig:example_core_leq2/5}
	\end{subfigure}
	\hfill
	\begin{subfigure}[b]{0.2\textwidth}
		\centering
		\includegraphics[width=\textwidth]{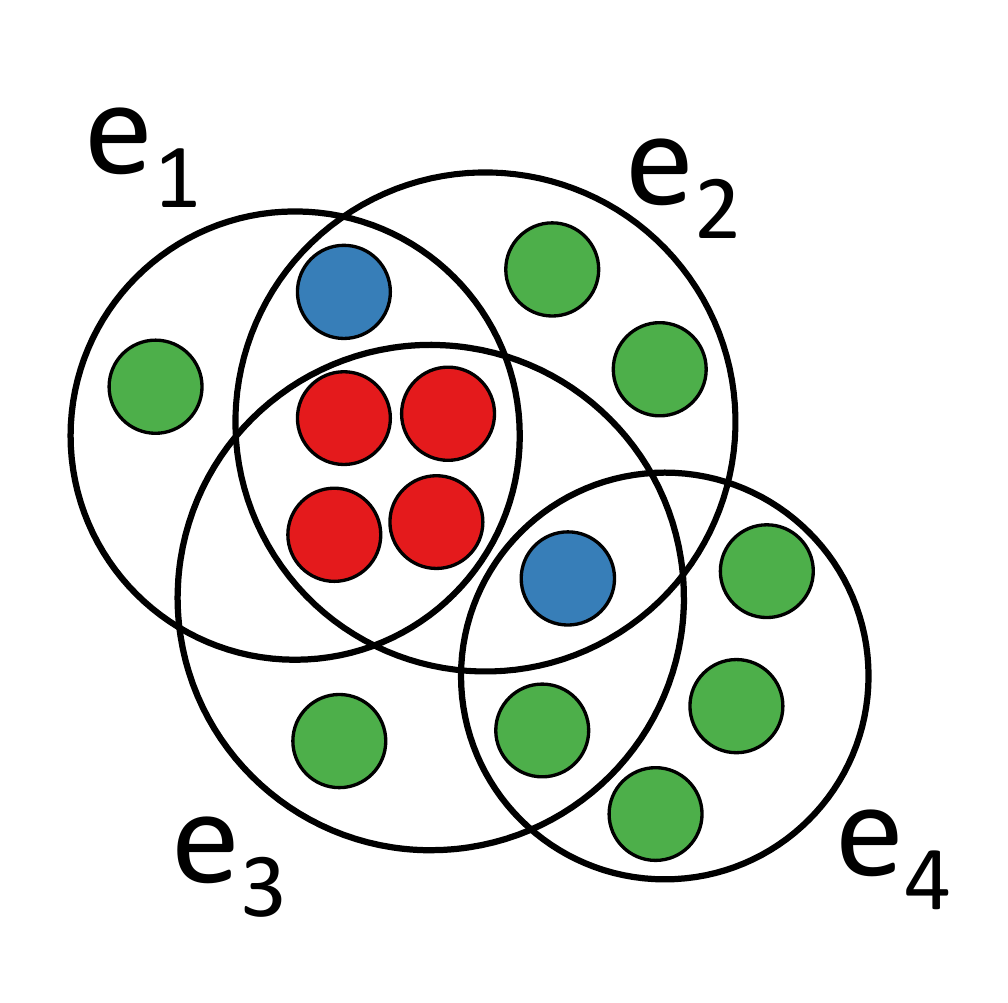}
		\caption{$\frac{2}{5} < t \leq \frac{4}{7}$}
		\label{fig:example_core_2/5-4/7}
	\end{subfigure}
	\hfill
	\begin{subfigure}[b]{0.2\textwidth}
		\centering
		\includegraphics[width=\textwidth]{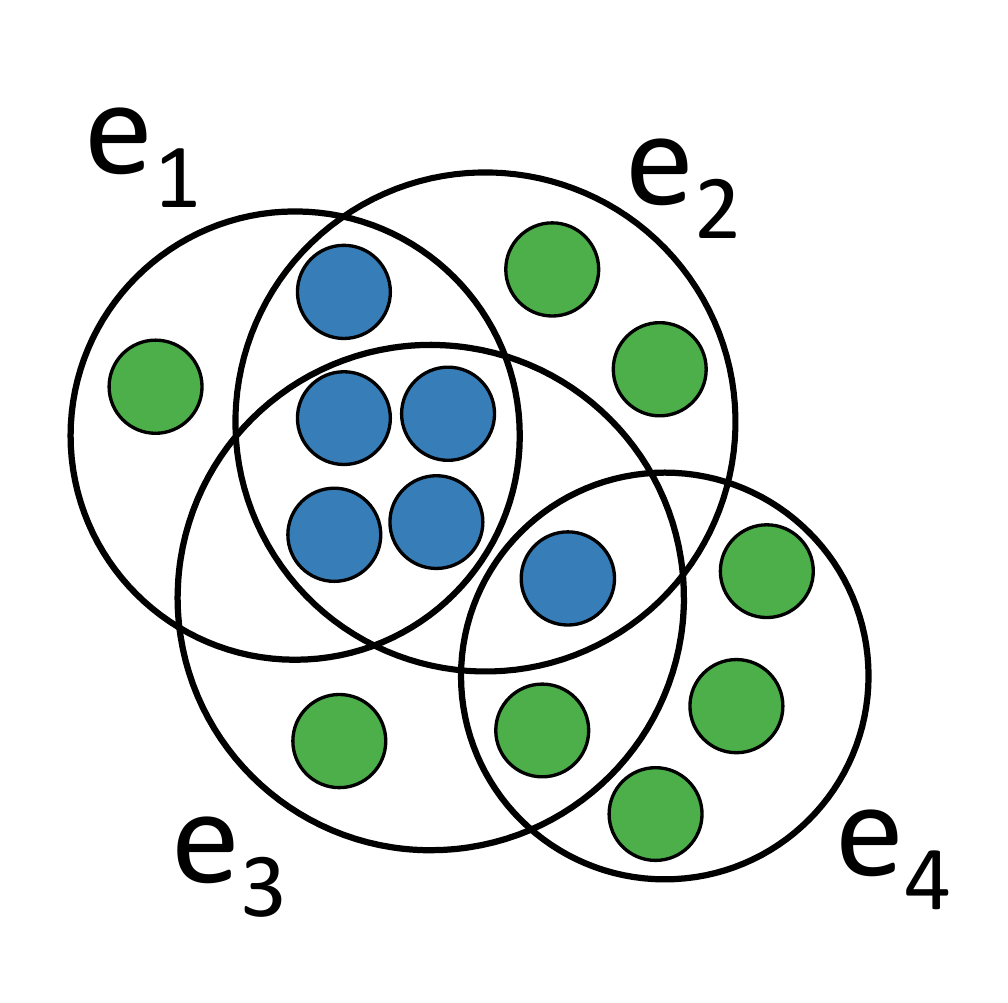}
		\caption{$\frac{4}{7} < t \leq \frac{5}{7}$}
		\label{fig:example_core_4/7-5/7}
	\end{subfigure}
	\hfill
	\begin{subfigure}[b]{0.2\textwidth}
		\centering
		\includegraphics[width=\textwidth]{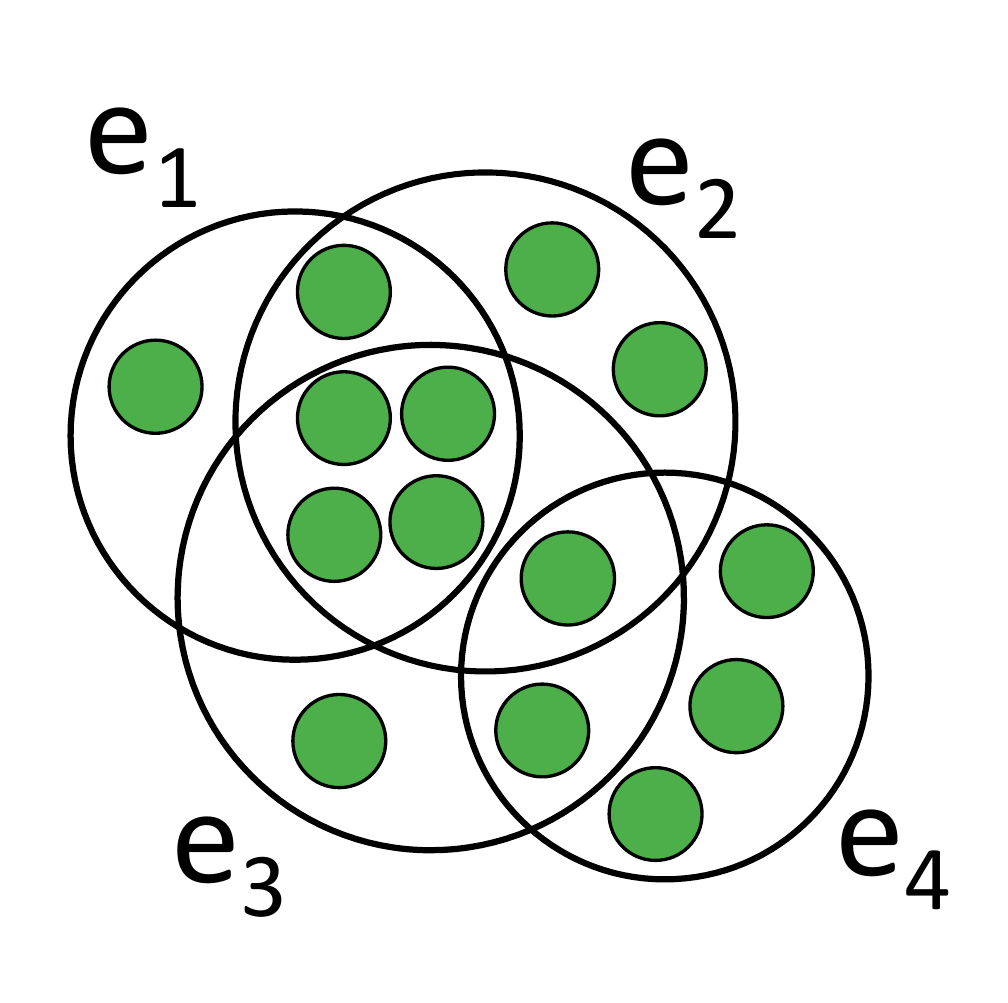}
		\caption{$t > \frac{5}{7}$}
		\label{fig:example_core_geq5/7}
	\end{subfigure} 
	\caption{\label{fig:exmaple_fragile_core}
		\bolden{An example of ($k, t$)-hypercores.}
		Assuming more robust hyperedges (i.e., decreasing the hyperedge-fraction threshold $t$) reveals cohesive substructures that are overlooked when fragile hyperedges are assumed.
		Notably, when fragile hyperedges are assumed (i.e., when $t=1$), every node has the same $t$-hypercoreness, as shown in (d).}
\end{figure}

There have been attempts to generalize the notion of $k$-cores to hypergraphs \citep{hua2023revisiting, luo2021hypercore, luo2022hypercore, gabert2021shared, gabert2021unifying, sun2020fully}, and the generalized notations, called \textit{hypercores}, commonly assume that hyperedges are \textit{fragile}.
That is, a hyperedge (i.e., a group relation) becomes obsolete as soon as \textit{any} constituent node opts out of it. Specifically, an entire hyperedge is ignored as soon as any node in it is removed during hypercore computation.
However, such an assumption is unrealistic and potentially leads to much information loss.
For example, an online group chat may remain active even if someone leaves it; and a recipe (i.e., a group of ingredients) may still produce a delicious result even if some ingredients are unavailable.
As another example, the hypergraph shown in Fig.~\ref{fig:exmaple_fragile_core}(d) cannot be decomposed into $k$-hypercores with different $k$, although the cohesiveness of subhypergraphs varies, since fragile hyperedges are assumed.

In order to better reveal the structural information in hypergraphs, we propose the notion of $(k, t)$-hypercores.
In addition to the node-degree threshold $k$, we introduce the \textit{hyperedge-fraction threshold} $t$ that determines how many constituent nodes suffice to maintain a hyperedge.
Specifically, given a hypergraph and thresholds $k$ and $t$, the \textit{$(k, t)$-hypercore} is defined as the maximal subhypergraph where (1) every node \red{is contained in at least $k$ hyperedges} in it and (2) at least $t$ fraction of the constituent nodes \red{(i.e., the nodes constituting the original hyperedge)} remain in every remaining hyperedge.
The larger the value of $t$ is, the more fragile the hyperedges are.
Based on the concept, we define the \textit{$t$-hypercoreness} of a node as the maximum $k$ {such} that the node is in the $(k, t)$-hypercore, and \textit{the $k$-fraction} of a node as the maximum $t$ {such} that the node is in the $(k, t)$-hypercore.
In Fig.~\ref{fig:exmaple_fragile_core}, we show an example where the $(k, t)$-hypercore structures change with $t$.
Notably, some other variants of hypercores have been considered.
The concept of $(k; \ell)$-hypercores has been considered by~\cite{limnios2021hcore},
where the $(k; \ell)$-hypercore requires that at least $\ell$ constituent nodes (instead of $t$ fraction required in the $(k, t)$-hypercore) remain in every remaining hyperedge.
The concept of neighbor-$k$-hypercores has been considered by~\cite{arafat2023neighborhood}, which 
focuses on the number of neighbors (i.e., nodes coexisting in at least one hyperedge) of each node,
and the concept is further extended to (neighbor, degree)-$(k, d)$-hypercores.
Compared to the existing concepts, our proposed concepts provide unique information on hypergraphs, as theoretically proven and empirically demonstrated.
\color{black}

We first show that the proposed concepts are well-defined
and have containment properties w.r.t both $k$ and $t$, then propose peeling-like computation algorithms for computing all the proposed concepts and show their correctness and time complexity.
In particular, we show that both the $t$-hypercoreness for {given $k$} and the $k$-fraction for {given $t$} of each node can be computed in time proportional to the sum of the sizes of hyperedges.

In order to demonstrate the usefulness of the proposed concepts, we investigate the $(k, t)$-hypercore structures of fourteen real-world hypergraphs in six different domains \citep{sinha2015overview, mastrandrea2015contact, leskovec2007graph} while varying $t$.
The examination leads to the following observations from different perspectives:
(1) \textit{domain-based patterns of $(k, t)$-hypercore sizes}: hypergraphs in the same domain show similar patterns of the $(k, t)$-hypercore sizes with different $k$ and $t$ values;
(2) \textit{heavy-tailed distributions of $t$-hypercoreness}: in most investigated real-world hypergraphs, the $t$-hypercoreness of nodes consistently follows heavy-tailed distributions regardless of $t$;
(3) \textit{heterogeneity of $t$-hypercoreness}: in the same real-world hypergraph, the $t$-hypercoreness with different $t$ provides statistically and information-theoretically distinct information.

We also utilize some properties of the proposed concepts in three applications:
(1) \textit{influential-node identification}: we generalize the SIR model in hypergraphs and use the model to show that $t$-hypercoreness is a reliable indicator to node-influence; 
(2) \textit{dense substructure discovery}: we show that $(k, t)$-hypercores generally have much higher density than the whole hypergraph and consider a generalized vertex cover problem to demonstrate that $t$-hypercoreness can be used to find dense substructures;
(3) \textit{vulnerability detection}: we generalize the core minimization problem to detect vulnerabilities in hypergraphs by finding the nodes whose removal reduces the size of the $(k,t)$-hypercore (for given $k$ and $t$) most, and to this end, we propose an efficient and effective algorithm.

In short, our contributions are three-fold:
\begin{itemize}
	\item \textbf{New concepts}.
	We propose the $(k, t)$-hypercore, a new substructure model for hypergraphs, together with $t$-hypercoreness and $k$-fraction (Defs.~\ref{def:kt_hypercore} to \ref{def:k_fraction}).
	In addition to the node-degree threshold $k$, the proposed concepts incorporate the hyperedge-fraction threshold $t$ to provide more comprehensive information. 
	\item \textbf{Properties and algorithms}.
	We show some theoretical properties of the proposed concepts, and computation algorithms (Algs.~\ref{alg:kt_hypercore_decomp} to \ref{alg:k_fraction}) for the proposed concepts with analyses of the correctness and time complexity (Thms.~\ref{thm:alg_kt} to \ref{thm:alg_k}).
	\item \textbf{Observations and applications}. 
	We investigate $14$ real-world hypergraphs, which leads to interesting observations (Sec.~\ref{sec:obss}), including a surprising similarity in the $(k, t)$-hypercores of hypergraphs in the same domain.
	We also show successful applications (Sec.~\ref{sec:apps}) of the proposed concepts to influence estimation, dense-substructure detection, and vulnerability detection.
\end{itemize}
\smallsection{Reproducibility.} The code and datasets are available online~\citep{onlineSuppl}.\footnote{\url{https://github.com/bokveizen/non-fragile-hypercore}}

\section{Preliminaries}\label{sec:prel}
In this section, we provide the mathematical background and preliminaries that are used {throughout} this paper.

\smallsection{Hypergraphs.}
A \textit{hypergraph} $H = (V, E)$ consists of a node set $V$ and a hyperedge multiset $E$.\footnote{A multiset is a set allowing duplicate elements.}
Given a hypergraph $H = (V, E)$, we associate each hyperedge with a distinct positive integer in $\bbN$, i.e.,
$E = \setbr{e_i: i \in I_E}$, where $I_E$ is called the \textit{index set} of $E$.
The \textit{degree} $d(v; H)$ of a node $v$ is the number of hyperedges that contain $v$, i.e., $d(v; H) = \abs{\setbr{i \in I_E: v \in e_i}}$.
The set $N(v; H)$ of \textit{neighbors} of a node $v$ is the number of nodes coexisting with $v$ in at least one hyperedge, i.e., $N(v; H) = \setbr{u \in V: u \neq v, \exists e \in E \text{~s.t.~} u, v \in e}$.
The \textit{constituent nodes} of a hyperedge $e \in E$, is the nodes in $e$.
The \textit{size} of a hyperedge $e \in E$, denoted by $\abs{E}$, is the cardinality of $E$ (i.e., the number of constituent nodes of $e$).
The \textit{size} of $H$, denoted by $\abs{H}$, is the number of nodes in $H$, i.e., $\abs{H} = \abs{V}$.
The \textit{total size} of $H$, denoted by $TS(H)$, is the sum of the size of each hyperedge in $H$ (i.e., $TS(H) = \sum_{i \in I_E} \abs{e_i}$).
\color{black}
All hypergraphs in this paper are finite, undirected, and unweighted; and in them, each node has degree at least $1$, i.e., $d(v; H) \geq 1, \forall v \in V$, and each hyperedge is of cardinality at least two, i.e., $\abs{e_i} \geq 2, \forall i \in I_E$.
If in a hypergraph $H = (V, E)$, each hyperedge is of cardinality exactly two, i.e., $\abs{e_i} = 2, \forall i \in I_E$, then $H$ is also called a (pairwise) graph.
\color{black}

\begin{definition}[Subhypergraph]
	A hypergraph $H' = (V', E')$ is a \textbf{subhypergraph} of $H = (V, E)$ if each hyperedge in $H'$ is a subset of the hyperedge with the same index in $H$, i.e., \red{$e_i' \subseteq e_i, \forall i \in I_{E'} \subseteq I_E$.}
	If {$e_i' = e_i, \forall i \in I_{E'}$,} we call $H'$ a \textbf{complete subhypergraph} of $H$.
\end{definition}
Note that a subhypergraph should be a hypergraph, and thus each hyperedge in a subhypergraph should also be of cardinality at least two.
\color{black}

We summarize the notations in Tbl.~\ref{tab:notations}.
{In the notations,} the input hypergraph $H$ may be omitted when the context is clear.

\begin{table}[t!]
	\begin{center}
		\caption{Notations.}\label{tab:notations}            
		\begin{tabular}{ll}
			\toprule
			\textbf{Notation} & \textbf{Definition}\\
			\midrule
			$H = (V, E)$    & a hypergraph with nodes $V$ and hyperedges $E$ \\
			$d(v; H)$       & the degree of $v$ in $H$\\
			$I_E$           & the index set of $E$\\
			$k, t$          & the degree and hyperedge-fraction thresholds \\
			$C_{k, t}(H)$   & the $(k, t)$-hypercore of $H$ \\
			$c_t(v; H), c_t^*(H)$ & the $t$-hypercoreness of $v$ in $H$, and that of $H$ \\
			$f_k(v; H), f_k^*(H)$     & the $k$-fraction of $v$ in $H$, and that of $H$ \\
			\bottomrule
		\end{tabular}
	\end{center}
\end{table}

\smallsection{Hypercores.}
In pairwise graphs, the concept of $k$-cores \citep{seidman1983network} {is widely used}.
Given a pairwise graph $G$ and $k \in \bbN$, the \textit{$k$-core} of $G$ is the maximal subgraph where each node has degree at least $k$ within it.
\begin{definition}[$k$-core]\label{def:k_core}
    Given a pairwise graph $G = (V, E)$ and $k \in \bbN$, the \bolden{$k$-core} of $H$, denoted by $C_k(G) = (V', E')$, is the maximal subgraph of $G$ where each node has degree at least $k$ (i.e., is incident to at least $k$ edges) within $C_k$. \footnote{In this work, the maximal subgraph (subhypergraph) satisfying some conditions means that every other graph (hypergraph) satisfying such conditions is a subgraph (subhypergraph) of the maximal one.}
\end{definition}
\color{black}
It is naturally generalized to hypergraphs~\citep{hua2023revisiting, luo2021hypercore, luo2022hypercore, gabert2021shared, gabert2021unifying, sun2020fully}, as {follows}.
\begin{definition}[$k$-hypercore]\label{def:k_hypercore}
	Given a hypergraph $H = (V, E)$ and $k \in \bbN$, the \bolden{$k$-hypercore} of $H$, denoted by $C_k(H) = (V', E')$, is the maximal \textit{complete subhypergraph} of $H$ where each node has degree at least $k$ \red{(i.e., is contained in at least $k$ hyperedges)} within $C_k$.
\end{definition}

Some variants of hypercores have been considered.
See Sec.~\ref{sec:ccpt} for some related discussions.
\color{black}

\smallsection{Clique expansion.}
One of the most common ways to convert hypergraphs into pairwise graphs is the clique expansion, where each hyperedge $e \in E$ is converted to a clique consisting of the nodes in $e$. 
Given a hypergraph $H = (V, E)$, {its} unweighted clique expansion is $G_{\mathrm{uc}}(H) = (V, \mathcal{E})$, and {its} weighted clique expansion is $G_{\mathrm{wc}}(H) = (V, \mathcal{E}, \omega)$, where the edge set $\mathcal{E} = \setbr{(u,v) \in \binom{V}{2}: \exists e \in E~s.t.~\setbr{u, v} \subseteq e}$, and the weight function
$\omega\left((u,v)\right) = \abs{\setbr{i \in I_E: \setbr{u, v} \subseteq e_i}}$.
Clique expansion provides an approach to make the hypergraphs easier to analyze, but the information on the higher-order interactions is lost,
which is natural since for a set of nodes $V$, there are $O(\abs{V}^2)$ possible pairs in $V$, while there are $O(2^{\abs{V}})$ possible subsets.
Two hypergraphs with obviously different structures may have the same clique expansions.

\smallsection{Star expansion.}
Each hypergraph can be represented as a bipartite graph, which is called its star expansion~\citep{zien1999multilevel}.
The star expansion of a hypergraph $H$ is the bipartite graph whose node set is the union of $V$ and $E$ and whose edge set consists of the incidence relations in $H$.
\begin{definition}[Star expansion]\label{def:bigraph_repr}
	Given a hypergraph $H = (V, E)$, its \textbf{star expansion} (i.e., bipartite-graph representation) is $G_{se}(H) = (V \cup E, E_{se}(H))$, where $E_{se}(H) = \setbr{(v, e): v \in V, e \in E, v \in e} \subseteq V \times E$.\footnote{Similar to clique expansion, we can also have weighted star expansion, which is, however, not used in this work.}
\end{definition}
As~\cite{yang2022semi} pointed out, although a star expansion contains all the incidence information in hypergraphs, the remaining heterogeneous structure has no explicit edges between nodes and is unsuitable for many well-studied graph algorithms designed for simple homogeneous graphs.
\color{black}

\section{Concepts}\label{sec:ccpt}
In this section, we introduce the proposed concepts and show some theoretical properties of them.
Moreover, we discuss the connections and differences between the proposed concepts and some existing related concepts.
\color{black}

\subsection{Proposed concepts}
In pairwise graphs, each edge represents a connection between two nodes, and thus the removal of either node naturally results in the complete nullification of the edge.
In contrast, a hyperedge with three or more nodes still represents the interactions among the remaining nodes even when some constituent nodes are removed.
As we have discussed and shown in Fig.~\ref{fig:exmaple_fragile_core}, the straightforward generalization in Def.~\ref{def:k_hypercore}
groundlessly assumes fragile hyperedges and suffers from information loss.
We seek to better reveal the structure of hypergraphs by considering \textit{non-fragile hyperedges}.

Therefore, we introduce the hyperedge-fraction threshold $t$ that determines the minimum proportion of constituent nodes required to maintain a hyperedge, which leads to Def.~\ref{def:kt_hypercore}.

\begin{definition}[$(k, t)$-hypercore]\label{def:kt_hypercore}
	Given $H = (V, E)$, $k \in \bbN$, and $t \in [0, 1]$, the \bolden{$(k, t)$-hypercore} of $H$, denoted by $C_{k, t}(H) = (V', E')$, is the maximal (in terms of total size) subhypergraph of $H$ where 
	(1) every node in has degree at least $k$ (i.e., is contained in at least $k$ hyperedges) within $C_{k, t}$ and 
	(2) at least $t$ proportion of the constituent nodes remain in every hyperedge of $C_{k, t}(H)$.
	Formally, $d(v; C_{k,t}(H)) \geq k, \forall v \in V'$ and $\abs{e_i' \cap e_i} \geq t\abs{e_i}, \forall i \in I_{E'} \subseteq I_E$.
\end{definition}

Note that the definition of $(k, t)$-hypercore requires that at least two nodes remain in each hyperedge because of the definition of subhypergraphs (see Sec.~\ref{sec:prel}).
See also Line~\ref{line:rem_edges_if} in Alg.~\ref{alg:kt_hypercore_decomp}.
\color{black}

\begin{definition}[$t$-hypercoreness]\label{def:t_hcoreness}
	Given $H = (V, E)$ and $t \in [0, 1]$, the \bolden{$t$-hypercoreness} of $v \in V$, denoted by $c_t(v; H)$, is the maximum positive integer such that $v$ is in the $(c_t(v), t)$-hypercore, i.e., 
	$c_t(v) = \max\setbr{k \in \bbN: v \in V(C_{k, t})}$.
	We call $c_t^*(H) \coloneqq \max \setbr{c_t(v): v \in V}$ the $t$-hypercoreness of $H$.
\end{definition}
\begin{definition}[$k$-fraction]\label{def:k_fraction}
	Given $H = (V, E)$ and $k \in \bbN$, the \bolden{$k$-fraction} of $v \in V$, denoted by $f_k(v; H)$, is the maximum real number in $[0, 1]$ such that $v$ is in the $(k, f_k(v))$-hypercore, i.e.,
	$f_k(v) = \max \setbr{t \in [0, 1]: v \in V(C_{k, t})}$.
	\red{For the completeness of definition, if $\setbr{t \in [0, 1]: v \in V(C_{k, t})} = \varnothing$, we let $f_k(v; H) = -1$.}
	We call $f_k^*(H) \coloneqq \max \setbr{f_k(v): v \in V}$ the $k$-fraction of $H$.
\end{definition}

Note that the proposed concepts are extendable to weighted hypergraphs. Specifically, as long as we have rigorous definitions of node degrees and hyperedge fractions on weighted hypergraphs, the extensions are straightforward.
\color{black}

\smallsection{Example.}
In Fig.~\ref{fig:exmaple_fragile_core}, the $t$-hypercoreness of each node changes when the $t$ changes.
Specifically, four nodes have $t$-hypercoreness $3$ when $t \leq \frac{4}{7}$.
They have $t$-hypercoreness $2$ when $\frac{4}{7} < t \leq \frac{5}{7}$, and have $t$-hypercoreness $1$ when $t > \frac{5}{7}$, which means that their $3$-fraction is $\frac{4}{7}$ and $2$-fraction is $\frac{5}{7}$.

The following propositions show that the $(k, t)$-hypercores are well-defined 
and have two-way containment properties.
\begin{proposition}[Existence and uniqueness]\label{prop:uniqueness}
	Given any hypergraph $H$, $k \in \bbN$, and $t \in [0, 1]$, $C_{k, t}$ uniquely exists \red{and is possibly empty}.
\end{proposition}
\begin{proof}
	See Appendix~\ref{subsec:pf_prop_uniqueness}.
\end{proof}

\begin{proposition}[Two-way containment] \label{prop:containment}
	Let $H$ be any hypergraph.
	Fix any $k \in \bbN$, for any $0 \leq t_1 < t_2 \leq 1$, $C_{k, t_2}(H)$ is a subhypergraph of $C_{k, t_1}(H)$.
	Similarly, fix any $t \in [0, 1]$, for any $k_1 < k_2 \in \bbN$, $C_{k_2, t}(H)$ is a subhypergraph of $C_{k_1, t}(H)$.
\end{proposition}
\begin{proof}
	See Appendix~\ref{subsec:pf_prop_containment}.
\end{proof}

\subsection{Related concepts}
Below, we discuss some existing related concepts, especially the connections and differences between them and our proposed concepts.

\smallsection{Existing variants of hypercores.}
As mentioned in Sec.~\ref{sec:prel} (see Def.~\ref{def:k_hypercore}), most previous works \citep{hua2023revisiting, luo2021hypercore, luo2022hypercore, gabert2021shared, gabert2021unifying, sun2020fully} are based on the straightforward generalization of {$k$-cores to hypergraphs} assuming fragile hyperedges (i.e., a hyperedge is removed when \textit{any} node leaves it), which is equivalent to the $(k, t)$-hypercore with $t = 1$ (i.e., a special case of $(k, t)$-hypercore).
\cite{limnios2021hcore} defined the $(k; \ell)$-hypercore of a given hypergraph $H$ as the maximal subhypergraph of $H$ where  each node has degree at least $k$ \red{(i.e., is contained in at least $k$ hyperedges)} within the subhypergraph and each hyperedge contains at least $\ell$ nodes.
Based on the concept, we can define $\ell$-hypercoreness.

\begin{definition}[$(k; \ell)$-hypercore]\label{def:kl_hcore}
	Given a hypergraph $H = (V, E)$ and $k, \ell \in \bbN$, the \bolden{$(k; \ell)$-hypercore} of $H$, denoted by $\tilde{C}_{k; \ell}(H)$, is the maximal subhypergraph of $H$ such that each node has degree at least $k$ (i.e., is contained in at least $k$ hyperedges) in within $\tilde{C}_{k; \ell}(H)$ and each hyperedge contains at least $\ell$ nodes.
\end{definition}

\begin{definition}[$\ell$-hypercoreness]\label{def:l_hypercoreness}
	Given $H = (V, E)$ and $\ell \in \bbN$, the \bolden{$\ell$-hypercoreness} of $v \in V$, denoted by $\tilde{c}_\ell(v)$, is the maximum positive integer such that $v$ is in the $(\tilde{c}_\ell(v); \ell)$-hypercore, i.e., $\tilde{c}_\ell(v) = \max\setbr{k \in \bbN: v \in V(\tilde{C}_{k; \ell})}$.	
\end{definition}

The concept of $(k; \ell)$-hypercores is equivalent to a $k$-core-like concept on bipartite graphs called $(\alpha, \beta)$-cores~\citep{liu2020efficient, sariyuce2018peeling}.
\begin{definition}[$(\alpha; \beta)$-core]\label{def:alpha_beta_core}
	Given a bipartite graph $G_{B} = (V_1 \cup V_2, E)$ and $\alpha, \beta \in \bbN$,
	the \bolden{$(\alpha; \beta)$-core} of $G_B$,
	denoted by $\hat{C}_{\alpha; \beta}(G_B) = (V'_1 \cup V'_2, E')$ where $V'_1 \subseteq V_1$ and $V'_2 \subseteq V_2$,
	is the maximal subgraph of $G_B$ such that
	each node in $V'_1$ has degree at least $\alpha$ within $\hat{C}_{\alpha; \beta}(G_B)$,
	and each node in $V'_2$ has degree at least $\beta$ within $\hat{C}_{\alpha; \beta}(G_B)$.
\end{definition}
\begin{lemma}\label{lem:kl_hypercore_equiv_bipartite_core}
	Given $H$, $k$, and $\ell$, the $(k; \ell)$-hypercore of $H$ is equivalent to the $(\alpha = k, \beta = \ell)$-core of $G_{bp}(H)$, the star expansion of $H$ (see Sec.~\ref{sec:prel}).
\end{lemma}
\begin{proof}
	See Appendix~\ref{subsec:pf_lem_kl_hypercore_equiv_bipartite_core}.	
\end{proof}

Notable, only the special case with $\ell = 2$ was actually used by~\cite{limnios2021hcore}, and such a special case (i.e., $(k; \ell=2)$-hypercore) was also previously considered by~\cite{vogiatzis2013influence}.
Also, $(k; \ell=2)$-hypercore is equivalent to the proposed $(k, t)$-hypercore with $t = 0$.
\begin{lemma}\label{lem:kl2_hypercore_equiv_kt0_hypercore}
	Given a hypergraph $H = (V, E)$ and $k \in \bbN$, $\tilde{C}_{k; \ell = 2}(H) = C_{k; t = 0}(H)$.
\end{lemma}
\begin{proof}
	See Appendix~\ref{subsec:pf_lem_kl2_hypercore_equiv_kt0_hypercore}.
\end{proof}

Essential differences exist between the concept of $(k; \ell)$-hypercores and the concept of $(k, t)$-hypercores proposed by us.
In Appendix~\ref{subsec:pf_lem_diff_kt_kl_hcore}, we theoretically analyze the limitations of the $(k; \ell)$-hypercores and the superiority of the proposed $(k, t)$-hypercores with empirical comparisons.
For example, Lem.~\ref{lem:diff_kt_kl_hcore} below tells us that the proposed concept of $(k, t)$-hypercores can provide unique information of a hypergraph, which is not contained in the existing concept of $(k; \ell)$-hypercores for any $\ell$.
\begin{lemma}\label{lem:diff_kt_kl_hcore}
	There exist $H, k, t$ such that $C_{k, t}(H) \neq \tilde{C}_{k; \ell}(H)$ for any $\ell$.
\end{lemma}
\begin{proof}
	See Appendix~\ref{subsec:pf_lem_diff_kt_kl_hcore}.	
\end{proof}

Recently, \cite{arafat2023neighborhood} proposed a variant of hypercores, where 
for each node, the number of neighbors (i.e., nodes coexisting in at least one hyperedge) of this node
(instead of the degree of this node) is considered,
which leads to the concept of neighbor-$k$-hypercores.
Based on the concept, we can define neighbor-hypercoreness.

\begin{definition}[neighbor-$k$-hypercores]\label{def:nbr_k_hypercore}
	Given a hypergraph $H = (V, E)$, and $k \in \bbN$, the \bolden{neighbor-$k$-hypercore} of $H$, denoted by $C^{nbr}_{k}(H)$, is the maximal complete subhypergraph of $H$ such that each node in $C^{nbr}_{k}(H)$ has at least $k$ neighbors (i.e., $\abs{N(v; C^{nbr}_{k}(H))} \geq k, \forall v \in C^{nbr}_{k}(H)$).
\end{definition}

\begin{definition}[neighbor-hypercoreness]\label{def:nbr_hypercoreness}
	Given $H = (V, E)$, the \bolden{neighbor-hypercoreness} of $v \in V$, denoted by $c^{nbr}(v)$, is the maximum positive integer such that $v$ is in the neighbor-$c^{nbr}(v)$-hypercore, i.e., $c^{nbr}(v) = \max\setbr{k \in \bbN: v \in V(C^{nbr}_k)}$.
\end{definition}

\cite{arafat2023neighborhood} further extended the concept of neighbor-$k$-hypercores by incorporating the information of the degree of each node, which leads to the concept of
(neighbor, degree)-$(k, d)$-hypercores.

\begin{definition}[(neighbor, degree)-$(k, d)$-hypercores]\label{def:kd_hypercore}
	Given a hypergraph $H = (V, E)$, and $k, d \in \bbN$, the \bolden{(neighbor, degree)-$(k, d)$-hypercore} of $H$, denote by $C^{nd}_{k, d}(H)$, is the maximal complete subhypergraph of $H$ such that each node in $C^{nd}_{k, d}(H)$ has at least $k$ neighbors and has degree at least $d$ (i.e., $\abs{N(v; C^{nd}_{k, d}(H))} \geq k \land d(v; C^{nd}_{k, d}(H)) \geq d, \forall v \in C^{nd}_{k, d}(H)$).
\end{definition}

Since two parameters are involved, we can have multiple ways to define the hypercoreness w.r.t (neighbor, degree)-$(k, d)$-hypercores, and an intuitive and straightforward way is as follows.
\begin{definition}[neighbor-degree-hypercoreness]\label{def:nbr_deg_hypercoreness}
	Given $H = (V, E)$, the \textbf{neighbor-degree-hypercoreness} of $v \in V$, denoted by $c^{nd}(v)$, is the maximum positive integer such that $v$ is in the neighbor-degree-$(c^{nd}(v), c^{nd}(v))$-hypercore, i.e., $c^{nd}(v) = \max\setbr{k \in \bbN: v \in V(C^{nd}_{k, k})}$.
\end{definition}

Notably, the above two concepts consider only complete subhypergraphs, i.e., they still assume fragile hyperedges.

\color{black}

\smallsection{Simplicial complexes.}
Another way to take the subsets of hyperedges into consideration is to use simplicial complexes~\citep{torres2021and}. 
For example, \cite{preti2021strud} considered the computation of $k$-trusses in simplicial complexes.
Similar to clique expansion, converting hypergraphs into simplicial complexes also brings information loss.
Our work shows that considering the subsets of relations is meaningful also when the data is modeled as hypergraphs;
when the data is modeled as hypergraphs, considering the subsets of the relations is also meaningful; 
and we provide a way to do so.

\begin{algorithm}[t!]
	\small
	\caption{{($k,t$)-Hypercore}}\label{alg:kt_hypercore_decomp}
	\hspace*{\algorithmicindent} \textbf{Input:} {$H = (V, E)$, $k$, $t$, and original hyperedge sizes $\mathcal{D}$} \\
	\hspace*{\algorithmicindent} \textbf{Output:} {$C_{k, t}(H)$: the $(k, t)$-hypercore of $H$}
	\begin{algorithmic}[1]
		\State $\mathcal{R} \leftarrow \setbr{v \in V: d(v; H) < k}$ \Comment{\blue{Nodes to remove}} \label{line:deg_check}	
		\While{$\mathcal{R} \neq \varnothing$}
		\State $\mathcal{R}' \leftarrow \varnothing$ \label{line:kt_while_start} \Comment{\blue{Nodes to remove in next round}}
		\For{\upshape\textbf{each} $e_i \in E$ \upshape\textbf{s.t.} $e_i \cap \mathcal{R} \neq \varnothing$}
		\State $e_i \leftarrow e_i \setminus \mathcal{R}$  \Comment{\blue{Remove nodes}} \label{line:rem_nodes}
		\If {$\abs{e_i} < t\mathcal{D}(i)$ \upshape\textbf{or} $\abs{e_i} < 2$} \label{line:rem_edges_if}
		\State $\mathcal{R}' \leftarrow \mathcal{R}' \cup \{v \in e_i : d(v; H) = k\}$ \label{line:add_node_to_remove}
		\State $E \leftarrow E \setminus \setbr{e_i}$ \label{line:rem_edges} \Comment{\blue{Remove hyperedge}}
		\EndIf			
		\EndFor
		\State $V \leftarrow V \setminus \mathcal{R}$
		\State $\mathcal{R} \leftarrow \mathcal{R}'$ \label{line:kt_while_end}		
		\EndWhile
		\State \Return{$H$}
	\end{algorithmic}
\end{algorithm}

\section{Computation Algorithms}\label{sec:algs}

In this section, we provide the computation algorithms of the proposed concepts: $(k, t)$-hypercore, $t$-hypercoreness, and $k$-fraction. We also show their correctness and time complexity.

\begin{algorithm}[t!]
	\small
	\caption{$t$-Hypercoreness}\label{alg:t_hypercoreness}
	\hspace*{\algorithmicindent} \textbf{Input:} {$H = (V, E)$ and $t$} \\
	\hspace*{\algorithmicindent} \textbf{Output:} {$t$-hypercoreness $c_t(v)$ for each node $v \in V$}
	\begin{algorithmic}[1]
		\State $\mathcal{R} \leftarrow \varnothing$
		\While{$H \neq \varnothing$} \label{line:until_empty_t}
		\If{$\mathcal{R} = \varnothing$}
		\State $k \leftarrow \min_{v \in V} d(v; H) + 1$ \label{line:increase_k}
		\State $\mathcal{R} \leftarrow \setbr{v \in V: d(v; H) = k - 1}$ \label{line:find_V_rem_t}		
		\Else
		\State $c_t(v) \leftarrow k - 1, \forall v \in \mathcal{R}$ \label{line:assign_t}
		\State $\mathcal{R}' \leftarrow \varnothing$ \Comment{\blue{Nodes to remove in next round}}
		\For{\upshape\textbf{each} $e_i \in E$ \upshape\textbf{s.t.} $e_i \cap \mathcal{R} \neq \varnothing$}
		\State $e_i \leftarrow e_i \setminus \mathcal{R}$  \Comment{\blue{Remove nodes}} 
		\If{$\abs{e_i} < t\mathcal{D}(i)$ \upshape\textbf{or} $\abs{e_i} < 2$}
		\State $\mathcal{R}' \leftarrow \mathcal{R}' \cup \{v \in e_i : d(v; H) = k\}$
		\State $E \leftarrow E \setminus \setbr{e_i}$ \Comment{\blue{Remove hyperedge}} 				
		\EndIf
		\EndFor
		\State $V \leftarrow V \setminus \mathcal{R}$
		\State $\mathcal{R} \leftarrow \mathcal{R}'$
		\EndIf
		\EndWhile
		\State \Return{$c_t(v)$ for each $v \in V$}
	\end{algorithmic}
\end{algorithm}

\subsection{Computation of $(k, t)$-hypercore}
Alg.~\ref{alg:kt_hypercore_decomp} shows the process of {finding a ($k, t$)-hypercore},
where $\mathcal{D}$ maps the index of a hyperedge to the original size of the hyperedge (in the original hypergraph $H = (V, E)$, $\mathcal{D}(i) = \abs{e_i}, \forall i \in I_E$).
During the process, we remove each node with degree less than $k$ from all its incident hyperedges (Line~\ref{line:rem_nodes}) and delete each hyperedge with the number of remaining nodes below the threshold (Lines~\ref{line:rem_edges_if} to \ref{line:rem_edges}).
Notably, in the threshold for hyperedges (Line~\ref{line:rem_edges_if}), we also require the cardinality to be at least $2$ because of the definition of hypergraphs.
When the degree of a node decreases from $k$ to $k - 1$, it is added to the set of nodes to be removed in the next round (Line~\ref{line:add_node_to_remove}).
\begin{theorem}\label{thm:alg_kt}
	Given $H = (V, E)$, $k \in \bbN$, and $t \in [0, 1]$, Alg.~\ref{alg:kt_hypercore_decomp} returns $C_{k, t}(H)$ in $O(\abs{V} + \abs{E} + (1-t) \sum_{e \in E} \abs{e})$ time.\footnote{We assume that the input hypergraph is in the memory and thus do not count the complexity of loading the hypergraph, which is $O(\sum_{e \in E} \abs{e})$.}
\end{theorem}
\begin{proof}
	See Appendix~\ref{subsec:pf_thm_alg_kt}.
\end{proof}

\subsection{Computation of $t$-hypercoreness}
Alg.~\ref{alg:t_hypercoreness} describes the process of computing $t$-hypercoreness.
Essentially, by the containment property w.r.t $k$, we {repeatedly} find the $(k, t)$-hypercore, while increasing $k$ until the {remaining} hypergraph becomes empty; and thus Alg.~\ref{alg:t_hypercoreness} can also output the $(k,t)$-hypercores for the given $t$ and all possible $k$ with the same time complexity as shown in Thm.~\ref{thm:alg_t}.
\begin{theorem}\label{thm:alg_t}
	Given $H = (V, E)$ and $t \in [0, 1]$, Alg.~\ref{alg:t_hypercoreness} returns $c_t(v)$ for {all} $v \in V$ in $O(c_t^* \abs{V} + \abs{E} + (1 - t) \sum_{e \in E} \abs{e})$ time.
\end{theorem}
\begin{proof}
	See Appendix~\ref{subsec:pf_thm_alg_t}.
\end{proof}

\begin{algorithm}[t!]
	\small
	\caption{$k$-Fraction}\label{alg:k_fraction}	
	\hspace*{\algorithmicindent} \textbf{Input:} {$H = (V, E)$ and $k$} \\
	\hspace*{\algorithmicindent} \textbf{Output:} {$k$-fraction $f_k(v)$ for each $v \in V$}
	\begin{algorithmic}[1]
		\State $\mathcal{D}(i) \leftarrow \abs{e_i}, \forall i \in I_E$ \Comment{\blue{Record original sizes}} \label{line:rec_dim}
		\State $H' = (V', E') \leftarrow$ {Alg.~\ref{alg:kt_hypercore_decomp} with $H$, $k$, $0$ and $\mathcal{D}$}
		\State $t \leftarrow 0$
		\State $\mathcal{R} \leftarrow \varnothing$ \label{line:Ck0}
		\While{$H' \neq \varnothing$} \label{line:while_start_k}
		\If{$\mathcal{R} = \varnothing$}
		\State $t \leftarrow \min_{e_i' \in E'} \abs{e_i'} / \mathcal{D}(i)$ \label{line:increase_t}
		\For{\upshape\textbf{each} $e_i' \in E'$ \upshape\textbf{s.t.} $\abs{e_i'} = t\mathcal{D}(i)$}\label{line:th_check_k}
		\State $\mathcal{R} \leftarrow \mathcal{R} \cup \setbr{v \in e: d(v; H') = k}$
		\State $E' \leftarrow E' \setminus \setbr{e}$ 
		\EndFor \label{line:th_check_k_end}
		\Else
		\State $f_k(v) \leftarrow t, \forall v \in \mathcal{R}$ \label{line:assign_k}
		\State $\mathcal{R}' \leftarrow \varnothing$ \Comment{\blue{Nodes to remove in next round}}
		\For{\upshape\textbf{each} $e_i \in E$ \upshape\textbf{s.t.} $e_i \cap \mathcal{R} \neq \varnothing$}
		\State $e_i \leftarrow e_i \setminus \mathcal{R}$  \Comment{\blue{Remove nodes}} 
		\If{$\abs{e_i} \leq t\mathcal{D}(i)$ \upshape\textbf{or} $\abs{e_i} < 2$}
		\State $\mathcal{R}' \leftarrow \mathcal{R}' \cup \{v \in e_i : d(v; H) = k\}$ 
		\State $E \leftarrow E \setminus \setbr{e_i}$ \Comment{\blue{Remove hyperedge}} 
		\EndIf
		\EndFor
		\State $V \leftarrow V \setminus \mathcal{R}$
		\State $\mathcal{R} \leftarrow \mathcal{R}'$\label{line:while_end_k}
		\EndIf
		\EndWhile
		\State \Return{$f_k(v)$ for each $v \in V$}
	\end{algorithmic}
\end{algorithm}

\subsection{Computation of $k$-fraction}
Alg.~\ref{alg:k_fraction} shows the process of computing $k$-fraction.
Similar to Alg.~\ref{alg:t_hypercoreness}, we {repeatedly} find the $(k, t)$-hypercore while increasing $t$ until an empty hypergraph remains.
\red{We first find the minimum fraction $t$ for the remaining hyperedges (Line~\ref{line:increase_t}), i.e., at least one hyperedge will be totally removed if we use any fraction strictly larger than $t$. We check the hyperedges that will be immediately removed and collect the nodes that will consequently be removed (Lines~\ref{line:th_check_k}-\ref{line:th_check_k_end}).}
Notably, Alg.~\ref{alg:k_fraction} can output the $(k,t)$-hypercores for the given $k$ and all possible $t$ with the same time complexity in Thm.~\ref{thm:alg_k}.
\begin{theorem}\label{thm:alg_k}
	Given $H = (V, E)$ and $k \in \bbN$, Alg.~\ref{alg:k_fraction} returns $f_k(v)$ for {all} $v \in V$ in $O(\sum_{e \in E} \abs{e})$ time.
\end{theorem}
\begin{proof}
    See Appendix~\ref{subsec:pf_thm_alg_k}.
\end{proof}

    There are some existing works on improving the efficiency of the computation of some related hypercore concepts~\citep{luo2021hypercore, luo2022hypercore, arafat2023neighborhood}.
    We leave potential improvements of our computation algorithms as future directions.	
\color{black}

\section{Observations}\label{sec:obss}
\begin{table}[t!]
	\begin{center}
		\caption{{The basic statistics of the $14$ real-world datasets from $6$ domains used in our empirical evaluations.} \red{See Table~\ref{tab:num_card_edges} in Appendix~\ref{appendix:extra_exp} for the number of hyperedges of different cardinality in each dataset.}}
		\label{tab:datasets}			\resizebox{0.6\linewidth}{!}{%
		\begin{tabular}{ lrrrr}
			\toprule
			\textbf{Dataset}& $\abs{V}$   & $\abs{E}$& max./avg. $d(v)$ & max./avg. $\abs{e}$ \\
			\midrule
			coauth-DBLP 	& 1,831,126 & 2,169,663 & 846 / 4.06 	    & 25 / 3.42 \\
			coauth-Geology 	& 1,087,111 & 908,516  & 716 / 3.21 	    & 25 / 3.84 \\
			\midrule
			NDC-classes 	& 1,149    & 1,047 	& 221 / 5.57	    & 24 / 6.11 \\
			NDC-substances  & 3,438    & 6,264 	& 578 / 14.51	    & 25 / 7.96 \\
			\midrule
			contact-high 	& 327     & 7,818    & 148 / 55.63	    & 5 / 2.33  \\
			contact-primary & 242	  & 12,704	& 261 / 126.98	    & 5 / 2.42  \\
			\midrule
			email-Enron		& 143	  & 1,457	& 116 / 31.43	    & 18 / 3.09 \\
			email-Eu		& 979	  & 24,399	& 910 / 86.93	    & 25 / 3.49 \\
			\midrule
			tags-ubuntu 	& 3,021 	  & 145,053  & 12,930 / 164.56  & 5 / 3.43  \\
			tags-math 		& 1,627 	  & 169,259  & 13,949 / 363.80  & 5 / 3.50  \\
			tags-SO			& 49,945	  & 5,517,054 & 520,468 / 427.77 & 5 / 3.87  \\
			\midrule
			threads-ubuntu  & 90,054   & 115,987  & 2,170 / 2.97		& 14 / 2.31 \\
			threads-math 	& 153,806  & 535,323  & 11,358 / 9.08  	& 21 / 2.61 \\
			threads-SO		& 2,321,751 & 8,589,420 & 34,925 / 9.75	& 25 / 2.64 \\
			\bottomrule %
		\end{tabular}
  }
	\end{center}
\end{table}

In this section, we present observations with regard to our proposed concepts, on real-world hypergraphs, from various perspectives. 
In particular, we show empirical properties and patterns that are pervasive or shared within each domain.

\smallsection{Datasets.}
In Tbl.~\ref{tab:datasets}, we report the basic statistics of the fourteen real-world hypergraph datasets in six different domains used in this work (source: \blue{\url{cs.cornell.edu/~arb/data}}).

For each dataset, we remove the hyperedges of cardinality $1$.
Although parallel hyperedges are allowed in our framework, we only keep one copy of each group of parallel hyperedges as in previous studies~\citep{kook2020evolution, lee2020hypergraph, do2020structural, lee2021hyperedges}.

\begin{figure*}[t]
	\centering    
	\begin{subfigure}[b]{\linewidth}
		\centering
		\includegraphics[scale=0.5]{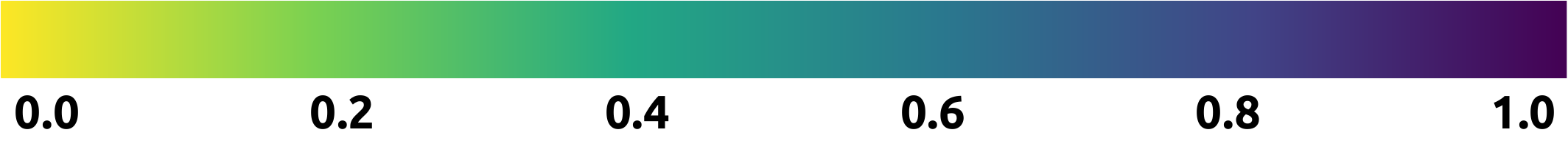}\\
	\end{subfigure}
	\begin{subfigure}[b]{0.48\linewidth}
		\centering
		\includegraphics[scale=0.35]{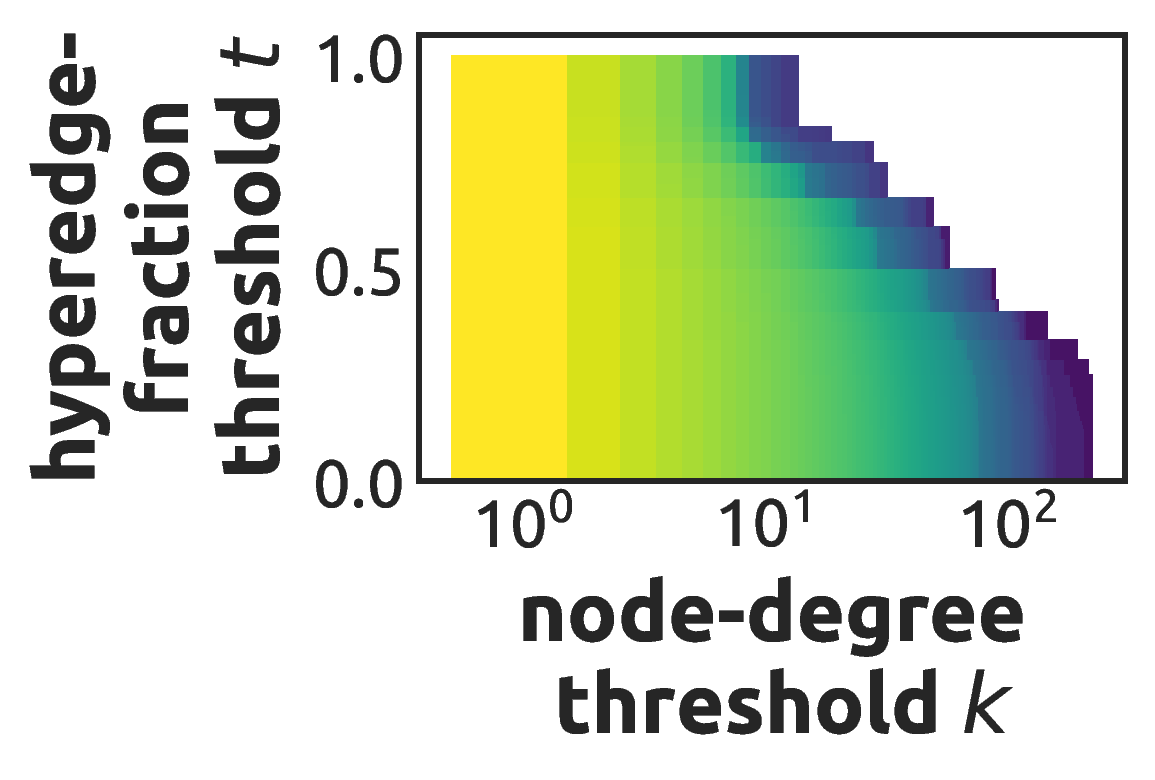}
		\includegraphics[scale=0.35]{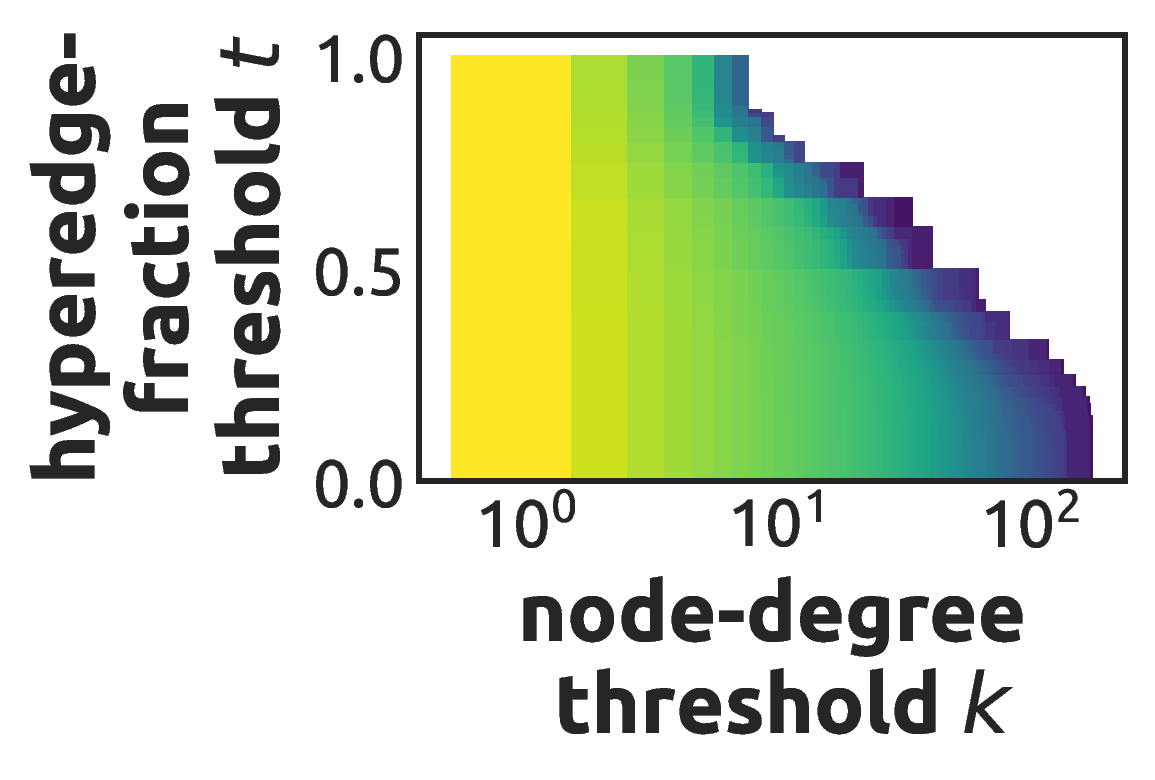}
		\caption{coauth-DBLP/Geology}
	\end{subfigure}
	\begin{subfigure}[b]{0.48\linewidth}
		\centering
		\includegraphics[scale=0.35]{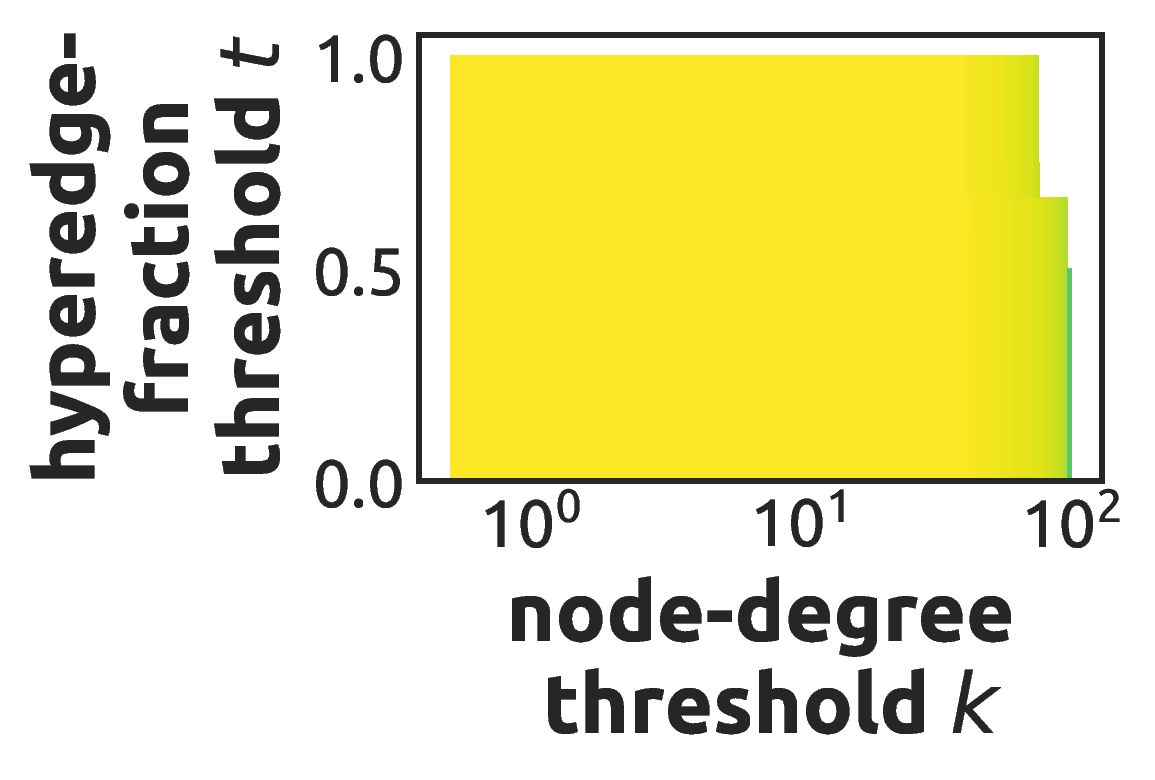}
		\includegraphics[scale=0.35]{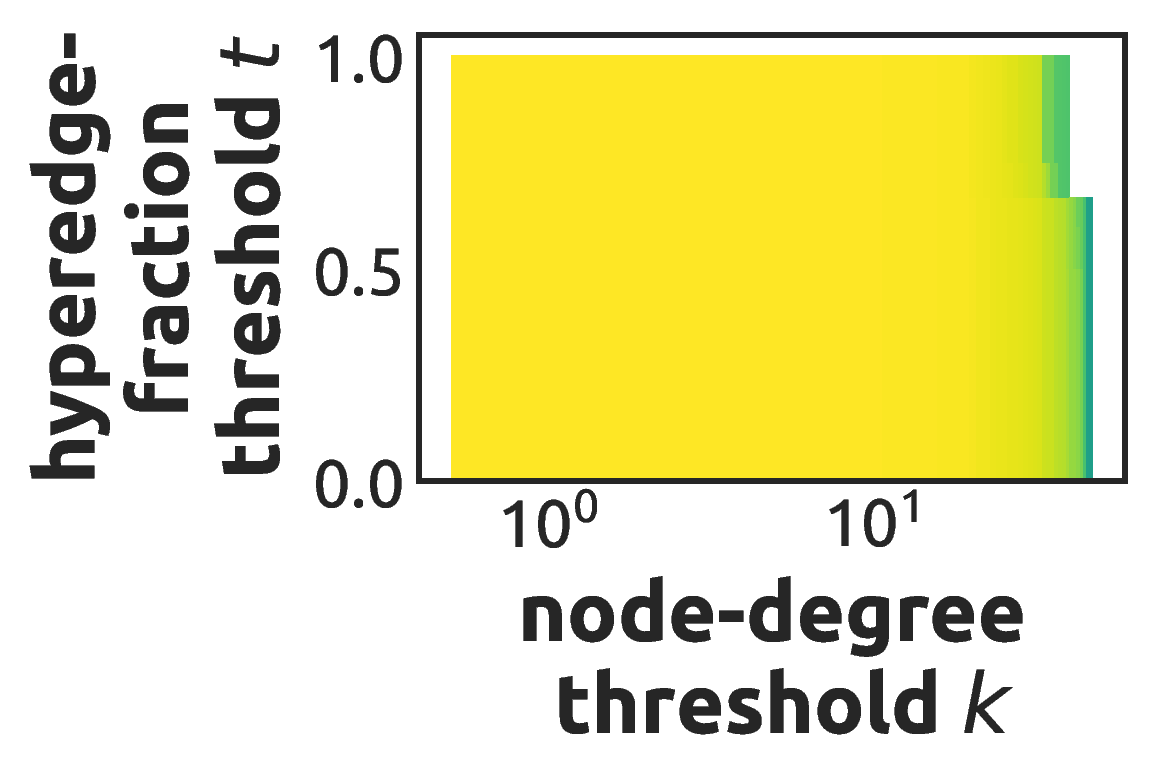}
		\caption{contact-primary/high}
	\end{subfigure}    
	\begin{subfigure}[b]{0.48\linewidth}
		\centering		
		\includegraphics[scale=0.35]{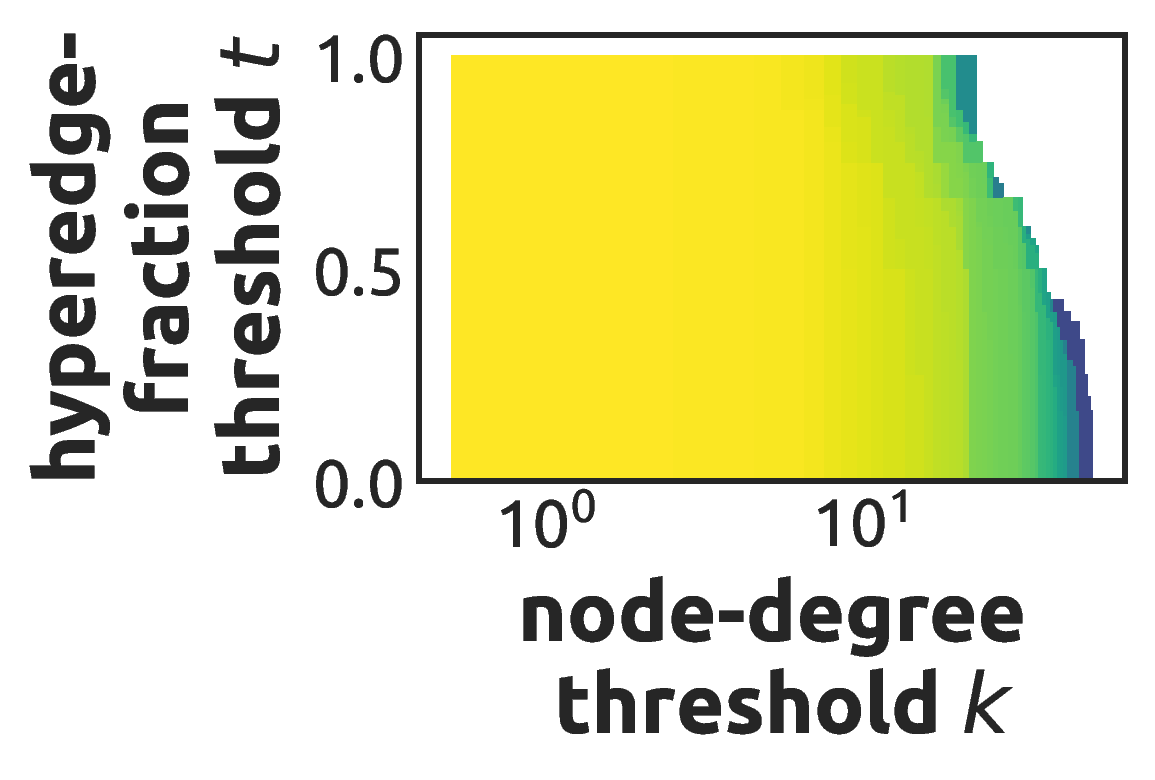}  
		\includegraphics[scale=0.35]{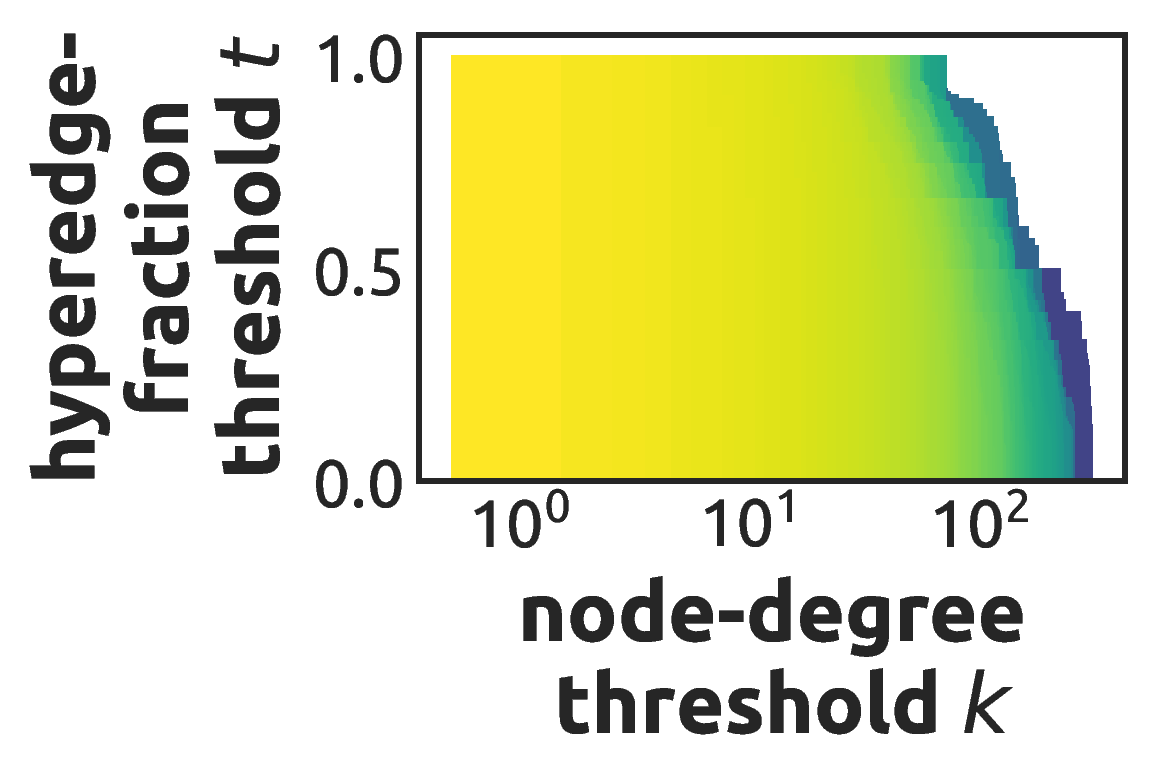}  
		\caption{email-Enron/Eu}
	\end{subfigure}
	\begin{subfigure}[b]{0.48\linewidth}
		\centering
		\includegraphics[scale=0.35]{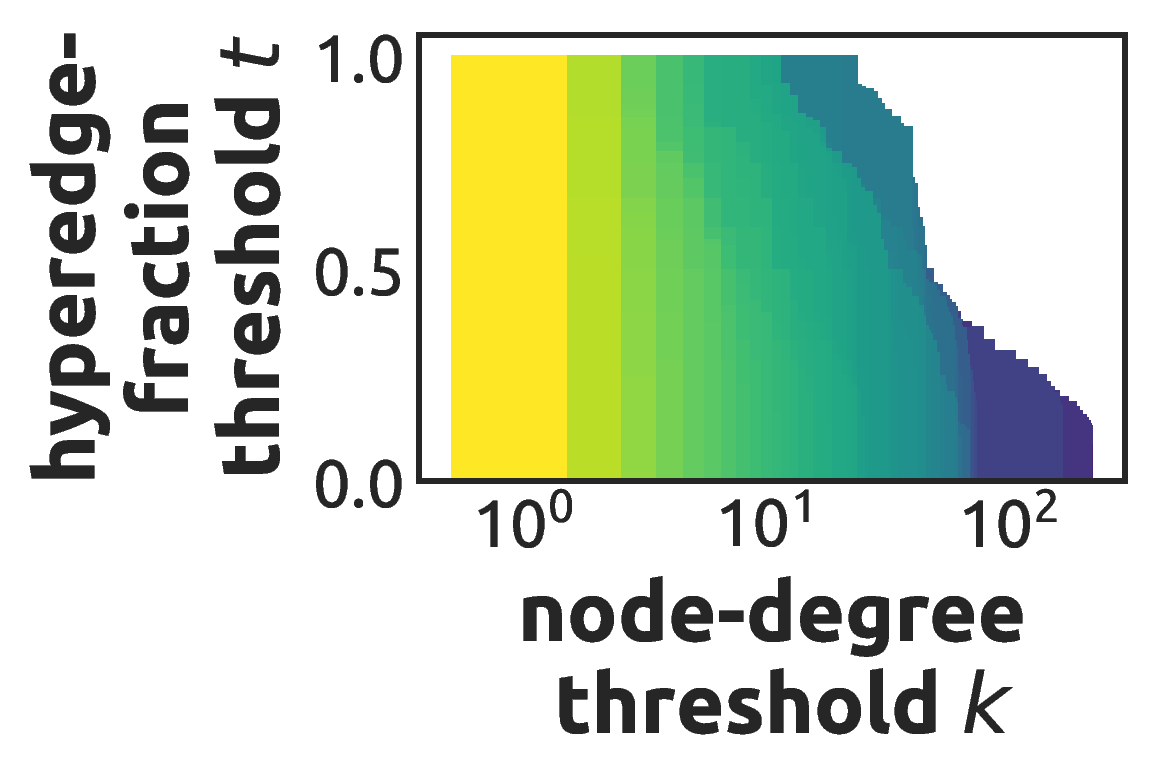}  
		\includegraphics[scale=0.35]{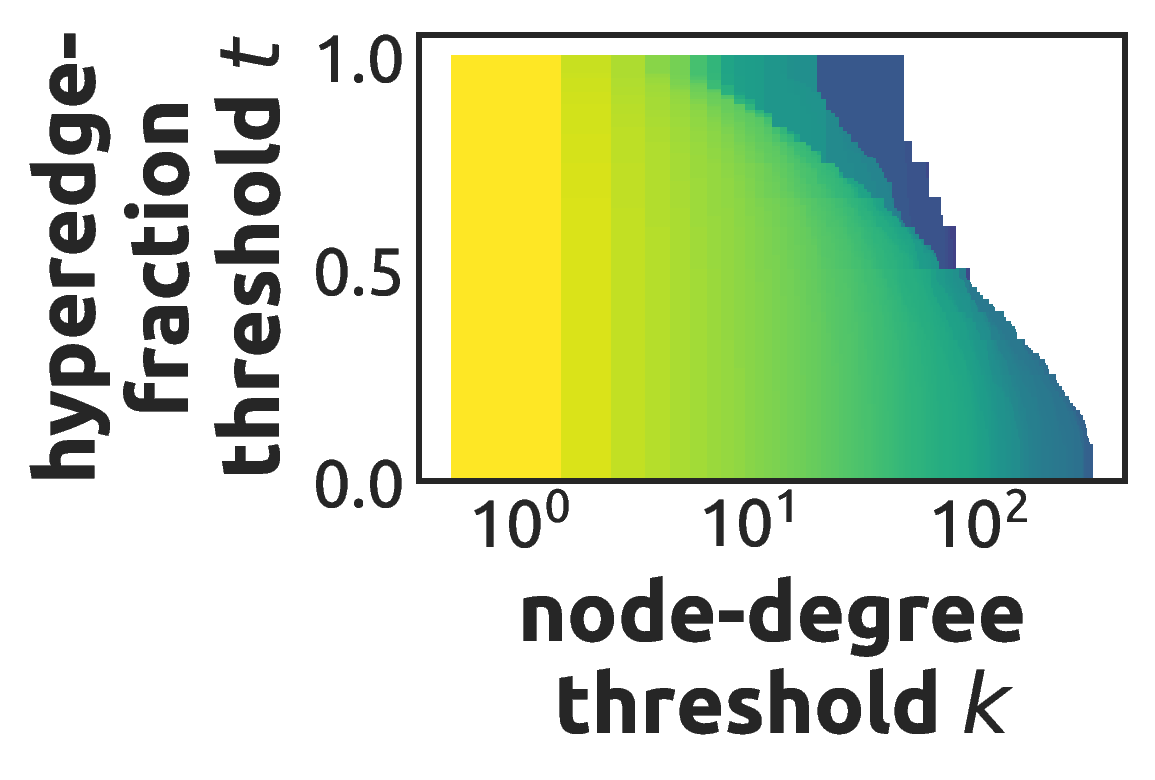}  
		\caption{NDC-classes/substances}
	\end{subfigure}   
	\begin{subfigure}[b]{0.48\linewidth}
		\centering
		\includegraphics[scale=0.35]{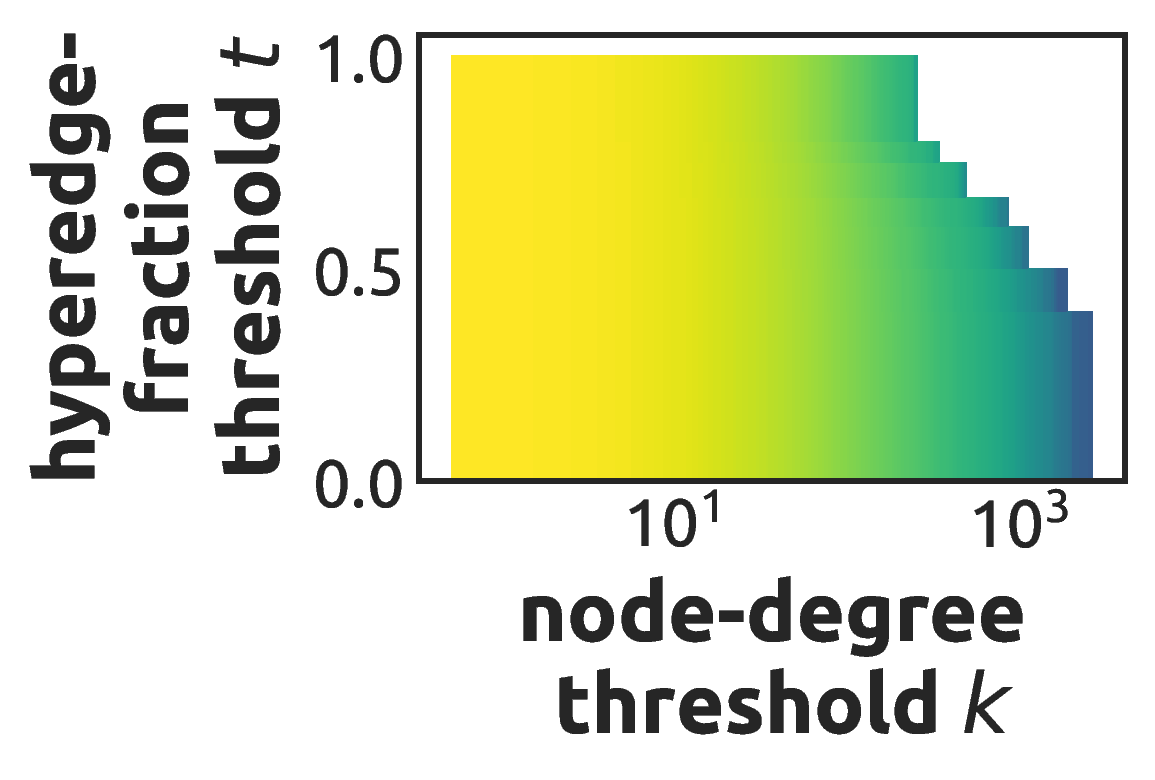}  
		\includegraphics[scale=0.35]{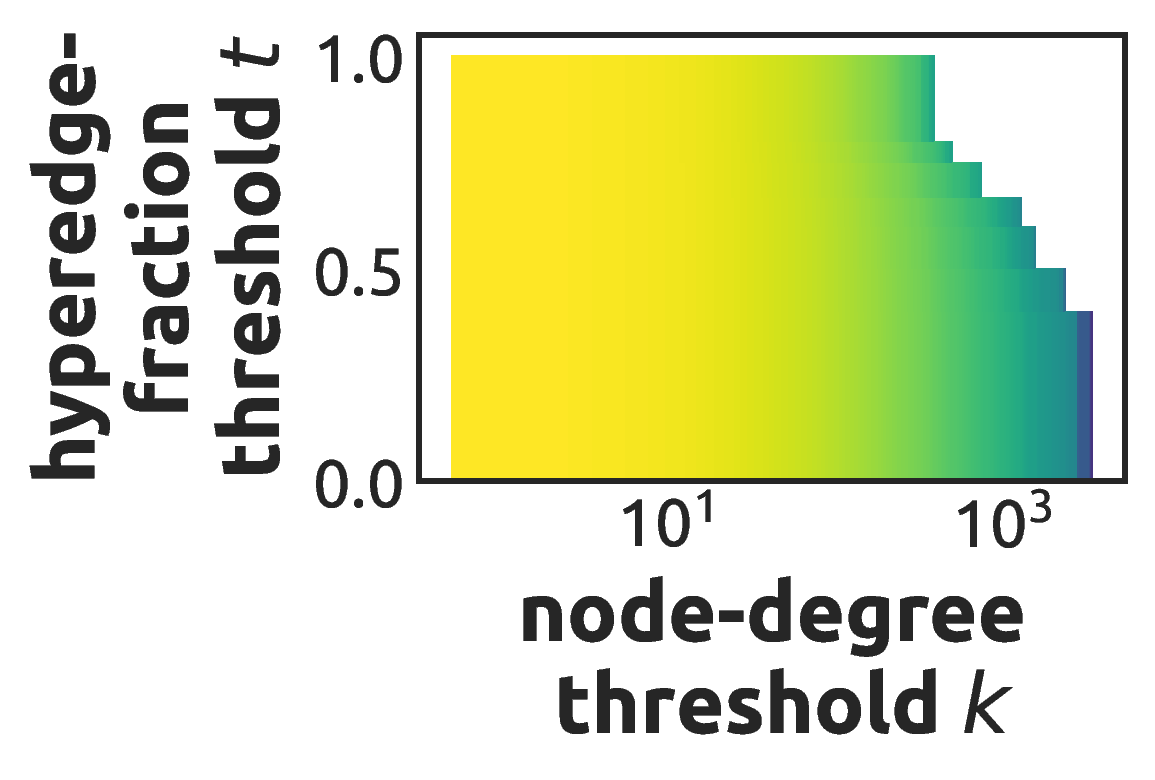}  
		\includegraphics[scale=0.35]{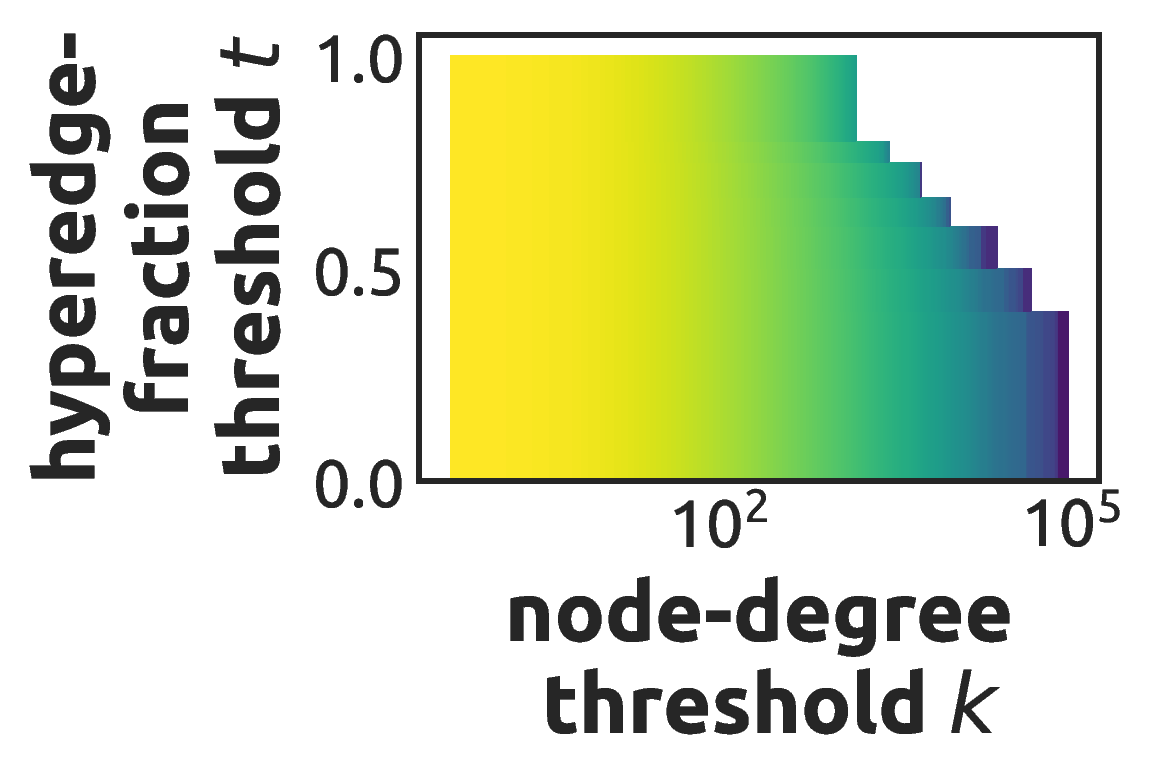}  
		\caption{tags-ubuntu/math/SO}
	\end{subfigure}   
	\begin{subfigure}[b]{0.48\linewidth}
		\centering
		\includegraphics[scale=0.35]{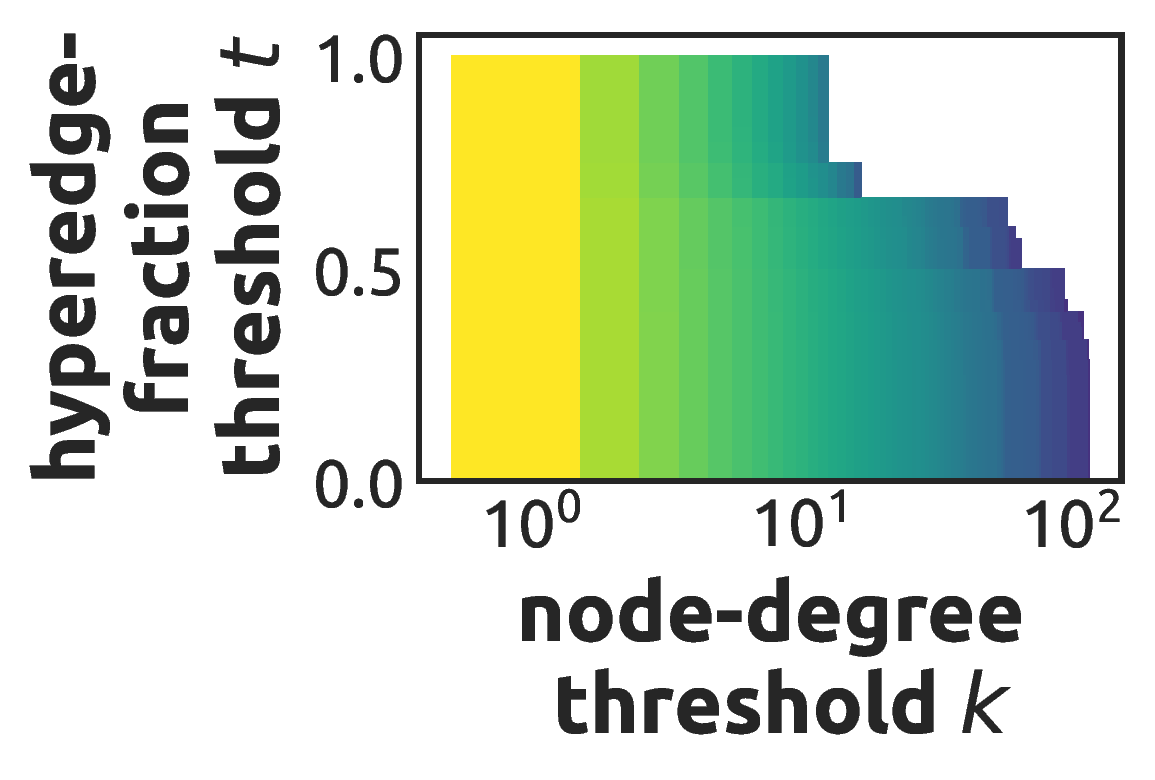}  
		\includegraphics[scale=0.35]{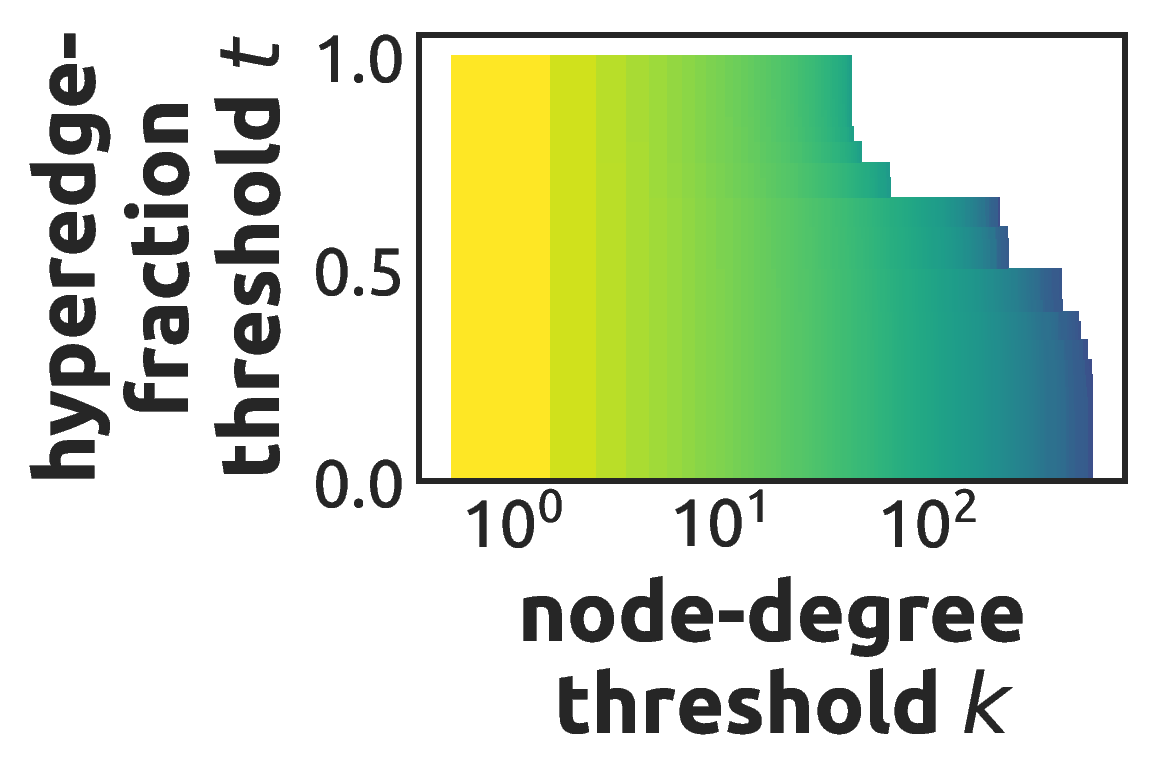}  
		\includegraphics[scale=0.35]{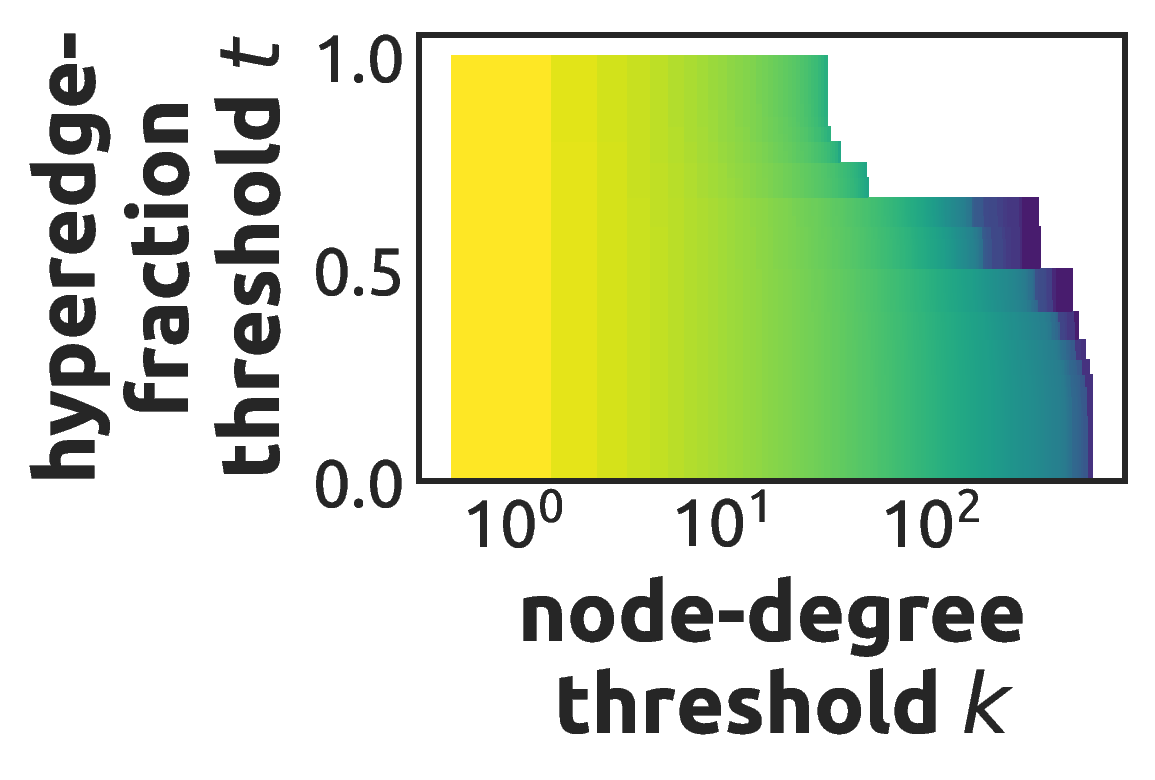}  
		\caption{threads-ubuntu/math/SO}
	\end{subfigure}    
	\caption{\bolden{Domain-based patterns of $(k,t)$-hypercore sizes.}
		The $(k, t)$-hypercore sizes vary depending on the node-degree threshold $k$ and the hyperedge-fraction threshold $t$ with datasets grouped by domains.
		The color indicates the size of the corresponding $(k, t)$-hypercore.
		The size and $k$ are in a log scale.}
	\label{fig:kt_coresize}
\end{figure*}

\subsection{Patterns of $(k,t)$-hypercore sizes}\label{subsec:obs_coresize}
Due to the newly introduced parameter $t$, we have hypercores of different sizes for different $(k, t)$ pairs.
In Fig.~\ref{fig:kt_coresize}, we report the hypercore sizes {(i.e., the number of nodes in the hypercore)} for different $k$ and $t$, where the color represents the size of the $(k, t)$-hypercore.
Specifically, the color of the position $(k, t)$ is the color assigned to $\Tilde{n}_{k, t} \coloneqq \log_{\abs{V}} \abs{V(C_{k, t})} \in [0, 1]$, for all $(k, t)$ such that $C_{k, t} \neq \varnothing$.
Fig.~\ref{fig:kt_coresize} also shows $f^*_k$ for each $k$ (see the boundary between the colored and empty regions in each subfigure).

Similarity within each domain is observed in Fig.~\ref{fig:kt_coresize}.
To numerically measure the similarity, we need to compare the size of all $(k, t)$-hypercores in different hypergraphs.
Since different hypergraphs may have different absolute sizes and thus have different ranges of $(k, t)$ pairs, normalization is needed.
Given any hypergraph $H = (V, E)$, by the containment properties (Proposition~\ref{prop:containment}), $1 \leq c_t^* \leq c_0^*, \forall t$.
Therefore, we can use the normalizer $\mathcal{N}_H: [0, 1] \rightarrow \setbr{1, 2, \ldots, c_0^*}$ defined by $\mathcal{N}_H(x) = \lceil (c_0^*)^x \rceil$.
We then define the dissimilarity between two hypercore sizes by their difference in log scale (as in Figure~\ref{fig:kt_coresize}), which is also normalized in $[0, 1]$.
Formally, the dissimilarity between two hypergraphs $H_1, H_2$ at the normalized point $(x, t)$ with $x, t \in [0, 1]$ is $\Tilde{d}(x, t; H_1, H_2) \coloneqq \min(\abs{\Tilde{n}_{\mathcal{N}_{H_1}(x), t}(H_1) - \Tilde{n}_{\mathcal{N}_{H_2}(x), t} (H_2)}, 1)$, where we let $\Tilde{n}_{k, t} = -1$ if $C_{k, t}$ is empty.
This dissimilarity can also be understood as the difference between the same position of two subfigures in Figure~\ref{fig:kt_coresize}.
Finally, we define the hypercore-size-mean-difference (HSMD) distance, {which lies between $0$ and $1$}, as follows:
\begin{definition}[Hypercore-size-mean-difference (HSMD) distance]\label{def:HSMD}
	Given two hypergraphs $H_1$ and $H_2$, the hypercore-size-mean-difference (HSMD) distance between $H_1$ and $H_2$ is defined as 
	\begin{align*}
		\operatorname{HSMD}(H_1, H_2) \coloneqq \sqrt{\int_{0}^{1} \int_{0}^{1} (\Tilde{d}(x, t; H_1, H_2))^2 \, dx \, dt}.
	\end{align*}
\end{definition}
See Fig.~\ref{fig:core_size_heatmap} for the HSMD distance between each pair of datasets, where the domain-based patterns are clearly shown by the small distance between those datasets in the same domain.

\begin{observation}[Domain-based patterns of $(k,t)$-hypercore sizes]\label{obs:domain:kt}
	Real-world hypergraphs in the same domain usually have similar patterns of the hypercore sizes with different $k$ and $t$ values, and the patterns vary from domain to domain.
\end{observation}

\begin{figure*}[t]
	\centering
	\includegraphics[scale=0.4]{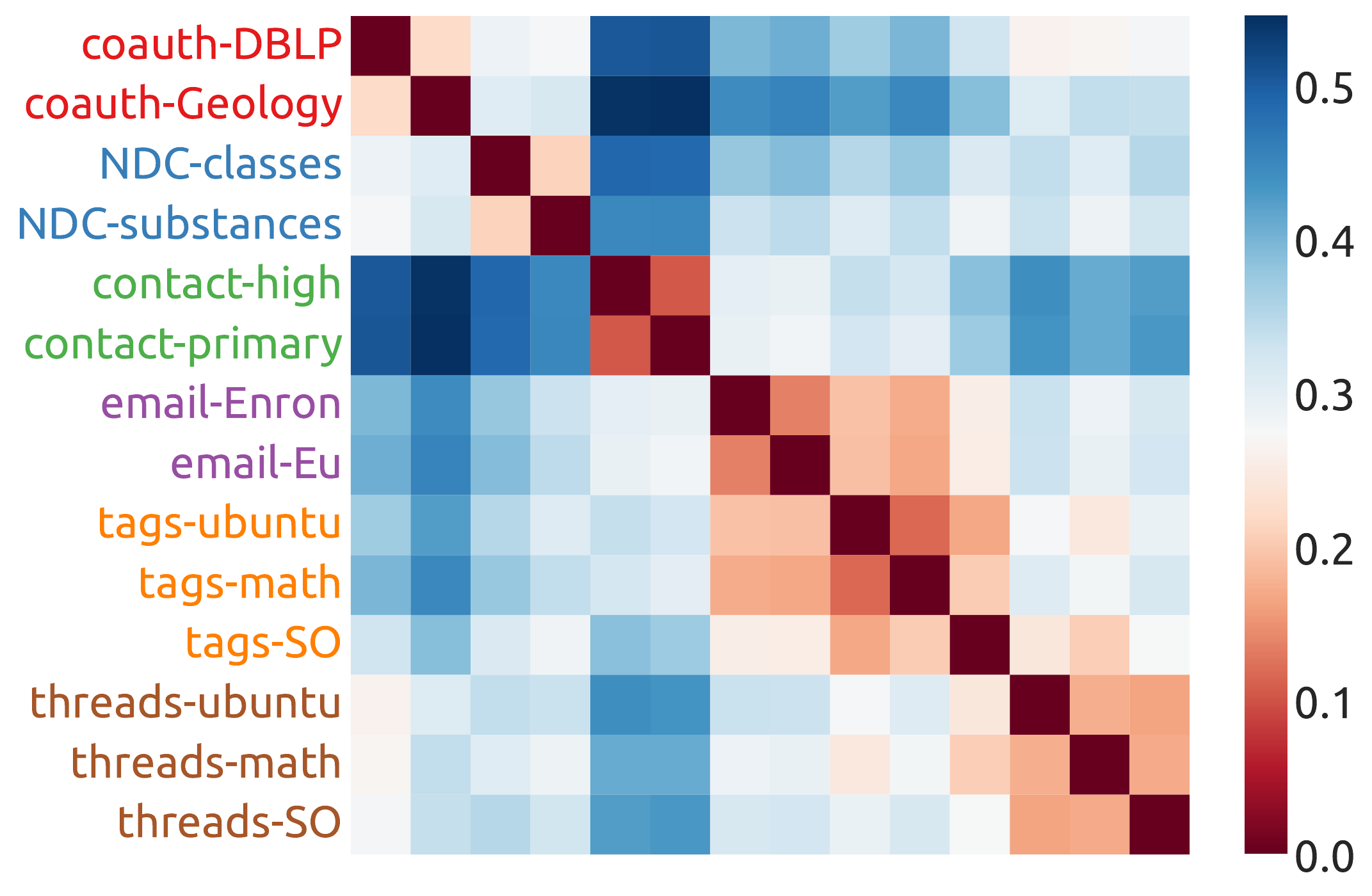}
	\caption{\textbf{Datasets in the same domain tend to have small HSMD distance, while ones in different domains usually have large HSMD distance.}
		The average within-domain distance $0.166$ and the global average distance $0.323$ are significantly different with $p = 8.6\mathrm{e}{-10}$ in the $t$-test.}
	\label{fig:core_size_heatmap}
\end{figure*}

\subsection{Distributions of $t$-hypercoreness}

We now investigate the distributions of the $t$-hypercoreness of nodes with different $t$ values, {which} show common patterns. 
Heavy-tailed distributions, especially power-law distributions, are observed in real-world (hyper)graphs w.r.t many different quantities \citep{mcglohon2008weighted, watts1998collective,albert2002statistical,adamic2001search,kook2020evolution,lee2021thyme+}.
In Fig.~\ref{fig:heavy_tailed_dist}, for the $t$-hypercoreness sequences of each dataset with $t \in \setbr{0, 0.2, 0.4, 0.6, 0.8, 1}$, we report the log-likelihood ratio ($R$-value) of heavy-tailed distributions against the exponential distribution, where a positive $R$-value indicates that heavy-tailed distributions are more promising.
In particular, we compute the log-likelihood ratio for two heavy-tailed distributions (power-law and log-normal) and take the maximum.
In most cases, the log-likelihood ratio is positive, which supports the possibility that the $t$-hypercoreness follows heavy-tailed distributions consistently regardless of the value of $t$.
Notably, regarding the distributions of $k$-fraction, {we could not find any systematic pattern}. 
Moreover, strong power-law distributions are observed in some datasets.
In Fig.~\ref{fig:powerlaw_example}, for two datasets, we show the numbers of nodes with $t$-hypercoreness at least $k$ with different $k$ values with different $t$ values,
together with the results of power-law fitting, i.e., linear regression in log-log scale{; and} consistent power-law distributions of the $t$-hypercoreness sequences are observed.
In Table~\ref{tab:dist_stats}, we provide the full results of the heavy-tailed distribution tests. Specifically, we report the log-likelihood ratio ($R$-value) of heavy-tailed distributions against the exponential distribution, where a positive $R$-value indicates that heavy-tailed distributions are more promising; and the $p$-values, where a small $p$-value indicates that the heavy-tailed or exponential distribution is significant.
\color{black}

\begin{figure}[t!]
	\centering	
	\includegraphics[scale=0.4]{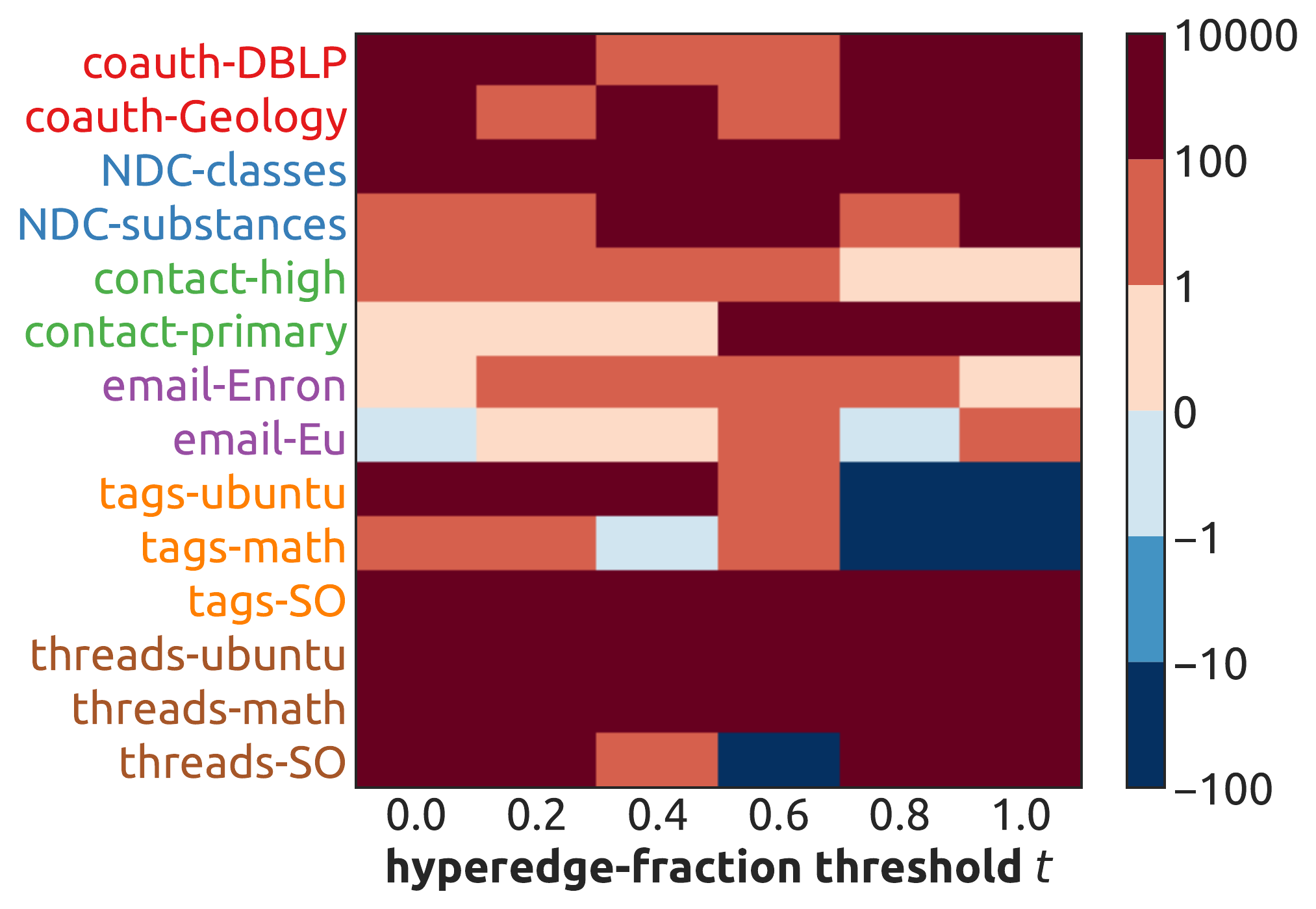}
	\caption{\textbf{$t$-Hypercoreness follows heavy-tailed distributions consistently.}
		The maximum log-likelihood ratio of two heavy-tailed distributions (power-law and log-normal) against the exponential distribution
		for the $t$-hypercoreness sequences with different $t$ values.}
	\label{fig:heavy_tailed_dist}
\end{figure}

\begin{figure*}[t!]
	\centering
	\begin{subfigure}[b]{0.9\linewidth}
		\centering		
		\includegraphics[scale=0.4]{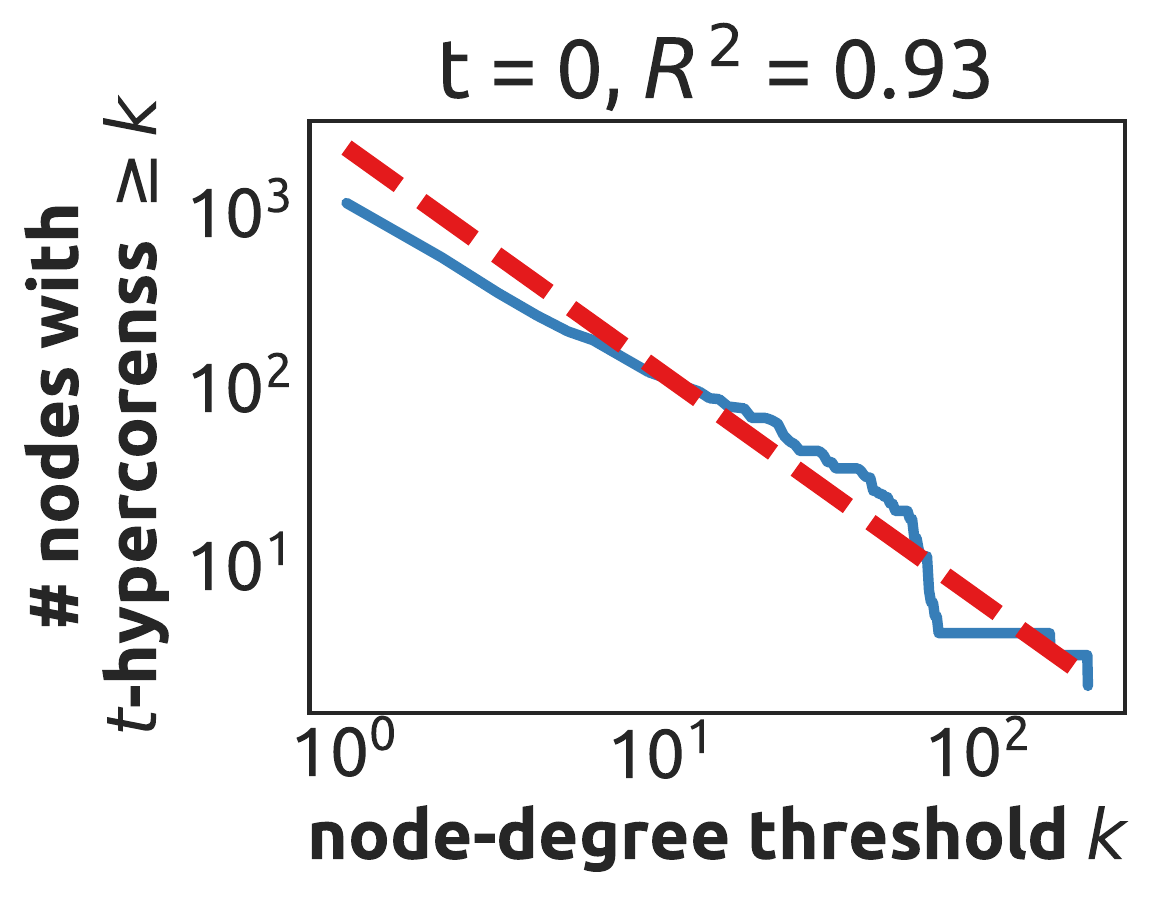}
		\includegraphics[scale=0.4]{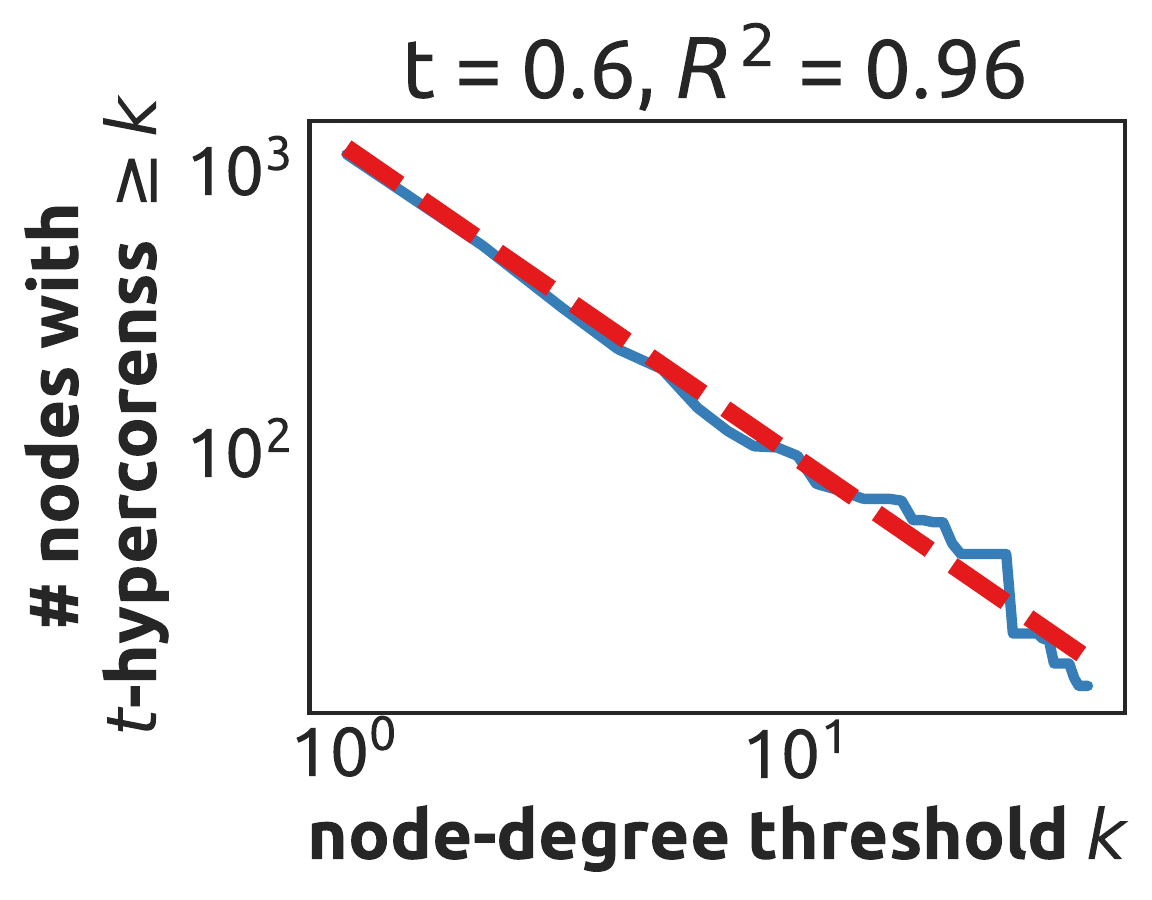}
		\includegraphics[scale=0.4]{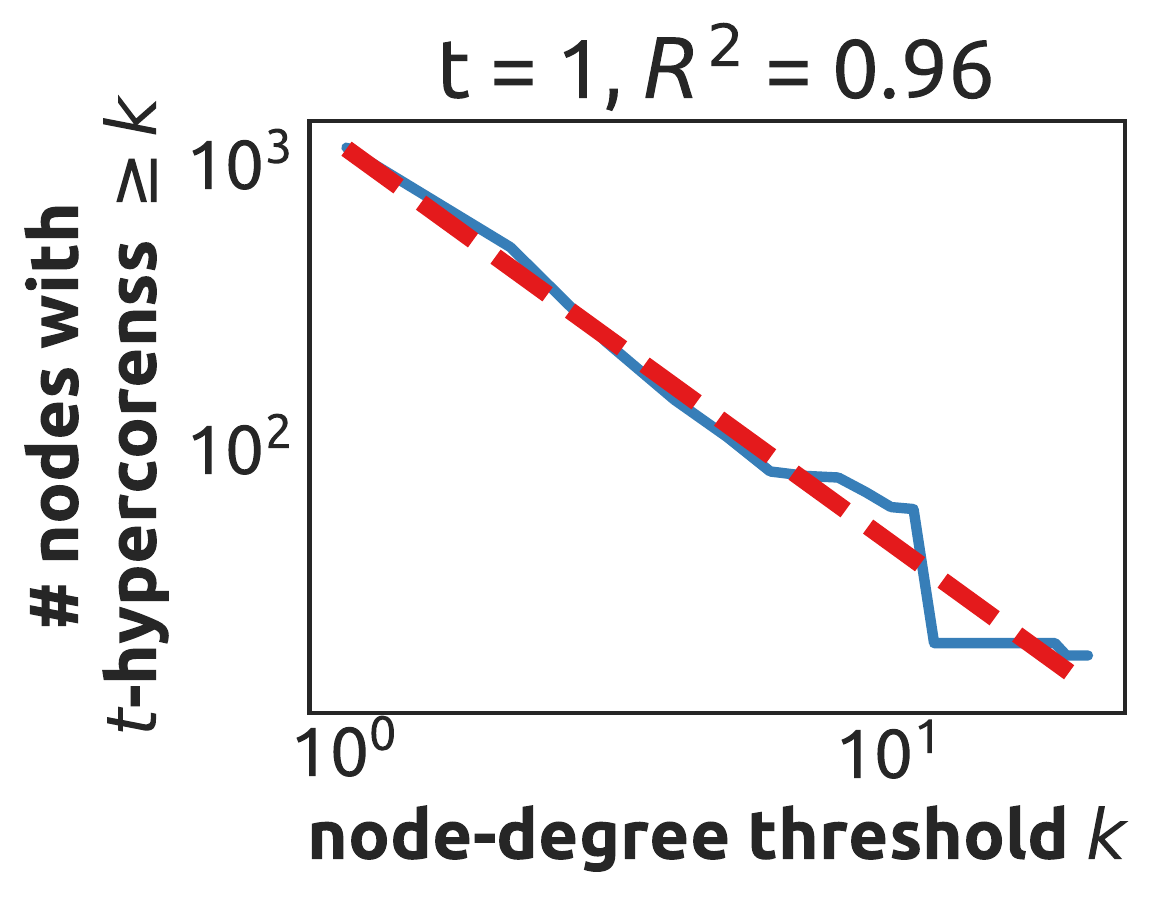}
		\caption{NDC-classes}
	\end{subfigure}
	\begin{subfigure}[b]{0.9\linewidth}
		\centering
		\includegraphics[scale=0.4]{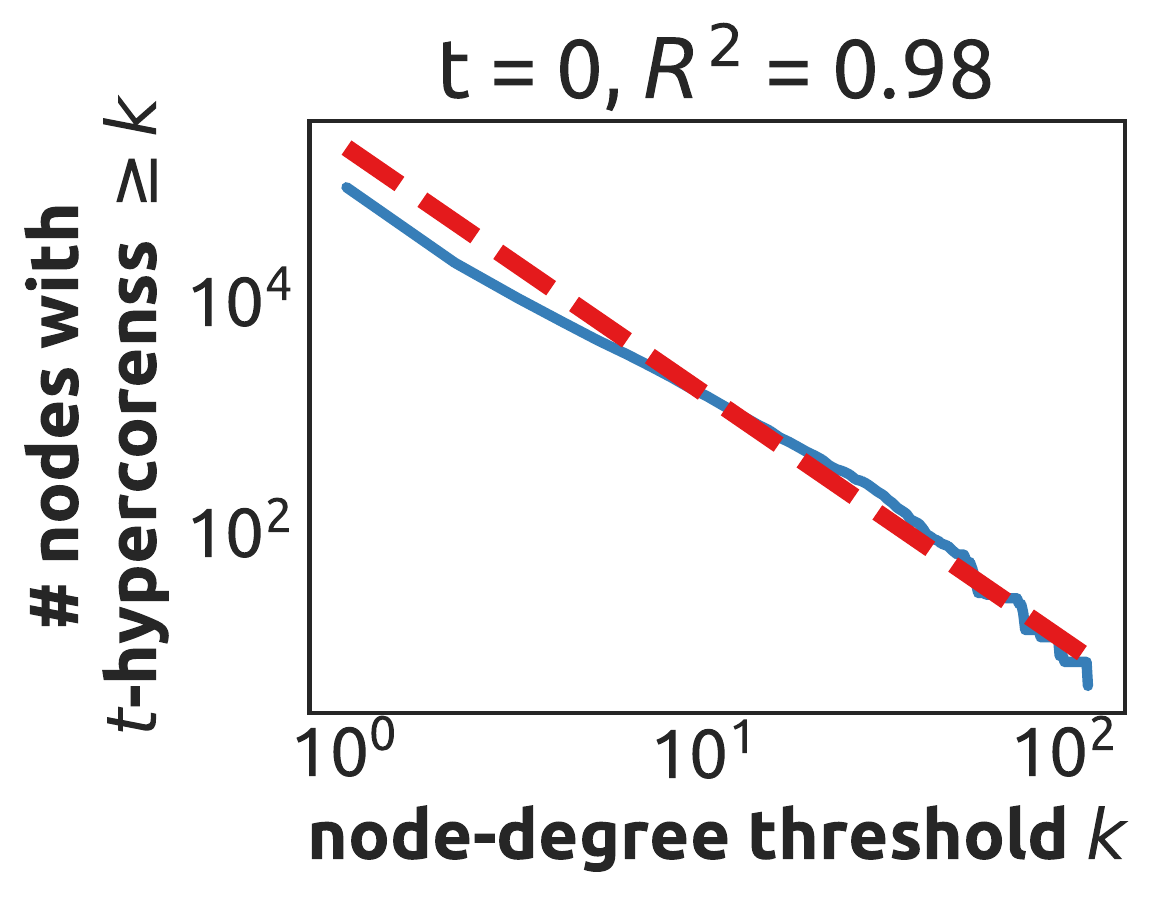}
		\includegraphics[scale=0.4]{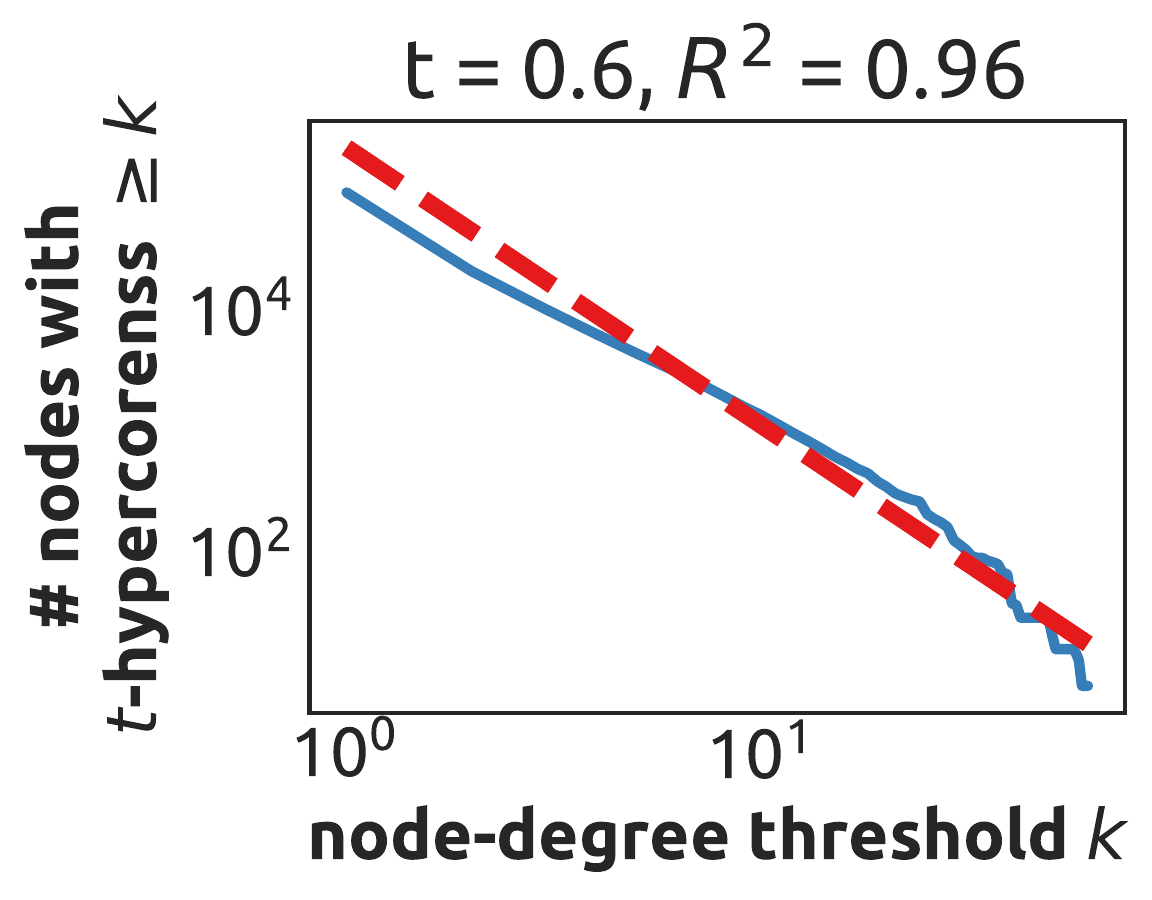}
		\includegraphics[scale=0.4]{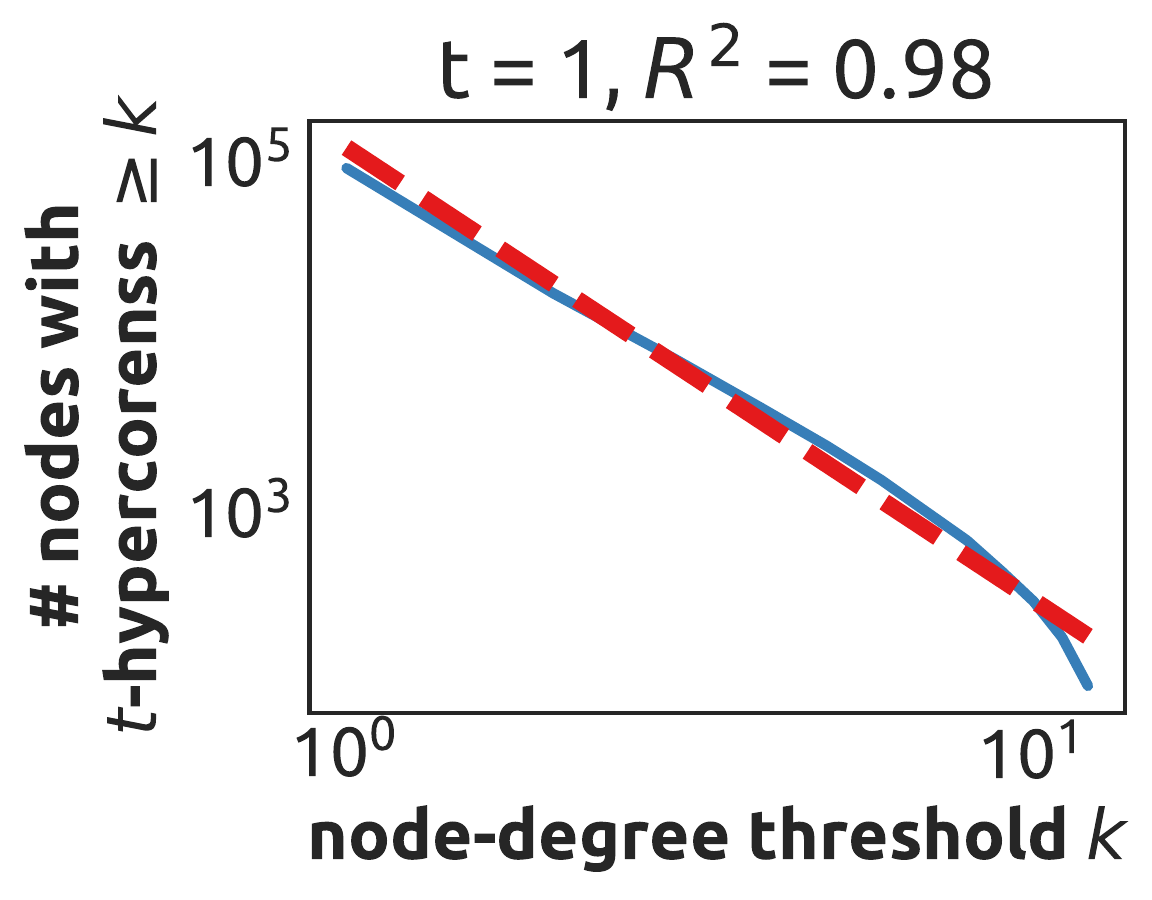}
		\caption{threads-ubuntu}
	\end{subfigure}
	\caption{\textbf{$t$-Hypercoreness consistently follows power-law distributions in some datasets.}
		For the NDC-classes and threads-ubuntu datasets with $t \in \setbr{0, 0.6, 1}$, we show the numbers of nodes with $t$-hypercoreness at least $k$ with different $k$ values.
		Each red dashed line represents the {result} of power-law fitting, i.e., the linear regression in log-log scale, with the $R^2$ value above each subfigure.
		In the two datasets, $t$-hypercoreness consistently and {strongly} follows a power law.}
	\label{fig:powerlaw_example}
\end{figure*}

\begin{table}[t!]
	\begin{center}
		\caption{\textbf{The detailed statistics on the heavy-tailed distribution tests.} For each dataset and each $t \in \setbr{0, 0.2, 0.4, 0.6, 0.8, 1}$, we report the log-likelihood ratio ($R$-value) of heavy-tailed distributions against the exponential distribution with its $p$-value. In most cases, the $R$-value is positive and the $p$-value is small, which implies the significance of the heavy-tailed distributions.} \label{tab:dist_stats}
		\resizebox{\columnwidth}{!}{%
			\begin{tabular}{ l rrrrrrrrrrrr }
				\toprule			
				& \multicolumn{2}{c}{$t = 0$} & \multicolumn{2}{c}{$t = 0.2$} & \multicolumn{2}{c}{$t = 0.4$} & \multicolumn{2}{c}{$t = 0.6$} & \multicolumn{2}{c}{$t = 0.8$} & \multicolumn{2}{c}{$t = 1$} \\ 
				\textbf{Dataset} & $R$-value & $p$-value & $R$-value & $p$-value & $R$-value & $p$-value & $R$-value & $p$-value & $R$-value & $p$-value & $R$-value & $p$-value \\
				\midrule
				coauth-DBLP 	&  156.75 & 7.24e-13 	& 184.96 & 7.28e-16 	& 139.27 & 3.36e-13 	& 50.80 & 0.001 	 	& 1685.14 & 1.45e-40 	& 117.76 & 4.75e-56 \\
				coauth-Geology 	&  106.38 & 3.06e-11 	& 83.21 & 7.70e-8   	& 31.52 & 6.86e-8	 	& 17.80 & 9.57e-5	 	& 1049.01 & 0.0	  		& 989.44 & 0.0 \\
				\midrule
				NDC-classes 	&  45.32 & 4.38e-6	 	& 364.85 & 5.80e-24 	& 103.93 & 4.90e-14 	& 282.71 & 1.06e-42 	& 290.16 & 4.21e-45  	& 242.25 & 1.28e-41 \\  
				NDC-substances  &  30.90 & 2.34e-5	 	& 26.95 & 0.00061	 	& 2608.99 & 4.33e-208	& 1884.07 & 6.54e-171	& 1175.06 & 8.51e-91	& 221.15 & 2.78e-24 \\
				\midrule
				contact-high 	&  16.15 & 3.20e-20	 	& 16.15 & 3.20e-20 		& 16.15 & 3.20e-20		& 16.76 & 0.0040		& 0.70 & 0.48			& 0.70 & 0.48 \\
				contact-primary &  0.19 & 0.23			& 0.19 & 0.23			& 0.19 & 0.23			& 136.51 & 1.75e-16		& 127.81 & 6.41e-13		& 127.81 & 6.41e-13 \\
				\midrule
				email-Enron		&  0.29 & 0.73			& 2.05 & 0.24			& 2.67 & 0.063			& 8.43 & 0.024			& 1.55 & 2.4e-267		& 0.22 & 0.76 \\
				email-Eu		& -0.47 & 0.60			& 0.05 & 0.97			& 2.40 & 3.77e-9		& 83.69 & 2.05e-11		& -0.28 & 0.36			& 11.26 & 0.005 \\
				\midrule
				tags-ubuntu 	& 201.24 & 4.28e-21		& 201.24 & 4.28e-21		& 201.24 & 4.28e-21		& 83.69 & 2.05e-11		& -14.68 & 5.76e-6		& -17.40 & 1.30e-32 \\
				tags-math 		& 8.81 & 0.06			& 8.81 & 0.06			& 8.81 & 0.06			& 15.31 & 0.027			& -17.96 & 1.03e-9		& -14.34 & 0.00052 \\
				tags-SO			& 616.59 & 2.41e-29		& 616.59 & 2.41e-29		& 3617.24 & 8.07e-222	& 2189.25 & 5.27e-234	& -17.40 & 1.30e-32 \\
				\midrule
				threads-ubuntu  & 279.41 & 2.09e-22 	& 278.50 & 2.71e-22 	& 259.53 & 8.96e-22 	& 130.95 & 1.15e-14 	& 119.14 & 8.15e-14 	& 226.74 & 6.06e-41 \\
				threads-math 	& 226.30 & 1.45e-23		& 225.66 & 1.68e-23		& 5192.50 & 2.08e-282	& 11461.10 & 0.0		& 3305.64 & 0.0			& 6632.53 & 0.0 \\
				threads-SO		& 444.93 & 3.47e-57		& 436.37 & 4.84e-56		& 153.14 & 6.50e-19		& -23.83 & 7.24e-8		& 6002.46 & 0.0			& 2682.90 & 4.64e-102 \\
				\bottomrule
			\end{tabular}
		}
	\end{center}
\end{table}

\begin{observation}[Heavy-tailed distributions of $t$-hypercoreness]
	In most real-world hypergraphs, $t$-hypercoreness
	follows heavy-tailed distributions regardless of $t$.
	In particular, in some datasets, the $t$-hypercoreness strongly follows a power law.
\end{observation}

\begin{figure*}[t!]
	\centering
	\begin{subfigure}[b]{0.9\textwidth}
		\centering
		\includegraphics[scale=0.375]{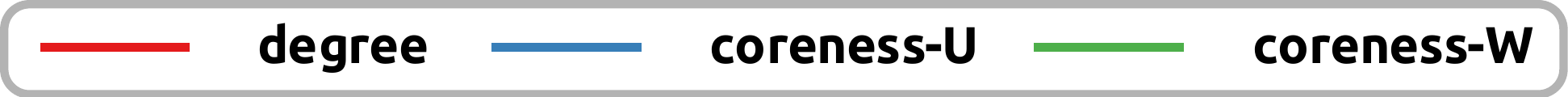}\\
	\end{subfigure}  
	\vspace{1mm}
	
	\begin{subfigure}[b]{0.48\linewidth}
		\centering
		\hspace{-5mm}
		\includegraphics[scale=0.4]{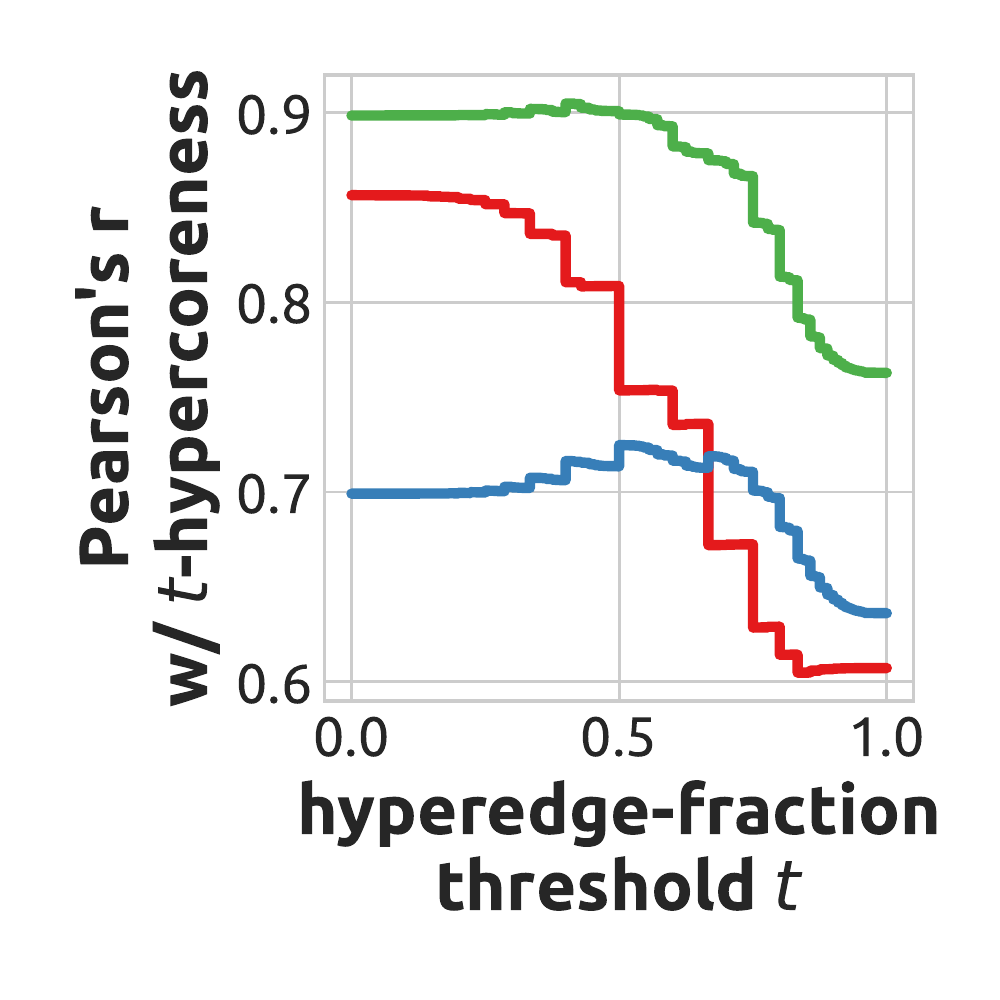}
		\includegraphics[scale=0.4]{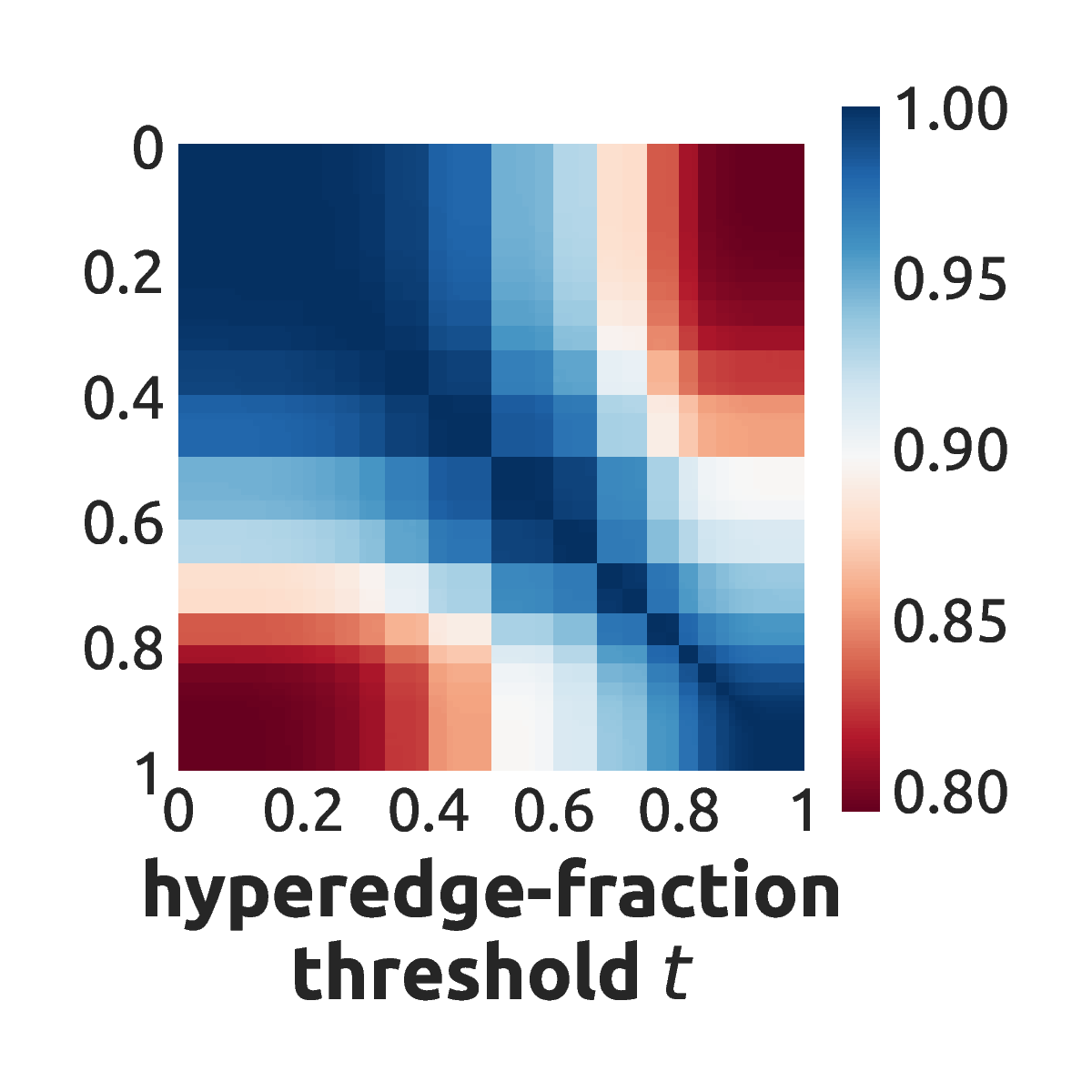}                
		\caption{coauth-DBLP}
	\end{subfigure}
	\begin{subfigure}[b]{0.48\linewidth}
		\centering
		\hspace{-5mm}
		\includegraphics[scale=0.4]{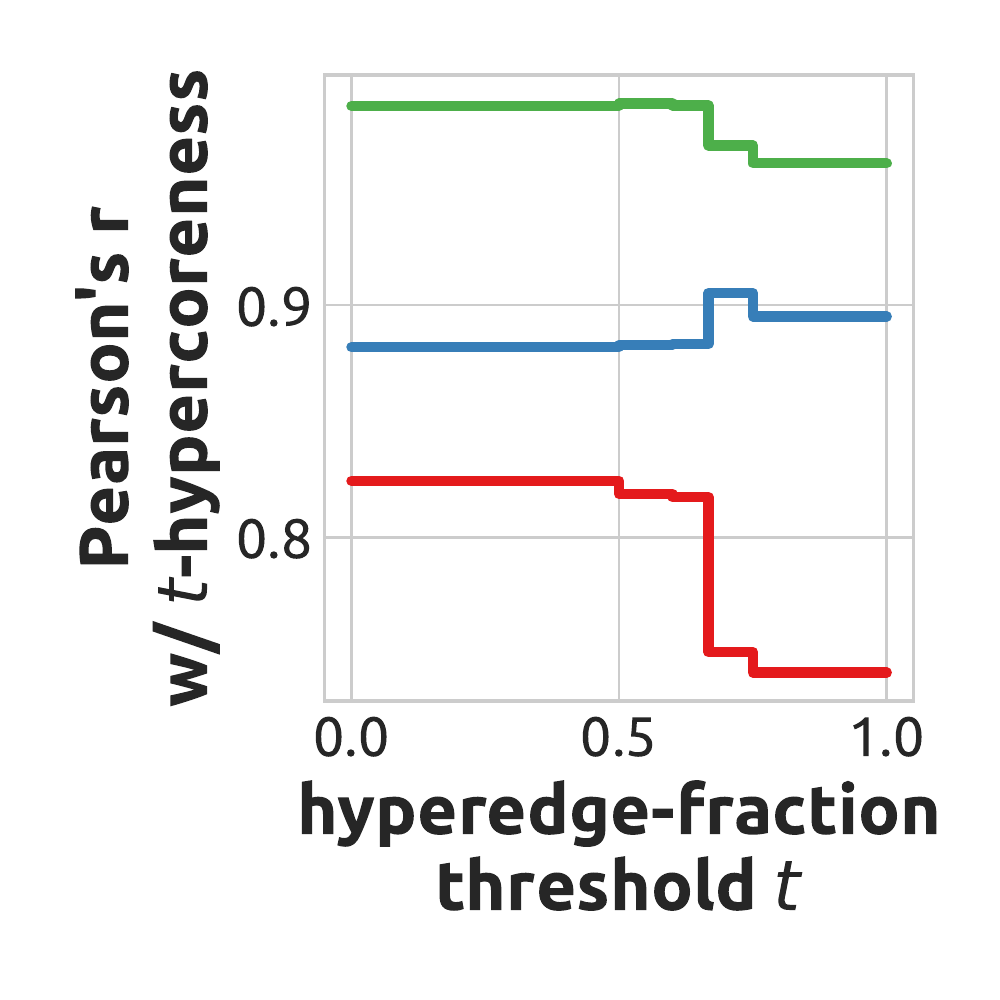}
		\includegraphics[scale=0.4]{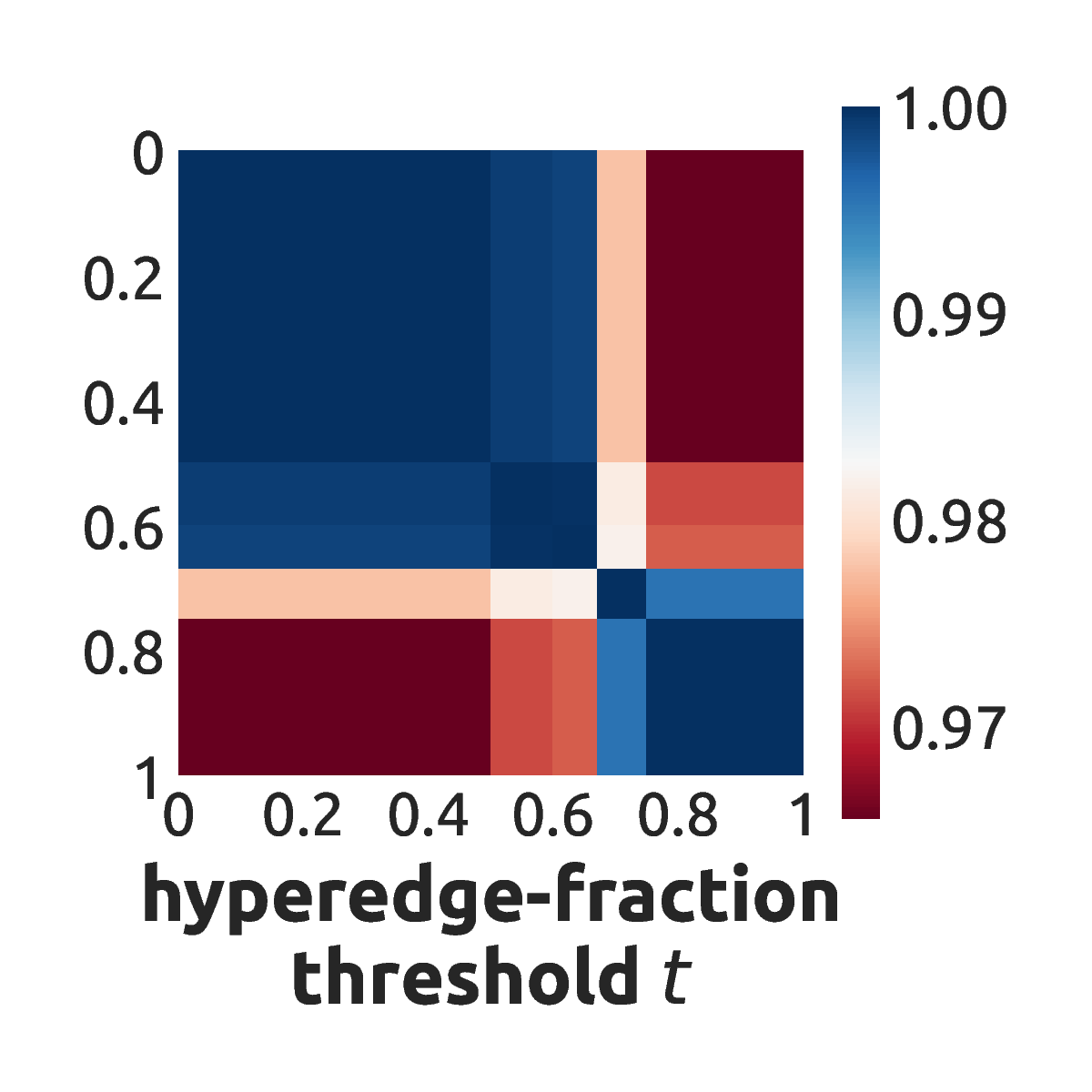}              
		\caption{contact-high}
	\end{subfigure}
	\begin{subfigure}[b]{0.48\linewidth}
		\centering
		\hspace{-5mm}
		\includegraphics[scale=0.4]{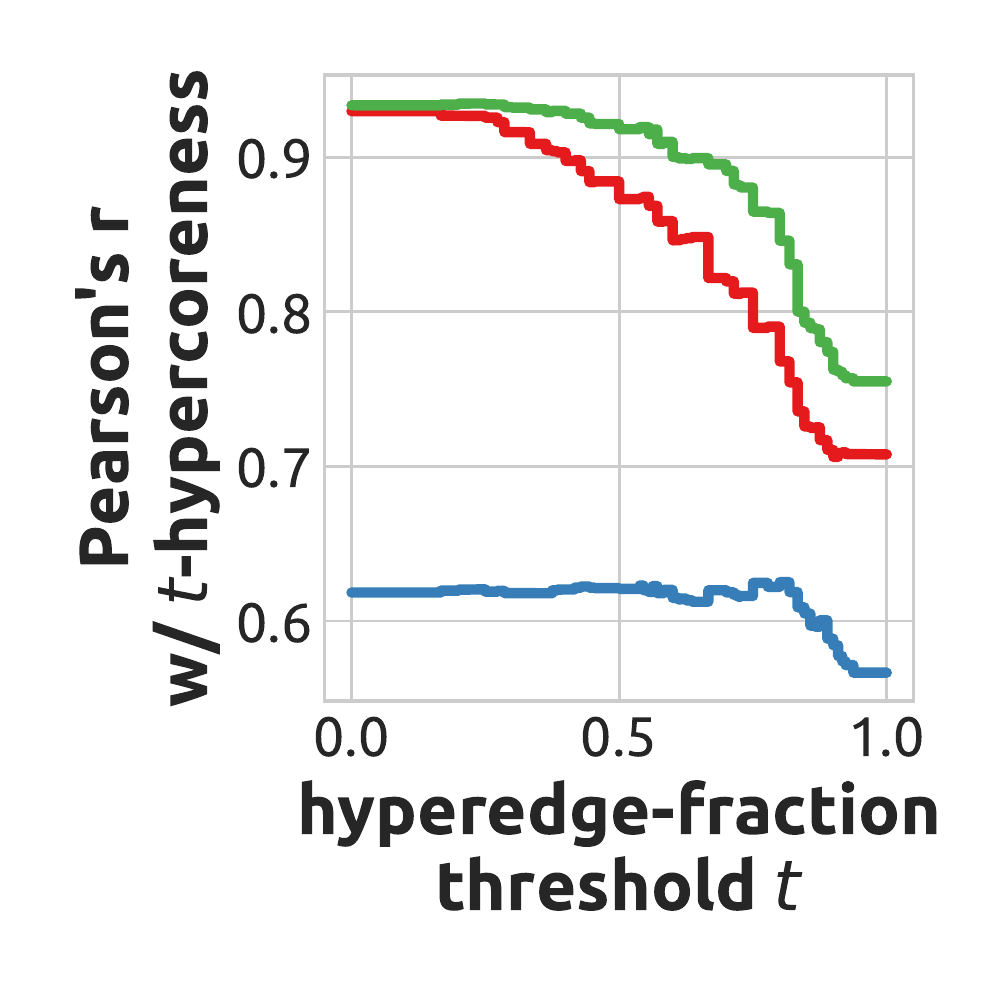}
		\includegraphics[scale=0.4]{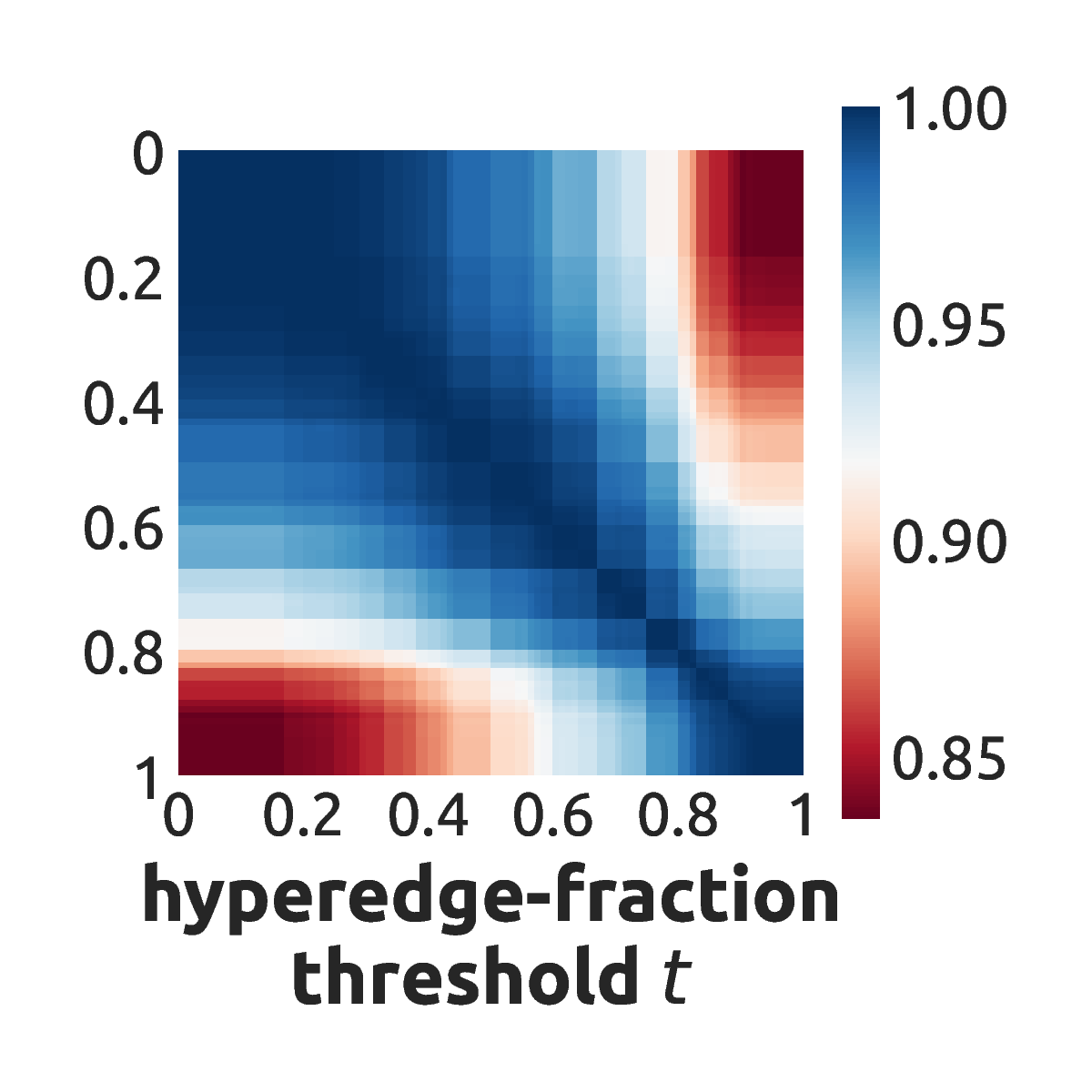}              
		\caption{email-Enron}
	\end{subfigure}
	\begin{subfigure}[b]{0.48\linewidth}
		\centering
		\hspace{-5mm}
		\includegraphics[scale=0.4]{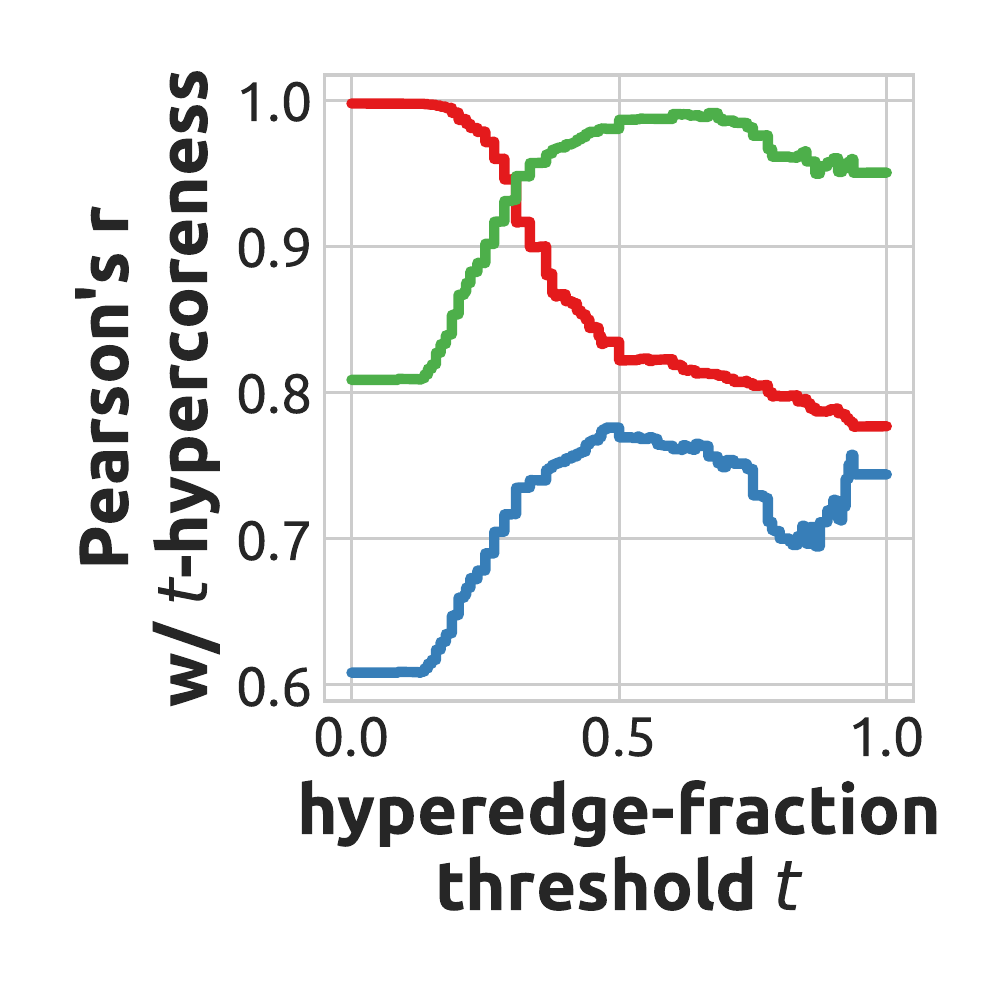}
		\includegraphics[scale=0.4]{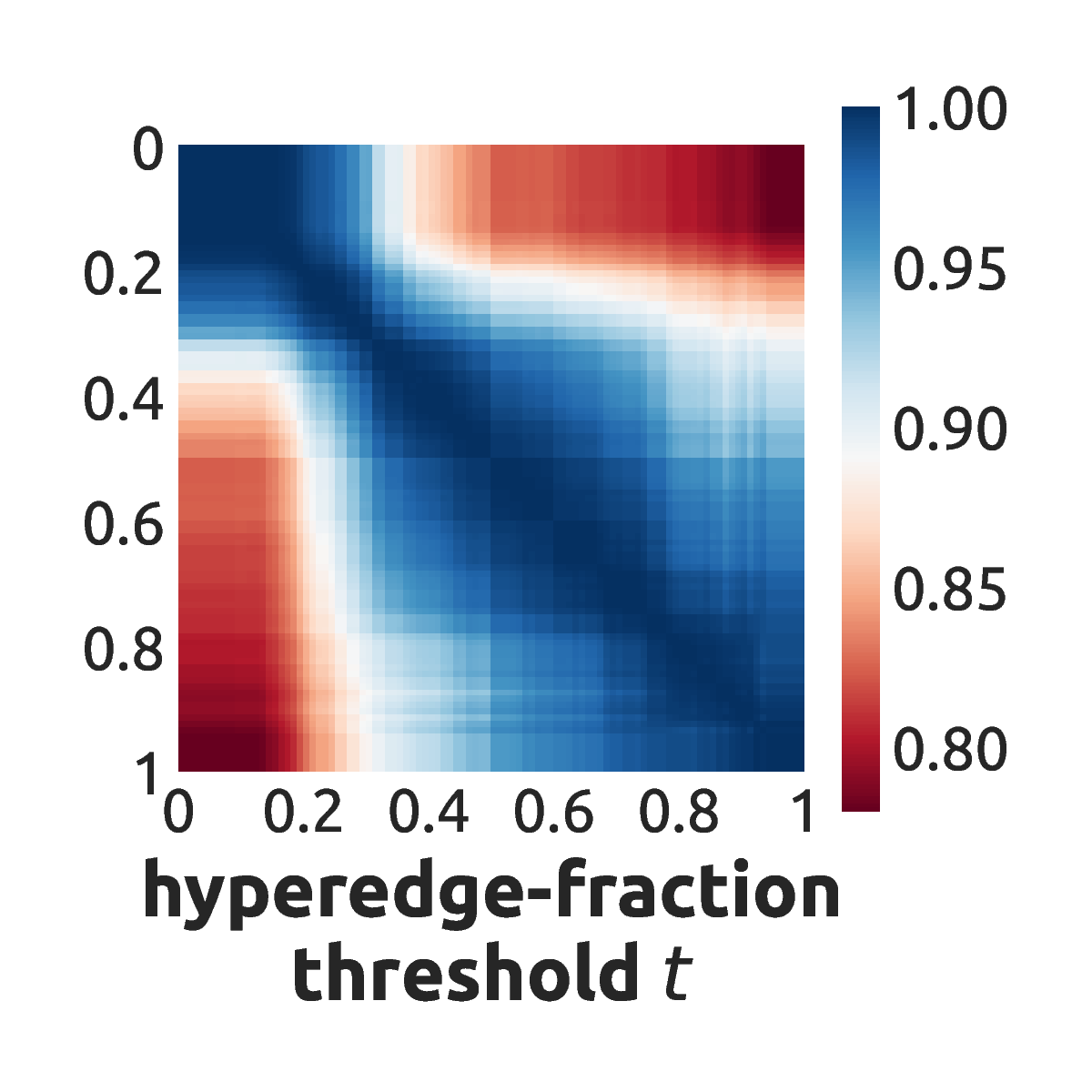}              
		\caption{NDC-classes}
	\end{subfigure}
	\begin{subfigure}[b]{0.48\linewidth}
		\centering
		\hspace{-5mm}
		\includegraphics[scale=0.4]{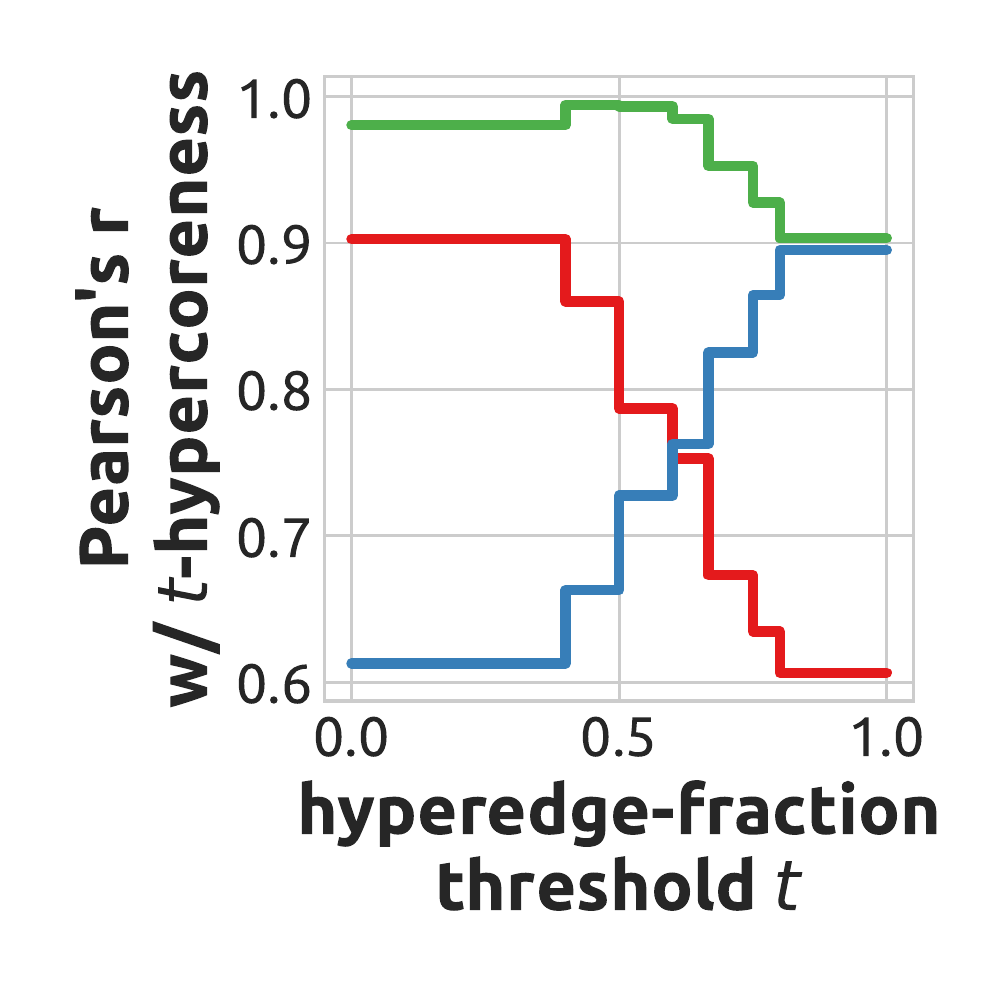}
		\includegraphics[scale=0.4]{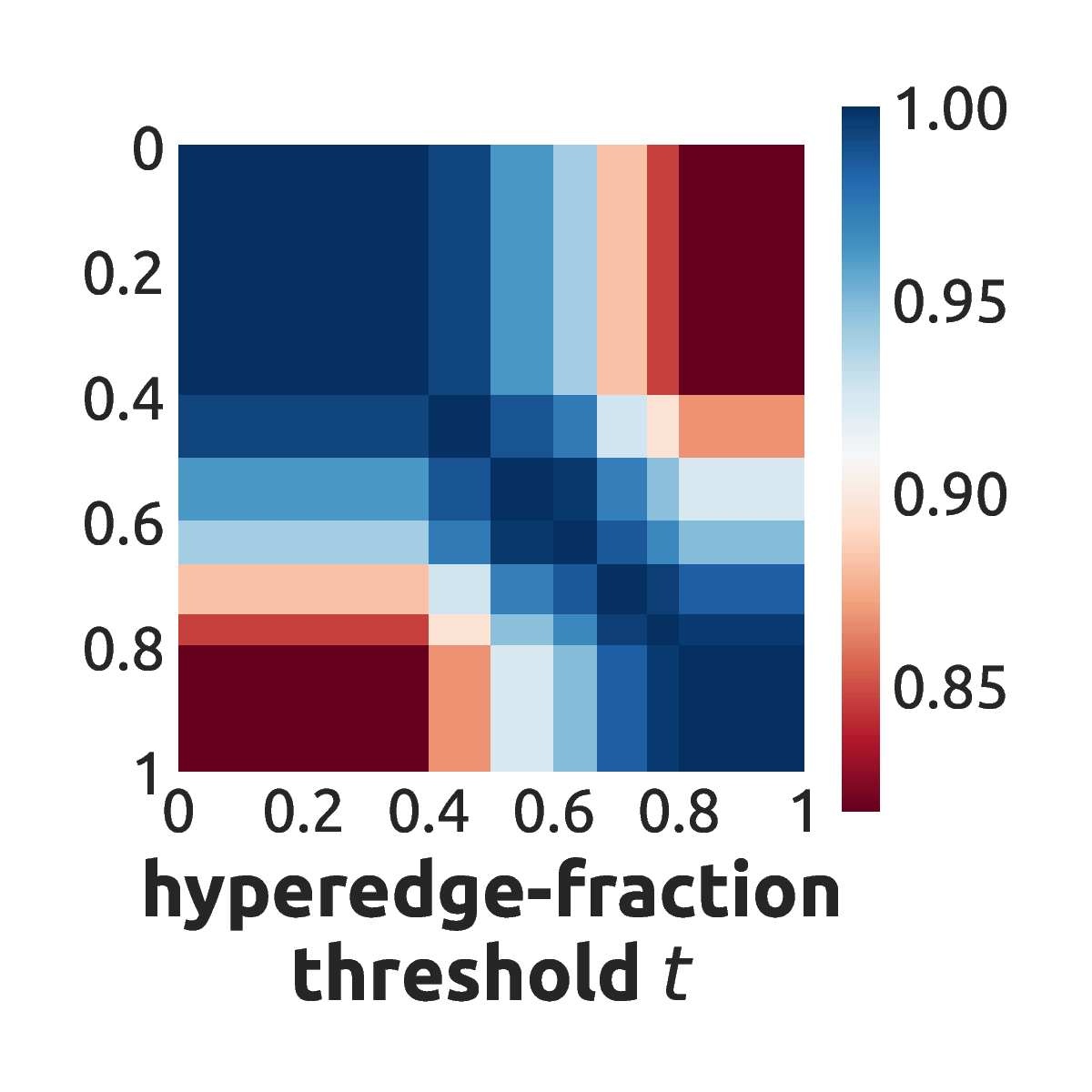}
		\caption{tags-ubuntu}
	\end{subfigure}    
	\begin{subfigure}[b]{0.48\linewidth}
		\centering
		\hspace{-5mm}
		\includegraphics[scale=0.4]{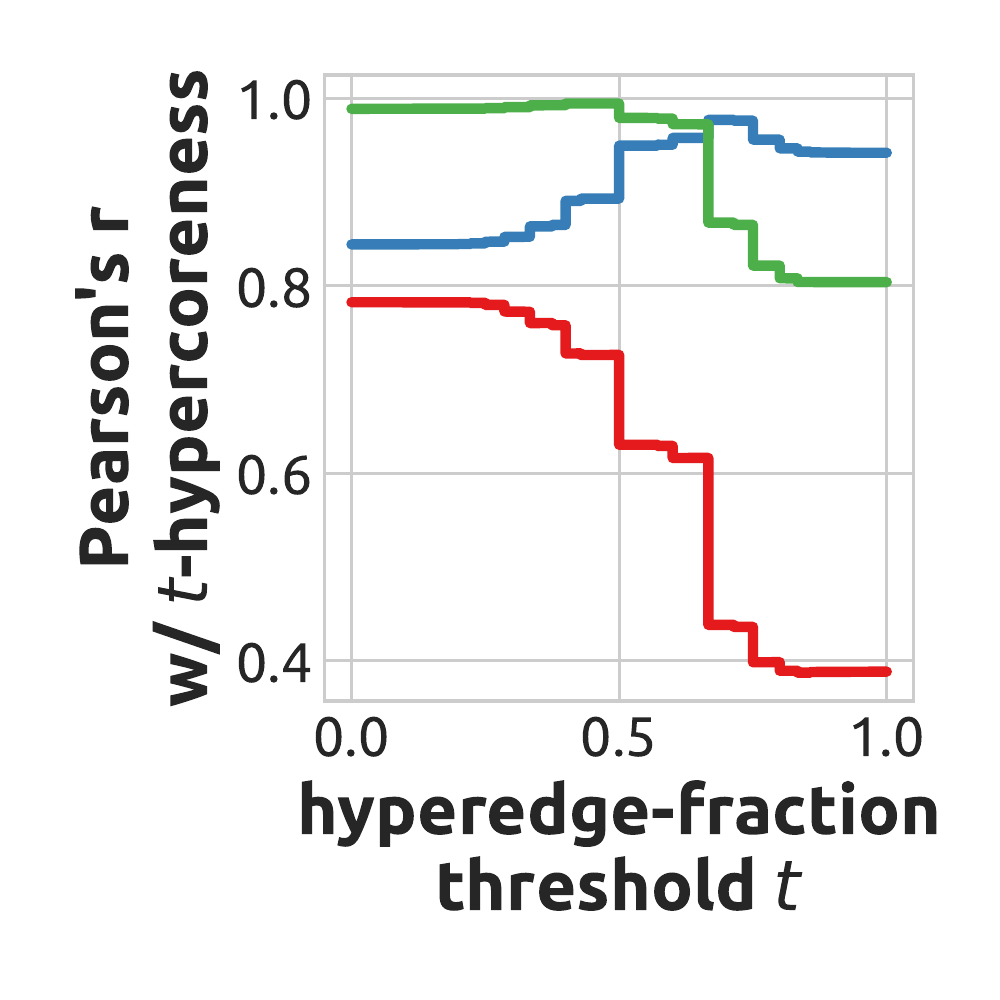}
		\includegraphics[scale=0.4]{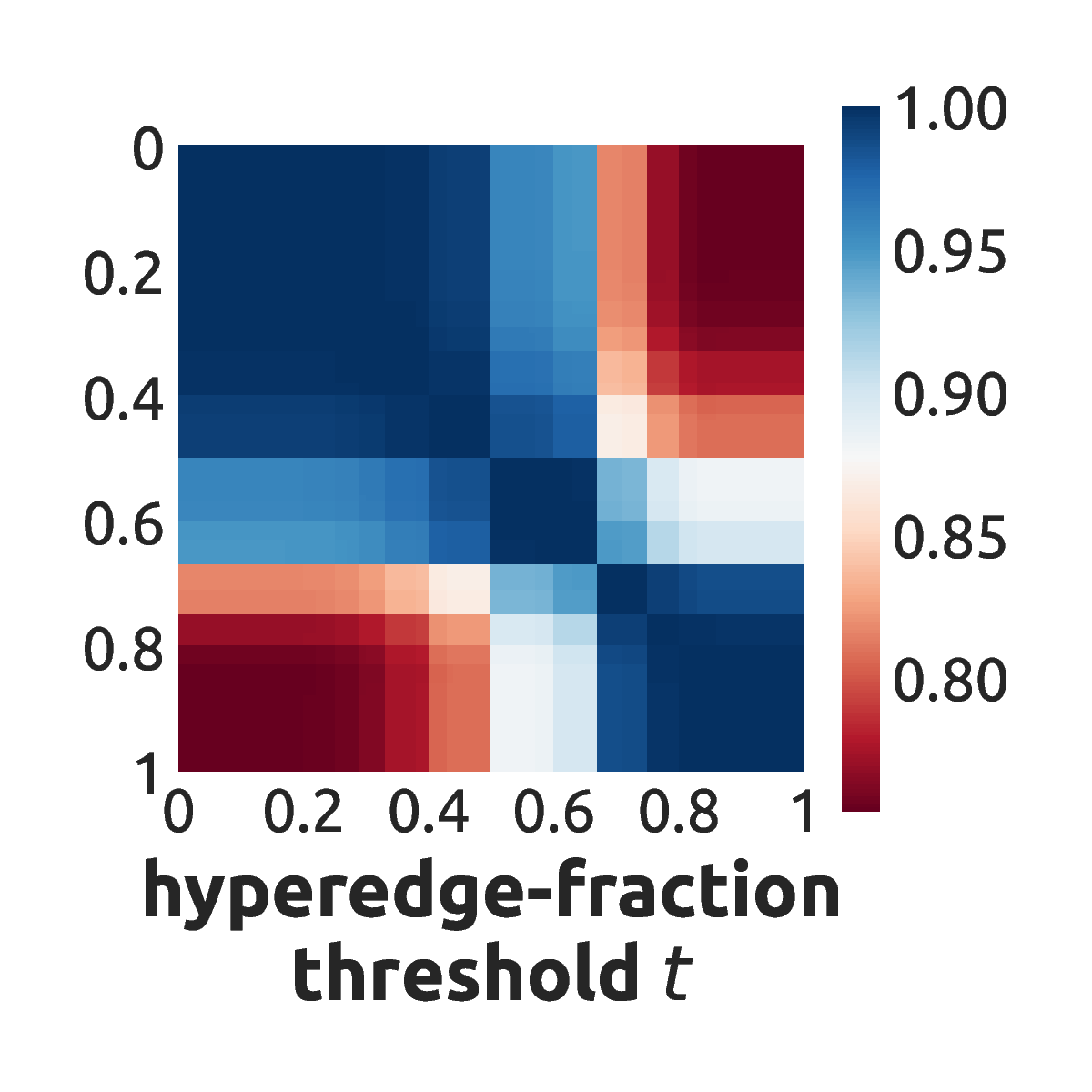}
		\caption{threads-math}
	\end{subfigure}
	\vspace{-1mm}
	\caption{\textbf{Statistical difference exists between $t$-hypercoreness and other centrality measures, as well as among $t$-hypercoreness with different $t$.}
		Left: the Pearson correlation coefficients between the $t$-hypercoreness sequences with different $t$ and each of the degree and coreness sequences in the unweighted (coreness-U) and weighted (coreness-W) clique expansions.
		Right: the Pearson correlation coefficient between each pair of $t$-hypercoreness sequences.
		\red{See Fig.~\ref{fig:rank_cor_sr} in Appendix~\ref{appendix:extra_exp} for the results on other datasets.}}
	\label{fig:rank_cor}
\end{figure*}

\subsection{Heterogeneity of $t$-hypercoreness}
We show that $t$-hypercoreness is statistically different from several existing centrality measures, and $t$-hypercoreness provides significantly different information depending on $t$.

\smallsection{Correlations.}
To show
(a) the distinctiveness of {$t$-}hypercoreness {from} existing centrality measures,
and
(b) the dissimilarity between {$t$-}hypercoreness with different $t$ values, we first measure {the Pearson correlation coefficients}.
In Fig.~\ref{fig:rank_cor}, we report Pearson's $r$ between the $t$-hypercoreness sequences with different $t$ values and each of the \textbf{degree} and \textbf{coreness} sequences in the \textit{unweighted} and \textit{weighted} clique expansions.
We also report Pearson's $r$ between each pair of $t$-hypercoreness sequences.
It is observed that even for the same hypergraph, the hypercoreness sequences with different $t$ values can be fairly dissimilar.

\begin{figure*}[t!]
	\centering
	\vspace{-1.2mm}
	\begin{subfigure}[b]{0.32\textwidth}
		\centering
		\includegraphics[scale=0.38]{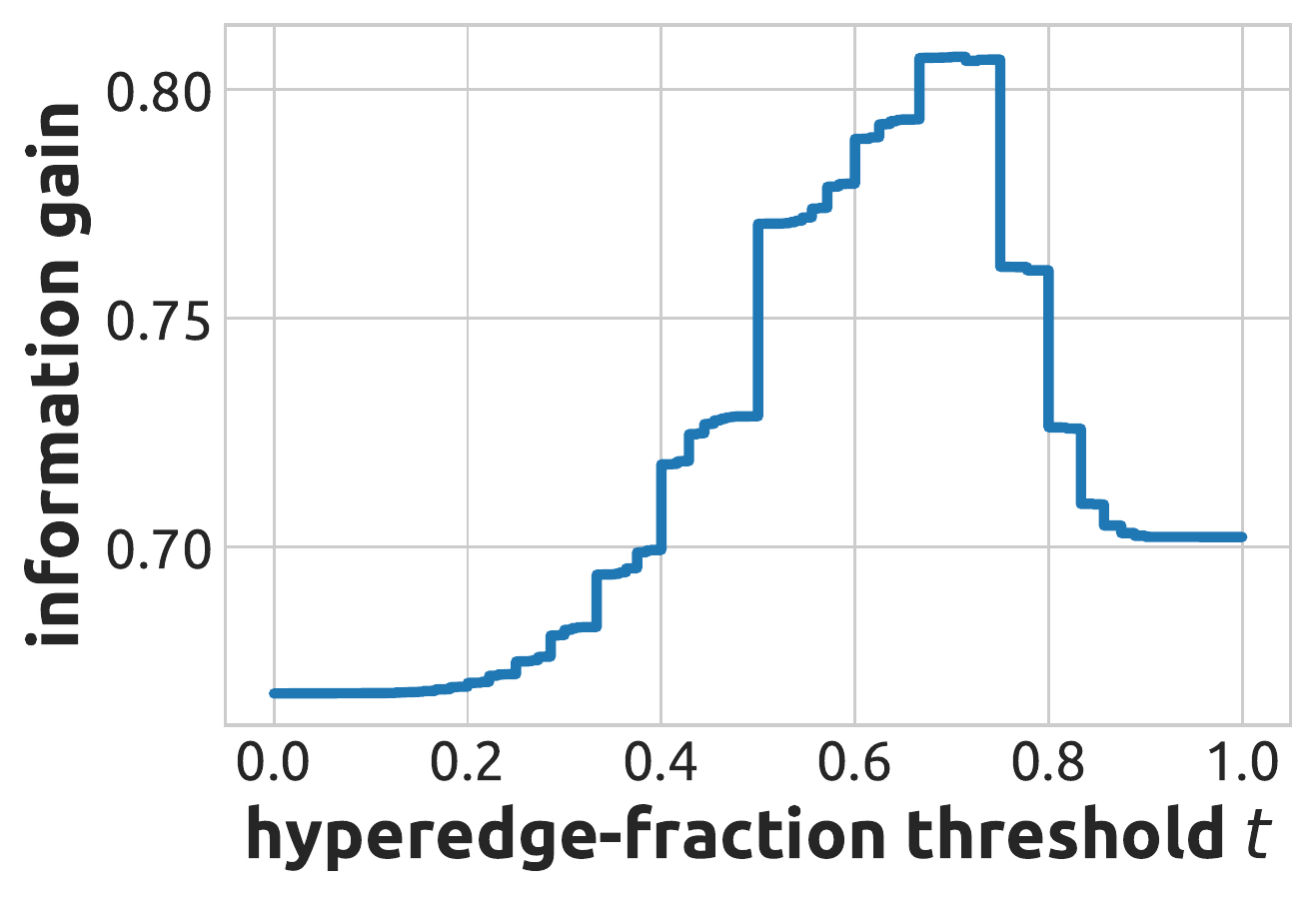}\\
		\caption{\Nq}
	\end{subfigure}
	\begin{subfigure}[b]{0.32\textwidth}
		\centering
		\includegraphics[scale=0.38]{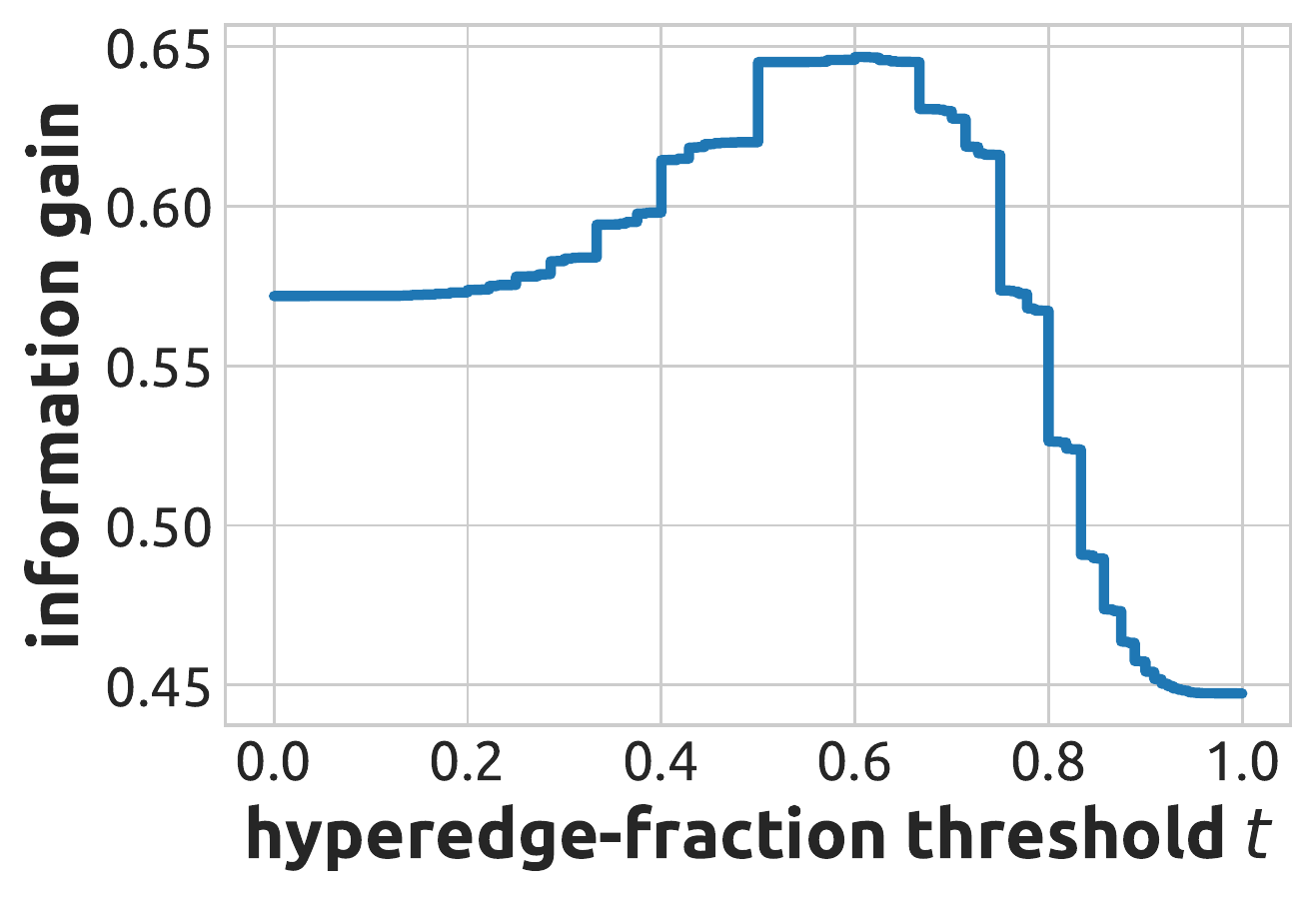}\\
		\caption{\Nw}
	\end{subfigure}
	\begin{subfigure}[b]{0.32\textwidth}
		\centering
		\includegraphics[scale=0.38]{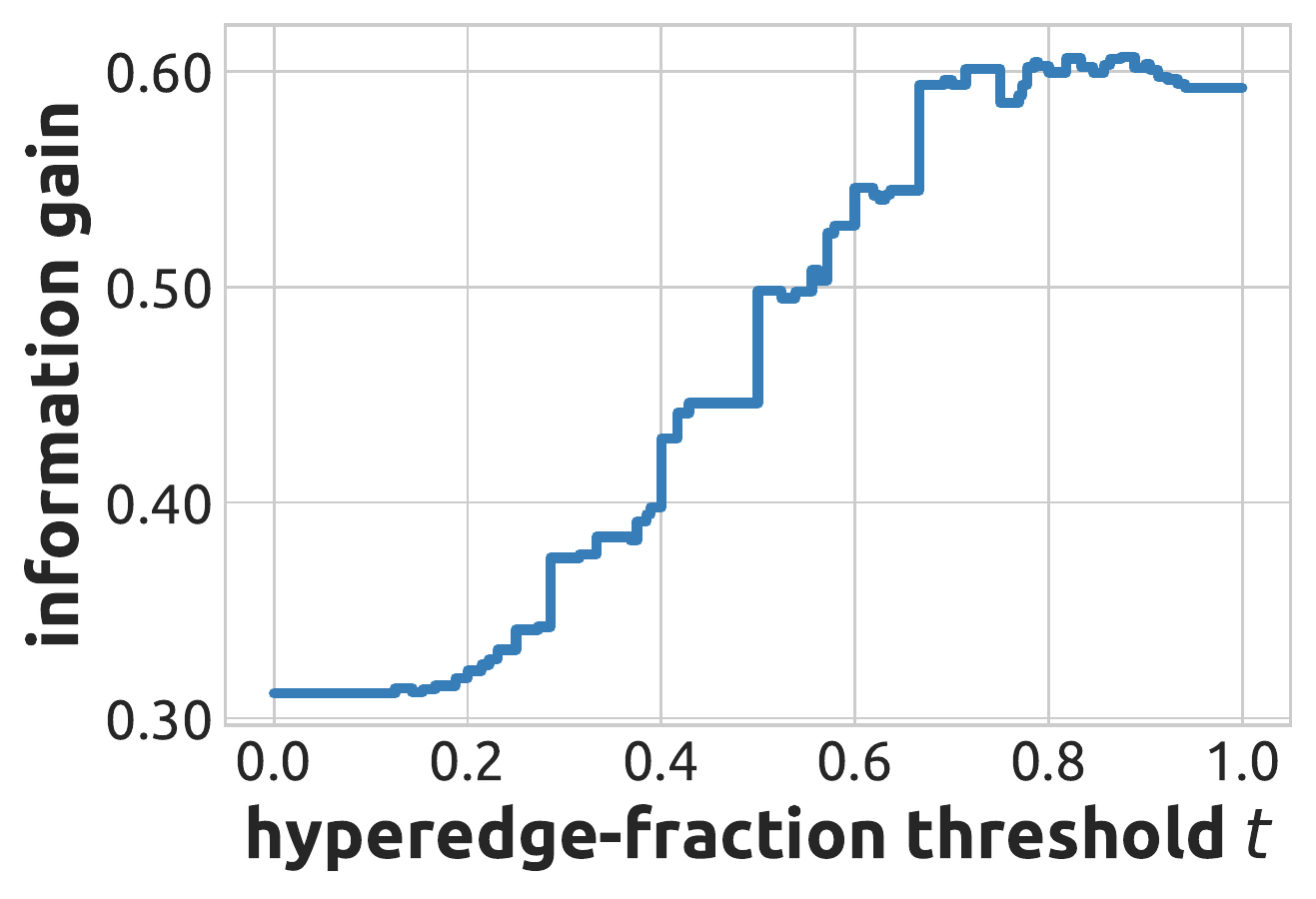}\\
		\caption{\Ne}
	\end{subfigure}
	\begin{subfigure}[b]{0.32\textwidth}
		\centering
		\includegraphics[scale=0.38]{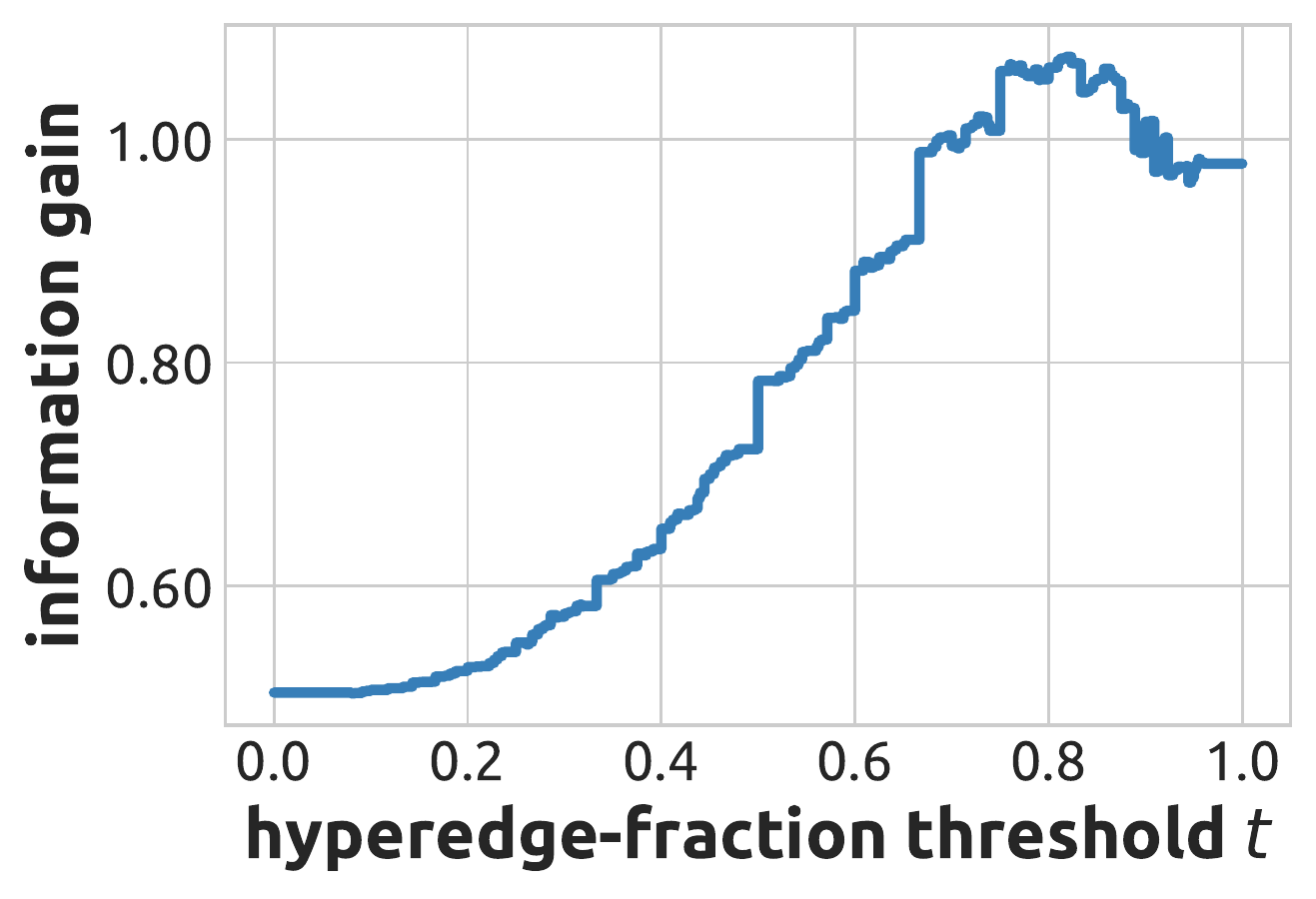}\\
		\caption{\Nr}
	\end{subfigure}
	\begin{subfigure}[b]{0.32\textwidth}
		\centering
		\includegraphics[scale=0.38]{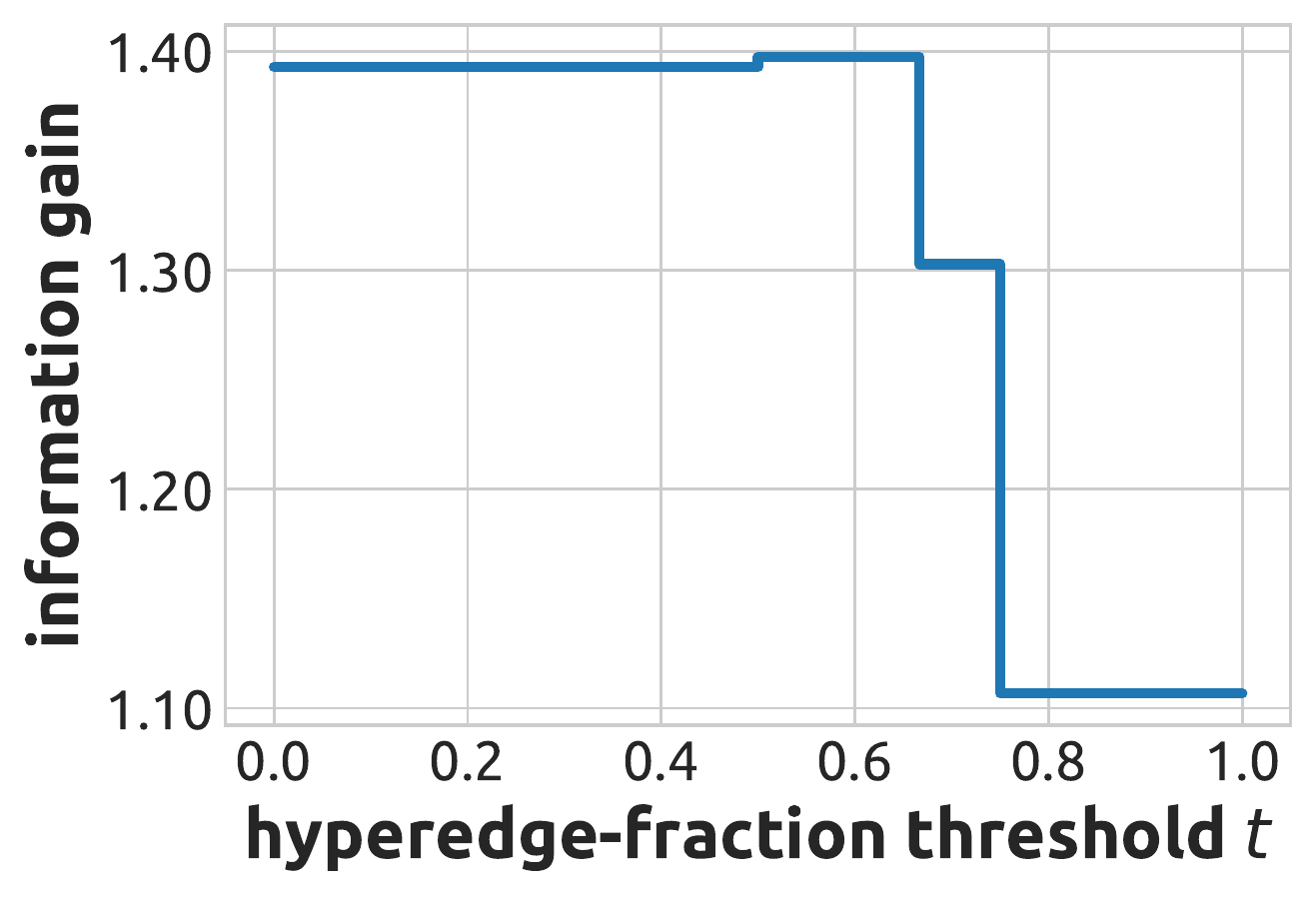}\\
		\caption{\Nt}
	\end{subfigure}        
	\begin{subfigure}[b]{0.32\textwidth}
		\centering
		\includegraphics[scale=0.38]{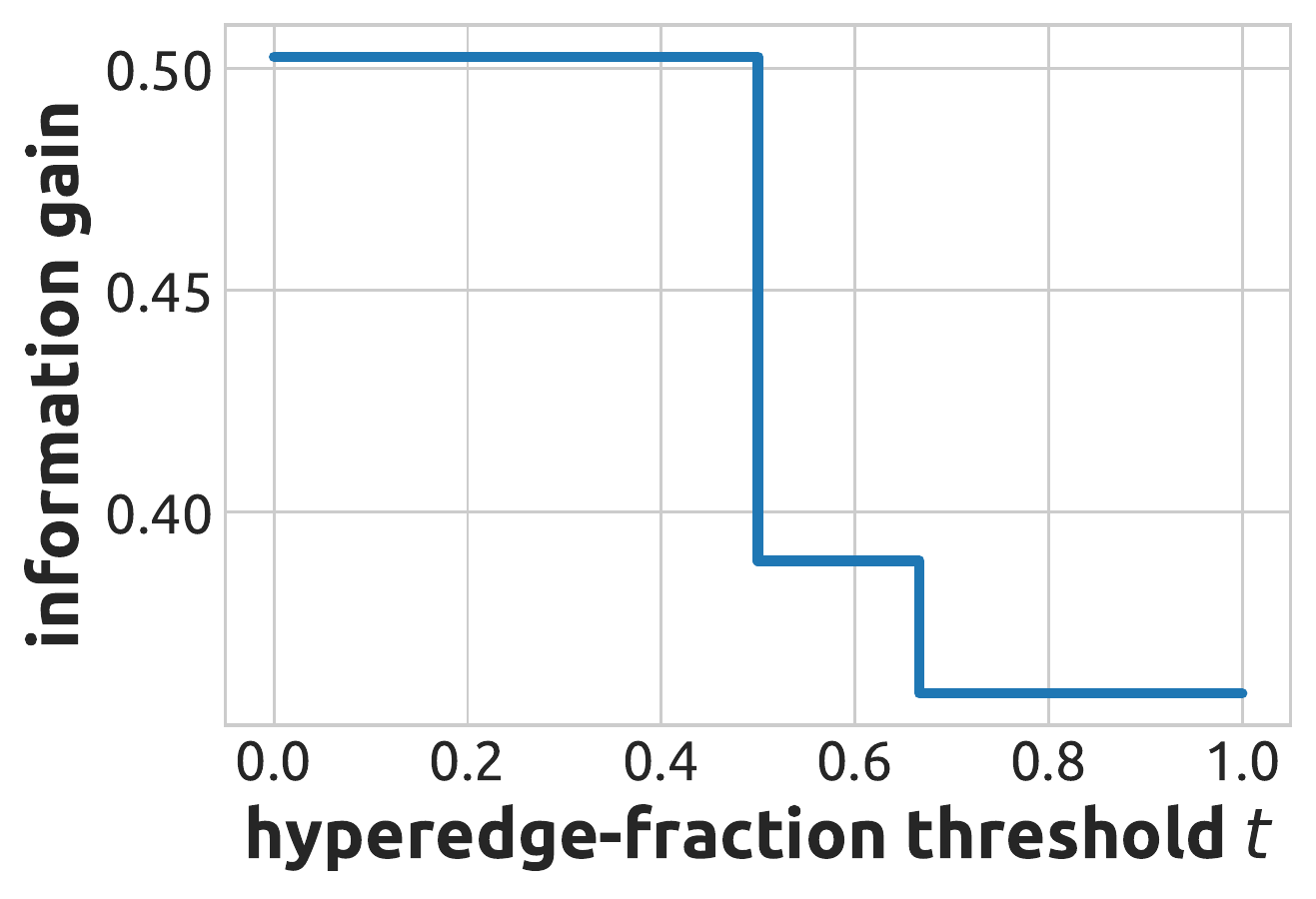}\\
		\caption{\Ny}
	\end{subfigure}
	\vspace{-0.3mm}
	\caption{\textbf{$t$-Hypercoreness has substantial information gain over degree, and it provides distinct information depending on $t$.}
		The average Pearson's $r$ between the information gain sequences is $0.232$ overall and $0.890$ within domains.
		The two values are significantly different with $5.0\mathrm{e}{-9}$ as the $p$-value of the $t$-test.
		\red{See Appendix~\ref{appendix:extra_exp} for the results on other datasets and  the results using other quantities.}}
	\label{fig:info_gain}
\end{figure*}

\smallsection{Information gain.}
We also show from the perspective of information theory that hypercoreness sequences with different $t$ values contain different information.
To this end, we define the \textit{information gain}.
\begin{definition}[Information gain \citep{quinlan1986induction}]\label{def:info_gain}
	Given $H = (V, E)$, for $i \in \bbN$, define $V_i \coloneqq \setbr{v \in V: d(v) = {i}}$ for $i \in \bbN$, and $V^t_i \coloneqq \setbr{v \in V: c_t(v) = i}$.
	The \textbf{information gain} of the $t$-hypercoreness sequence over the degree sequence is
	\begin{equation*}
		\mathcal{H}^t(H) \coloneqq - \sum_{i, j \in \bbN} \frac{\abs{V_i \cap V^t_j}}{n} \log_2 \frac{\abs{V_i \cap V^t_j}}{n} + \sum_{i \in \bbN} \frac{\abs{V_i}}{n} \log_2 \frac{\abs{V_i}}{n}.
	\end{equation*}
\end{definition}
The higher the information gain a hypercoreness sequence has, the more finely the nodes can be divided by the corresponding degree-hypercoreness pairs.
In Fig.~\ref{fig:info_gain}, we report the information gain for different $t$ values.
The highest information gain is achieved by different $t$ values in different datasets, and hypergraphs in the same domain show similar patterns.
In summary:
\begin{observation}[Heterogeneity of $t$-hypercoreness]
	In real-world hypergraphs, the $t$-hypercoreness of nodes provides statistically and information-theoretically distinct information depending on $t$.
\end{observation}

\section{Applications}\label{sec:apps}
In this section, we present some successful applications of our proposed concepts to demonstrate their usefulness.

\subsection{Influential-node identification}\label{subsec:inf_node_id}
It is well-known that in pairwise graphs, coreness is a good indicator of influential nodes \citep{kitsak2010identification}.
However, influential-node identification in hypergraphs is still underexplored, while some trials have been done \citep{zhu2018social, antelmi2021social}.
We use the SIR model, a widely-used epidemic model.
The model is straightforwardly generalized so that it can be used on hypergraphs, where the probability of a susceptible node being infected by the infected nodes in a hyperedge is proportional to the proportion of infected nodes in the hyperedge.
At each time step, each infected node recovers with a given probability ($\gamma$) independently.
We simulate the \textsc{hyperSIR} (see Alg.~\ref{alg:hyperSIR}) process assuming a single initially infected node.
In Alg.~\ref{alg:hyperSIR}, we show the process of \textsc{hyperSIR}.
In the relatively large datasets (coauth-DBLP, coauth-Geology, and threads-SO), we randomly draw $10\%$ of the nodes, and perform the simulation $100$ times for each seed node.
In the other datasets, we simulate $10,000$ times for each node as the seed. 

\begin{algorithm}[t!]
	\small
	\caption{\textsc{hyperSIR}} \label{alg:hyperSIR}
	\hspace*{\algorithmicindent} \textbf{Input:} {$H = (V, E)$, seed node $v^*$, transmission rate $\beta$, and recovery rate $\gamma$} \\
	\hspace*{\algorithmicindent} \textbf{Output:} {number of ever-infected nodes $\abs{R}$}
	\begin{algorithmic}[1]
		\State $S \leftarrow V \setminus \setbr{v^*}; I \leftarrow \setbr{v^*}; R \leftarrow \varnothing$
		\While{$I \neq \varnothing$}
		\State $P_s(v_s) \leftarrow 1, \forall v_s \in S$
		\For{\upshape\textbf{each} $e \in E$ \upshape\textbf{s.t.} $e \cap I \neq \varnothing \land e \cap S \neq \varnothing$}
		\State $I_e \leftarrow e \cap I; S_e \leftarrow e \cap S$
		\State $P_s(u) \leftarrow P_s(u) (1 - 2 \beta \abs{I_e} / \abs{e}), \forall u \in S_e$ \label{line:linear_prob_hyperSIR} 
		\EndFor
		\State $v_i$ moves to $R$ with probability $\gamma$, $\forall v_i \in I$
		\State $v_s$ moves to $I$ with probability $1 - P_s(v_s)$, $\forall v_s \in S$
		\EndWhile
		\State \Return {$\abs{R}$}
	\end{algorithmic}
\end{algorithm}

We investigate the relations between the average number of ever-infected nodes and the following quantities of the seed node in addition to \bolden{$t$-hypercoreness} and \textbf{degree}:
\begin{itemize}
	\item \red{\textbf{Neighbor-hypercoreness}~\citep{arafat2023neighborhood}: see Def.~\ref{def:nbr_hypercoreness} (abbreviation: nbr-hypercoreness);}
	\item \red{\textbf{Neighbor-degree-hypercoreness}~\citep{arafat2023neighborhood}: see Def.~\ref{def:nbr_deg_hypercoreness} (abbreviation: nd-hypercoreness);}	
	\item \textbf{Coreness} in the unweighted (coreness-U) / weighted (coreness-W) clique expansion;
	\item \textbf{Eigencentrality} in unweighted (eigencentrality-U)/ weighted (eigencentrality-W) clique expansion;
	\item \textbf{Hyper-eigencentrality}~\citep{tudisco2021node}: three different versions, linear (hyperEC-L), log-exp (hyperEC-LE), and max (hyperEC-M);
	\item \red{\bolden{$\ell$-hypercoreness}~\citep{limnios2021hcore}: see Def.~\ref{def:l_hypercoreness}.\footnote{Recall that $\ell$-hypercoreness with $\ell = 2$ is include in $t$-hypercoreness with $t = 0$. For each dataset, we apply min-max normalization to all the possible $\ell$ values with $\ell \geq 3$ so that $t$-hypercoreness and $\ell$-hypercoreness can fit in the same $x$-axis with the range $[0, 1]$.}}
\end{itemize}

We also consider two supervised machine-learning methods.
Specifically, we apply \textbf{node2vec}~\citep{grover2016node2vec} to the unweighted clique expansion of each dataset,
and we apply a self-supervised hypergraph learning method \textbf{TriCL}~\citep{lee2022m} (which is based on the architecture proposed by~\cite{feng2019hypergraph}) directly to the original hypergraphs.
Both additional baseline methods output node embeddings of dimension $128$.
For each dataset, we sample $10\%$ of the nodes (for the three relatively large datasets where we only use $10\%$ of the nodes, we sample $1\%$ of the total nodes, i.e., $10\%$ of the $10\%$) uniformly at random and provide the ground-truth influence of the sampled nodes.\footnote{The average performance over five independent trials is reported.}
For both methods, we apply linear regression using the node embeddings as the features, and then we use the fitted linear regression model to predict the influence of the nodes.
Due to the scalability issues, results of them are unavailable on some large datasets.
\color{black}

\begingroup
\renewcommand{\arraystretch}{1.5}
\begin{table}[t!]
	\caption{Results of $t$-hypercoreness by choosing the $t$ values based on sampled nodes.
    Sampled nodes: the results where the $t$ values are chosen based on sampled nodes.
    Ground-truth best: the results where for each dataset, the best $t$ value is chosen among the candidate values.
    Best $t$: the most indicative $t$ value (in each of the five trials, or among the candidate values).
    Perm.: the average performance (the Pearson correlation coefficient; the higher the better) over the five trials.
    Rank: the rank (the lower the better) among all the baseline methods and each considered one (i.e., the result based on sampled nodes or using the ground-truth best $t$ value).
 } \label{tab:inf_node_id_suppl_res}
	\centering
	\resizebox{0.7\columnwidth}{!}{%
		\begin{tabular}{l|ccc|ccc}
			\hline
                & \multicolumn{3}{c|}{sampled nodes} & \multicolumn{3}{c}{ground-truth best} \\                
			dataset & best $t$ & perm. & rank & best $t$ & perm. & rank \\
			\hline		
                coauth-DBLP & $(\frac{1}{2}, \frac{1}{2}, \frac{1}{2}, \frac{1}{2}, \frac{1}{2})$ & $0.927 \pm 0.000$ & 1 & $\frac{1}{2}$ & $0.927$ & 1 \\
                coauth-Geology & $(\frac{1}{2}, \frac{1}{2}, \frac{1}{2}, \frac{1}{2}, \frac{1}{2})$ & $0.930\pm0.000$ & 1 & $\frac{1}{2}$ & $0.930$ & 1 \\
                NDC-classes & $(\frac{1}{2}, \frac{1}{2}, \frac{1}{2}, \frac{2}{3}, \frac{2}{3})$ & $0.939\pm0.001$ & 2 & $\frac{1}{2}$ & $0.940$ & 2 \\
                NDC-substances & $(\frac{2}{3}, \frac{2}{3}, \frac{2}{3}, \frac{2}{3}, \frac{2}{3})$ & $0.959\pm0.000$ & 1 & $\frac{2}{3}$ & $0.959$ & 1 \\
                contact-high & $(0, 0, 0, 0, 0)$ & $0.947\pm0.000$ & 1 & $0$ & $0.947$ & 1 \\
                contact-primary & $(0, 0, \frac{2}{3}, \frac{2}{3}, 1)$ & $0.970\pm0.007$ & 1 & $\frac{2}{3}$ & $0.975$ & 1\\
                email-Enron & $(\frac{2}{3}, \frac{2}{3}, \frac{2}{3}, \frac{2}{3}, \frac{2}{3})$ & $0.960\pm0.000$ & 1 & $\frac{2}{3}$ & $0.960$ & 1 \\
                email-Eu & $(\frac{1}{2}, \frac{1}{2}, \frac{1}{2}, \frac{2}{3}, \frac{2}{3})$ & $0.975\pm0.003$ & 1 & $\frac{2}{3}$ & $0.977$ & 1 \\
                tags-ubuntu & $(1, 1, 1, 1, 1)$ & $0.970\pm0.000$ & 2 & $1$ & $0.970$ & 2 \\
                tags-math & $(1, 1, 1, 1, 1)$ & $0.990\pm0.000$ & 1 & $1$ & $0.990$ & 1 \\
                tags-SO & $(1, 1, 1, 1, 1)$ & $0.844\pm0.000$ & 2 & $1$ & $0.844$ & 2 \\
                threads-ubuntu & $(\frac{2}{3}, \frac{2}{3}, \frac{2}{3}, \frac{2}{3}, \frac{2}{3})$ & $0.938\pm0.000$ & 1 & $\frac{2}{3}$ & $0.938$ & 1 \\
                threads-math & $(\frac{2}{3}, \frac{2}{3}, \frac{2}{3}, \frac{2}{3}, \frac{2}{3})$ & $0.971\pm0.000$ & 1 & $\frac{2}{3}$ & $0.971$ & 1 \\
                threads-SO & $(\frac{2}{3}, \frac{2}{3}, \frac{2}{3}, \frac{2}{3}, \frac{2}{3})$ & $0.962\pm0.000$ & 1 & $\frac{2}{3}$ & $0.962$ & 1 \\
			\hline
		\end{tabular}%
	}
\end{table}
\endgroup

We take the largest connected component of each dataset, as in previous works on pairwise graphs~\citep{kitsak2010identification}.
For simplicity, we use $\gamma = 1$, and choose $\beta \in \setbr{0.05, 0.025, 0.01, 0.005, 0.0025}$ to avoid the cases when almost all seed nodes lead to similar results.
For the \textit{email-Eu} dataset, Fig.~\ref{fig:inf_email_Eu} shows the detailed relations between the average number of ever-infected nodes (i.e., final $\abs{R}$) and each of the aforementioned quantities, with the best-fitted lines.
Fig.~\ref{fig:inf_summary} shows the Pearson correlation coefficient between the average number of ever-infected nodes and each quantity.
The comparison between $t$-hypercoreness and the coreness in clique expansions validates the information loss brought by the clique expansions.
On most of the datasets, at least one among the $t$-hypercoreness with $t \in \setbr{0, \frac{1}{2}, \frac{2}{3}, 1}$ works better than all the other baseline methods.
\red{On the remaining datasets, $t$-hypercoreness with a proper $t$ value ranks second.}
\red{Moreover, even if we always use the $t$-hypercoreness with $t = \frac{1}{2}$, $t$-hypercoreness still outperforms all the baseline methods on 10 out of 14 datasets.}
\red{In practice, we may sample a small number of nodes and choose the $t$ value that is most influence-indicative (w.r.t the Pearson correlation coefficient) on the sampled set of nodes.
For this purpose, we use the same $10\%$ nodes (or $1\%$ for some large datasets) that are used as a training set for the machine-learning methods.
In Table~\ref{tab:inf_node_id_suppl_res}, for each dataset, we show 
(1) the most indicative $t$ value in each of the five trials,
(2) the performance and rank of $t$-hypercoreness averaged on the five trials,\footnote{We count $\ell$-hypercoreness with each $\ell$ value as a separate method ($\ell = 2$ is not counted since it is already included in the concept of $t$-hypercoreness with $t = 0$).}
and (3) the performance and rank of $t$-hypercoreness with the best $t$ values among the four candidate values. 
We can observe that a well-performing $t$ value is always found (although the chosen $t$ values may vary), and $t$-hypercoreness performs well and stably, almost always outperforming all the baselines.}

\begin{observation}[Influence indicativeness of $t$-hypercoreness] \label{obs:influence}
	In real-world hypergraphs, $t$-hypercoreness identifies influential nodes well.
	In most cases, $t$-hypercoreness with a proper $t$ is the best indicator of influential nodes among all considered centrality measures.
	In different hypergraphs, the $t$ value maximizing the correlation between $t$-hypercoreness and node influence varies, and in most cases, such $t$ is neither $0$ nor $1$.
\end{observation}

\begin{figure*}[t!]
	\centering	
	\begin{subfigure}[b]{\textwidth}
		\centering
		\includegraphics[scale=0.5]{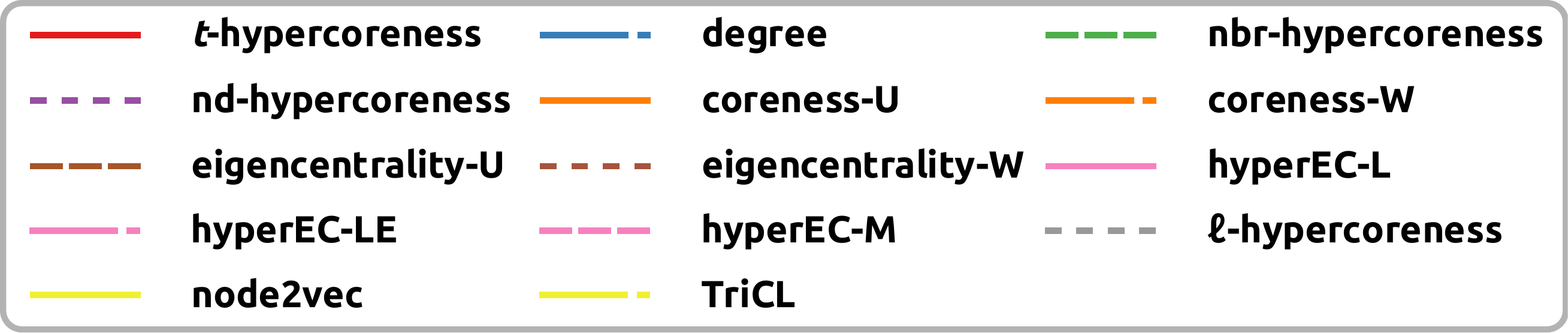}\\        
	\end{subfigure}    
	\vspace{1mm}
	\begin{subfigure}[b]{0.32\textwidth}
		\centering
		\includegraphics[scale=0.4]{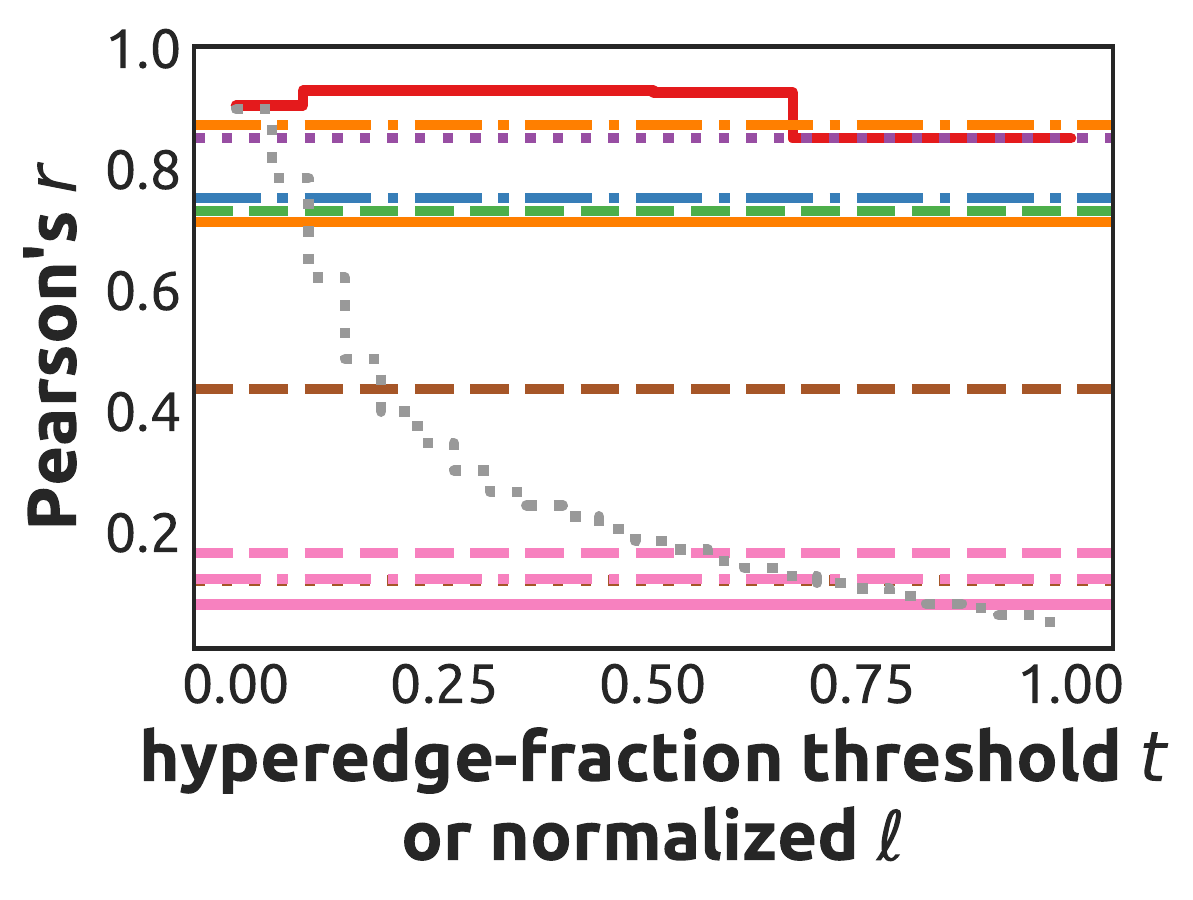}
		\caption{\Nq}
	\end{subfigure}
	\begin{subfigure}[b]{0.32\textwidth}
		\centering
		\includegraphics[scale=0.4]{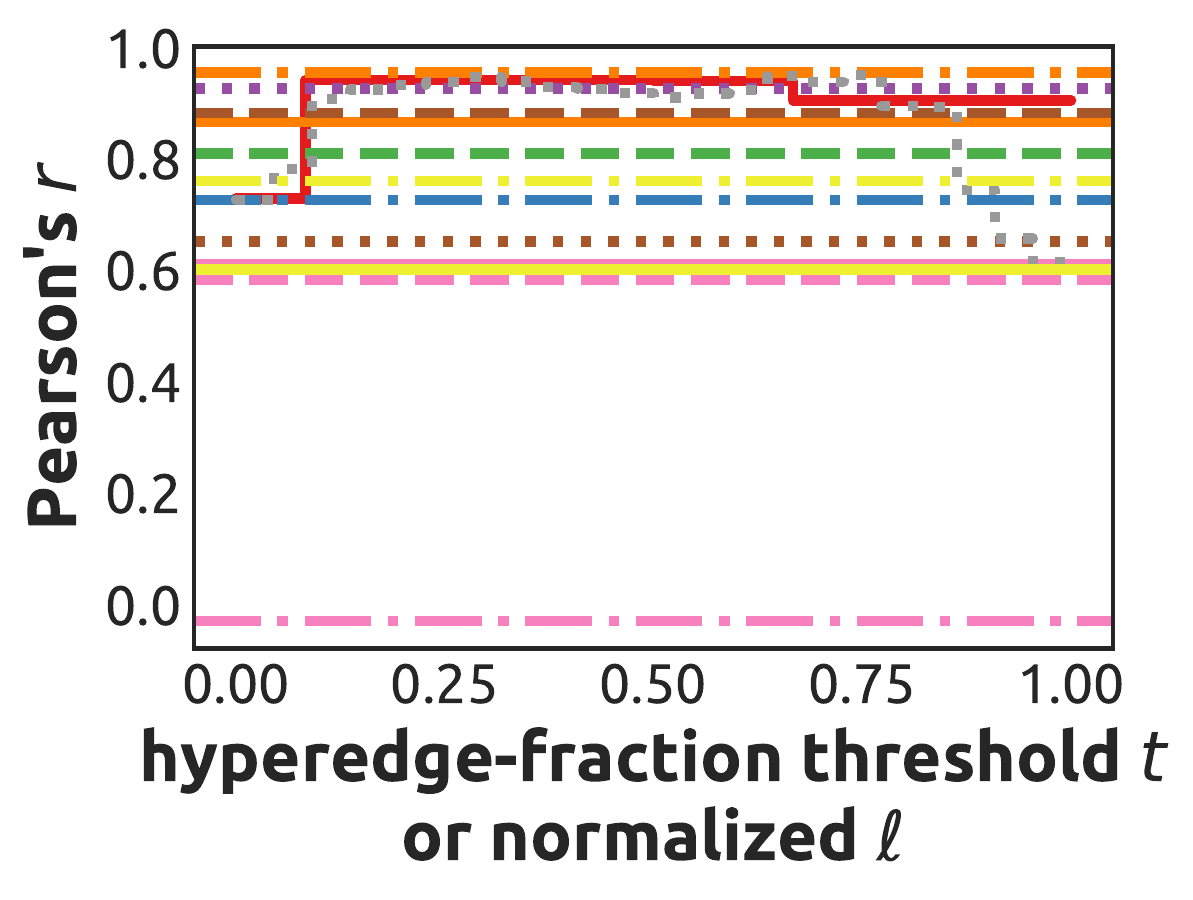}
		\caption{\Ne}
	\end{subfigure}
	\begin{subfigure}[b]{0.32\textwidth}
		\centering
		\includegraphics[scale=0.4]{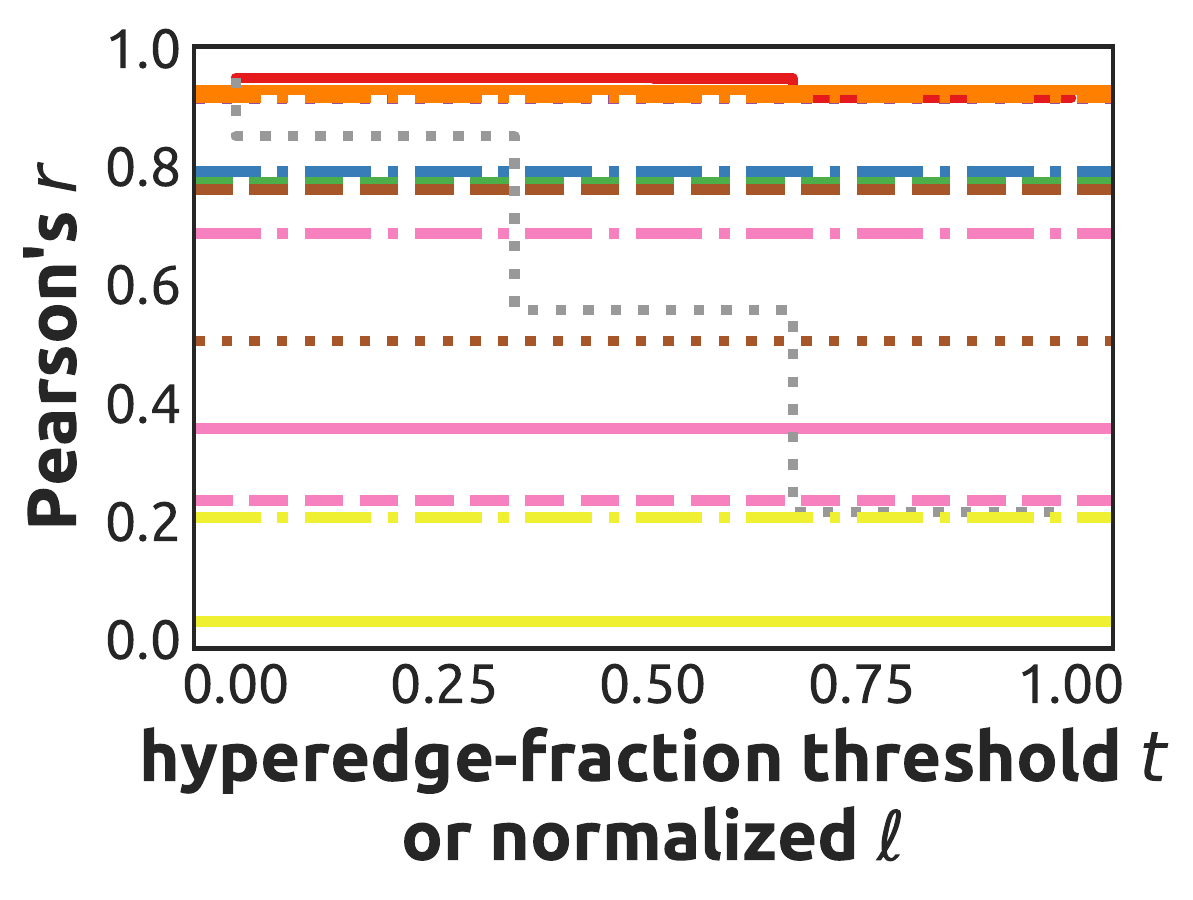}
		\caption{\Nt}
	\end{subfigure}
	\begin{subfigure}[b]{0.32\textwidth}
		\centering
		\includegraphics[scale=0.4]{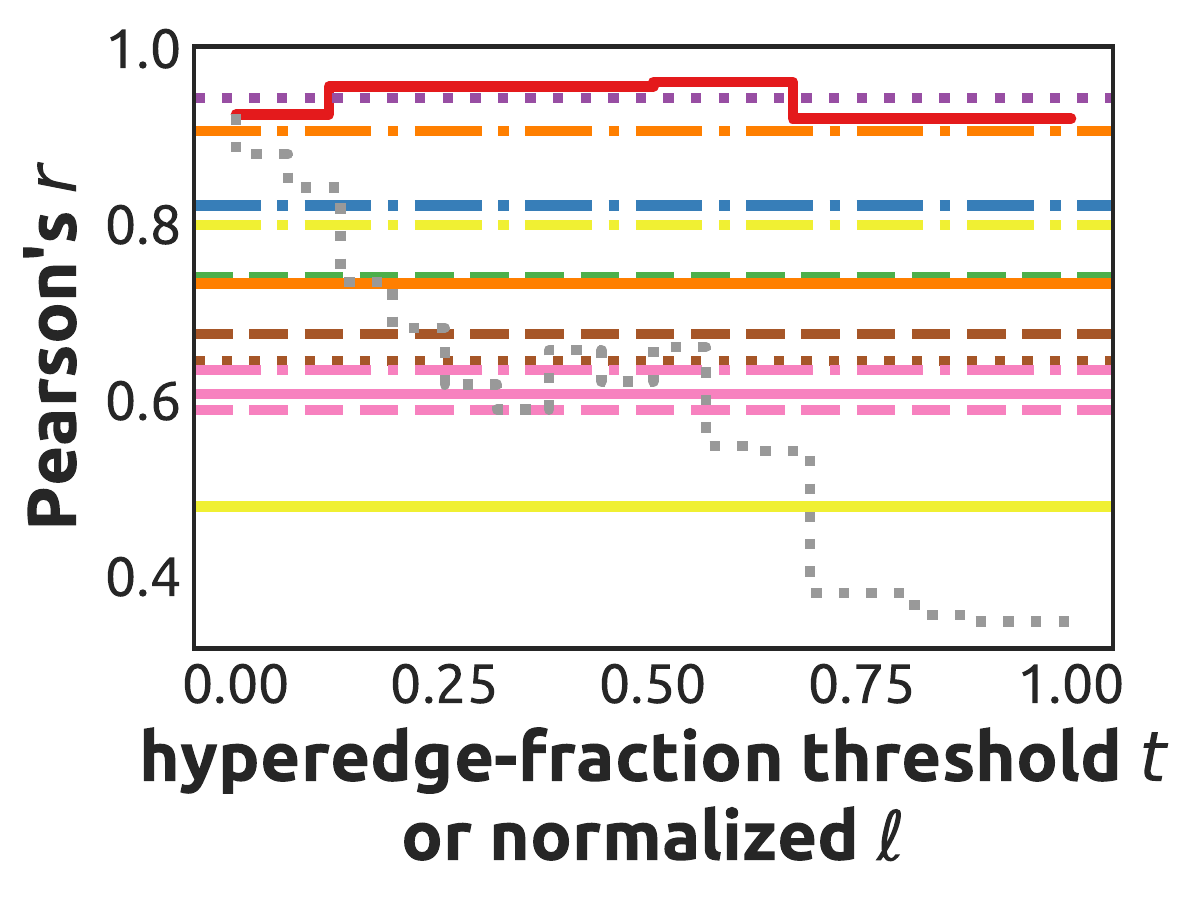}
		\caption{\Nu}
	\end{subfigure}
	\begin{subfigure}[b]{0.32\textwidth}
		\centering
		\includegraphics[scale=0.4]{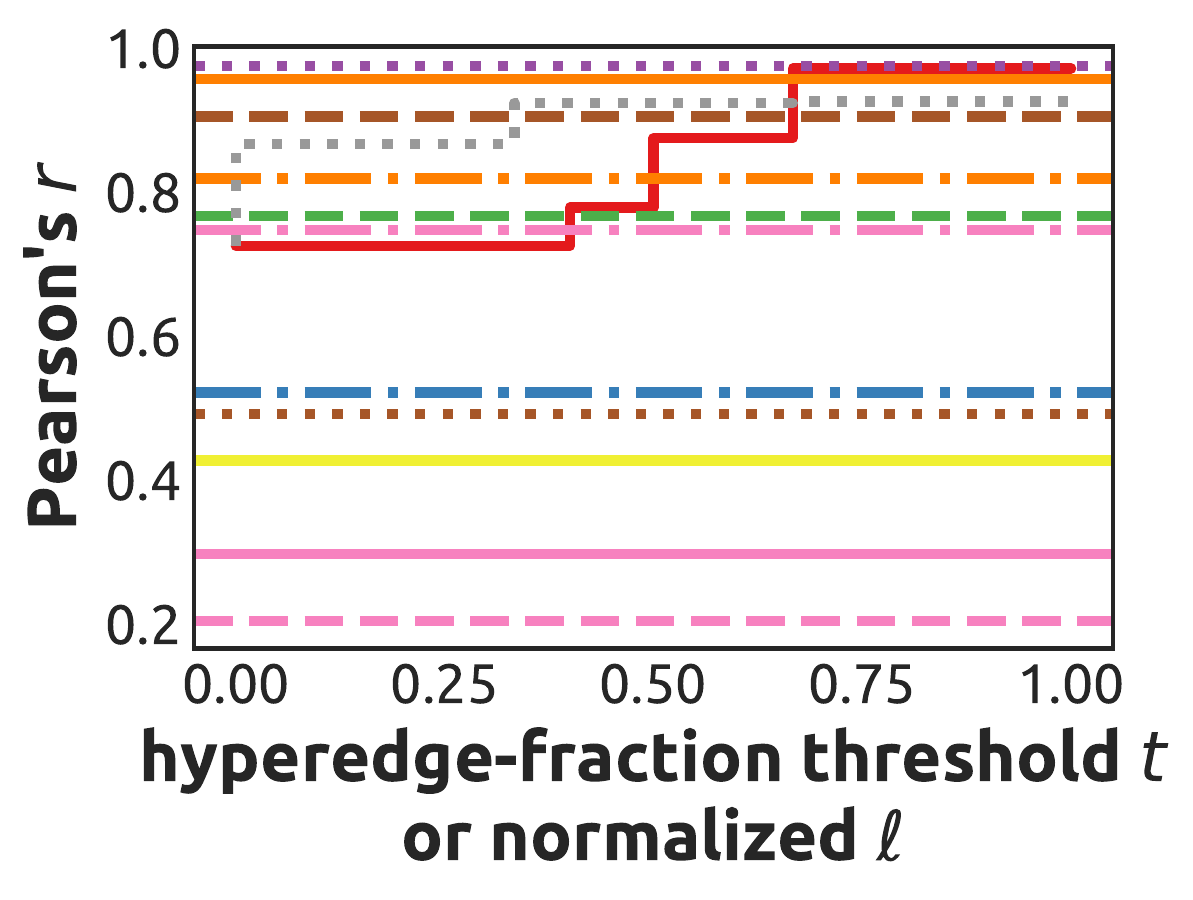}
		\caption{tags-ubuntu}
	\end{subfigure}
	\begin{subfigure}[b]{0.32\textwidth}
		\centering
		\includegraphics[scale=0.4]{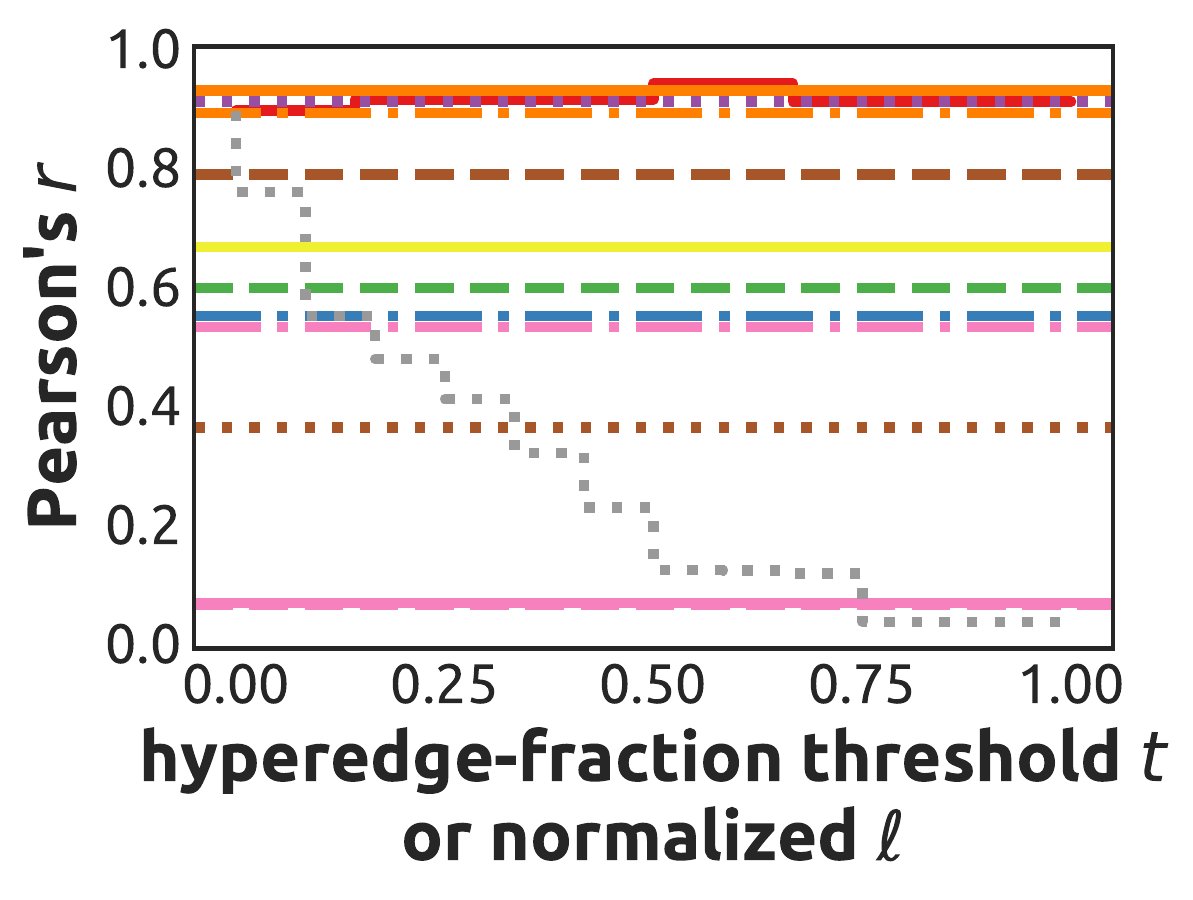}
		\caption{threads-ubuntu}
	\end{subfigure}
	\caption{\textbf{$t$-Hypercoreness is consistently indicative of influence in all datasets.} 
		We show the Pearson correlation coefficients between the average number of ever-infected nodes and each of the considered quantity of the seed node.
		See Fig.~\ref{fig:inf_summary_sr} in Appendix~\ref{appendix:extra_exp} for the results on other datasets.}
	\label{fig:inf_summary}
\end{figure*}

\begin{figure*}[t!]
	\centering
	\begin{subfigure}[b]{0.32\linewidth}
		\centering
		\includegraphics[scale=0.4]{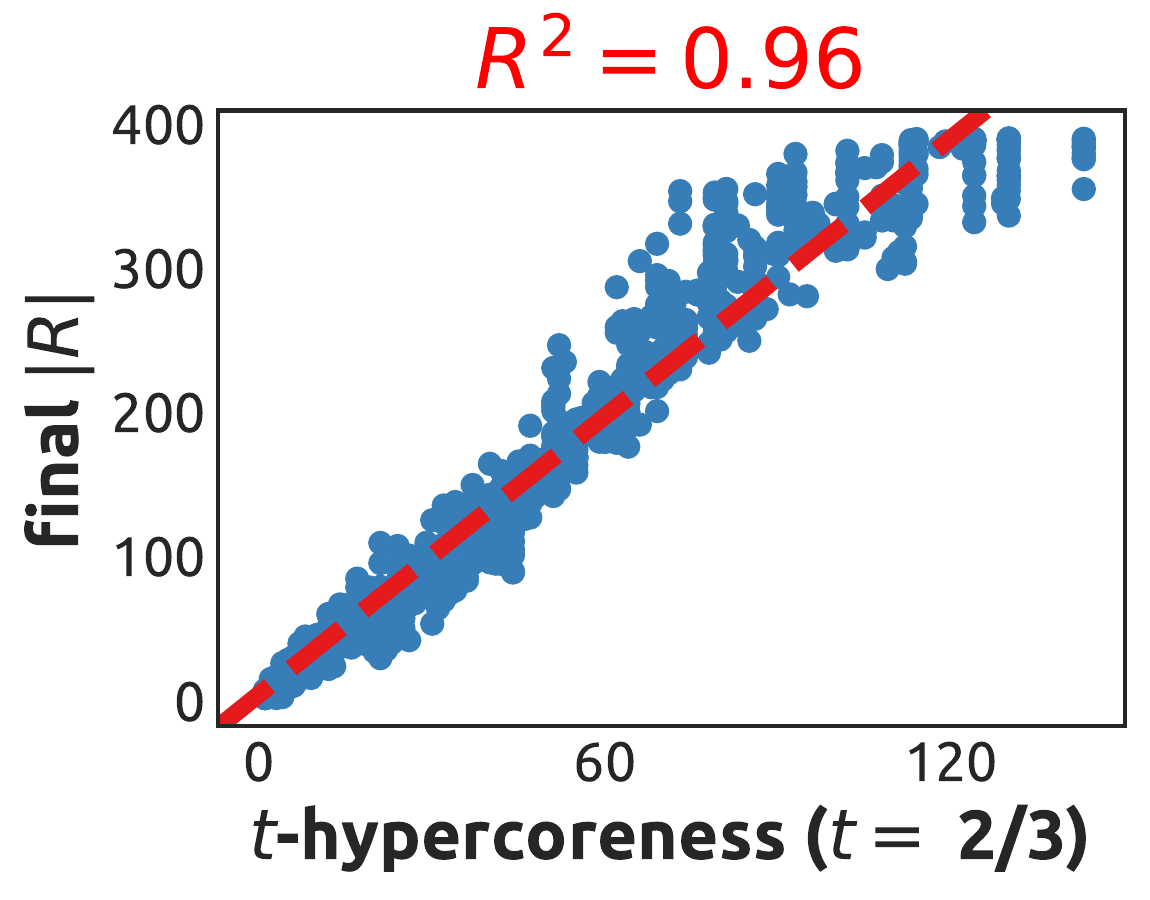}
	\end{subfigure}
	\begin{subfigure}[b]{0.32\linewidth}
		\centering
		\includegraphics[scale=0.4]{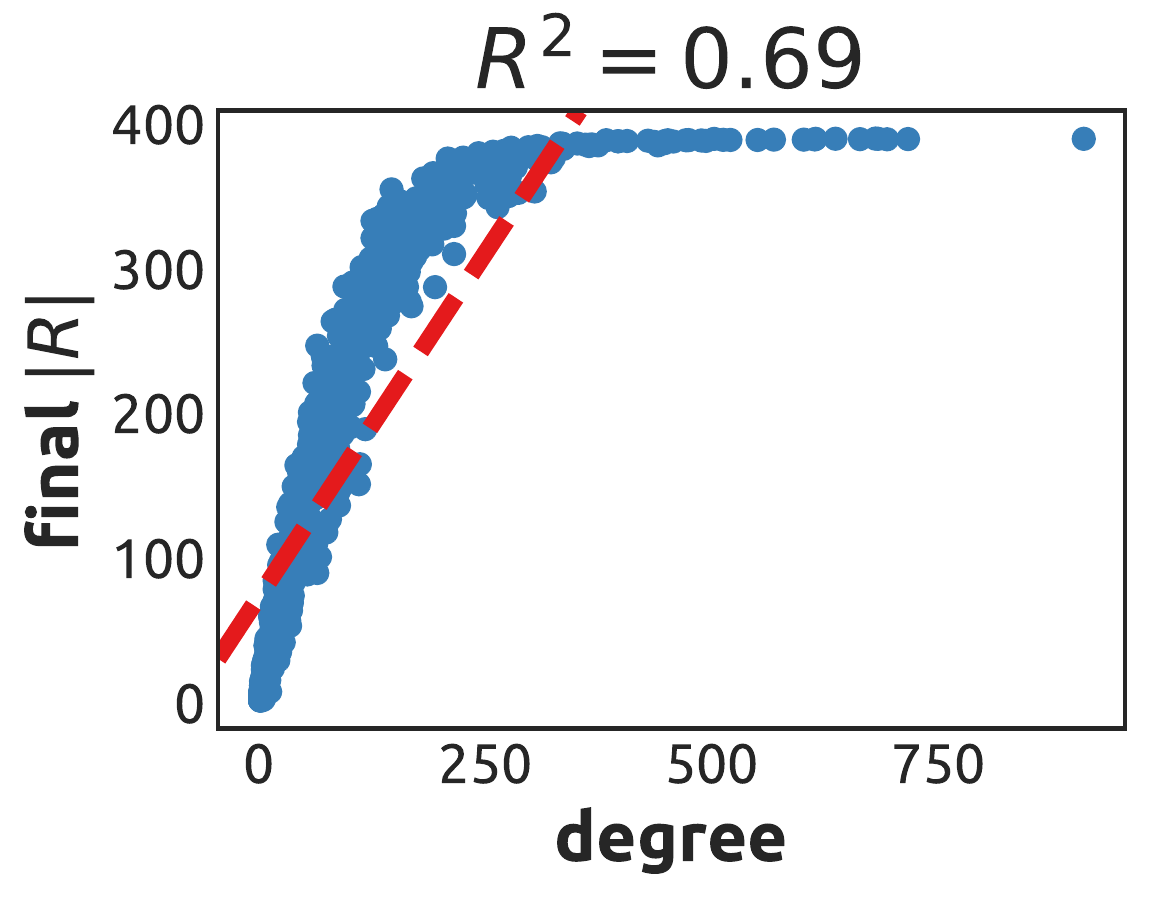}
	\end{subfigure}
	\begin{subfigure}[b]{0.32\linewidth}
		\centering
		\includegraphics[scale=0.4]{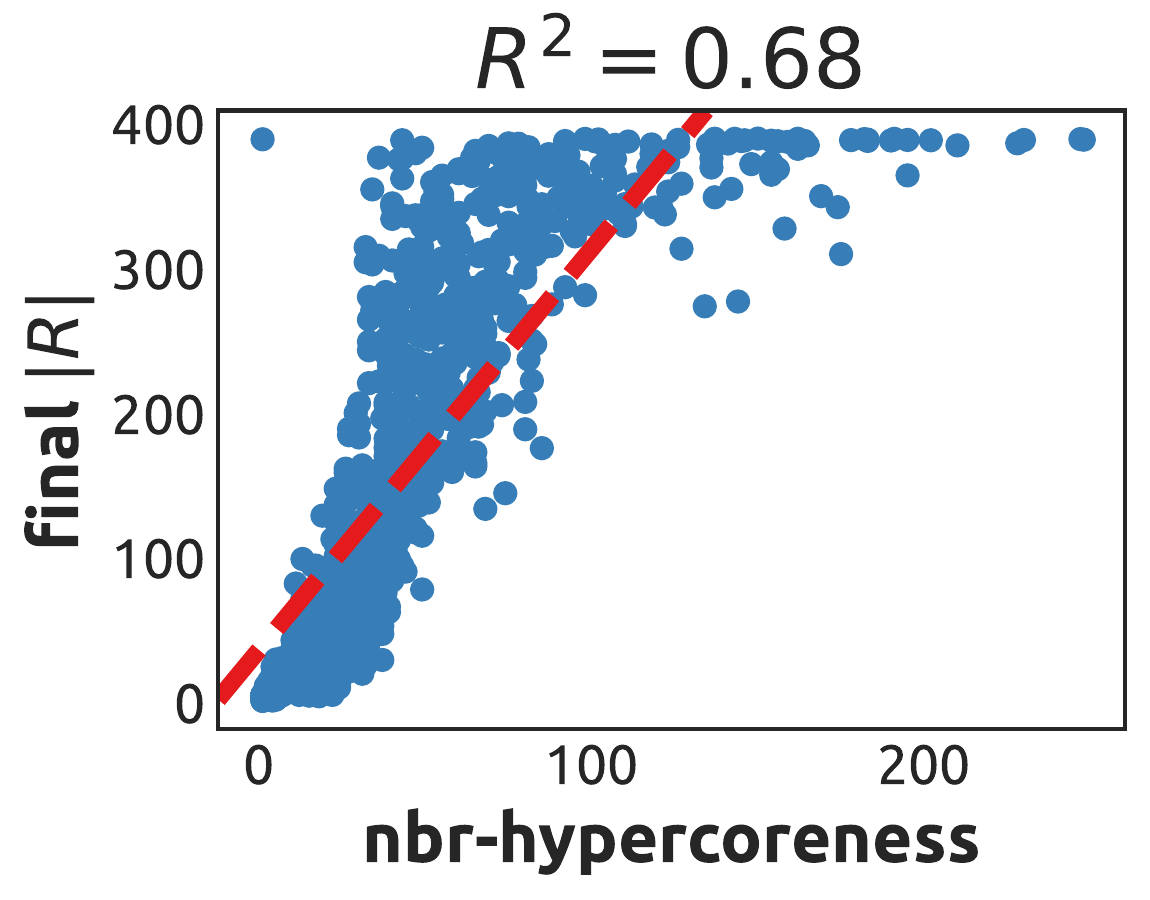}
	\end{subfigure}
	\begin{subfigure}[b]{0.32\linewidth}
		\centering
		\includegraphics[scale=0.4]{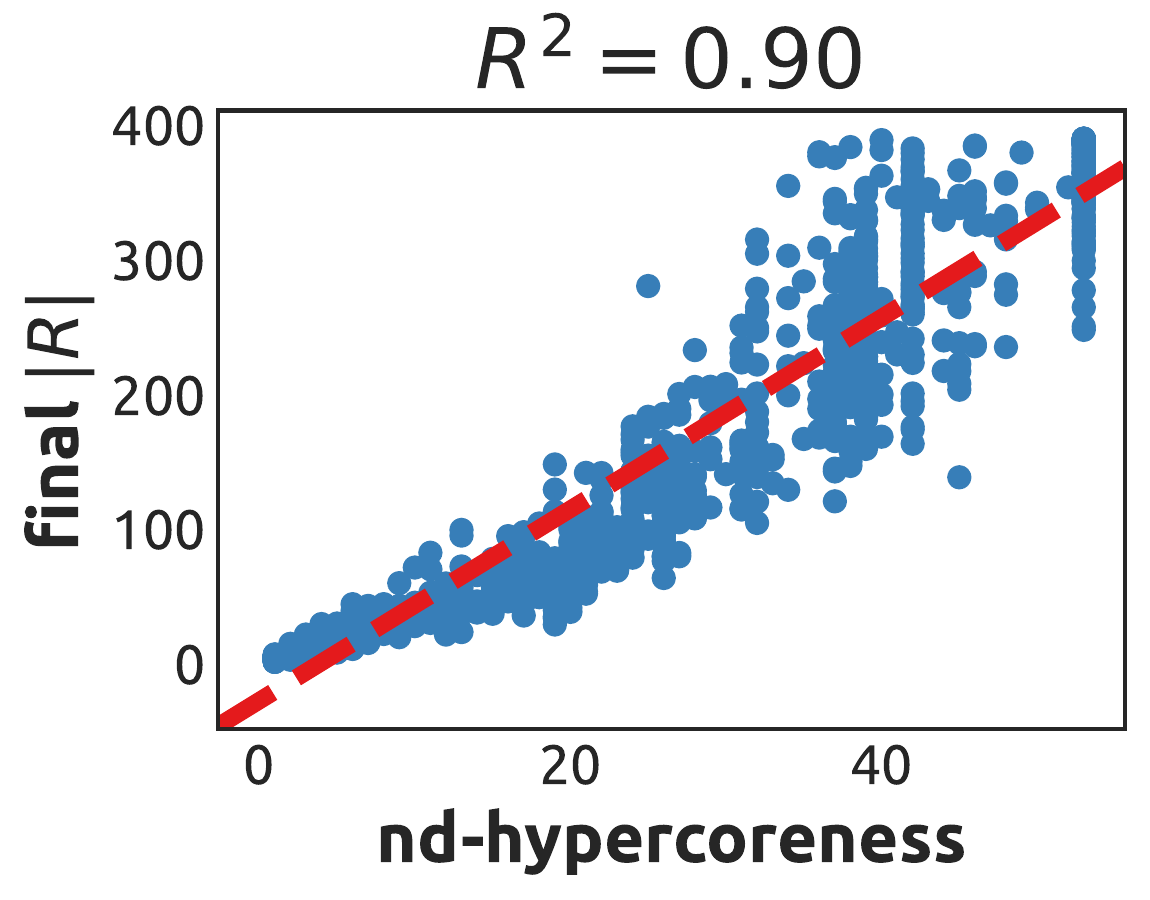}
	\end{subfigure}
	\begin{subfigure}[b]{0.32\linewidth}
		\centering
		\includegraphics[scale=0.4]{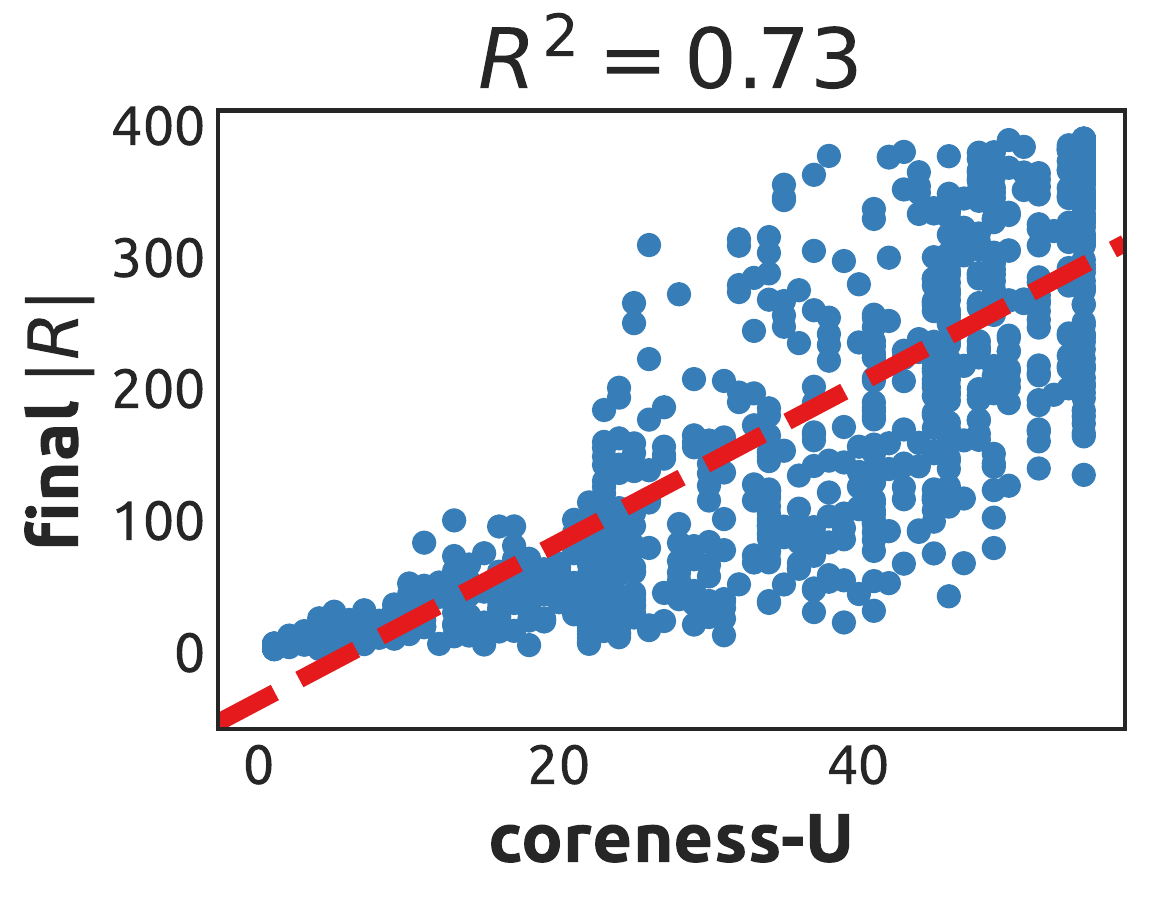}
	\end{subfigure}
	\begin{subfigure}[b]{0.32\linewidth}
		\centering
		\includegraphics[scale=0.4]{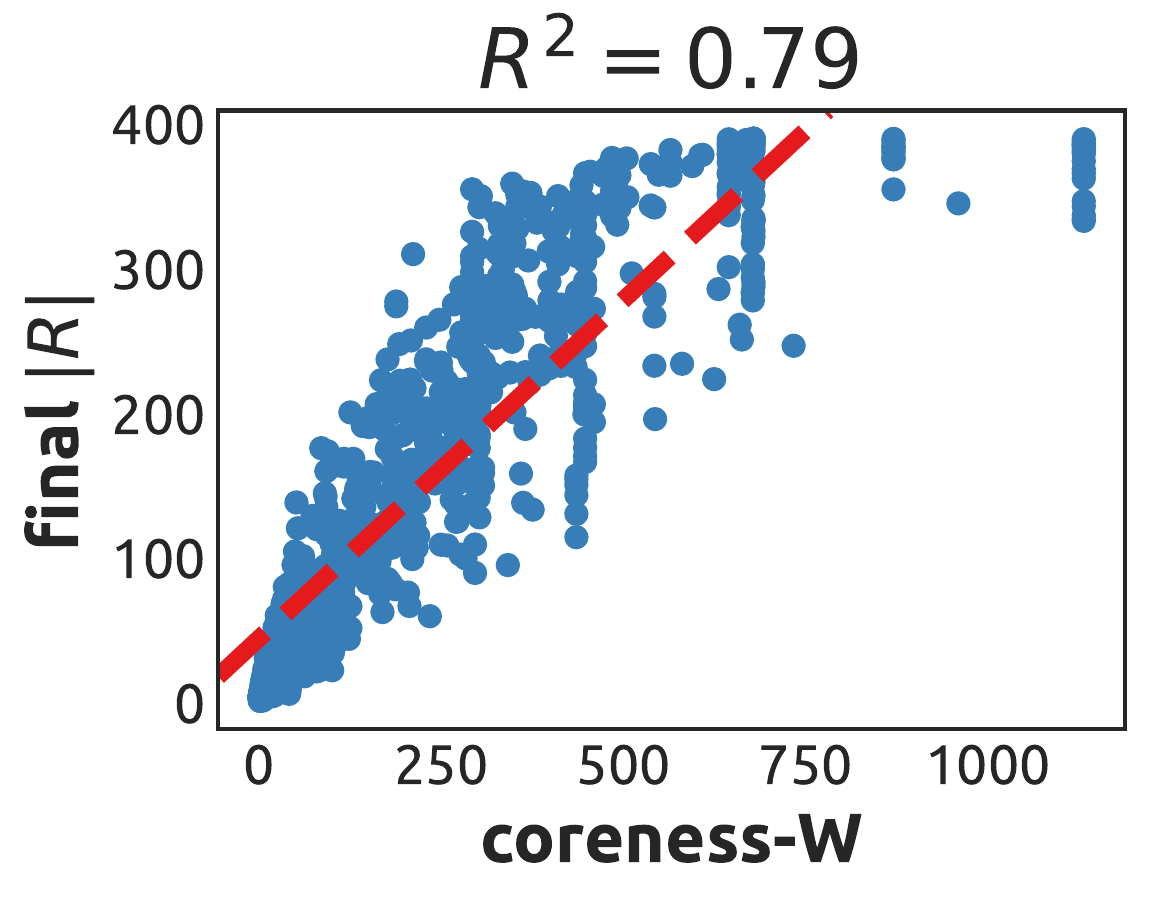}
	\end{subfigure}
	\begin{subfigure}[b]{0.32\linewidth}
		\centering
		\includegraphics[scale=0.4]{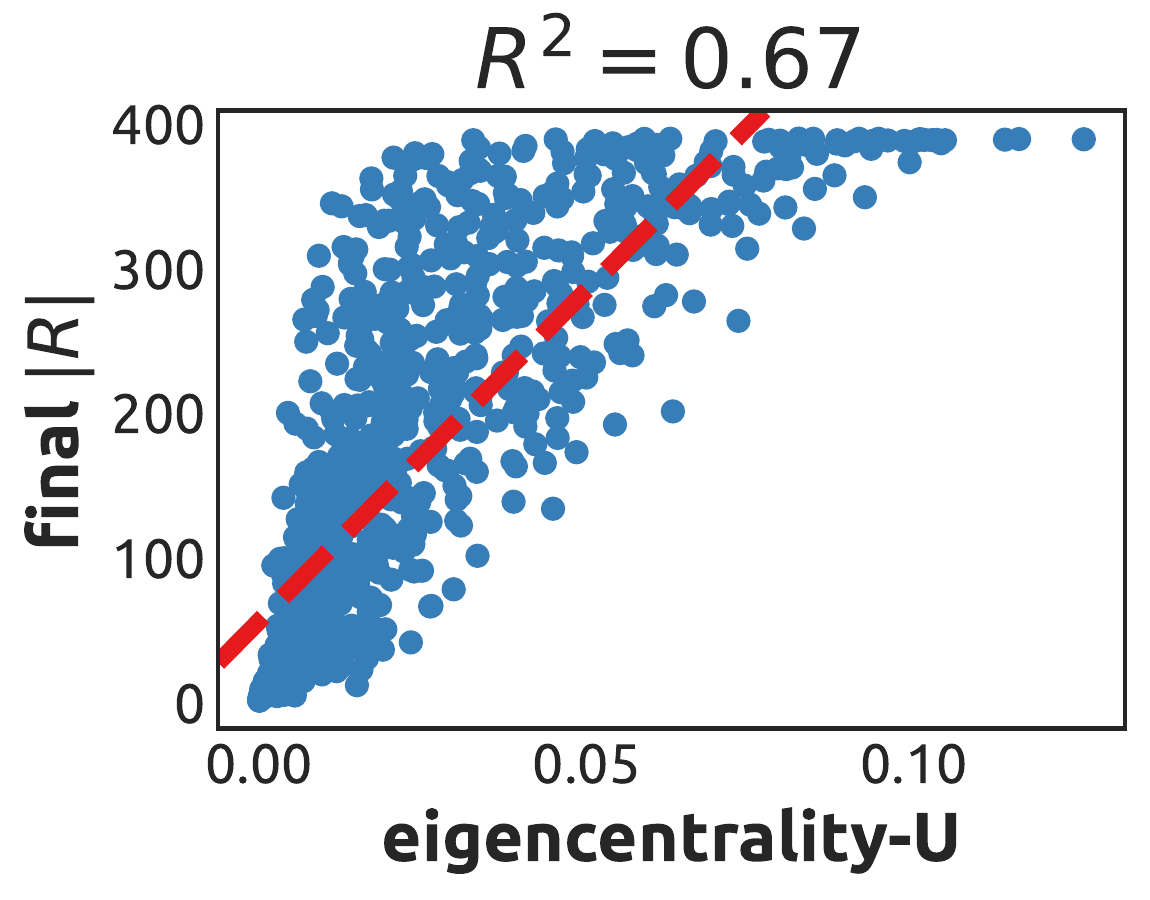}
	\end{subfigure}
	\begin{subfigure}[b]{0.32\linewidth}
		\centering
		\includegraphics[scale=0.4]{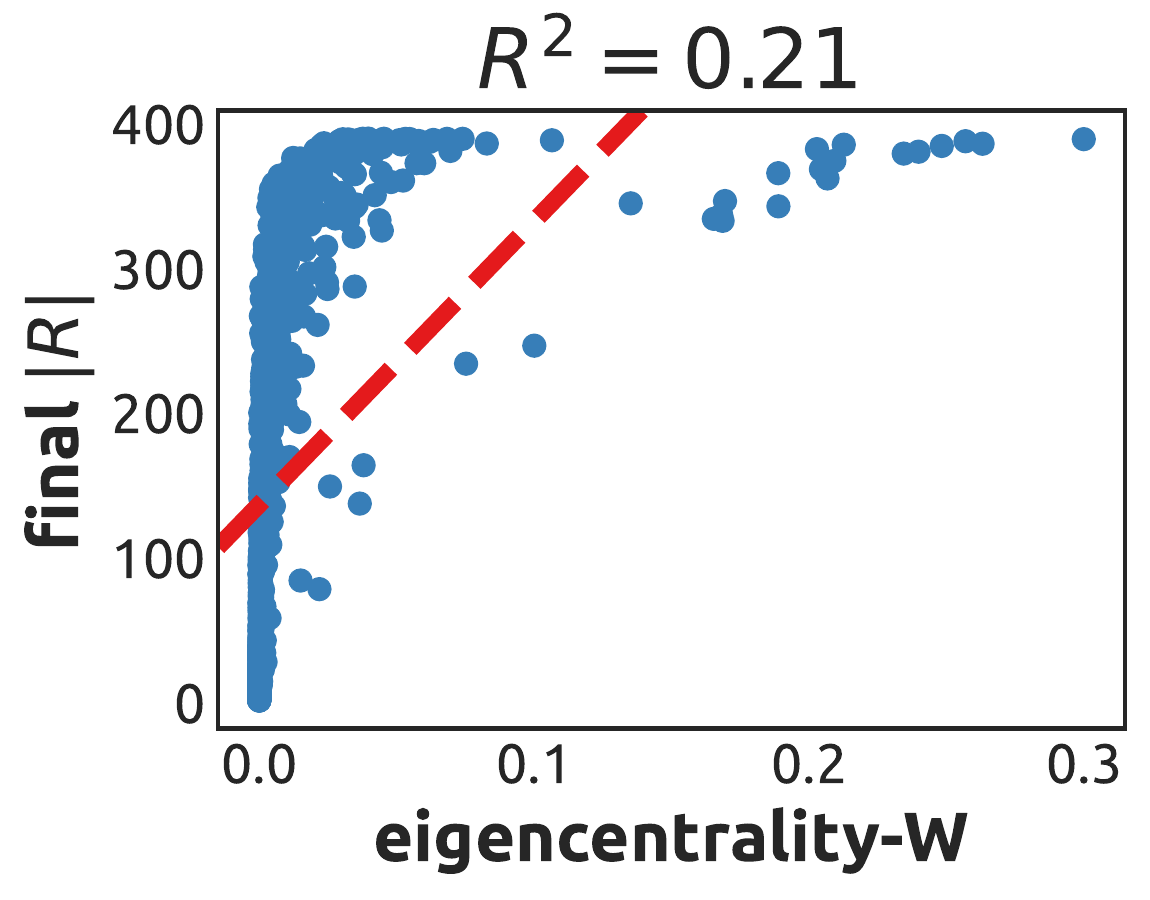}
	\end{subfigure}
	\begin{subfigure}[b]{0.32\linewidth}
		\centering
		\includegraphics[scale=0.4]{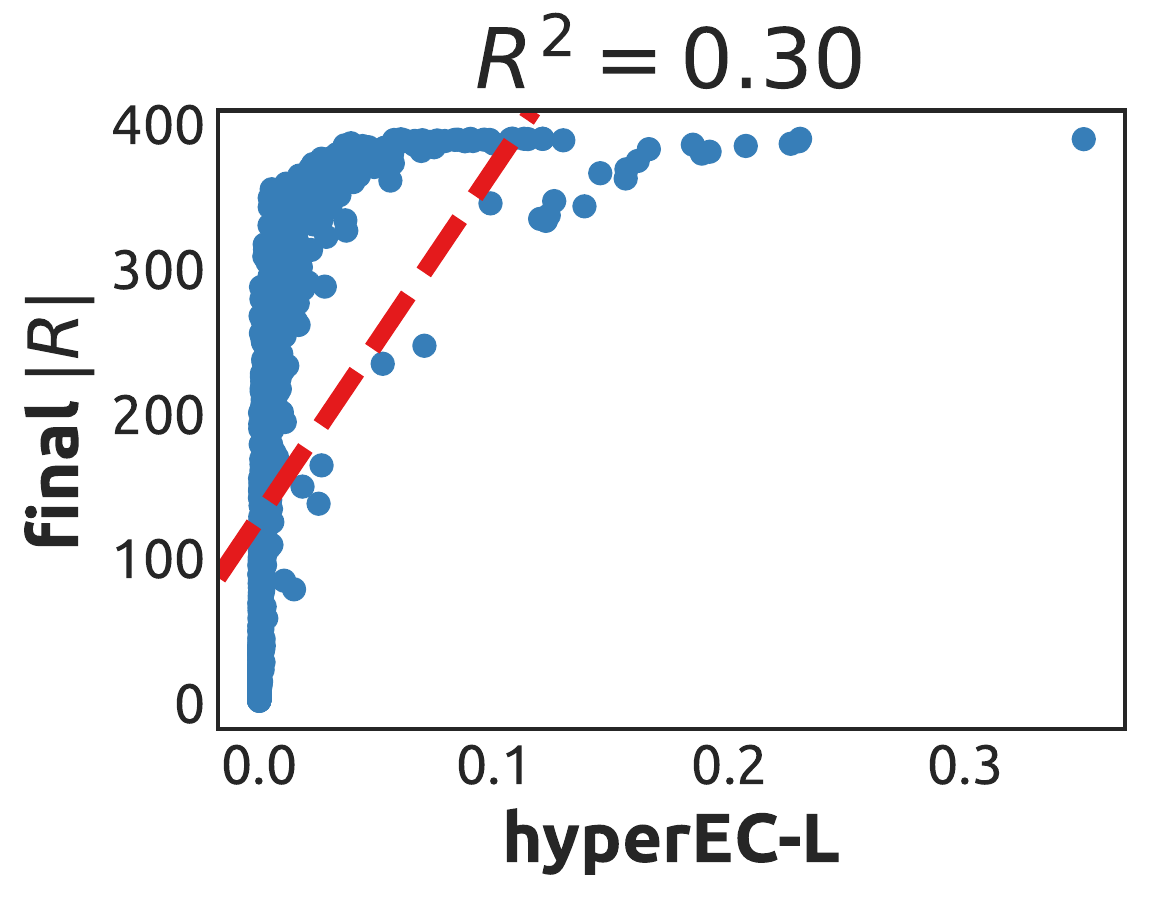}
	\end{subfigure}
	\begin{subfigure}[b]{0.32\linewidth}
		\centering
		\includegraphics[scale=0.4]{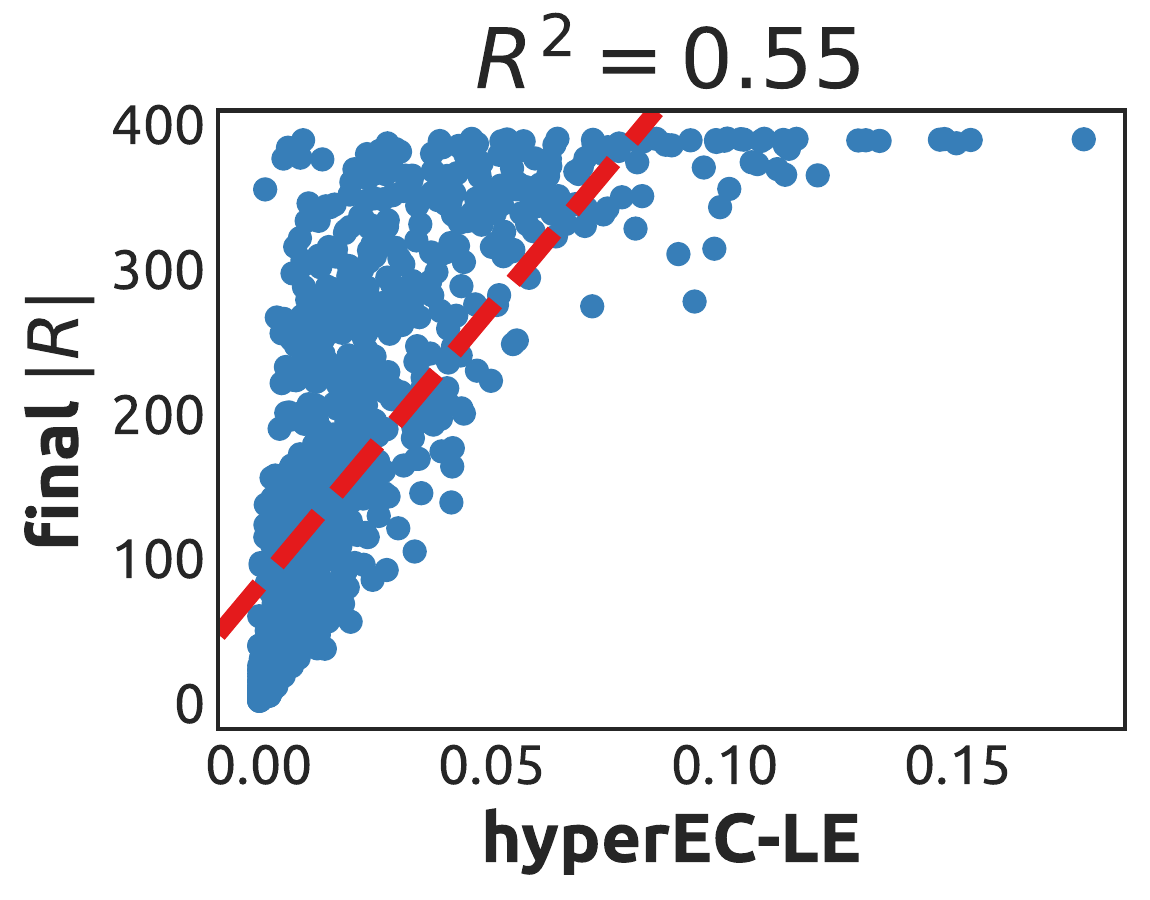}
	\end{subfigure}
	\begin{subfigure}[b]{0.32\linewidth}
		\centering
		\includegraphics[scale=0.4]{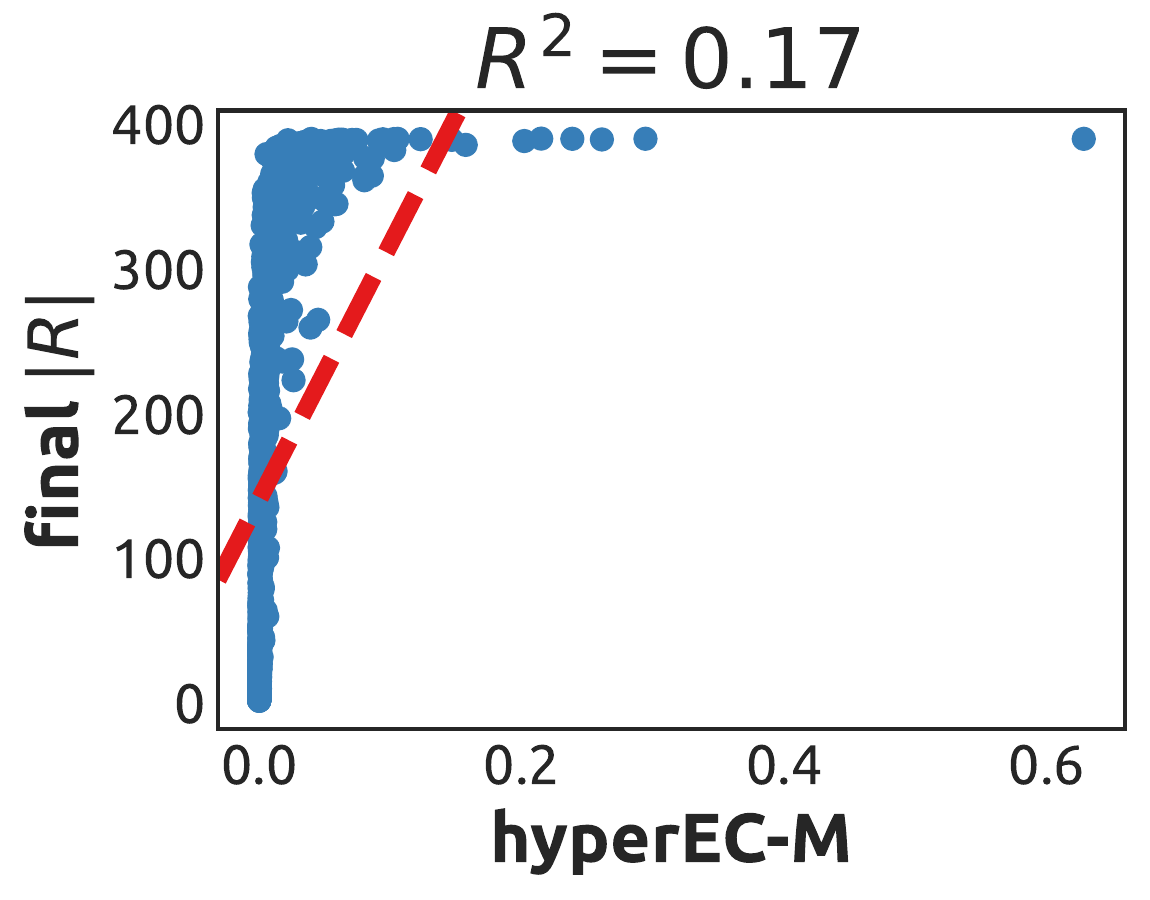}
	\end{subfigure}
	\begin{subfigure}[b]{0.32\linewidth}
		\centering
		\includegraphics[scale=0.4]{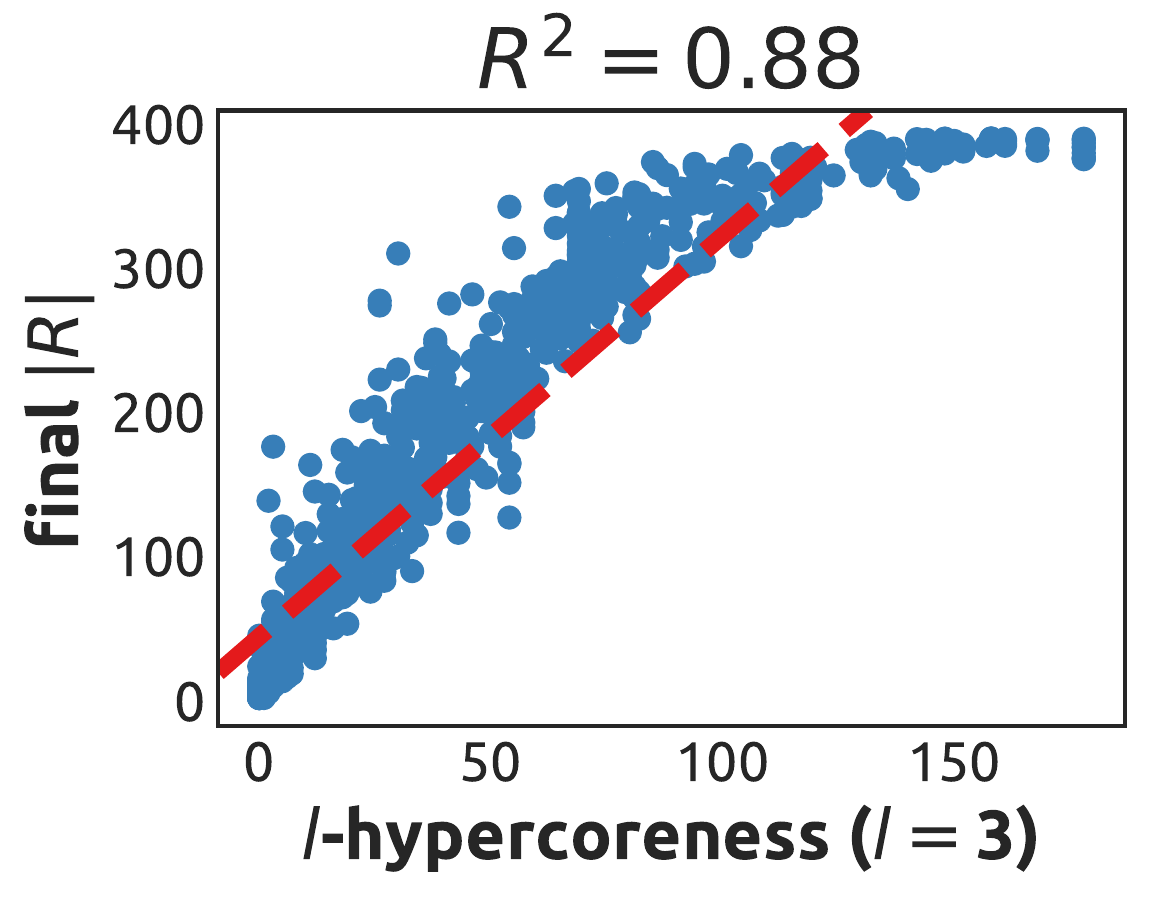}
	\end{subfigure}
	\begin{subfigure}[b]{0.32\linewidth}
		\centering
		\includegraphics[scale=0.4]{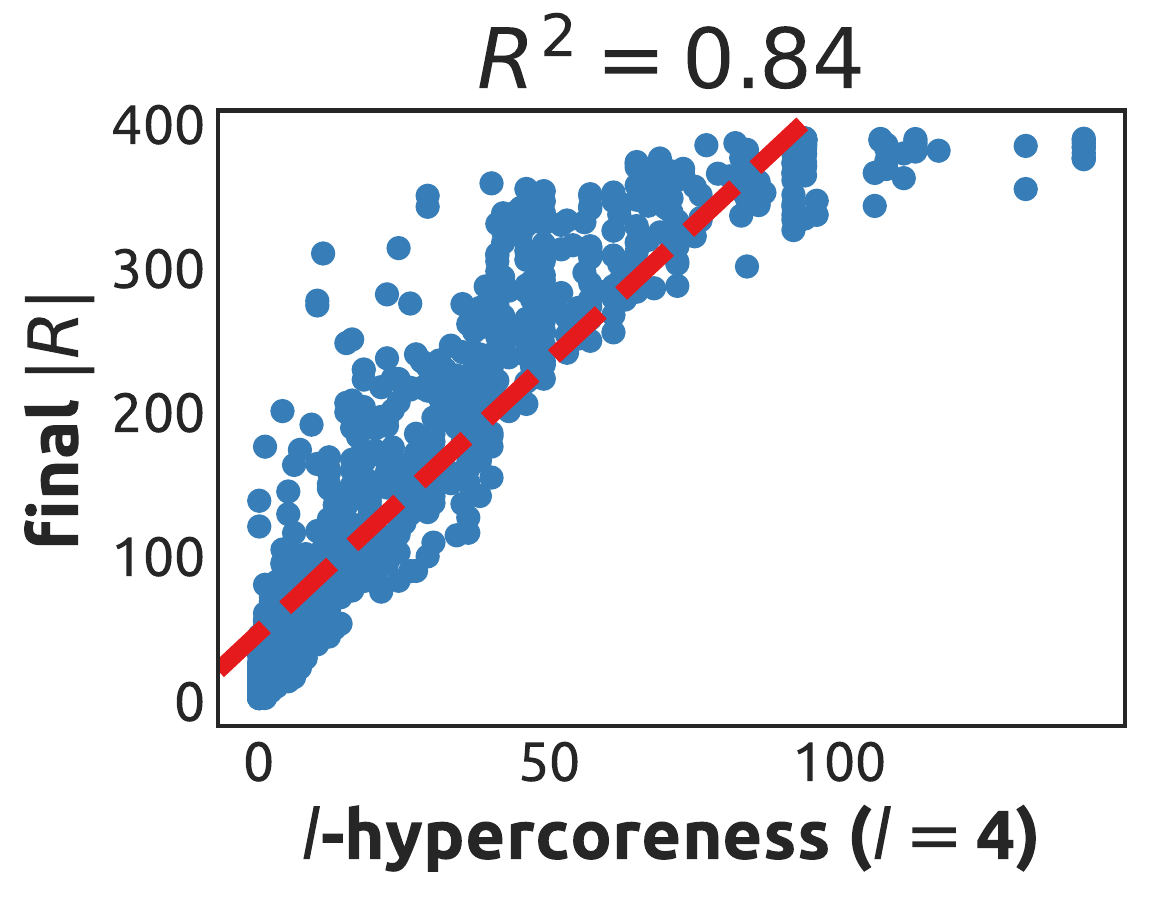}
	\end{subfigure}
        \begin{subfigure}[b]{0.32\linewidth}
		\centering
		\includegraphics[scale=0.4]{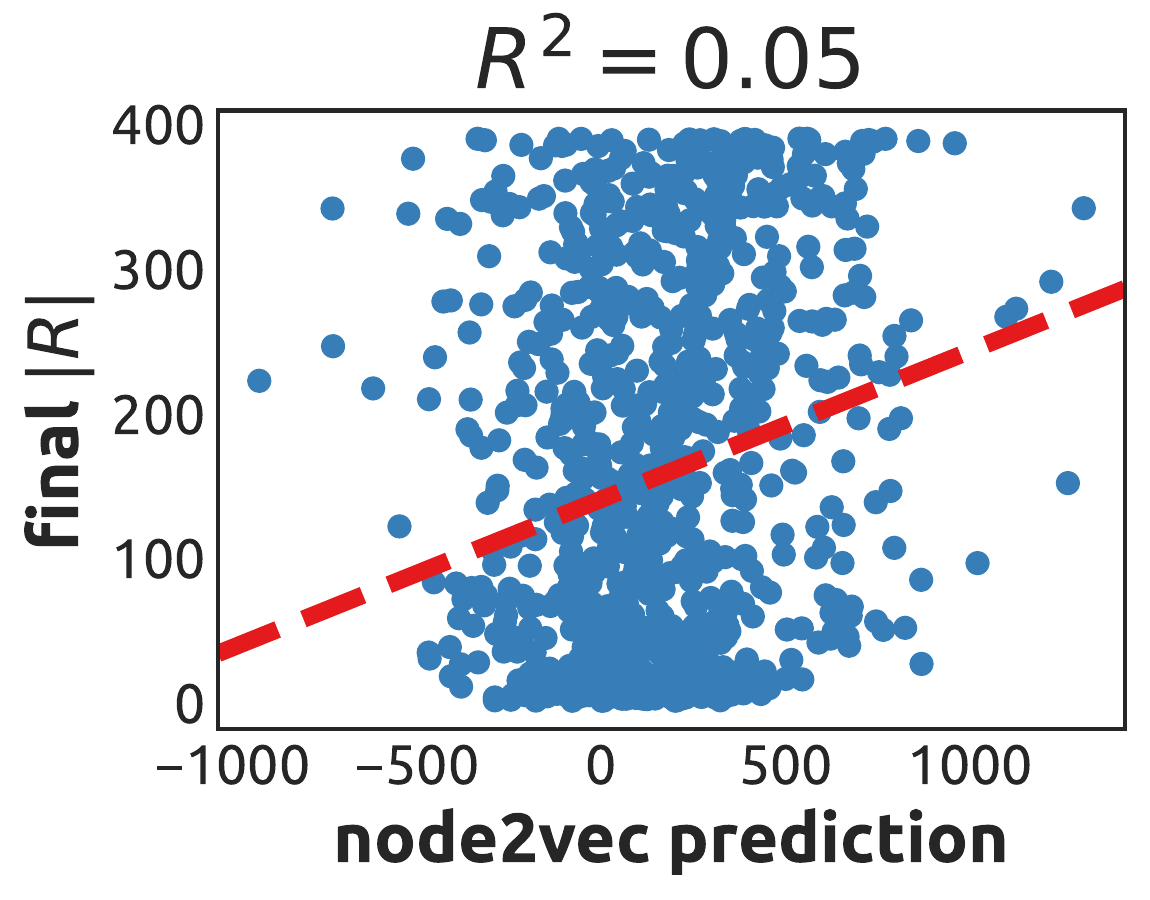}
	\end{subfigure}
        \begin{subfigure}[b]{0.32\linewidth}
		\centering
		\includegraphics[scale=0.4]{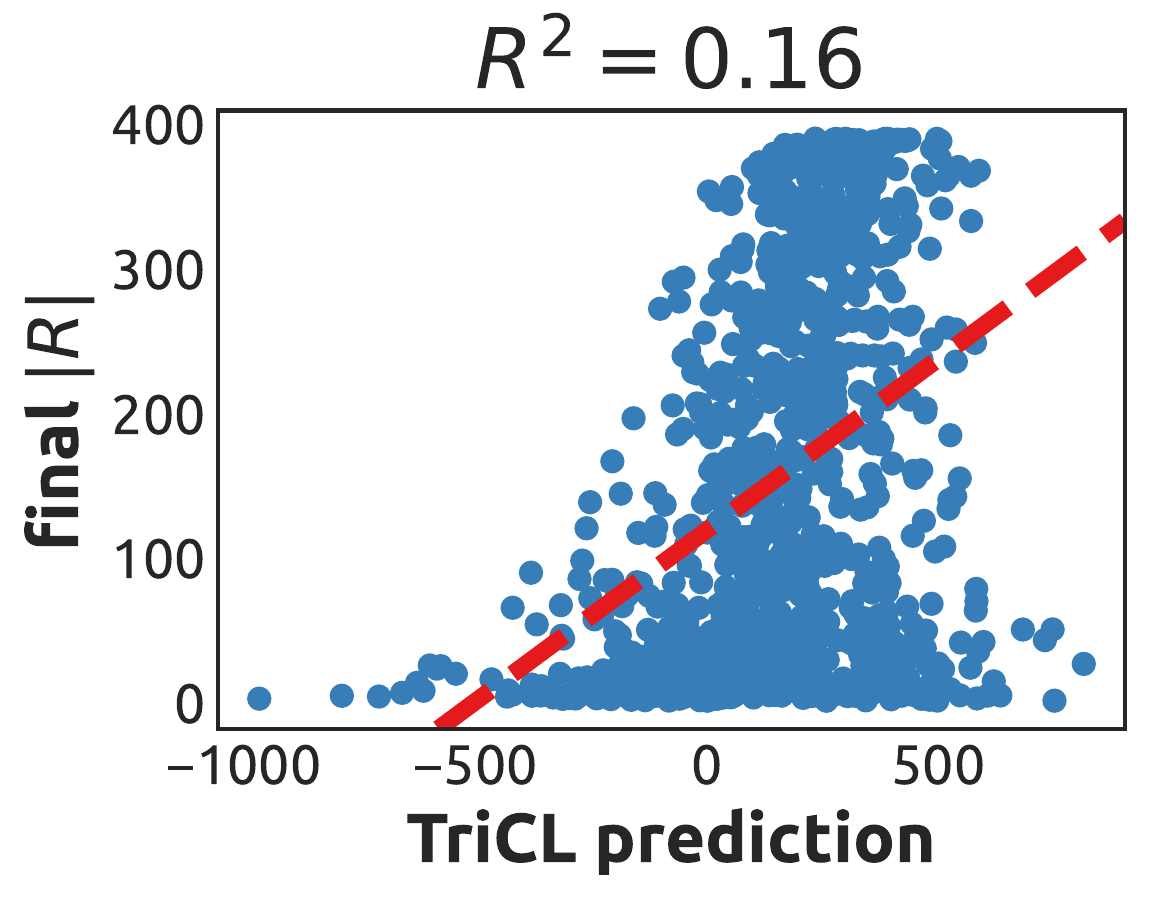}
	\end{subfigure}
	\caption{\textbf{$t$-Hypercoreness with a proper $t$ value is the best indicator of influence among all considered centrality measures (Dataset: email-Eu).} The red dashed line represents the best-fitted line, and the $R^2$ value is shown above each subfigure. 
		The full results on all the datasets are in the supplementary document~\citep{onlineSuppl}.}
	\label{fig:inf_email_Eu}
\end{figure*}

\subsection{Dense substructure discovery}
Intuitively, $(k,t)$-hypercores are not limited to complete subhypergraphs.
Thus, they can be denser than complete subhypergraphs, which previous works \citep{hua2023revisiting, luo2021hypercore, luo2022hypercore, gabert2021shared, gabert2021unifying, sun2020fully} focus on.

Given $H = (V, E)$, we define its density as $\delta(H) = \abs{E} / \abs{V}$.
In Fig.~\ref{fig:core_density}, for each dataset and each $t \in [0, 1]$, we show the relative density of the $(c_t^*, t)$-hypercore, which is defined as $\Tilde{\delta}_t = \delta(C_{c_t^*, t}) / \delta(H)$.
Note that the hypercores are significantly denser than the whole hypergraph, especially when $t$ is small.
In addition, except for the \textit{tags-SO} dataset, the similarity between hypergraphs in the same domain is observed.
Similar to the normalized hypercore-size-mean-difference (HSMD) distance used in Sec.~\ref{subsec:obs_coresize}, we define the relative-density-mean-difference (RDMD) distance between two hypergraphs to measure the similarity of the patterns.
\begin{definition}[Relative-density-mean-difference (RDMD) distance]\label{def:RDMD}
	Given two hypergraphs $H_1$ and $H_2$, the relative-density-mean-difference (RDMD) distance between $H_1$ and $H_2$ is defined as 
	\[
	\operatorname{RDMD}(H_1, H_2) \coloneqq \sqrt{\int_{0}^{1} (\log \Tilde{\delta}_t(H_1) - \log \Tilde{\delta}_t(H_2))^2 \, dt}.
	\]
\end{definition}
See Fig.~\ref{fig:dense_sim} for the RDMD distance between each pair of datasets.

\begin{observation}[Density of $(k, t)$-hypercores]\label{obs:density}
	In real-world hypergraphs, $(k,t)$-hypercores are dense, and the density tends to decrease as $t$ increases. The relative density with respect to $t$ tends to be similar in hypergraphs in the same domain.
\end{observation}

We utilize the {high} density of $(k, t)$-hypercores for the \textit{max $(k_{c}, t_{c})$-vertex cover problem} below, where we say a hyperedge $e$ is $t_{c}$-\textit{covered} by a set of nodes $V'$ if $\abs{e \cap V'} \geq t_{c} \abs{e}$.
\begin{problem}[max $(k_{c}, t_{c})$-vertex cover problem]
	Given a hypergraph $H = (V, E)$, $k_{c} \in \bbN$ and $t_{c} \in (0, 1]$,
	the \textbf{max $(k_{c}, t_{c})$-vertex cover problem} aims to find 
	$V^* \in \binom{V}{k_{c}} \coloneqq \setbr{V' \subset V: \abs{V'} = k_{c}}$
	such that the number of hyperedges $t_{c}$-covered by $V^*$ is maximized.
\end{problem}
In our experiments, we compare three different algorithms:
\begin{itemize}
	\item \textbf{$t_{c}$-Hypercoreness}: $k_{c}$ nodes with highest $t_{c}$-hypercoreness in $H$ are chosen (tie broken by node-degrees);
	\item \textbf{Degree}: $k_{c}$ nodes with highest degree in $H$ are chosen;
	\item \textbf{Greedy}: it first chooses the node with the highest degree and greedily chooses a node that increases the number of $t_{c}$-covered hyperedges most until $k_{c}$ nodes are chosen. 
\end{itemize}

\begin{figure}[t]
	\centering
	
	\begin{subfigure}[b]{\textwidth}
		\centering
		\includegraphics[scale=0.5]{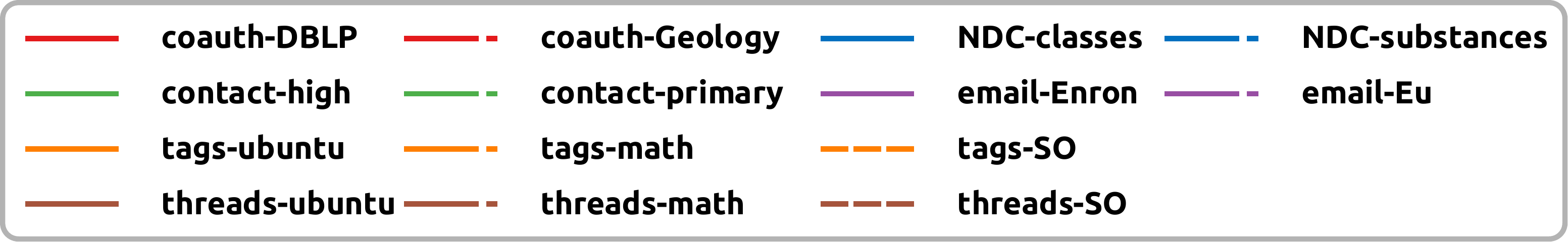}\\        
	\end{subfigure}
	
	\begin{subfigure}[b]{\textwidth}
		\centering
		\includegraphics[scale=0.4]{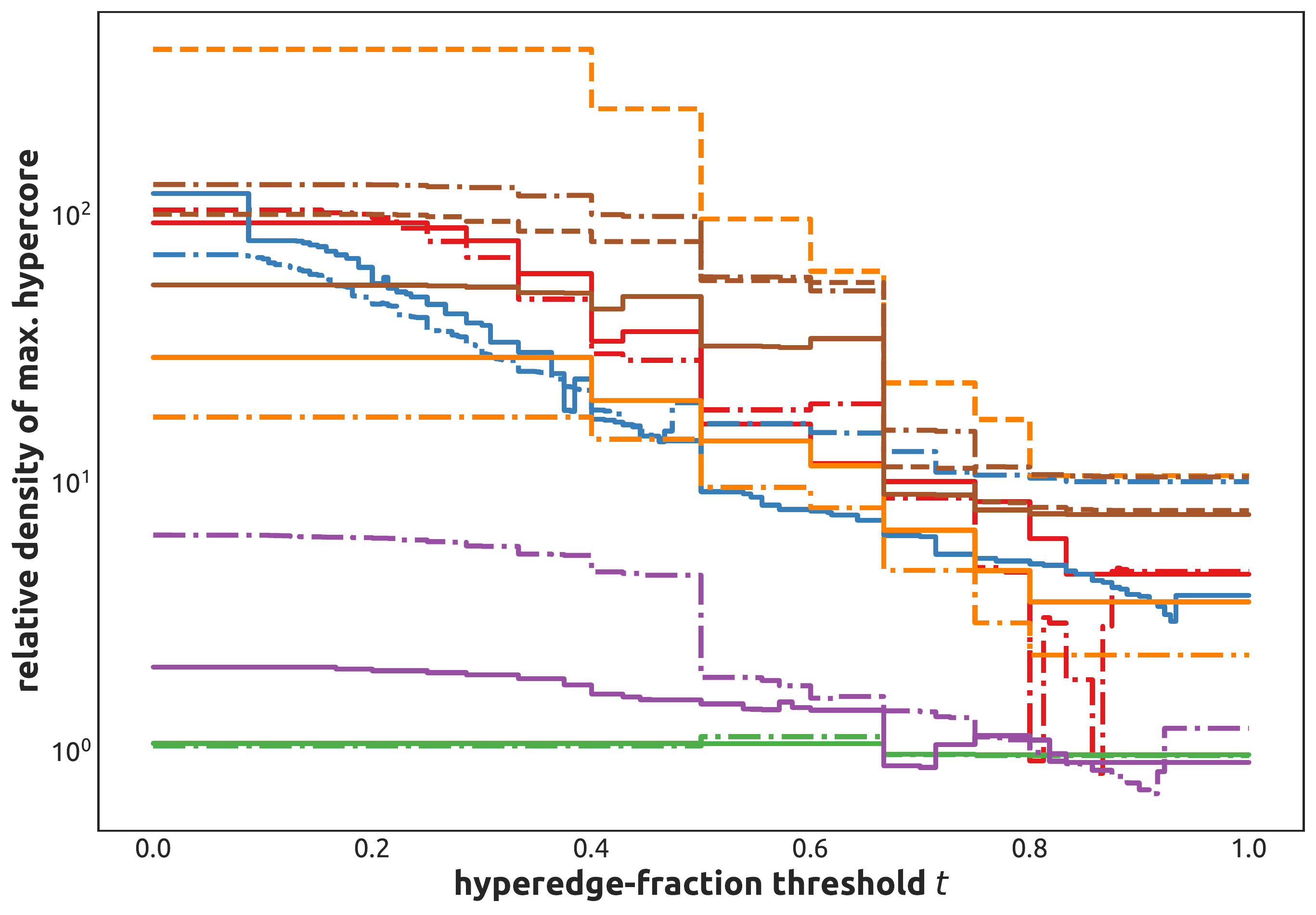}
	\end{subfigure}
	
	\caption{\textbf{Overall, hypercores are much denser than the whole hypergraph, and the density decreases as $t$ increases.}
		For each dataset, we report the relative density of the $(c_t^*, t)$-hypercore (i.e., the $(k, t)$-hypercore with maximal $k$) w.r.t $t$.}
	\label{fig:core_density}
\end{figure}

\begin{figure}
	\centering
	\includegraphics[scale=0.4]{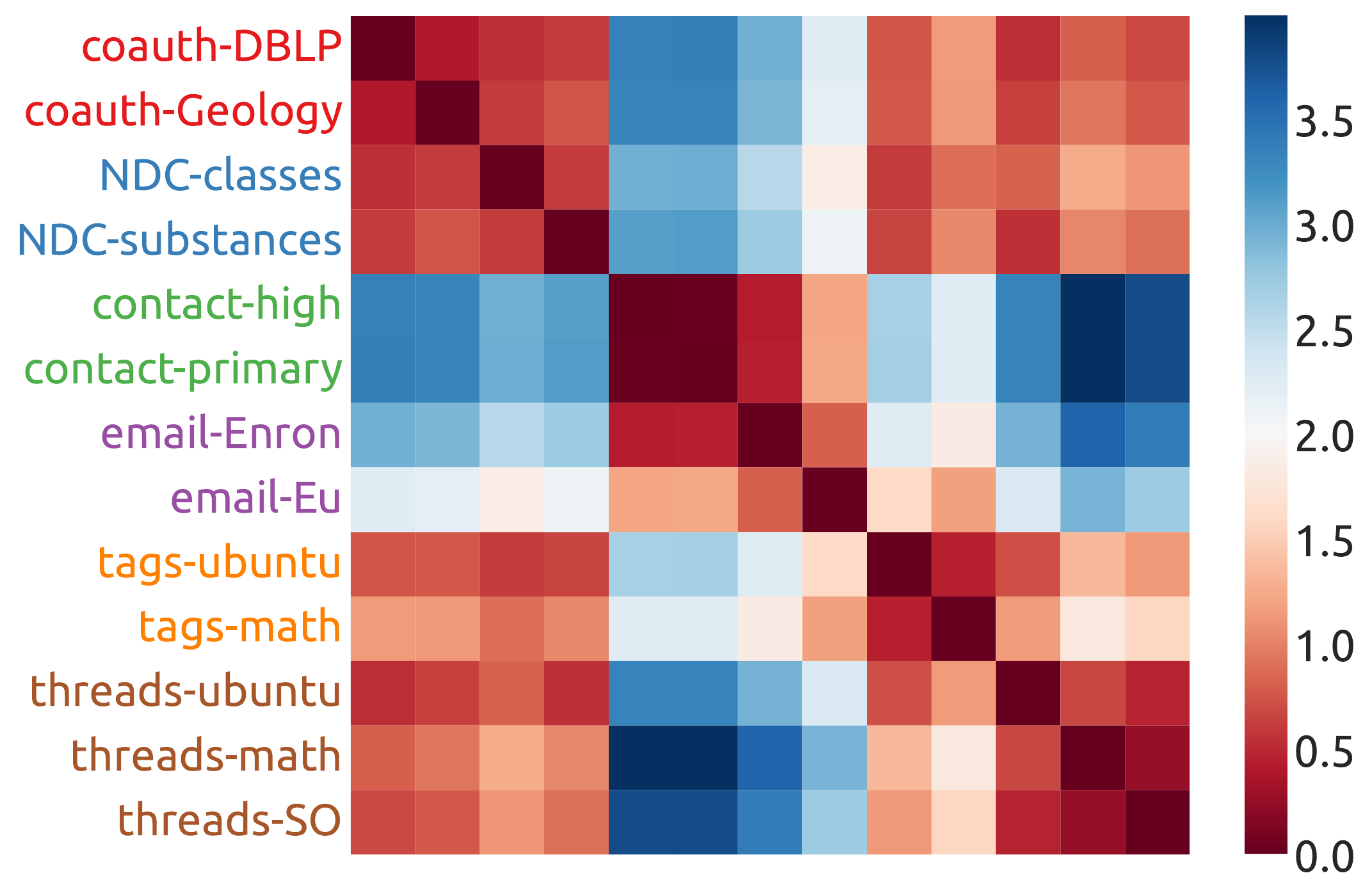}
	\vspace{-2mm}
	\caption{\textbf{The RDMD distance is small between datasets in the same domain ($0.456$ in average) while the overall average is $1.741$; the two means are significantly different with $p = 0.0035$ in the $t$-test.} We report the RDMD distance between each pair of datasets except for \textit{tag-SO}.}
	\label{fig:dense_sim}
\end{figure}

In each dataset, we track the count of $t_{c}$-covered hyperedges by the $k_{c}$ nodes chosen by each algorithm while varying $k_{c}$ from $10$ to $100$. 
Then, we divide each count by the count obtained by the \textit{degree} algorithm in the same setting.
The relative counts are averaged over all datasets for $t_{c} \in \setbr{0.6, 0.7, 0.8}$ and reported in Fig.~\ref{fig:density}.
We choose those $t_c$ values because they require a majority of, but not all of, the constituent nodes to cover a hyperedge.
On average, the algorithm \textit{$t_{c}$-hypercoreness} outperforms the other two algorithms, with clear superiority when $t_{c} \in \setbr{0.6, 0.7}$.

\subsection{Hypergraph vulnerability detection}
Through the observations and applications, we have shown the significance of the proposed concepts and the importance of nodes in the $(k, t)$-hypercore with large $k$ values.
Thus, in order to reinforce the engagement of nodes in a hypergraph (e.g., user engagement in online social networks), intuitively, the $(k,t)$-hypercores should be paid close attention to.
From another perspective, we should protect the nodes whose deletions will cause a large number of nodes to leave the $(k,t)$-hypercores.
For example, online social network providers should try to make such nodes stay.
Based on such ideas, in pairwise graphs, the \textit{collapsed $k$-core problem} \citep{zhang2017finding} and its variants \citep{zhu2018k, zhu2019pivotal} have been considered to find the critical users whose deletions reduce the size of $k$-core most, i.e., the most \textit{vulnerable} nodes in the $k$-core.
We generalize the problem to hypergraphs {based on} our proposed concepts.

\begin{figure}[t!]
	\centering        
	\includegraphics[scale=0.5]{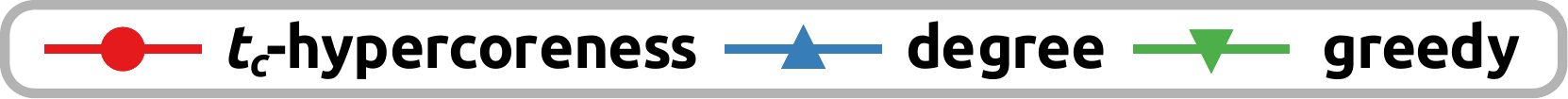}\\       
	\vspace{2mm}       
	\includegraphics[scale=0.4]{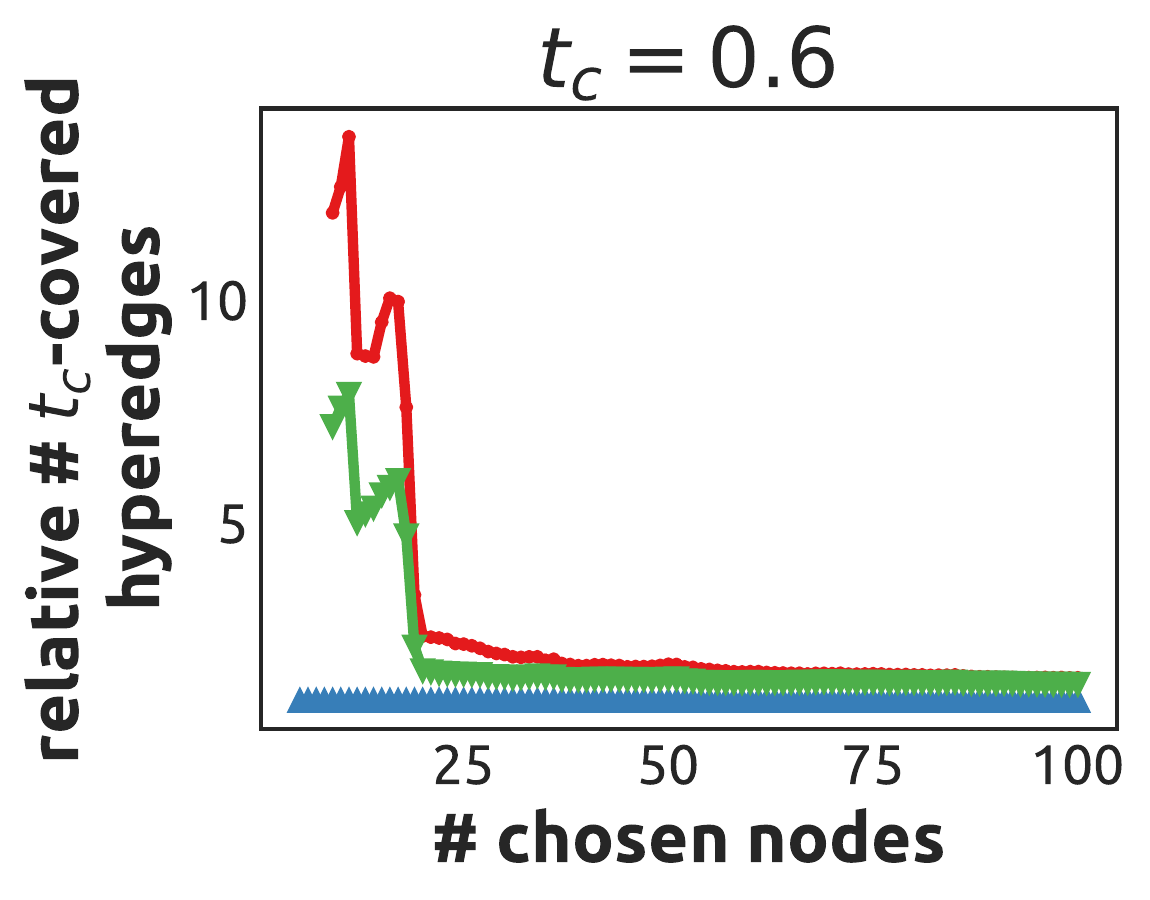}
	\includegraphics[scale=0.4]{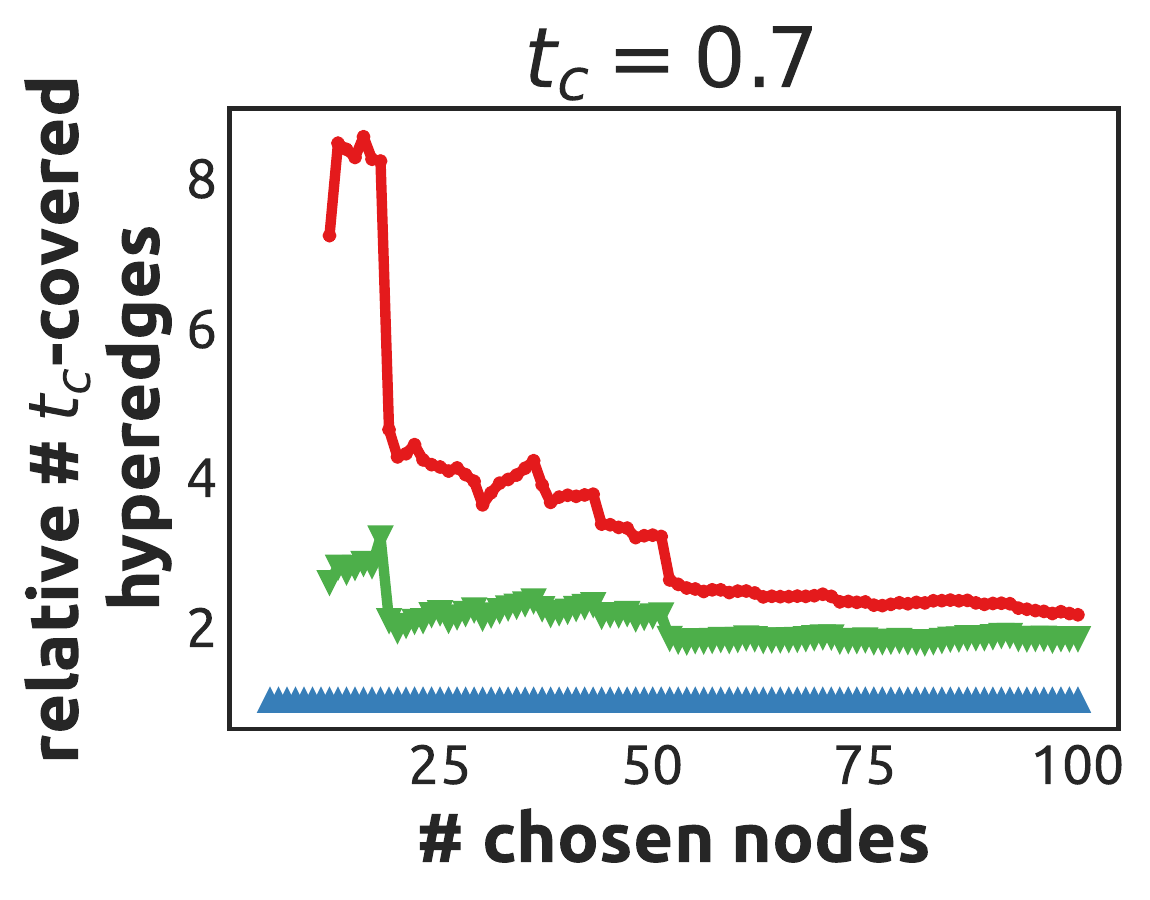}
	\includegraphics[scale=0.4]{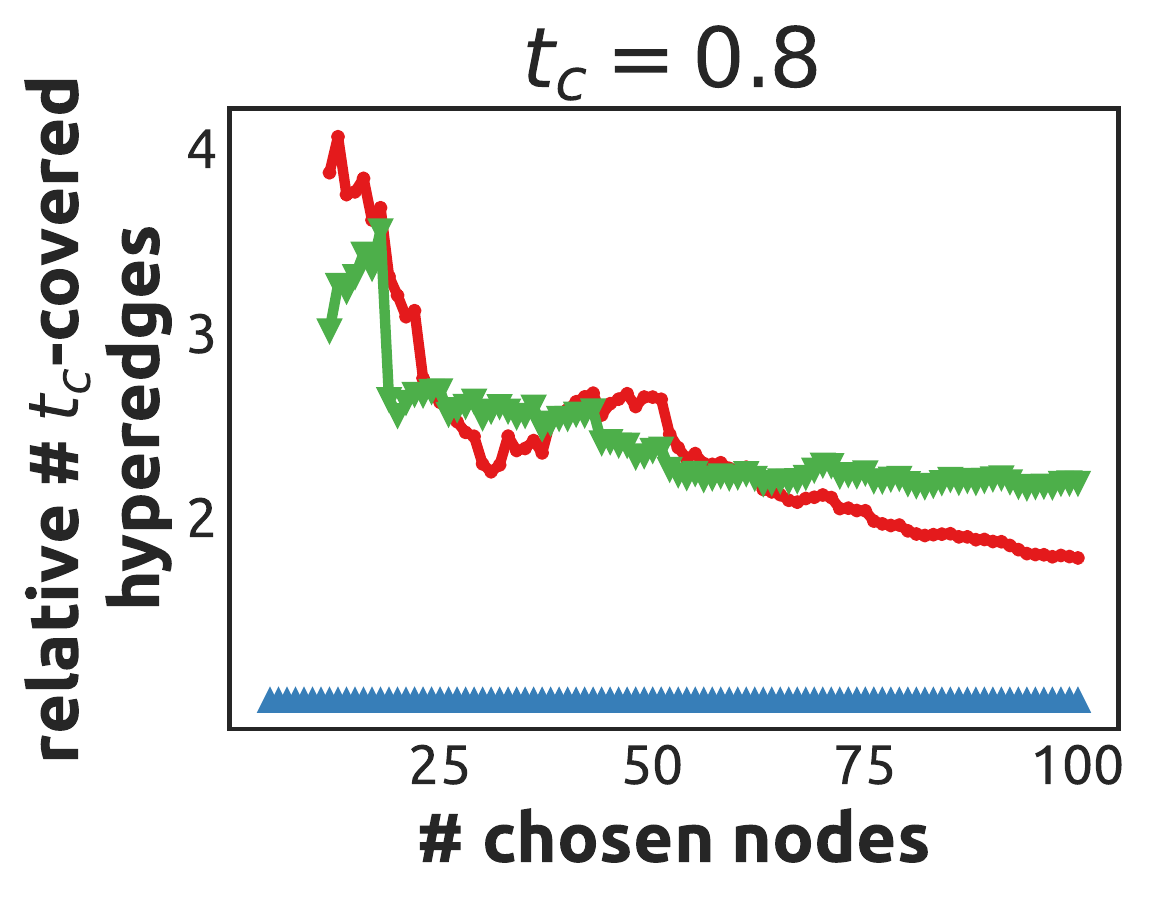}
	\caption{\textbf{Nodes chosen by $t_{c}$-hypercoreness cover most hyperedges.}	The performance of \textit{$t_{c}$-hypercoreness}, \textit{degree}, and \textit{greedy} in solving the max $(k_{c}, t_{c})$-vertex cover problem.}
	\label{fig:density}
\end{figure}

\begin{figure*}[t]
	\vspace{-2mm}
	\centering
	\includegraphics[scale=0.5]{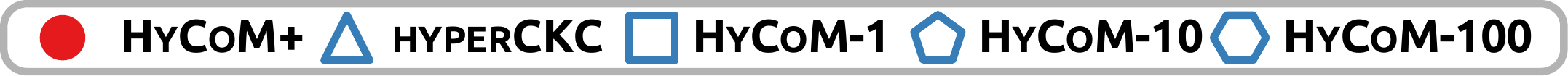}\\
	\vspace{2mm}
	\begin{subfigure}[b]{0.32\linewidth}
		\centering
		\includegraphics[scale=0.4]{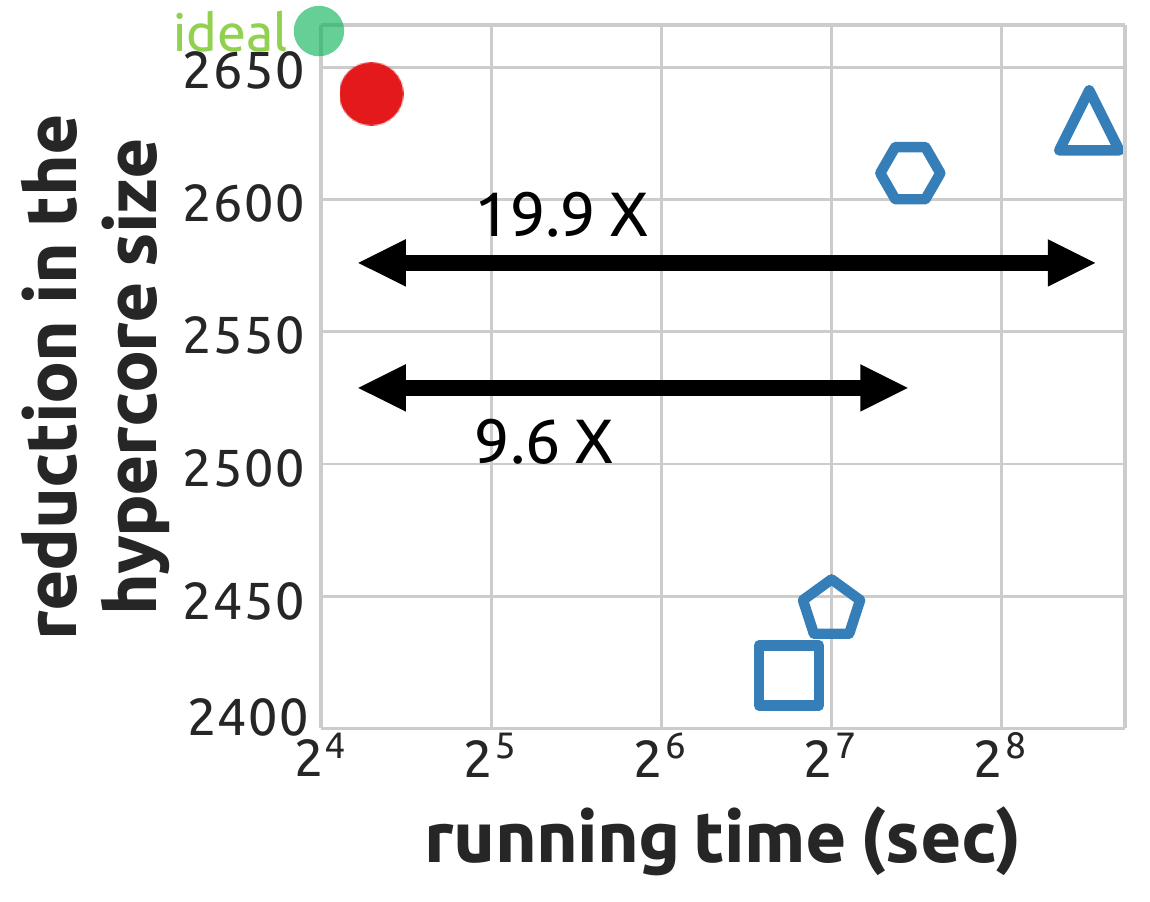}
		\caption{tags-SO}
	\end{subfigure}
	\begin{subfigure}[b]{0.32\linewidth}
		\centering
		\includegraphics[scale=0.4]{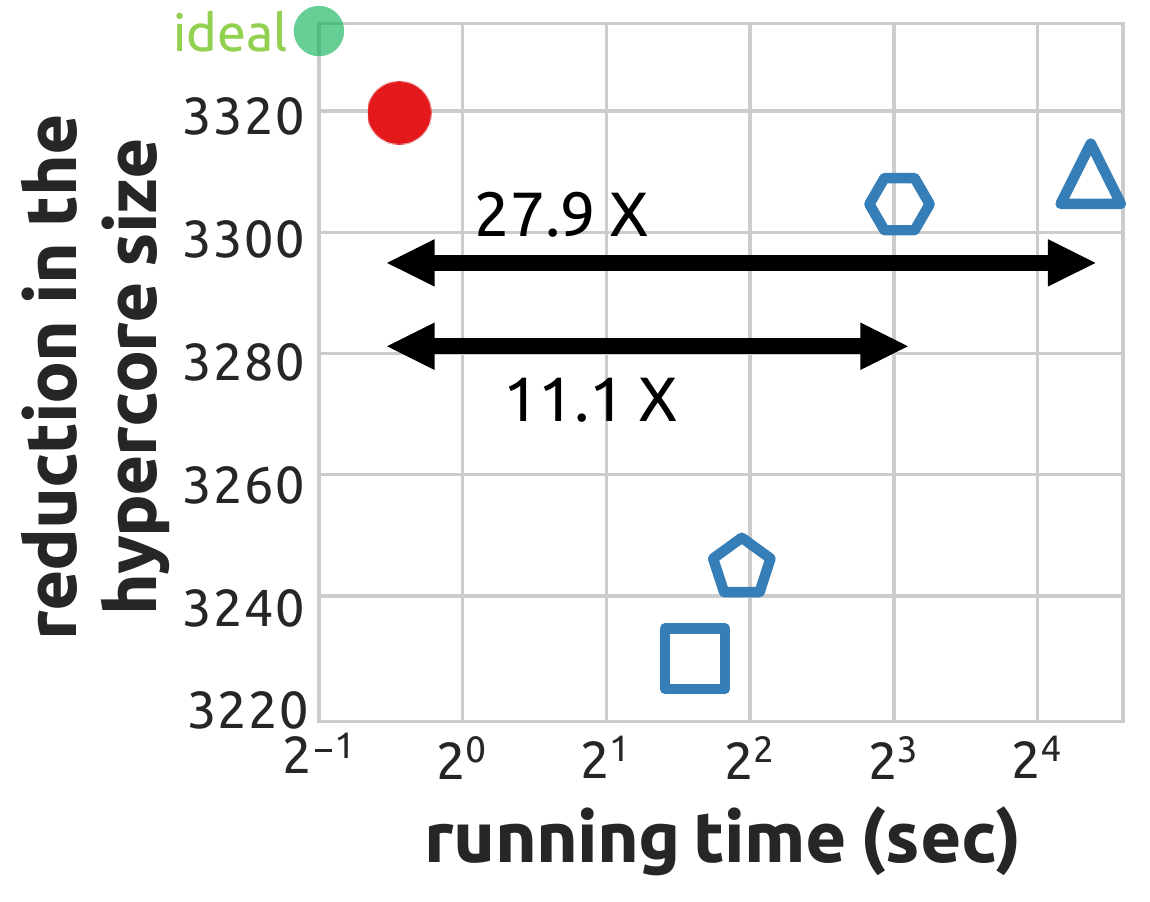}
		\caption{threads-math}
	\end{subfigure}
	\begin{subfigure}[b]{0.32\linewidth}
		\centering
		\includegraphics[scale=0.4]{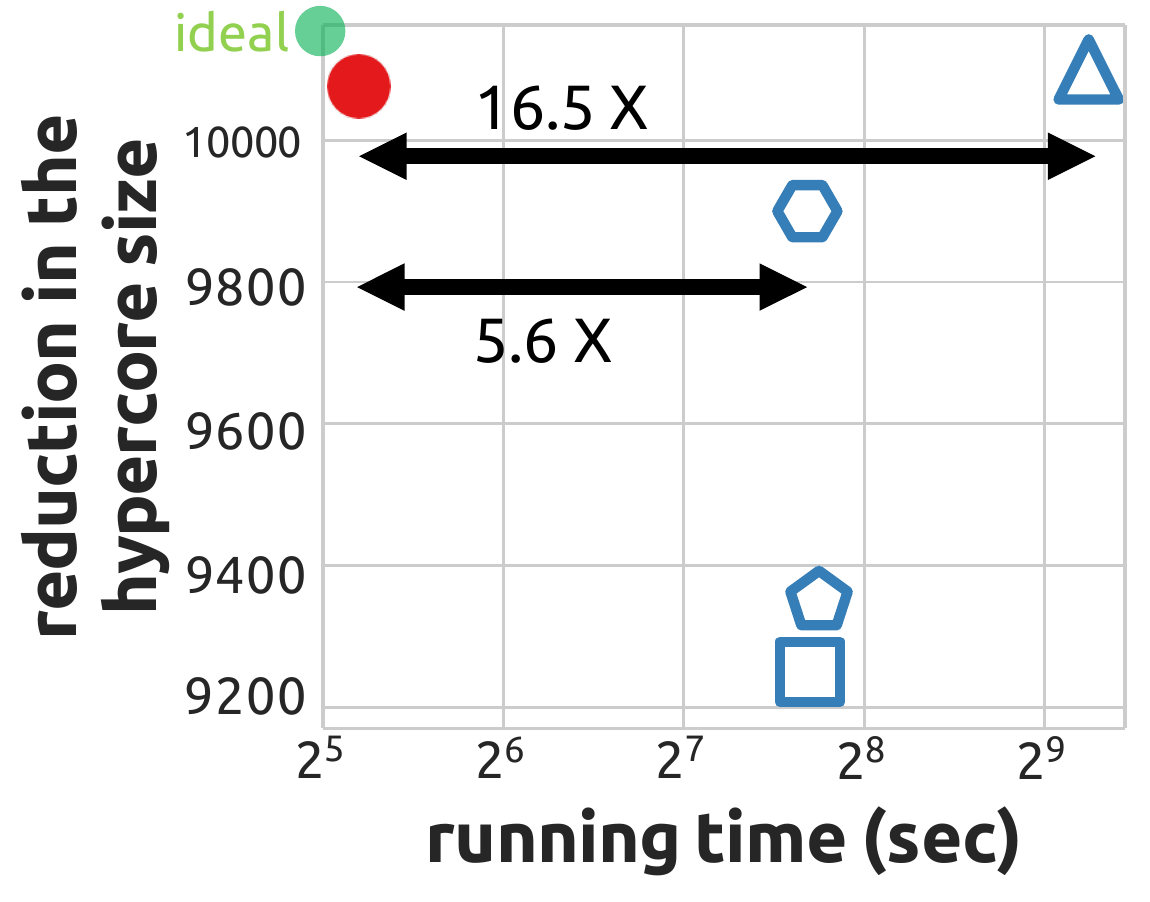}	    
		\caption{threads-SO}
	\end{subfigure}
	\begin{subfigure}[b]{0.32\linewidth}
		\centering
		\includegraphics[scale=0.4]{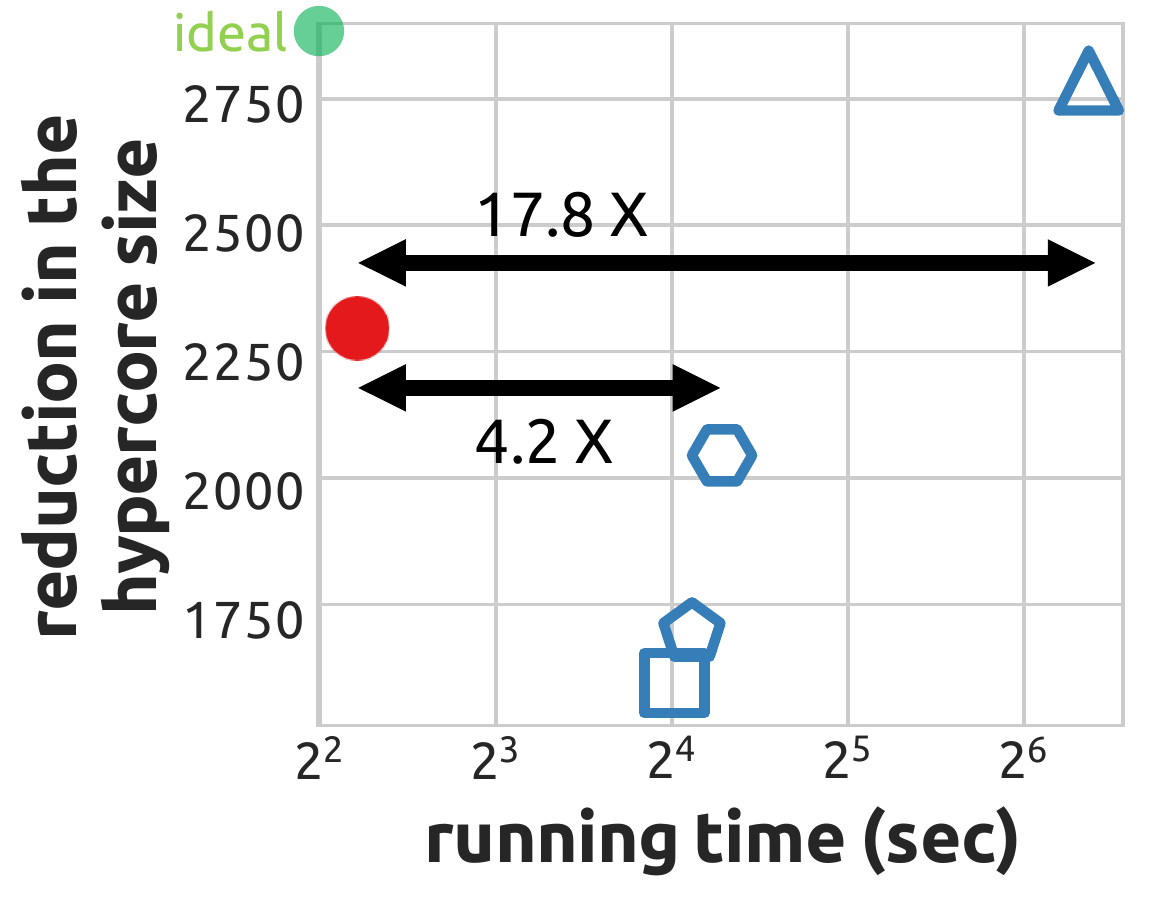}
		\caption{coauth-DBLP}
	\end{subfigure}
	\begin{subfigure}[b]{0.32\linewidth}
		\centering
		\includegraphics[scale=0.4]{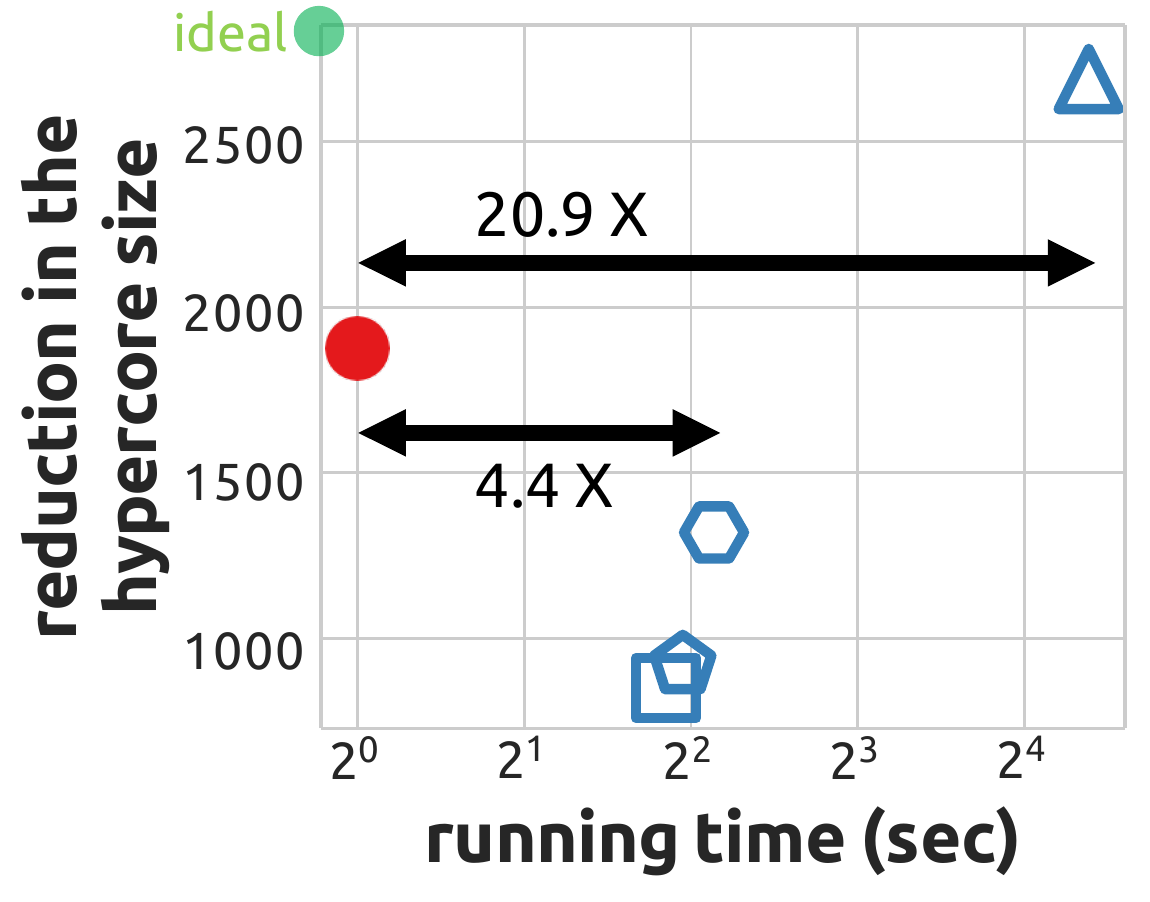}
		\caption{coauth-Geology}
	\end{subfigure}
	\caption{\textbf{\textsc{HyCoM+} shows outstanding efficiency and comparable effectiveness.}
		Given budget $b = 100$, we show the average running time over ten trials and the amount of reduction in the size of the $(10, 0.6)$-hypercore by different algorithms.}
	\label{fig:core_min_res}
\end{figure*}

\begin{figure}[htb]
	\centering
	\includegraphics[scale=0.5]{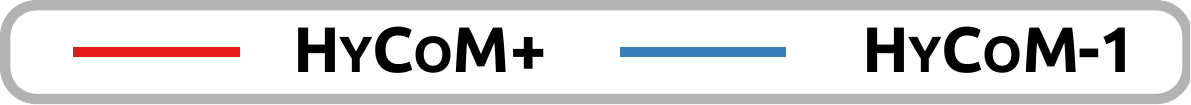}\\     
	\vspace{2mm}
	\begin{subfigure}[b]{0.48\linewidth}
		\centering
		\includegraphics[scale=0.4]{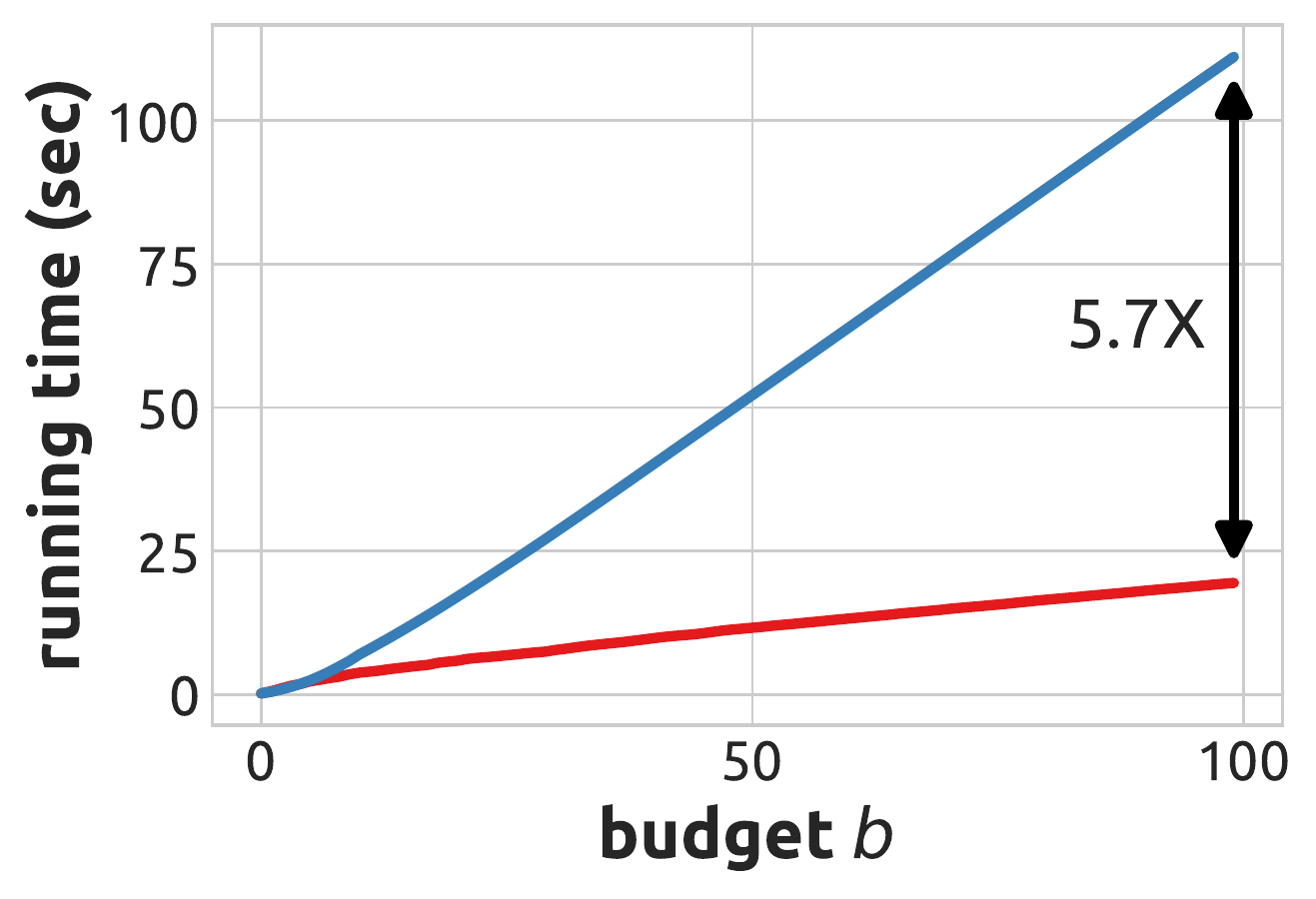}
	\end{subfigure}
	\begin{subfigure}[b]{0.48\linewidth}
		\centering
		\includegraphics[scale=0.4]{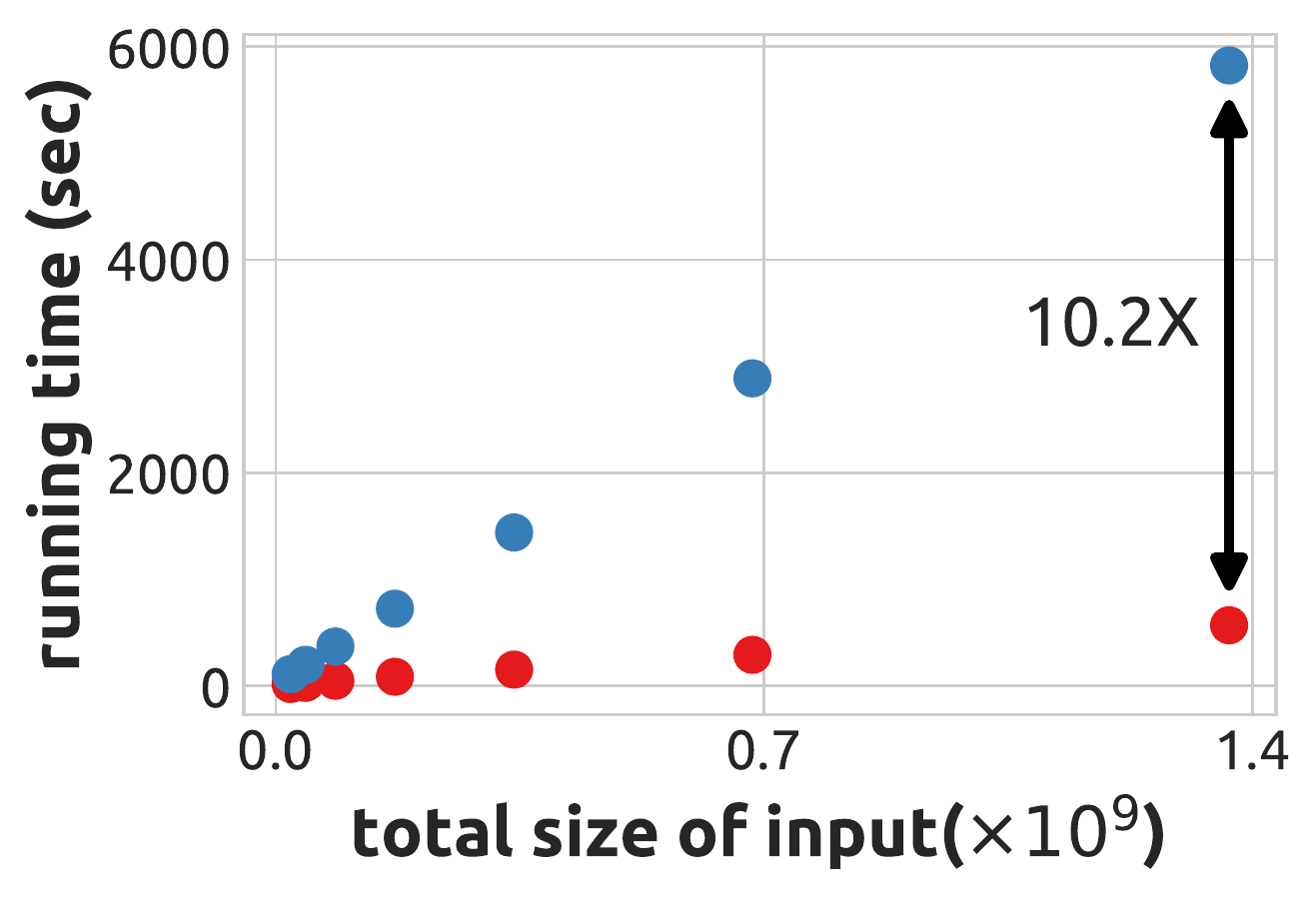}
	\end{subfigure}
	\caption{\textbf{\textsc{HyCoM}+ has linear scalability w.r.t the budget and the hypercore size.}
		On the left, we show the running time of \textsc{HyCoM}-$1$ and \textsc{HyCoM+} with $b$ increasing.
		On the right, we show the running time of \textsc{HyCoM}-$1$ and \textsc{HyCoM+} while upscaling the tags-SO dataset ($b = 100$).
		\textsc{HyCoM+} takes less than $10$ minutes (574 seconds) when the total size of the input hypergraph is $1.37$B ($64 \times$ upscaled).
		\vspace{-1.5mm}}
	\label{fig:core_min_time}
\end{figure}

\begin{problem}[Collapsed $(k, t)$-hypercore problem]
	Given a hypergraph $H = (V, E)$, $k \in \bbN$, $t \in [0, 1]$, and $b \in \bbN$.
	The \textbf{collapsed $(k, t)$-hypercore problem} aims to find $B \in \binom{V}{b}$ so that the size (i.e., the number of nodes) of $(k, t)$-hypercore is minimized when all nodes in $B$ are removed from $H$.
\end{problem}

\begin{algorithm}[t!]
	\small
	\caption{\textsc{HyCoM} / \blue{\textsc{HyCoM+}}} \label{alg:hyperCKC}
	\hspace*{\algorithmicindent} \textbf{Input:} {$H = (V, E)$, $k$, $t$, budget $b$, and max. \# candidates to check $n_c$} \\
	\hspace*{\algorithmicindent} \textbf{Output:} {set of the chosen collapsers $\mathcal{C}$}
	\begin{algorithmic}[1]
		\State $\mathcal{D}(i) \leftarrow \abs{e_i}; \tilde{\mathcal{T}}(i) \leftarrow \max(\ceil{t\abs{e_i}}, 2), \forall i \in I_E$
		\State $H' = (V', E') \leftarrow$ {Alg.~\ref{alg:kt_hypercore_decomp} with $H$, $k$, $t$, and $\mathcal{D}$}
		\State $\mathcal{C} \leftarrow \varnothing$
		\For{$i \in \setbr{1, 2, \ldots, b}$}
		\State $v^{*}\leftarrow \textsc{bestCollapser}()$
		\State $H' = (V', E') \leftarrow$ {Alg.~\ref{alg:kt_hypercore_decomp} with $H' \setminus \setbr{v^*}$, $k$, $t$, and $\mathcal{D}$}
		\State $\mathcal{C} \leftarrow \mathcal{C} \cup \setbr{v^*}$
		\EndFor
		\State \Return {$\mathcal{C}$}
		\Algphase{\textsc{bestCollapser}: find the best collapser}
		\Procedure{\textsc{bestCollapser}}{} \Comment{\blue{The blue parts are for \textsc{HyCoM+}}}
		\State \blue{\textbf{if} $i > 1$ \textbf{then} \textbf{goto} Line~\ref{line:T_comp}} \Comment{\blue{Construct $\tilde{E}$ once}} \label{line:co-edge}
		\State $\tilde{E}(u, v) \leftarrow \setbr{e_i' \in E': \setbr{u, v} \subseteq e_i', \abs{e_i'} = \tilde{\mathcal{T}}(i)}, \forall u, v$ \label{line:tilde_E}
		\State $T \leftarrow \setbr{v \in V': \exists u~s.t.~\abs{\tilde{E}(u, v)} > d(u; H') - k}$ \label{line:T_comp}
		\State \blue{$\mathcal{F}(v) \leftarrow \setbr{u \in V': \abs{\tilde{E}(u, v)} > d(u; H') - k}, \forall v \in T$} \label{line:F_comp}
		\State \textbf{sort} $T$ in the descending order by \blue{$\abs{\mathcal{F}(\cdot)}$ and} degree \label{line:T_sort}
		\State $n_{col}, v_{col}, n_t  \leftarrow \abs{V}, -1, 0$
		\While{$T \neq \varnothing$ {\upshape\textbf{and} $n_t \neq n_c$}}
		\State $v_0 \leftarrow$ the first element in $T$
		\State $H^* \leftarrow$ Alg.~\ref{alg:kt_hypercore_decomp} with $H' \setminus \setbr{v_0}$, $k$, $t$, and $\mathcal{D}$ \label{line:decomp_cand}
		\State \textbf{if} $\abs{V(H^*)} < n_{col}$ \textbf{then} $n_{col}, v_{col} \leftarrow \abs{V(H^*)}, v_0$
		\State $T \leftarrow T \setminus (V \setminus V(H^*))$
		\State $n_t \leftarrow n_t + 1$
		\EndWhile
		\State \blue{\textbf{update} $\tilde{E}$ by the difference between $H'$ and $H_{col}$} \label{line:E_update}
		\State \Return {$v_{col}$}
		\EndProcedure
	\end{algorithmic}
\end{algorithm}

Alg.~\ref{alg:hyperCKC} (with $n_c = -1$) shows the generalization of CKC \citep{zhang2017finding}, which was originally designed for the collapsed $k$-core problem, to the collapsed $(k, t)$-hypercore problem.
Following CKC, in each round, we find {a} best collapser (i.e., {a} node that reduces the size of $(k, t)$-hypercore most) in the candidate set $T$, and update the $(k,t)$-hypercore after removing the chosen collapser, until all $b$ collpasers are chosen.
However, the naive generalization encounters the following problems:
\begin{itemize}
	\item CKC only considers simple pairwise graphs. In simple pairwise graphs, where at most one edge exists between each node pair, the candidate set $T$, i.e., the set of nodes whose deletion will result in the removal of some other node, simply consists of the neighbors of the nodes with degree $k$.
	In hypergraphs, two nodes may co-exist in \textit{multiple} hyperedges. Therefore, we need to additionally check the number of \textit{endangered hyperedges} in the set $\tilde{E}$, where endangered hyperedges are the ones with exactly the minimum size satisfying the threshold determined by $t$.
	Furthermore, we need to count each node pair in each endangered hyperedge, which is time-consuming with time complexity $O(\sum_{e \in E'} \abs{e}^2)$ (Line~\ref{line:tilde_E}).
	To make the situation worse, this process is repeated in each round ($b$ times in total).
	\item CKC computes the $k$-core after removing each candidate to evaluate the candidates.
	Similarly, $(k, t)$-hypercore computation is required for \textit{each} candidate (Line~\ref{line:decomp_cand}), where the number of all candidates can be large.
	Compared to core computation with time complexity linear in the number of edges \citep{batagelj2003m}, as shown in Thm.~\ref{thm:alg_kt}, $(k, t)$-hypercore computation has considerably higher time complexity.
\end{itemize}

We propose \textsc{HyCoM} (\textbf{Hy}per-\textbf{Co}re \textbf{M}inimization) and the further-optimized \textsc{HyCoM+}, which are described in Alg.~\ref{alg:hyperCKC}, to address the above problems with the following improvements:
\begin{itemize}
	\item \textbf{Only checking the most promising candidates.}
	Although we may have a large number of candidates, not {every} candidate is likely to be the best collapser.
	Intuitively, we may set a maximum number of candidates to check in each round ($n_c \ll T$) and only check the most promising ones.
	The technique reduces the time of hypercore decompositions in each round from $O(T)$ to $O(n_c)$, which gives \textsc{HyCoM}.
	We further improve the algorithm by the following two techniques and have \textsc{HyCoM+}.
	\item \textbf{Sorting candidates by the number of direct followers.}
	In \textsc{HyCoM}, the degrees are used to sort the candidates and to find the promising ones.
	However, the degree does not always imply a node's ability in (hyper)core minimization.
	The \textit{direct followers} ($\mathcal{F}(\cdot)$ in Line~\ref{line:F_comp}, i.e., the nodes that will leave the hypercore immediately due to the deletion) of each candidate are accessible without additional cost during the process of finding candidates,
	and the number of direct followers provides a lower bound for the total number of followers.
	Thus, we use the number of direct followers to find the most promising candidates (Line~\ref{line:T_sort}).
	\item \textbf{Incremental update of the endangered hyperedges.}
	It is necessary to find the endangered hyperedges, and we incrementally update the set {whenever} the hypercore is updated (Line \ref{line:E_update}) instead of computing {it} from scratch in each round.
	By doing so, during the whole process, $\tilde{E}$ is constructed from scratch only once.
	The number of hyperedges needed to be checked in each round is the size of the \textit{symmetric difference} between the current set of endangered hyperedges and that in the previous round, which is empirically much less than the total number of hyperedges in the hypercore.
\end{itemize}

\smallsection{Algorithms.} The algorithms to compare are as follows:
\begin{itemize}
	\item \textsc{hyperCKC}: the naive generalization of CKC,
	which is equivalent to \textsc{HyCoM} with $n_c = -1$;
	\item \textsc{HyCoM}-$n_c$: \textsc{HyCoM} with $n_c \in \setbr{1, 10, 100}$;
	\item \textsc{HyCoM+}: \textsc{HyCoM+} with $n_c = 1$, i.e., the fastest version.
\end{itemize}

\smallsection{Settings.}
We conducted all the experiments on a machine with i9-10900K CPU and $64$GB RAM.
All algorithms are implemented in C++, and complied by G++ with O3 optimization.

\smallsection{Results.}
We show the results on five relatively large datasets: coauth-DBLP, coauth-Geology, tags-SO, threads-math, and threads-SO, and use $k = 10$ and $t = 0.6$.
Full results, where we use different datasets and different $(k, t)$ values, are \red{in Table~\ref{tab:min_res} in Appendix~\ref{appendix:extra_exp}.}
In Fig.~\ref{fig:core_min_res}, we report the running time and the reduction in the size of $(k,t)$-hypercore size when using different algorithms with $b = 100$, where \textsc{HyCoM+} shows outstanding efficiency and competent effectiveness.
{We do not count the time used on the initial $(k,t)$-hypercore computation since it is common in all algorithms.} 
In particular, in the tag-SO, thread-math, and threads-SO datasets, the performance of \textsc{HyCoM+} is comparable or even better than that of \textsc{hyperCKC} while \textsc{HyCoM+} is $16.5$-$27.9\times$ faster than \textsc{hyperCKC}.
Besides, on the largest dataset threads-SO whose input hypercore has $301$K nodes and $5.7$M hyperedges, \textsc{HyCoM+} takes only $38.2$ seconds.
In Fig.~\ref{fig:core_min_time}, we show the linear scalability w.r.t the budget and the hypercore size of \textsc{HyCoM} and \textsc{HyCoM+},
where we generate synthetic hypergraphs by upscaling the original ones.
In particular, we duplicate each hyperedge up to $64\times${, which is simple and generates realistic hypergraphs}.

\section{Related Work}\label{sec:relw}
\smallsection{$k$-Hypercores.}
The concept of $k$-cores in pairwise graphs was first proposed in \citep{seidman1983network} and has been used for various applications \citep{shin2016corescope, alvarez2008k, alvarez2006large, peng2014accelerating, corominas2014detection, luo2009core, wood2015minimal,malliaros2020core}.
Most previous works \citep{hua2023revisiting, luo2021hypercore, luo2022hypercore, gabert2021shared, gabert2021unifying, sun2020fully} are {based on} the straightforward generalization of {$k$-cores to hypergraphs} assuming fragile hyperedges (i.e., a hyperedge is removed when \textit{any} node leaves it), which is included in the proposed $(k, t)$-hypercore with $t = 1$.
\cite{limnios2021hcore} and \cite{vogiatzis2013influence}
considered a variant where each hyperedge is kept until only one node remains in it, which is equivalent to the $(\alpha, \beta)$-core in bipartite graphs \citep{liu2020efficient, sariyuce2018peeling} with $(\alpha, \beta) = (k, 2)$, and the proposed $(k, t)$-hypercore with $t = 0$.
No existing work has investigated the spectrum between the two extreme cases above, which is covered by our proposed concepts.

\smallsection{Generalized $k$-cores.}
\cite{zhang2012extracting} generalized $k$-cores to triangle $k$-cores, which are also known as $k$-trusses, to extract the information in pairwise graph.
\cite{peng2018efficient} generalized $k$-cores on uncertain graphs, where each edge exists in a probabilistic way.
Specifically, they considered the problem of $k$-core decomposition on uncertain graphs, and propose the concept of $(k, \theta)$-cores.
\cite{wang2018efficient} generalized $k$-cores on geo-social networks. 
Specifically, they proposed the radius-bounded $k$-core by taking the spatial constraints into consideration.
\cite{zhang2020exploring} generalized $k$-cores to $(k, p)$-cores. 
Specifically, given $k$ and $p$, they further required each node in the $(k, p)$-core to have at least $p$ fraction of its neighbors in the $(k, p)$-core.
\cite{lu2022time} further investigated $(k, p)$-cores on dynamic graphs.
\cite{bonchi2019distance} generalized $k$-cores to $(k, h)$-cores.
Specifically, they relaxed the node-degree condition by requiring each node in the $(k, h)$-cores to have at least $k$ other nodes at a distance at most $h$, i.e., to have at least $k$ $h$-hop neighbors.
\cite{dai2021scaling} further investigated $(k, h)$-cores.
\cite{zhang2017engagement} generalized $k$-cores to $(k, r)$-cores, where they took the similarity between each pair of nodes w.r.t the attributes also into consideration.
\cite{victor2021alphacore} generalized $k$-cores by combining multiple node properties and introducing the notion of data depth. 
\cite{chen2021higher} generalized triangle $k$-cores, i.e., $k$-trusses to $(k, \tau)$-trusses, by taking the $h$-hop neighbors of each node into consideration, which is based on a similar idea of the $(k, h)$-cores \citep{bonchi2019distance}.
\cite{sariyuce2018peeling} and \cite{shin2018fast} proposed to find dense substructures in bipartite graphs and tensors, respectively, by adapting the standard `peeling' algorithm for obtaining the $k$-core.
\cite{gabert2021unifying} used $k$-nuclei, a generalization of $k$-cores and $k$-trusses, to detect dense substructures. 
\cite{preti2021strud} generalized $k$-trusses to simplicial complexes.
\color{black}

\smallsection{Patterns in real-world hypergraphs.}
\cite{do2020structural} proposed to convert hypergraphs into pairwise graphs where each $k$-subset of the node set is regarded as a node in the converted graph and found some pervasive structural patterns.
\cite{lee2020hypergraph} defined hypergraph motifs that depict the connectivity patterns among each of three connected hyperedges. 
They revealed that the frequencies of hypergraph motifs are similar in hypergraphs in the same domain.
\cite{lotito2022higher} also studied hypergraph motifs using different definitions,
and \cite{Kim2023SC3} recently studied motifs in simplicial complexes.\footnote{Simplicial complexes can be seen as a special class of hypergraphs.}
\cite{lee2021hyperedges} defined the degree of overlaps of hyperedges and found some patterns related to the overlap. 
Moreover, temporal patterns have also been explored \citep{kook2020evolution,benson2018simplicial,benson2018sequences}. 
Structural properties (e.g., node centrality measures, the number of graph motifs involving each node) have been used as features of nodes in pairwise graphs~\citep{cui2022positional, he2021learning}.
We believe that structural properties on hypergraphs can also be useful for feature representation~\citep{arya2020hypersage}, especially as inputs of hypergraph neural networks~\citep{feng2019hypergraph,jiang2019dynamic,liao2021hypergraph,bai2021hypergraph,huang2021unignn,chien2021you,gao2022hgnn,kim2022equivariant,xia2022hypergraph,lee2022m,wu2023self,han2023intra}.
\color{black}

\smallsection{Influential-node identification in hypergraphs.}
Besides the trivial degree centrality, a variety of node centrality measures (e.g., eigenvector centrality \citep{bonacich2001eigenvector} and coreness \citep{kitsak2010identification})
have been used to find influential nodes in pairwise graphs~\citep{rossi2015spread}, and some have been generalized to hypergraphs \citep{benson2019three}.
Overall, influential-node identification in hypergraphs is still underexplored, 
although some theoretical analyses have been {made} \citep{zhu2018social, antelmi2021social}.
We provide an efficient and effective metric for practical use.

\section{Conclusion}\label{sec:ccls}
In this paper, we proposed the notion of $(k, t)$-hypercores and some related concepts (Definitions~\ref{def:kt_hypercore}-\ref{def:k_fraction})
for which we presented 
the theoretical properties (Propositions~\ref{prop:uniqueness}-\ref{prop:containment}) and 
computation algorithms (Algorithms~\ref{alg:kt_hypercore_decomp}-\ref{alg:k_fraction}) 
with analyses (Theorems~\ref{thm:alg_kt}-\ref{thm:alg_k}).
Through extensive experiments on real-world hypergraphs, we presented interesting findings from various perspectives (Observations~\ref{obs:domain:kt}-\ref{obs:density}), 
including striking similarities of the hypercore structure within each domain.
We also demonstrated the usefulness of the proposed concepts in identifying influential nodes (Figures~\ref{fig:inf_summary}-\ref{fig:inf_email_Eu}), detecting dense substructures (Figure~\ref{fig:density}), and revealing vulnerabilities (Figures~\ref{fig:core_min_res}-\ref{fig:core_min_time}).
For reproducibility, we made the code and datasets publicly available online {\citep{onlineSuppl}}. 

\section*{Declarations}
\subsection*{Funding}
This work was supported by National Research Foundation of Korea (NRF) grant funded by the Korea government (MSIT) (No. NRF-2020R1C1C1008296) and Institute of Information \& Communications Technology Planning \& Evaluation (IITP) grant funded by the Korea government (MSIT) (No. 2022-0-00871, Development of AI Autonomy and Knowledge Enhancement for AI Agent Collaboration) (No. 2019-0-00075, Artificial Intelligence Graduate School Program (KAIST)).

\subsection*{Competing interests}
The authors have no relevant financial or non-financial interests to disclose.

\bibliographystyle{plainnat}
\bibliography{ref}
\clearpage

\appendix

\section{Proofs}\label{sec:proofs}

\subsection{Proof of Proposition~\ref{prop:uniqueness}}\label{subsec:pf_prop_uniqueness}
\begin{proof}
	{Since $H$ is finite, the number of subhypergraphs of $H$ is also finite.
    Therefore, there exists \textit{one} subhypergraph with maximal total size (which is possibly an empty hypergraph) where each node has degree at least $k$ and at least $t$ proportion of the constituent nodes remain in each hyperedge, completing the proof of existence.}
	To show the uniqueness, suppose the opposite, and let $C^1 = (V^1, E^1)$ and $C^2 = (V^2, E^2)$ be two distinct $(k, t)$-hypercores of $H$.
	Then we consider the hypergraph $C' = (V', E')$ with $E' = \setbr{e^1_i \cup e^2_i: i \in I_{E^1} \cup I_{E^2}}$.
	Clearly, $C'$ is a subhypergraph of $H$ with a larger total size that satisfies the node-degree and hyperedge-fraction conditions, which {contradicts}
	the maximality and completes the proof.
\end{proof}

\subsection{Proof of Proposition~\ref{prop:containment}}\label{subsec:pf_prop_containment}
\begin{proof}
	Suppose that $C_{k, t_2}(H)$ is not a subhypergraph of $C_{k, t_1}(H)$.
	Then we take the union $C_{k, t_2}(H) \cup C_{k, t_1}(H)$ and we obtain a hypergraph that is strictly larger than $C_{k, t_1}(H)$ and satisfies the conditions of $(k, t_1)$-hypercore, which contradicts with the maximality, completing the proof.
	The second statement can be proved similarly.
\end{proof}

\subsection{Proof of Lemma~\ref{lem:kl_hypercore_equiv_bipartite_core}}\label{subsec:pf_lem_kl_hypercore_equiv_bipartite_core}
\begin{proof}
	This equivalence is immediate by two facts.
	First, for each node $v \in V$, the degree of $v$ in $H$ is equal to the degree of $v$ in $G_{bp}(H)$.
	Second, for each hyperedge $e \in E$, the number of nodes in $e$ is equal to the degree of $e$ in $G_{bp}(H)$.
	With the above two facts, this equivalence immediately follows.
\end{proof}

\subsection{Proof of Lemma~\ref{lem:kl2_hypercore_equiv_kt0_hypercore}}\label{subsec:pf_lem_kl2_hypercore_equiv_kt0_hypercore}
\begin{proof}
	By Def.~\ref{def:kt_hypercore}, when $t = 0$, the definition of $C_{k; t = 0}(H)$ is the maximal subhypergraph of $H$ where 
	(1) every node in $C_{k, t=0}(H)$ has degree at least $k$ and 
	(2) at least two nodes remain in every hyperedge of $C_{k, t=0}(H)$.
	Such a definition exactly coincides with $\tilde{C}_{k; \ell = 2}(H)$, completing the proof.
\end{proof}

\subsection{Proof of Lemma~\ref{lem:diff_kt_kl_hcore}}\label{subsec:pf_lem_diff_kt_kl_hcore}
We can understand the differences between $(k; \ell)$-hypercores and $(k, t)$-hypercores by two intuitions.
	When we obtain the $(k; \ell)$-hypercore of a given $H$ with $\ell > 2$, all hyperedges of cardinality $2$ are removed in the first place. Therefore, if we want to find an $\ell$ such that $\tilde{C}_{k; \ell} = C_{k, t}$ where $C_{k, t}$ contains any hyperedge of cardinality $2$, the only possible $\ell$ value is $\ell = 2$.
	Since the threshold in the $(k, t)$-hypercore is proportional, it imposes different absolute cardinality thresholds for hyperedges of different sizes. On the contrary, the $(k; \ell)$-hypercore imposes the same absolute cardinality threshold for each hyperedge.
We shall show two counterexamples from the two intuitions above.
\begin{proof}
	Consider $H = (V, E)$ with
	$E = \setbr{\setbr{1, 2}, \setbr{1, 3} \setbr{1, 2, 3, 4}, \setbr{1, 3, 4, 5, 6}}.$
	The $(k=2, t=3/4)$-hypercore of $H$ consists of the hyperedges
	$\setbr{\setbr{1,2}, \setbr{1,3}, \setbr{1, 2, 3}},$
	where the hyperedge $\setbr{1,3,4,5,6}$ is totally removed since only $3/5 < t = 3/4$ of the constituent nodes remain by the node-degree threshold $k = 2$. For the $(k; \ell)$-hypercore, the $(k=2; \ell=2)$-hypercore of $H$ consists of the hyperedges
	$\setbr{\setbr{1, 2}, \setbr{1, 3}, \setbr{1, 2, 3, 4}, \setbr{1, 3, 4}}$;
    the $(k=2; \ell=3)$-hypercore of $H$ consists of the hyperedges
	$\setbr{\setbr{1, 3, 4}, \setbr{1, 3, 4}}$;
    when $\ell \geq 4$, the $(k=2;\ell)$-hypercore of $H$ is empty, completing the proof.
	
 We show another counterexample. Consider $H = (V, E)$ with
	$$E = \setbr{\setbr{1,2,3,4}, \setbr{1,2,5,6}, \setbr{5,6,7,8},\setbr{3,4,9,10,11},\setbr{1,2,3,4,5,6,7,8}}.$$
	The $(k=3, t=1/2)$-hypercore of $H$ consists of the hyperedges 
	$$\setbr{\setbr{1,2}, \setbr{1,2,5,6}, \setbr{5,6}, \setbr{1,2,5,6}}.$$
	For the $(k; \ell)$-hypercore, when $\ell = 2$, the $(k=3; \ell=2)$-hypercore of $H$ consists of the hyperedges $$\setbr{\setbr{1,2,3,4},\setbr{1,2,5,6},\setbr{5,6},\setbr{3,4},\setbr{1,2,3,4,5,6}};$$
	when $\ell \geq 3$, the $(k=3; \ell)$-hypercore of $H$ is empty, completing the proof.
\end{proof}

\begin{remark}\label{rem:more_disc_kt_kl_hypercore}
	Our proposed $(k, t)$-hypercore allows arbitrarily fine-grained adjustment since the value of $t$ is continuous in $[0, 1]$, while $\ell$ must be an integer.
	In real-world hypergraphs, many hyperedges are of cardinality $2$. Therefore, the $(k; \ell)$-hypercore with $\ell > 2$ is significantly less meaningful than the $(k, t)$-hypercore since many hyperedges are not taken into consideration at all.
	See Tbl.~\ref{tab:num_card_edges} for the detailed number of hyperedges of different cardinality in each dataset we have used.
	See Figs.~\ref{fig:inf_summary} and \ref{fig:inf_email_Eu} for the performance of hypercoreness w.r.t $(k; \ell)$-hypercore to indicate the influence of nodes.
	Note again that the $(k; \ell=2)$-hypercore is included in our proposed concept as the $(k, t=0)$-hypercore.
	We observe that in most datasets, the $(k; \ell)$-hypercores become less meaningful and fail to indicate the influence of nodes when $\ell$ becomes large, as expected.
\end{remark}

\color{black}

\subsection{Proof of Theorem~\ref{thm:alg_kt}}\label{subsec:pf_thm_alg_kt}
\begin{proof}
	\noindent \textbf{Correctness.}
	The size of a hyperedge changes only when some node in $\mathcal{R}$ is removed from it, and the degree of a node changes only when some incident hyperedge is removed.
	Therefore, when Algorithm~1 ends, each node has degree at least $k$, otherwise it must have been included in $\mathcal{R}$ and removed, and each hyperedge satisfies the hyperedge-fraction condition, otherwise it must have been removed.
	This implies that the output of Algorithm~1 satisfies both the node-degree and hyperedge-fraction conditions w.r.t $H$, $k$, and $t$.
	We now show the maximality.
	Suppose not, and let $(v, e_i)$ be the first node-hyperedge pair that appears during the process of Algorithm~1 with $v \in e_i \in E(C_{k, t})$ but $v \notin e_i' \in E'$, where $C' = (V', E')$ is the returned hypergraph.
	This implies that $v$ is removed from $e$, and thus $v$ is included in $\mathcal{R}$ because its degree has been below $k$.
	However, by the definition of the $(k, t)$-hypercore and the assumption that $(v, e_i)$ is the first pair, before the deletion, the degree of $v$ is at least $k$, which completes the proof by contradiction.
	
	\noindent \textbf{Time complexity.}
	We assume the input hypergraph has been loaded in the memory and thus do not count the complexity of loading the hypergraph.
	Checking the initial degrees (Line~1) takes $O(\abs{V})$.
	In the while loop, each node is added to the set of nodes to be removed at most once since each node is added exactly when its degree decreases from $k$ to $k - 1$. Therefore, this process takes $O(\abs{V})$.
	{By checking the incident edges of each node in $\mathcal{R}$, we find all $e'_i$s intersecting with $\mathcal{R}$, which takes $O(\abs{\mathcal{R}}) = O(\abs{V})$.
		Hash tables are used to implement the sets.}
	Before a hyperedge $e \in E$ is totally removed, \red{at least $\max\left(\ceil{t\abs{e}}, 2\right)$ nodes remain in it (otherwise it has been removed earlier)}, and thus
	it can be visited at most $\abs{e} - \max\left(\ceil{t\abs{e}}, 2\right) + 1$ times \red{(because one node is removed at each time)}.
	This process takes $O(\sum_{e \in E} (\abs{e} - \max\left(\ceil{t\abs{e}}, 2\right) + 1)) = O(\abs{E} + (1 - t) \sum_{e \in E} \abs{e})$.
	\red{Therefore, the total time complexity is $O(\abs{V}) + O(\abs{E} + (1 - t) \sum_{e \in E} \abs{e}) = O(\abs{E} + (1 - t) \sum_{e \in E} \abs{e})$.}
\end{proof}

\subsection{Proof of Theorem~\ref{thm:alg_t}}\label{subsec:pf_thm_alg_t}
\begin{proof}
	\noindent \textbf{Correctness.}
	For each {node} $v$, the assignment of $c_t(v)$ happens only once when $v \in \mathcal{R}$, i.e., before its deletion.
	By Theorem~1, $c_t(v) = k - 1$ implies that $v$ is not in the $(k, t)$-hypercore but in the previous $(k', t)$-hypercore where each node has degree at least $k - 1$, i.e., $v$ is in the $(k-1, t)$-hypercore.\\
	
	\noindent \textbf{Time complexity.}
	The values of $k$ increases $O(c_t^*)$ times, thus the process in Lines~4 and 5 is repeated for $O(c_t^*)$ times and takes $O(c_t^* \abs{V})$.
	The assignment of $t$-hypercoreness of each node (Line~7) takes $O(\abs{V})$.
	As shown in the proof of Theorem~1, each hyperedge is visited at most $\abs{e} - \max(\ceil{t\abs{e}}, 2) + 1$ times before being deleted and each node is added to the set of nodes to be removed only once. Therefore, the remaining process takes $O(\abs{V} + \abs{E} + (1 - t) \sum_{e \in E} \abs{e})$.
\end{proof}

\subsection{Proof of Theorem~\ref{thm:alg_k}}\label{subsec:pf_thm_alg_k}
\begin{proof}
	\noindent \textbf{Correctness.}
	For each node $v$, the assignment of $f_k(v)$ happens only once when $v \in \mathcal{R}$, i.e., before its deletion.
	By Theorem~1, $f_k(v) = t$ implies that $v$ is in the $(k, t)$-hypercore with degree $k$ and is in at least one hyperedge that is in {the $(k,t)$-hypercore but} not in any $(k, t')$-hypercore with $t' > t$.
	Therefore, $v$ is not in any $(k, t')$-hypercore with $t' > t$.
	
	\noindent \textbf{Time complexity.}
	Recording the hyperedge sizes (Line~1) takes $O(\abs{E})$.
	By Theorem~1, computing $C_{k, 0}$ (Line~2) takes $O(\sum_{e \in E} \abs{e})$.
	As shown in the previous proofs, the while loop (Lines~4 to 12) takes $O(\abs{V} + \abs{E} + \sum_{e \in E} \abs{e}) = O(\sum_{e \in E} \abs{e})$.
\end{proof}

%

\section{Details of datasets}
In this section, we provide more details of the datasets used in our experiments.

\begin{itemize}
	\item \smallsection{coauth-DBLP/Geology.} In these two \textit{coauthorship} hypergraphs, each hyperedge represents a publication, and the constituent nodes of a hyperedge represent the authors of the corresponding publication.
	\item \smallsection{NDC-classes/substances.} In these two hypergraphs from the \textit{National Drug Code (NDC) Directory}, each hyperedge represents a drug (with its unique NDC code), and the constituent nodes of a hyperedge represent the class labels (for NDC-classes) or the ingredients (for NDC-substances) of the drug. 
	\item \smallsection{contact-high/primary.} In these two \textit{contact} hypergraphs, each hyperedge represents a group of interacting individuals (the constituent nodes) within a predetermined time period. 
 	\item \smallsection{email-Enron/Eu.} In these two \textit{email} hypergraphs, each hyperedge represents an email (possibly sent to multiple people individually at the same time), which contains the sender and all the receivers as its constituent nodes.
	\item \smallsection{tags-ubuntu/math/SO.} In these three \textit{tags} hypergraphs from \url{https://stackoverflow.com/}, each node represents a tag, and each hyperedge represents a question, where each constituent node represents a tag applied to the question.
	\item \smallsection{threads-ubuntu/math/SO.} In these three \textit{threads} hypergraphs also from \url{https://stackoverflow.com/}, each hyperedge represents a thread, where each constituent node represents a person that participates in it.
\end{itemize}

We have used the preprocessed version of the datasets where each hyperedge consists of at most $25$ nodes.
In Table~\ref{tab:num_card_edges}, we report the number of hyperedges of different cardinality in each dataset.
Notably, on \url{https://www.cs.cornell.edu/~arb/data/}, the \textit{full} version of the datasets, in which the cardinality of the hyperedges is not limited, is also available.

\begin{table}[t!]
	\begin{center}
		\caption{\textbf{The number of hyperedges of different cardinality in each dataset.} For each dataset, we list the number of hyperedges of each specific size. Specifically, in most datasets, a large number of hyperedges are of cardinality $2$. We use $E_s$ to denote the set of hyperedges of cardinality $s$, for each $s$.} \label{tab:num_card_edges}
		\resizebox{\columnwidth}{!}{%
			\begin{tabular}{ lrrrrrr }
				\toprule
				\textbf{Dataset}& $\abs{E}$ & $\abs{E_2}$ & $\abs{E_3}$ & $\abs{E_4}$ & $\abs{E_5}$ & $\abs{\bigcup_{s > 5} E_s}$ \\
				\midrule
				coauth-DBLP 	& 2,169,663 & 693,364 (31.96\%) & 667,302 (30.76\%) & 419,431 (19.33\%) & 205,965 (09.49\%) 	& 183,601 (08.46\%) \\
				coauth-Geology 	& 908,516   & 275,736 (30.35\%) & 227,950 (25.09\%) & 159,509 (17.56\%) & 99,140 (10.91\%) 	& 146,181 (16,09\%) \\
				\midrule
				NDC-classes 	& 1,047 	& 297 (28.37\%)		& 121 (11.56\%)		& 125 (11.94\%)		& 94 (08.98\%)		& 410 (39.16\%) \\
				NDC-substances  & 6,264 	& 1,130 (18.04\%)	& 745 (11.89\%)		& 535 (08.54\%)		& 500 (07.98\%)		& 3,354 (53.54\%) \\
				\midrule
				contact-high 	& 7,818     & 5,498 (70.32\%)	& 2,091 (26.75\%)	& 222 (02.84\%)		& 7 (00.09\%)		& 0 (00.00\%) \\
				contact-primary & 12,704	& 7,748 (60.99\%)	& 4,600 (36.21\%)	& 347 (02.73\%)		& 7 (00.09\%)		& 0 (00.00\%) \\
				\midrule
				email-Enron		& 1,457		& 809 (55.53\%)		& 317 (21.76\%)		& 138 (09.47\%)		& 63 (04.32\%)		& 130 (08.92\%) \\
				email-Eu		& 24,399	& 12,753 (52.27\%)	& 4,938 (20.24\%)	& 2,294 (09.40\%)	& 1,359 (05.57\%)	& 3,055 (12.52\%) \\
				\midrule
				tags-ubuntu 	& 145,053   & 28,138 (19.40\%)  & 52,282 (36.04\%)	& 39,158 (27.00\%)	& 25,475 (17.56\%)	& 0 (00.00\%) \\
				tags-math 		& 169,259   & 25,253 (14.92\%)  & 63,870 (37.74\%)	& 50,892 (30.07\%)	& 29,244 (17.28\%)	& 0 (00.00\%) \\
				tags-SO			& 5,517,054 & 399,051 (07.23\%)	& 1,537,702 (27.87\%)& 1,947,542 (35.30\%) & 1,632,759 (29.59\%) & 0 (00.00\%) \\
				\midrule
				threads-ubuntu  & 115,987  	& 88,301 (76.13\%)	& 21,621 (18.64\%)	& 4,560 (03.93\%)	& 1,117 (00.96\%) 	& 388 (00.33\%) \\
				threads-math 	& 535,323  	& 319,601 (59.70\%)	& 142,065 (26.54\%)	& 49,198 (09.19\%)	& 16,402 (03.06\%)   & 8,057 (01.51\%) \\
				threads-SO		& 8,589,420 & 5,210,916 (60.67\%)& 2,102,208 (24.47\%) & 787,701 (09.17\%) & 299,172 (03.48\%)& 189,423 (02.21\%) \\
				\bottomrule
			\end{tabular}
		}
	\end{center}
\end{table}

\begin{table}[t!]
	\begin{center}
		\caption{For each dataset, we report the information gain over degree, for each of the considered quantities: $\ell$-hypercoreness with $\ell \in \setbr{3, 4, 5}$, neighbor-hypercoreness, and neighbor-degree-hypercoreness.} \label{tab:info_gain_others}
			\begin{tabular}{ lccccc }
				\toprule
				\textbf{Dataset}& $\ell = 3$ & $\ell = 4$ & $\ell = 5$ & neighbor & neighbor-degree \\
				\midrule
				coauth-DBLP 	& 1.386 & 1.470 & 1.273 & 3.053 & 0.701 \\
				coauth-Geology 	& 1.239 & 1.394 & 1.313 & 3.262 & 0.447 \\
				\midrule
				NDC-classes 	& 1.186 & 1.190 & 1.159 & 2.515 & 0.592 \\
				NDC-substances  & 1.221 & 1.459 & 1.607 & 3.974 & 0.975 \\
				\midrule
				contact-high 	& 1.456 & 0.893 & 0.173 & 1.685 & 1.231 \\
				contact-primary & 0.725 & 0.534 & 0.204 & 0.886 & 0.545 \\
				\midrule
				email-Enron		& 1.228 & 1.176 & 1.114 & 1.340 & 1.020 \\
				email-Eu		& 1.706 & 1.654 & 1.627 & 1.929 & 1.305 \\
				\midrule
				tags-ubuntu 	& 2.104 & 2.176 & 1.919 & 2.734 & 1.467 \\
				tags-math 		& 1.489 & 1.683 & 1.531 & 1.876 & 1.203 \\
				tags-SO			& 2.530 & 3.213 & 3.007 & 4.046 & 2.219 \\
				\midrule
				threads-ubuntu  & 0.967 & 0.451 & 0.207 & 1.253 & 0.409 \\
				threads-math 	& 1.510 & 1.061 & 0.616 & 2.146 & 0.626 \\
				threads-SO		& 1.738 & 1.261 & 0.801 & 2.394 & 0.930 \\
				\bottomrule
			\end{tabular}
	\end{center}
\end{table}

\section{Additional experimental results}
\label{appendix:extra_exp}
In this section, we provide additional experimental results supplementing the main text.
In Fig.~\ref{fig:rank_cor_sr}, we report the results regarding the statistical difference between $t$-hypercoreness and
other centrality measures, as well as among $t$-hypercoreness with different $t$, on the datasets not covered in the main text.
In Fig.~\ref{fig:info_gain_sr}, we report the results regarding the information gain over degree, on the datasets not covered in the main text.
In Table~\ref{tab:info_gain_others}, for each dataset, we report the information gain over degree for the following quantities: $\ell$-hypercoreness with $\ell \in \setbr{3, 4, 5}$,\footnote{The case $\ell = 2$ is included in the proposed concept of $t$-hypercoreness with $t = 0$.} neighbor-hypercoreness, and neighbor-degree-hypercoreness (Defs.~\ref{def:l_hypercoreness}, \ref{def:nbr_hypercoreness}, and \ref{def:nbr_deg_hypercoreness}).
Notably, we do not claim that higher information gain is always better, since degree is still a reason measure by cohesiveness, and being too different from degrees can be negative as a cohesiveness measure.
In Fig.~\ref{fig:inf_summary_sr}, we report the results on influential-node identification, on the datasets not covered in the main text.
In Table~\ref{tab:min_res}, we report the full results of the collapsed $(k, t)$-hypercore problem.

\color{black}

\begin{figure*}[htb]
	\centering
	\begin{subfigure}[b]{0.9\textwidth}
		\centering
		\includegraphics[scale=0.5]{fig6_correlation_legend.pdf}\\
	\end{subfigure}  
	\begin{subfigure}[b]{0.48\linewidth}
		\centering
		\hspace{-5mm}
		\includegraphics[scale=0.4]{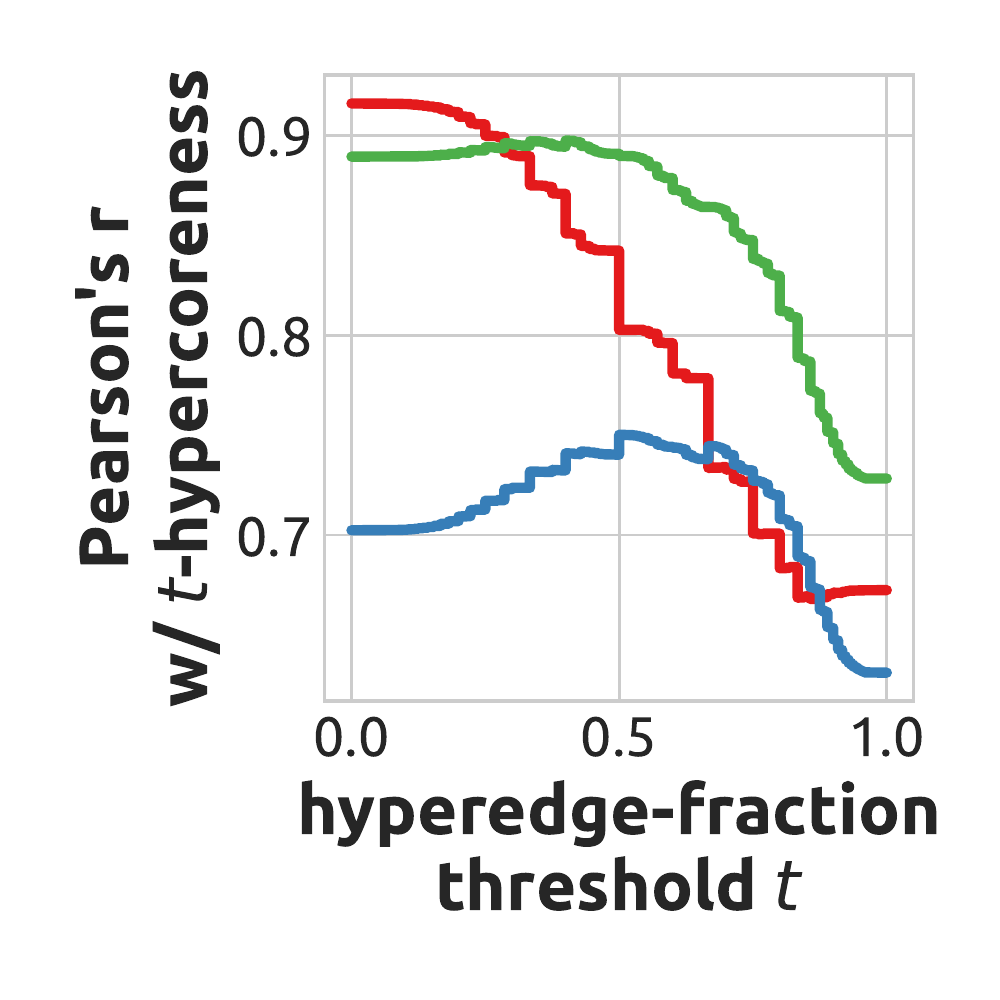}
		\includegraphics[scale=0.4]{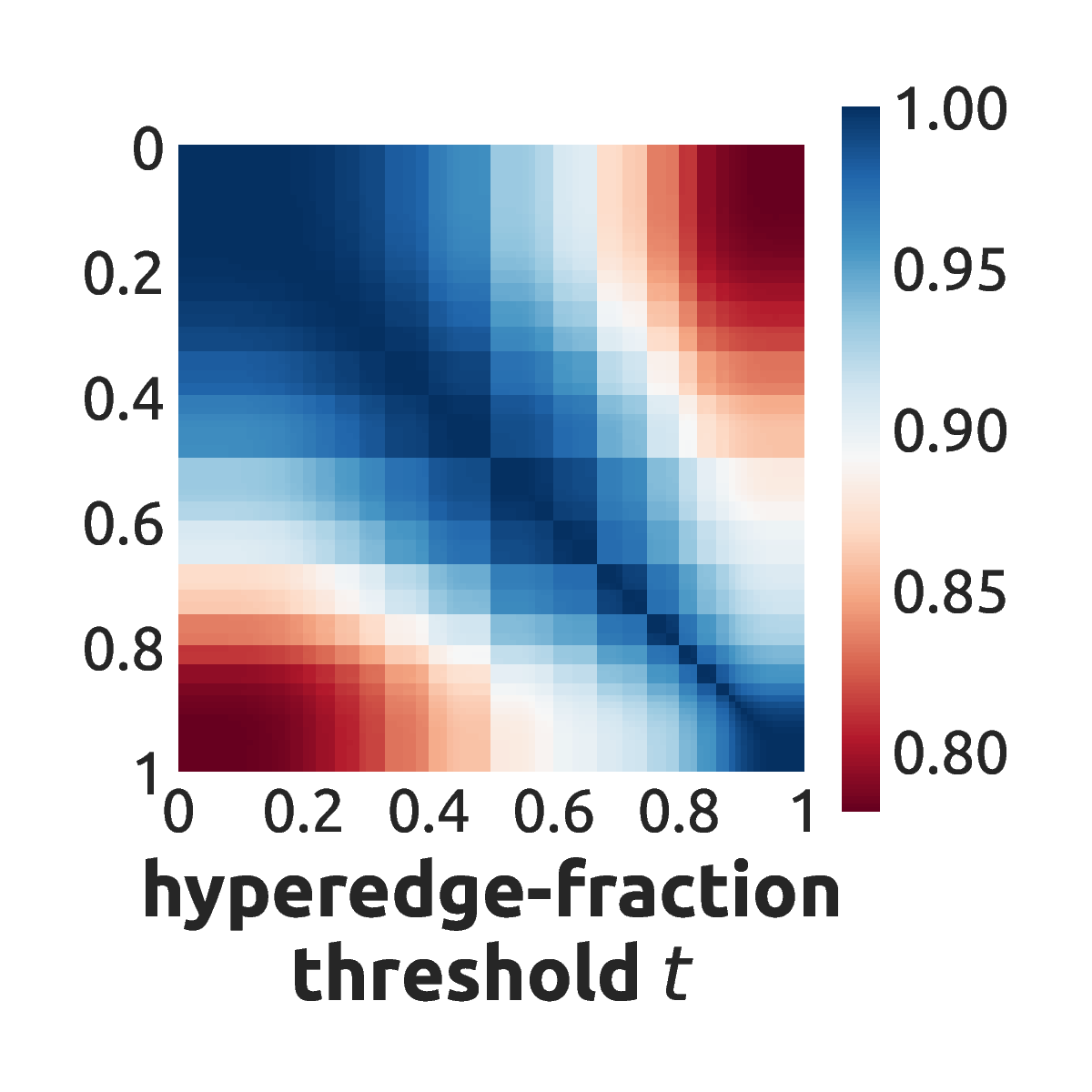}
		\vspace{-3mm}
		\caption{coauth-Geology}
	\end{subfigure}
	\begin{subfigure}[b]{0.48\linewidth}
		\centering
		\hspace{-5mm}
		\includegraphics[scale=0.4]{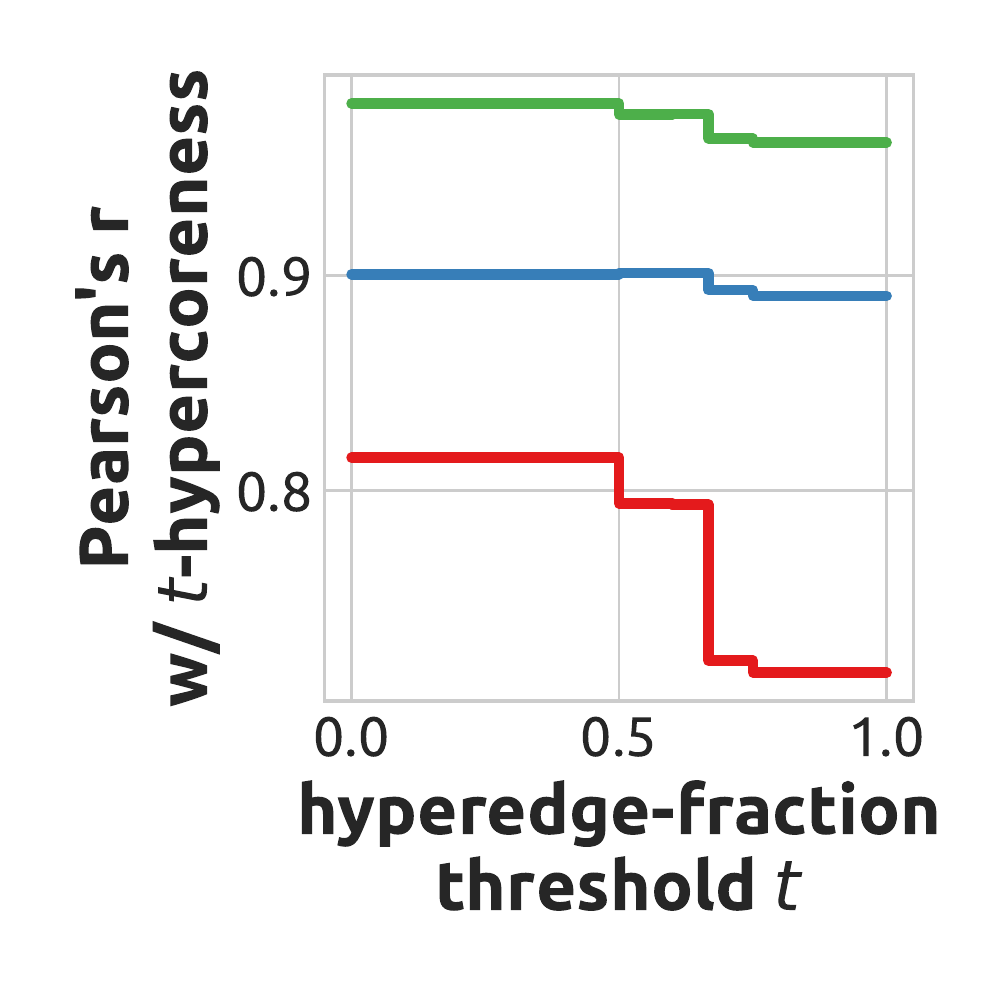}
		\includegraphics[scale=0.4]{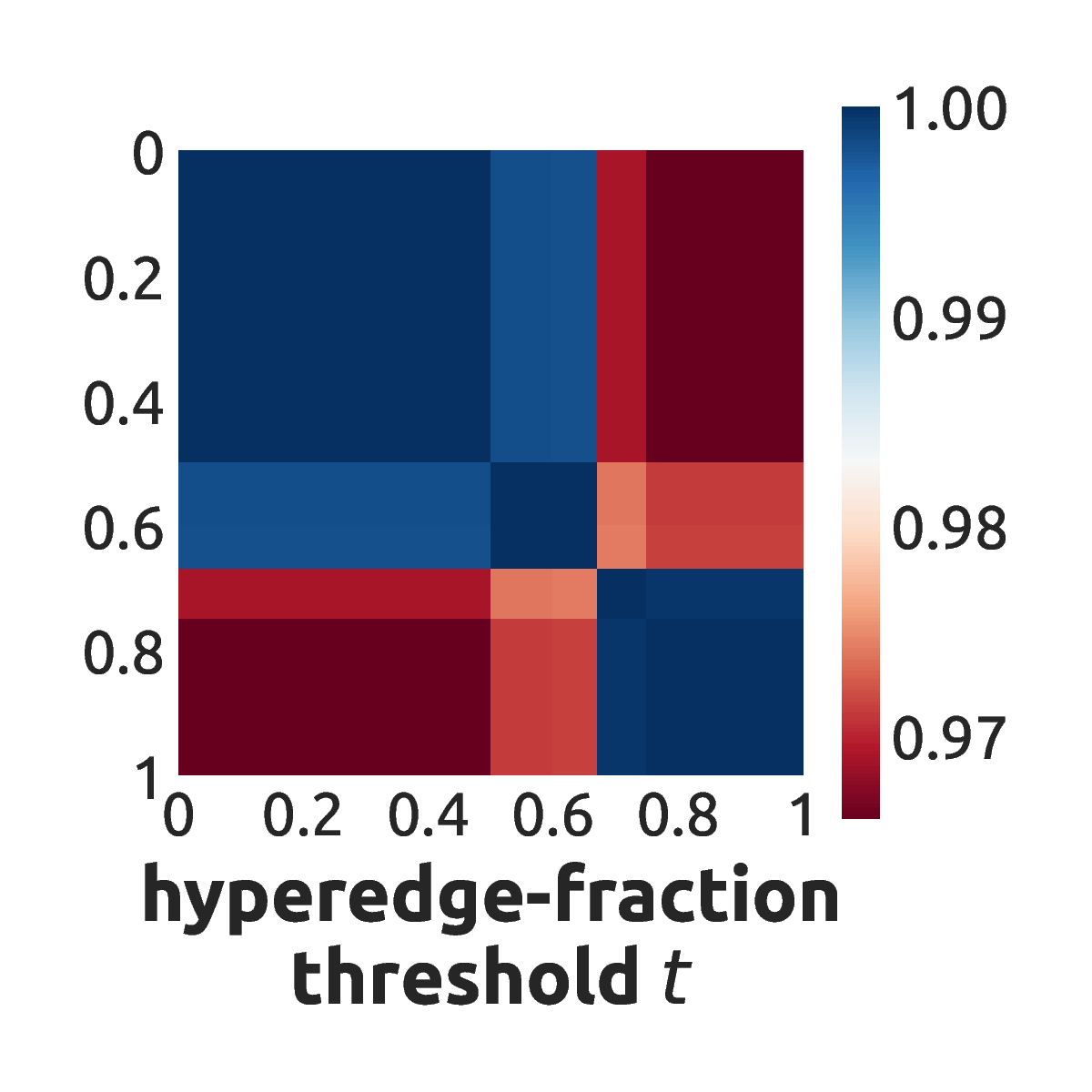}   
		\vspace{-3mm}
		\caption{contact-primary}
	\end{subfigure}
	\begin{subfigure}[b]{0.48\linewidth}
		\centering
		\hspace{-5mm}
		\includegraphics[scale=0.4]{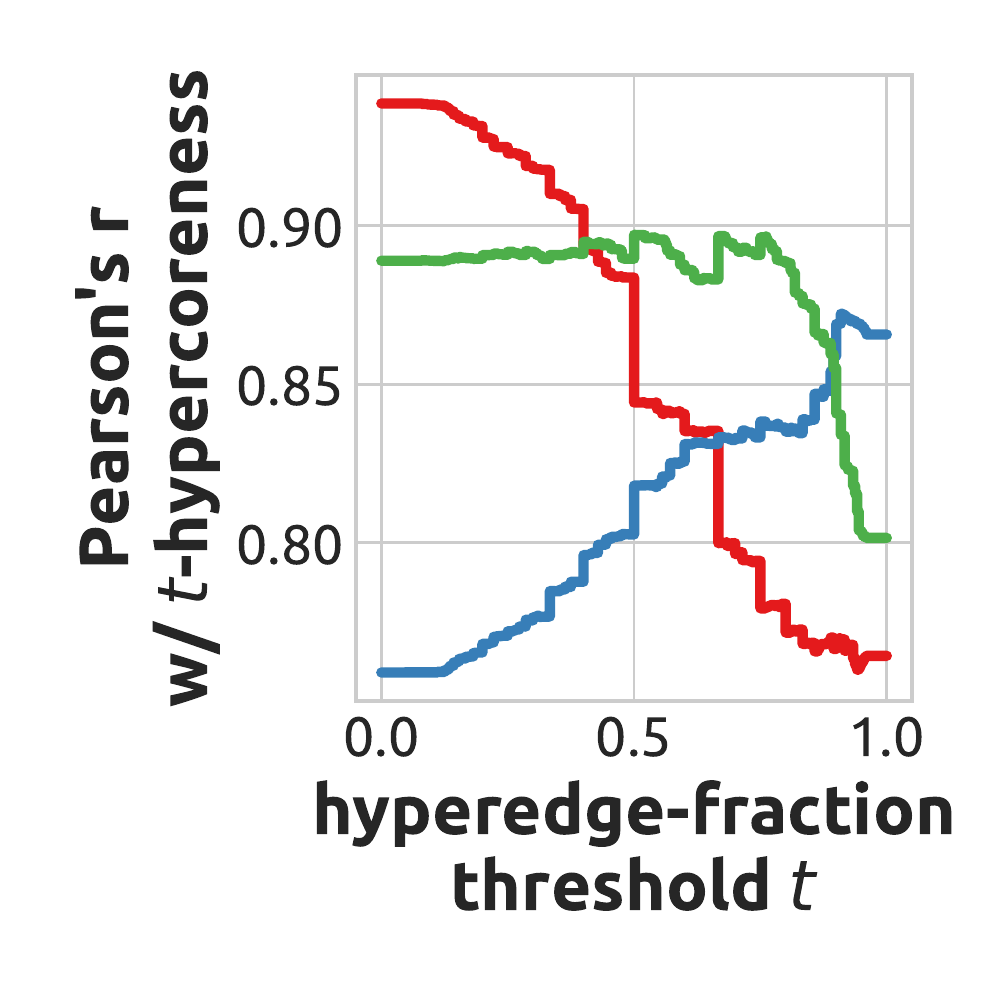}
		\includegraphics[scale=0.4]{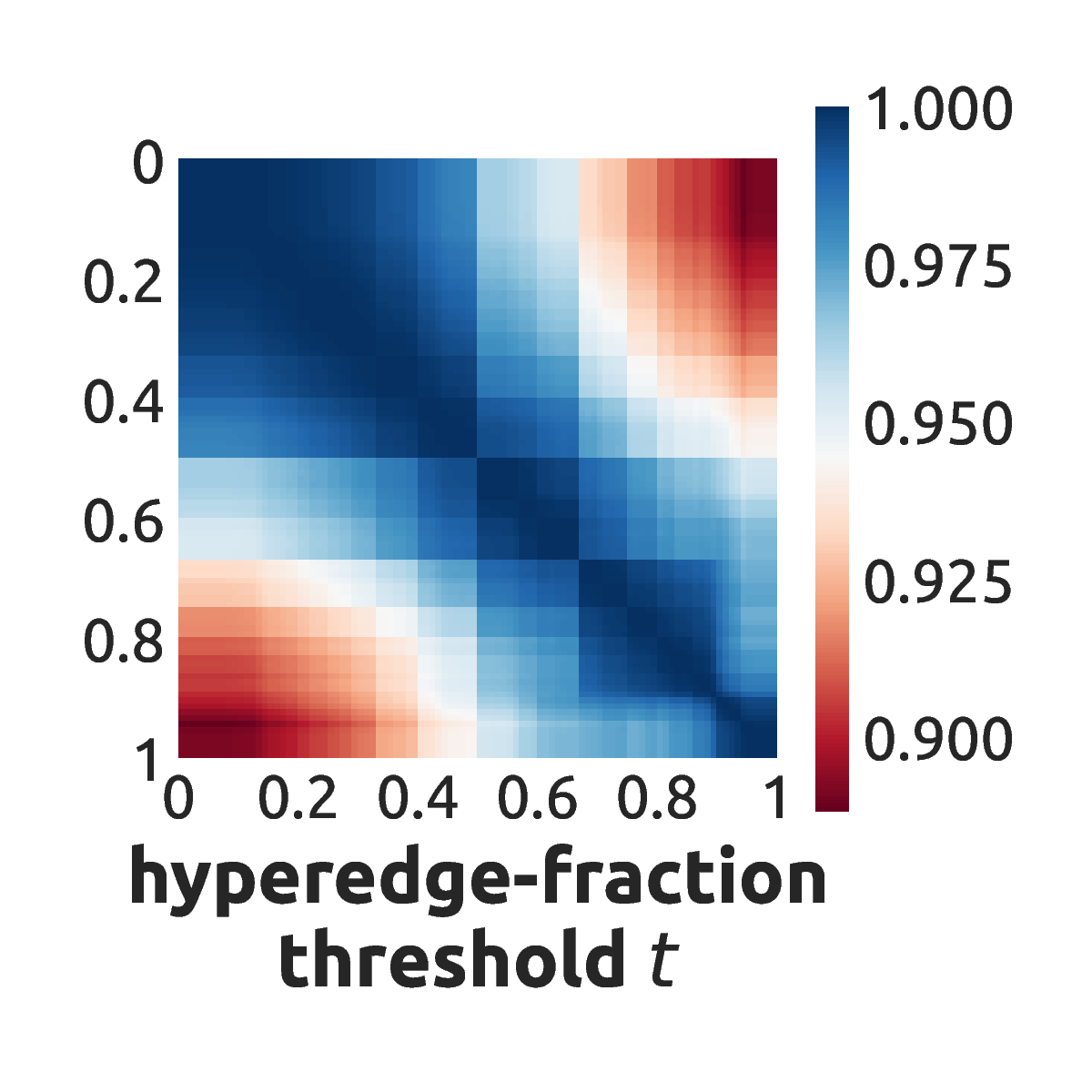}             
		\vspace{-3mm}
		\caption{email-Eu}
	\end{subfigure}
	\begin{subfigure}[b]{0.48\linewidth}
		\centering
		\hspace{-5mm}
		\includegraphics[scale=0.4]{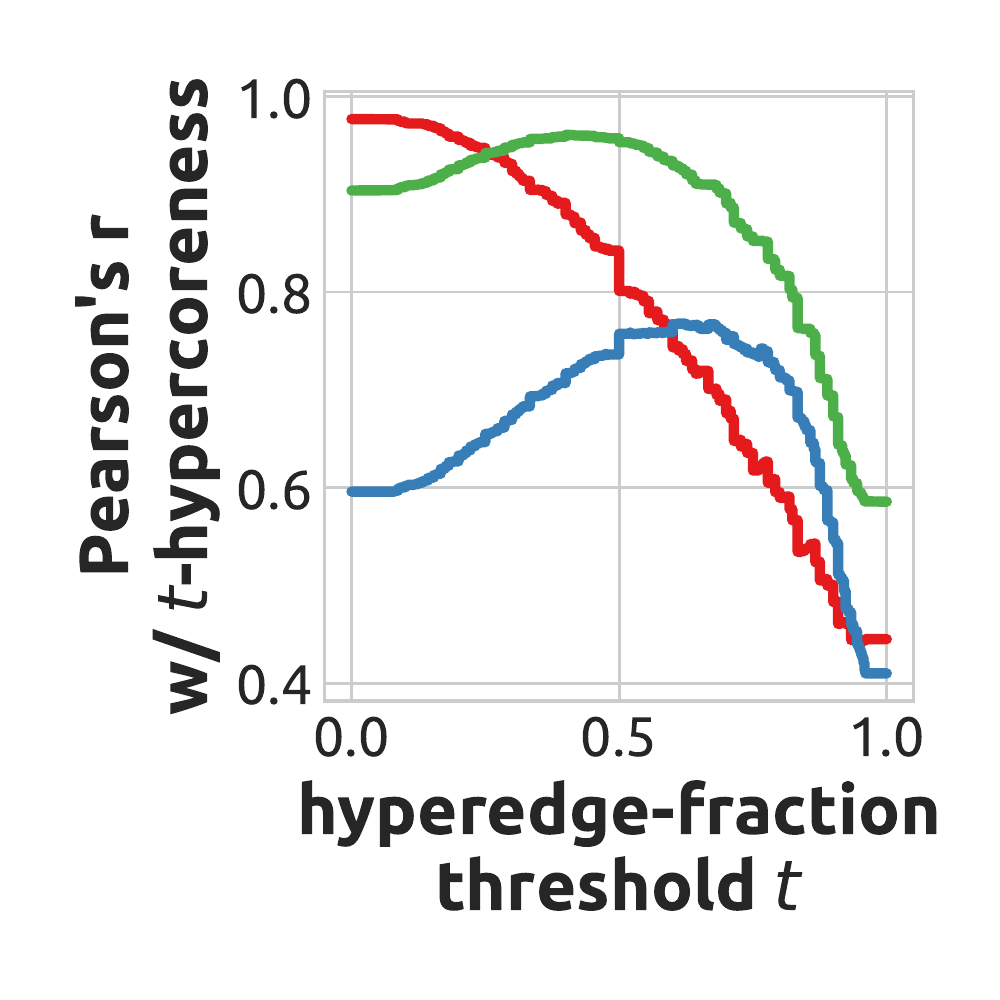}
		\includegraphics[scale=0.4]{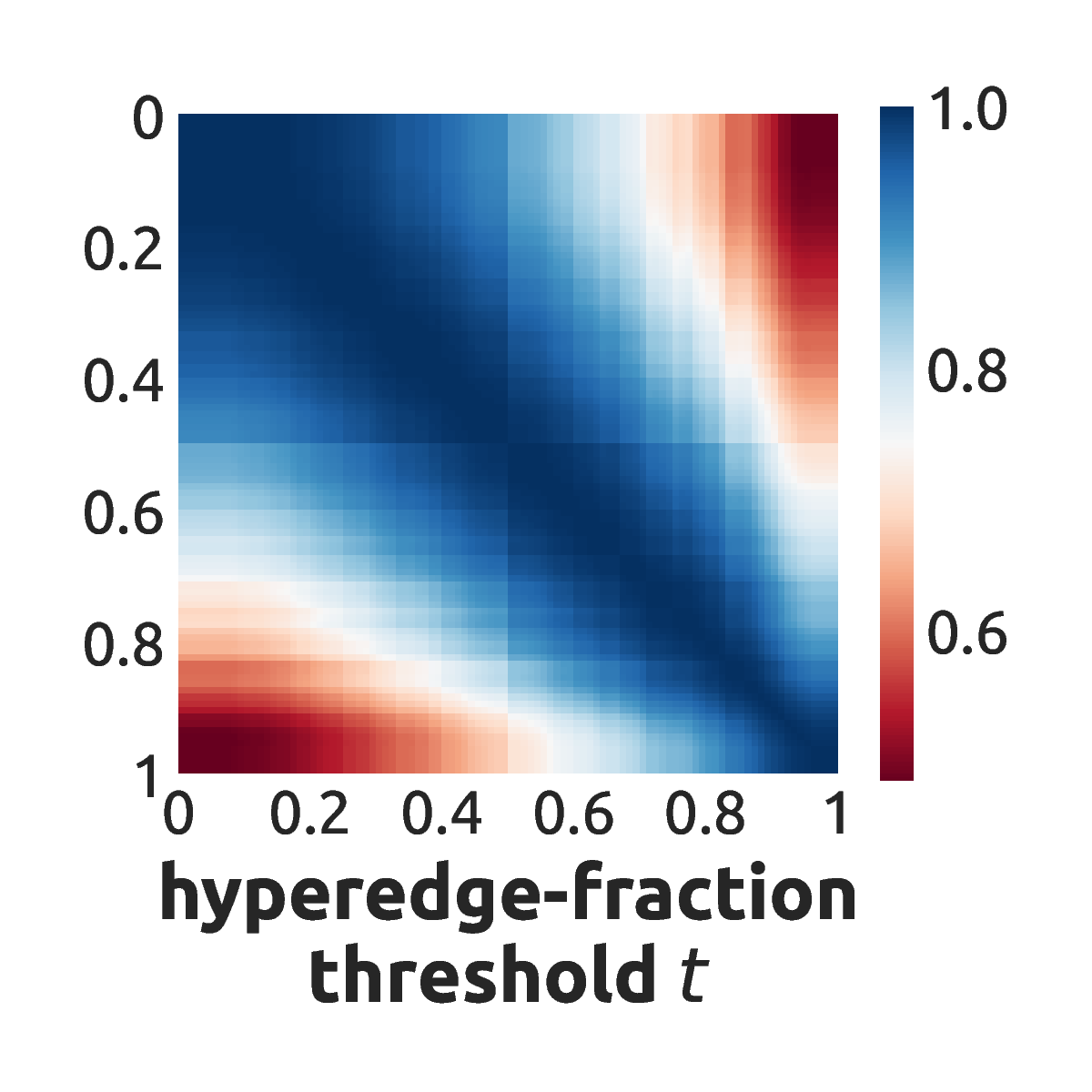}             
		\vspace{-3mm}		
		\caption{NDC-substances}
	\end{subfigure}
	\begin{subfigure}[b]{0.48\linewidth}
		\centering
		\hspace{-5mm}
		\includegraphics[scale=0.4]{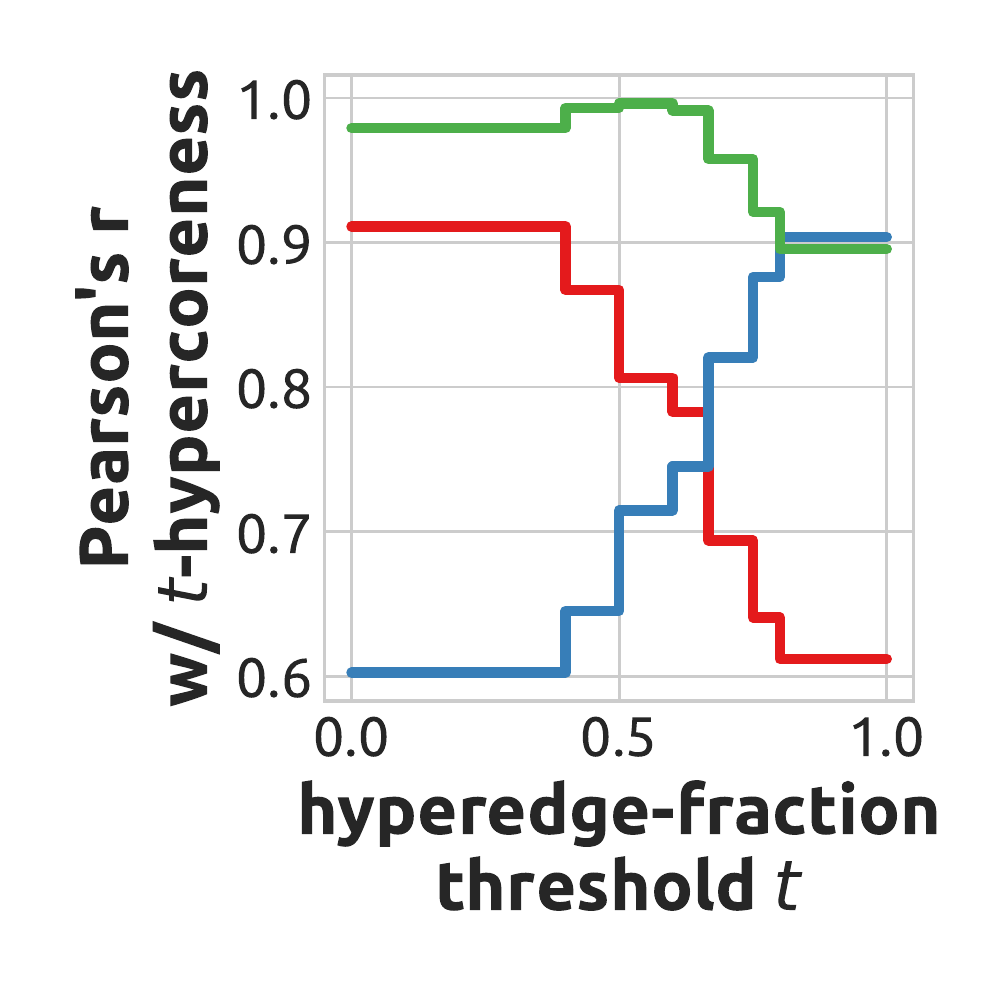}
		\includegraphics[scale=0.4]{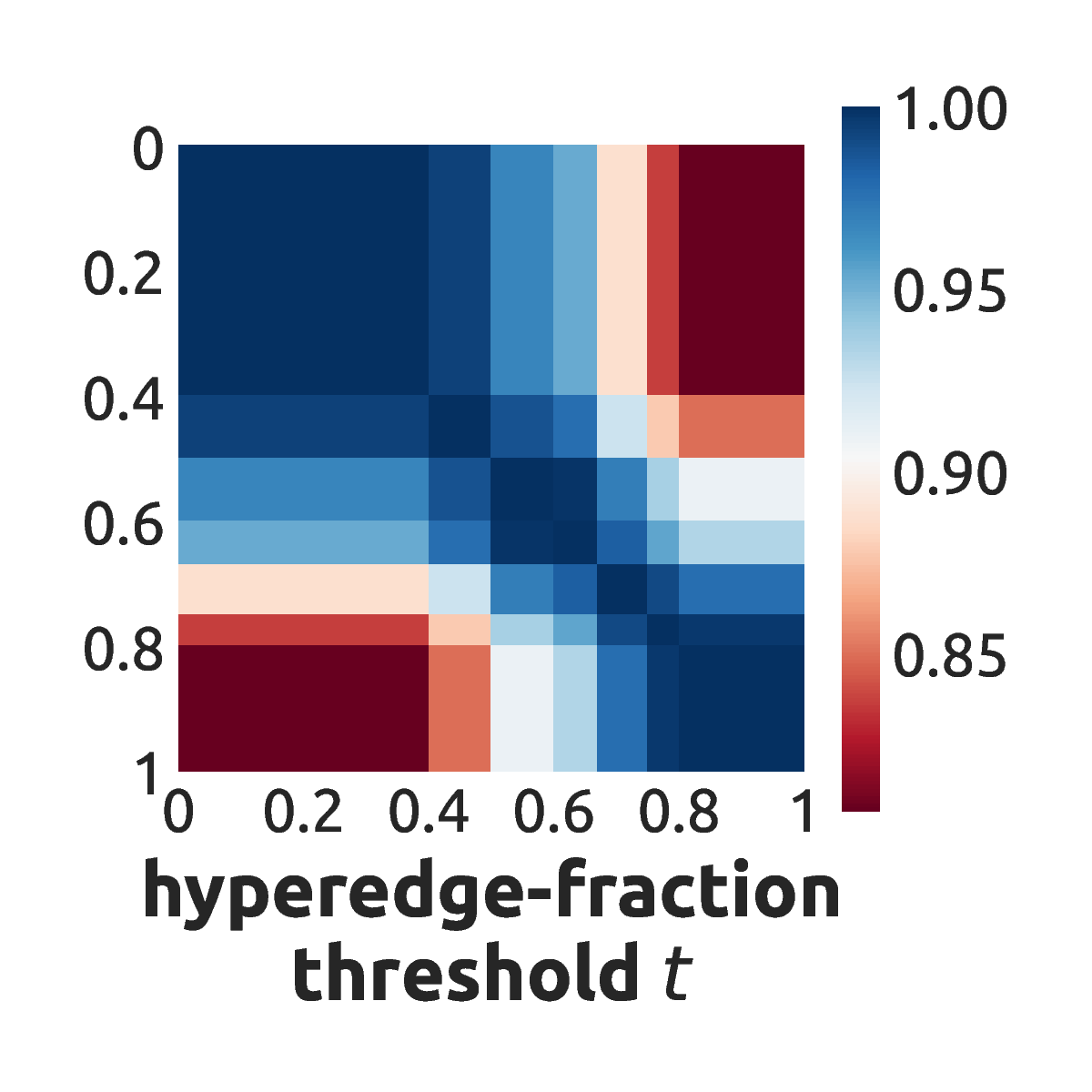}
		\vspace{-3mm}
		\caption{tags-math}
	\end{subfigure}    
	\begin{subfigure}[b]{0.48\linewidth}
		\centering
		\hspace{-5mm}
		\includegraphics[scale=0.4]{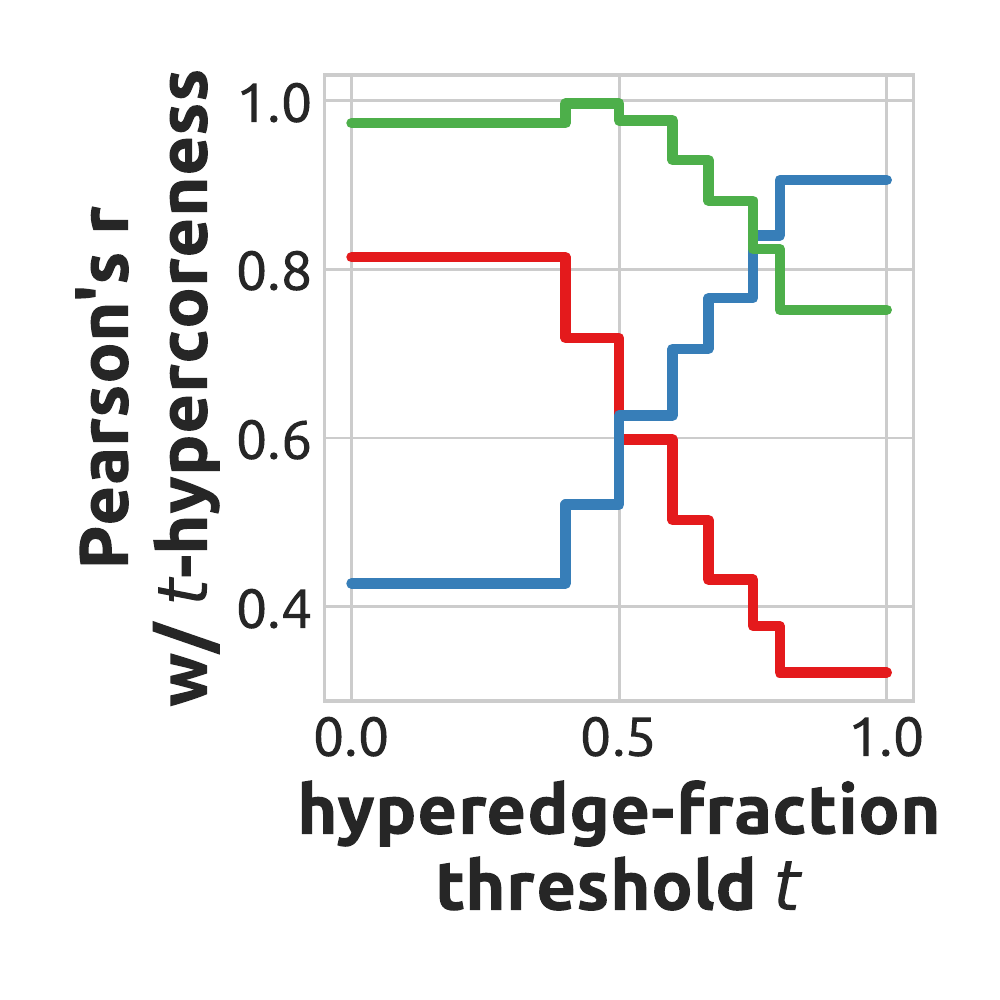}
		\includegraphics[scale=0.4]{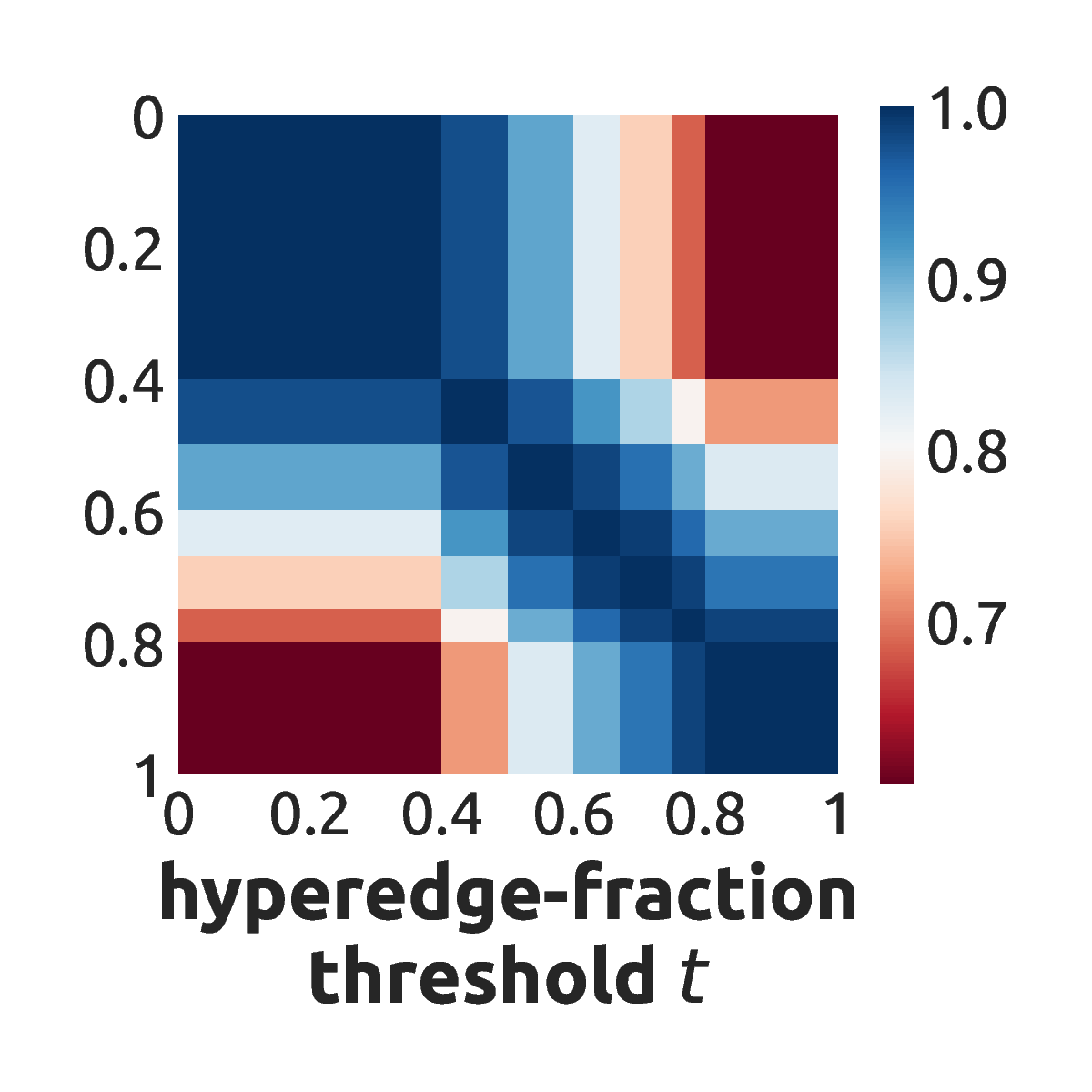}
		\vspace{-3mm}
		\caption{tags-SO}
	\end{subfigure}    
	\begin{subfigure}[b]{0.48\linewidth}
		\centering
		\hspace{-5mm}
		\includegraphics[scale=0.4]{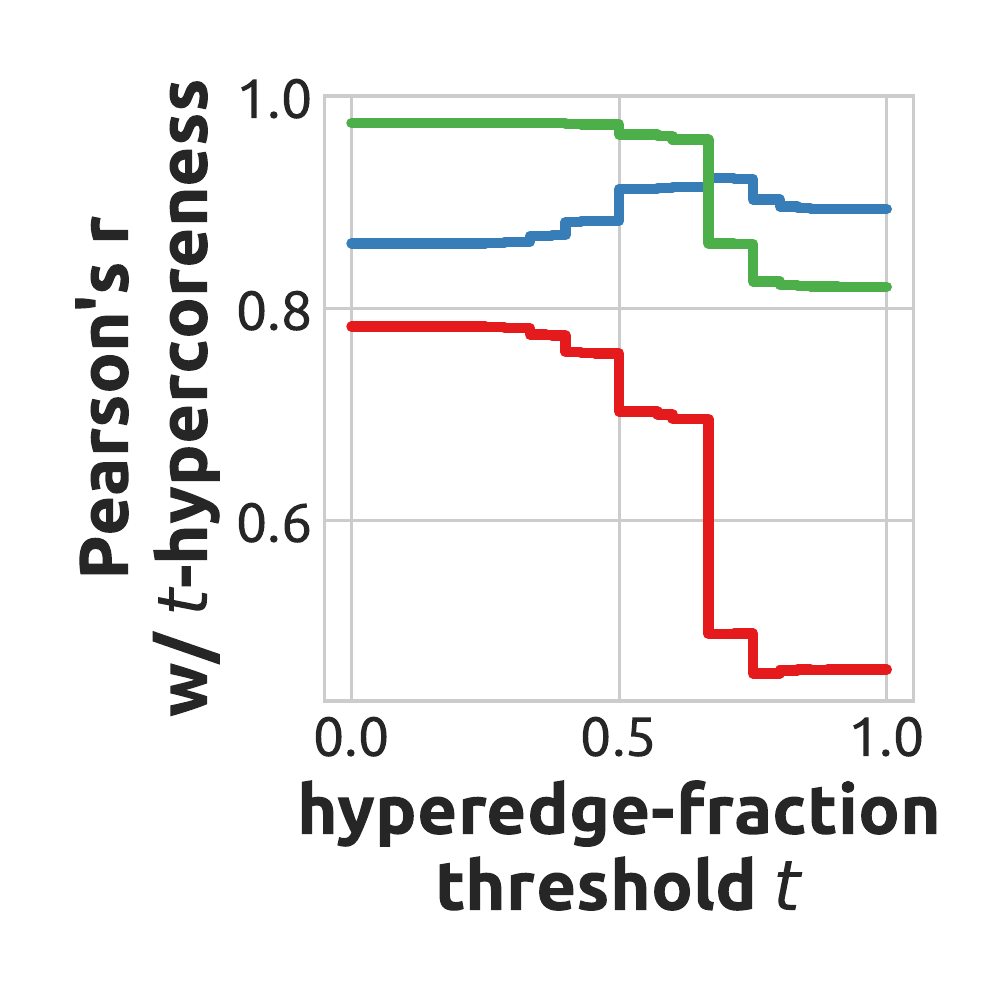}
		\includegraphics[scale=0.4]{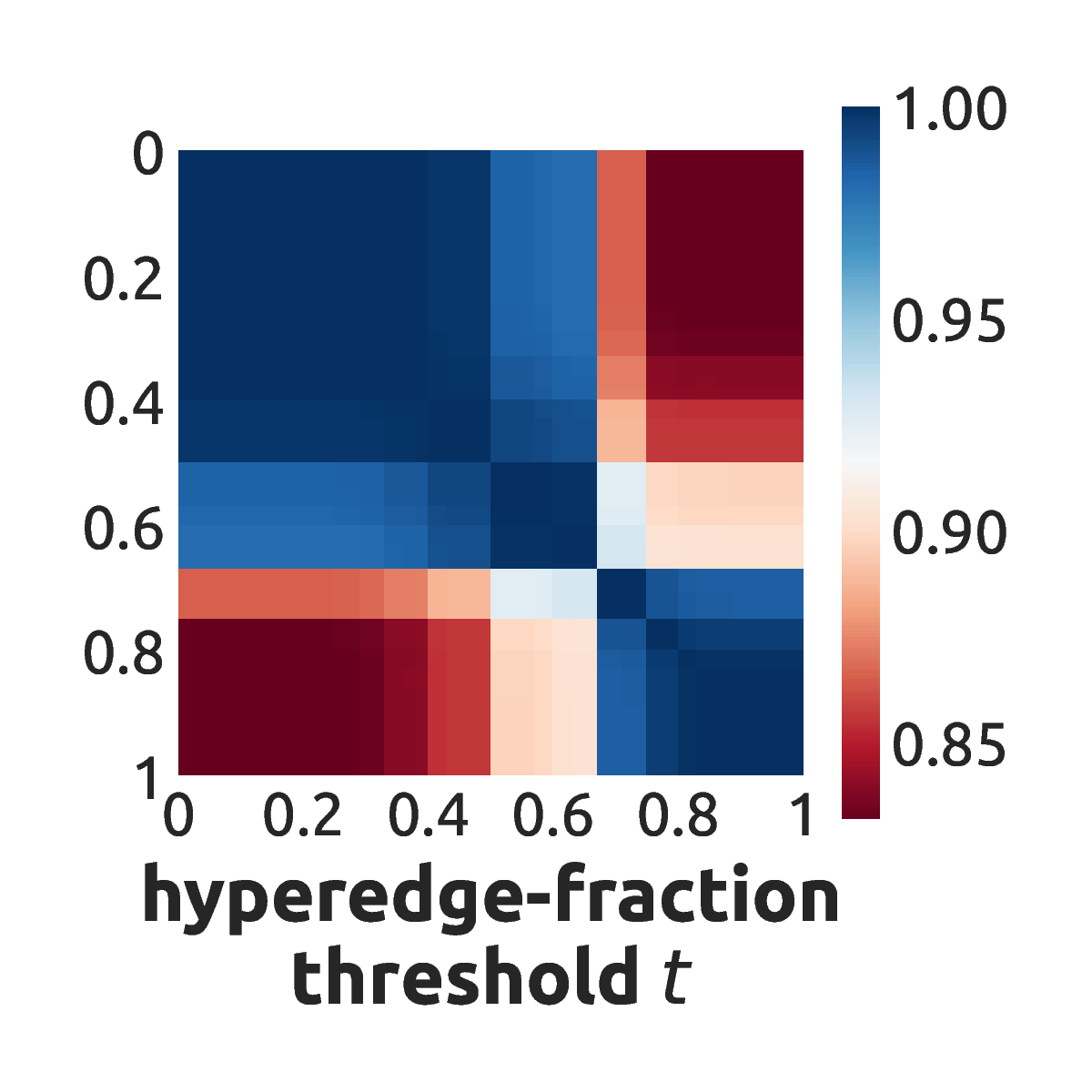}
		\vspace{-3mm}
		\caption{threads-ubuntu}
	\end{subfigure}
	\begin{subfigure}[b]{0.48\linewidth}
		\centering
		\hspace{-5mm}
		\includegraphics[scale=0.4]{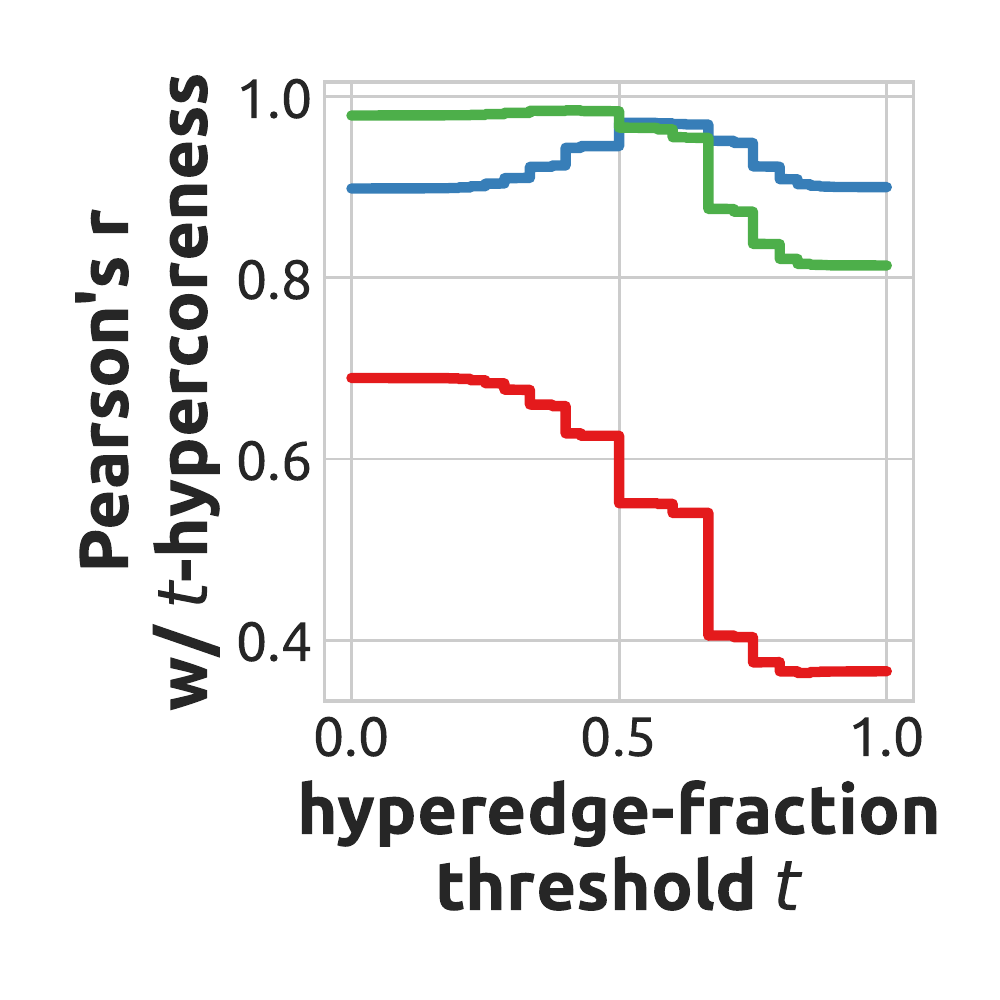}
		\includegraphics[scale=0.4]{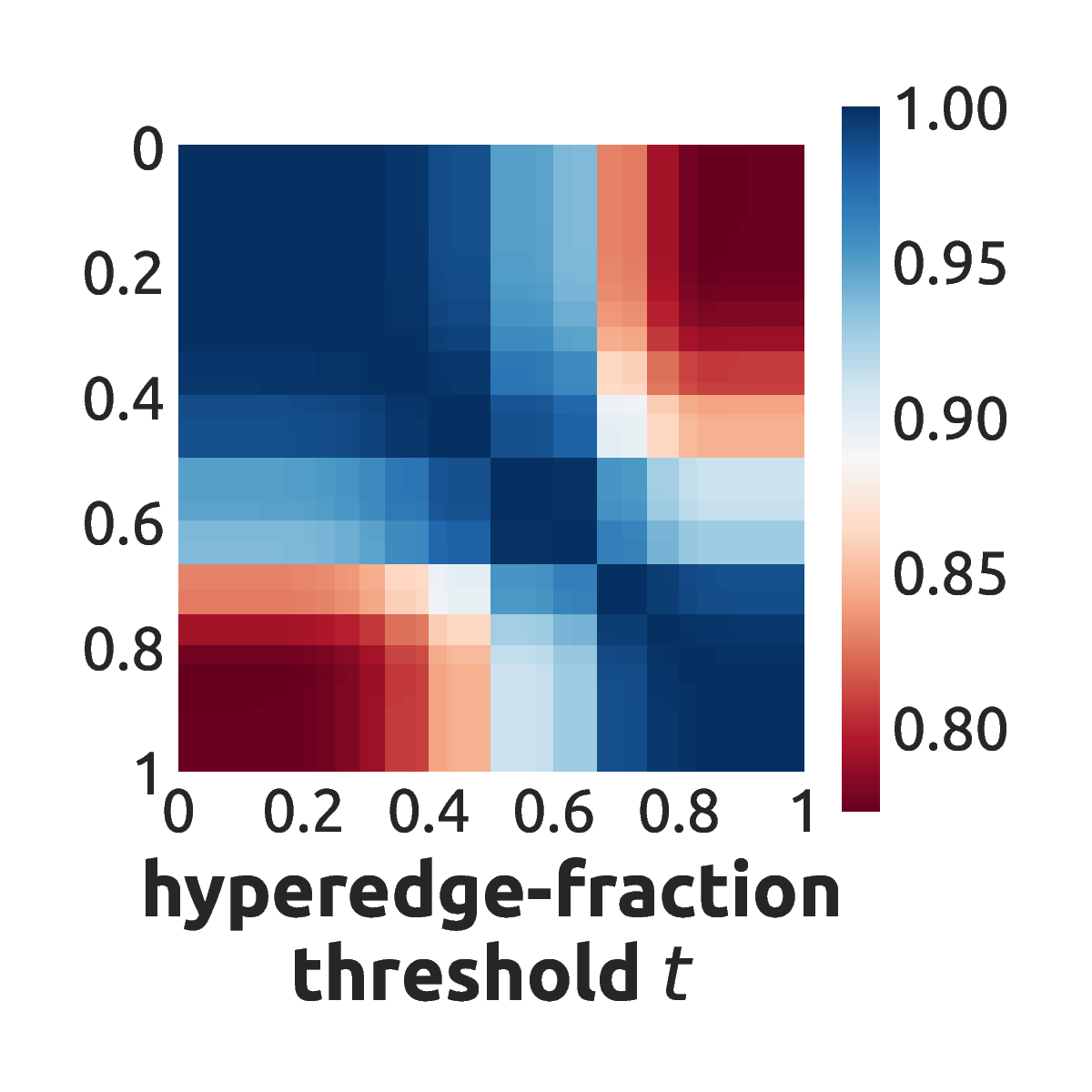}
		\vspace{-3mm}
		\caption{threads-SO}
	\end{subfigure}
	\vspace{-1mm}
	\caption{Supplementary results for Fig.~\ref{fig:rank_cor}.}
	\label{fig:rank_cor_sr}
\end{figure*}

\begin{figure*}[t]
	\centering
	\begin{subfigure}[b]{0.32\textwidth}
		\centering
		\includegraphics[scale=0.38]{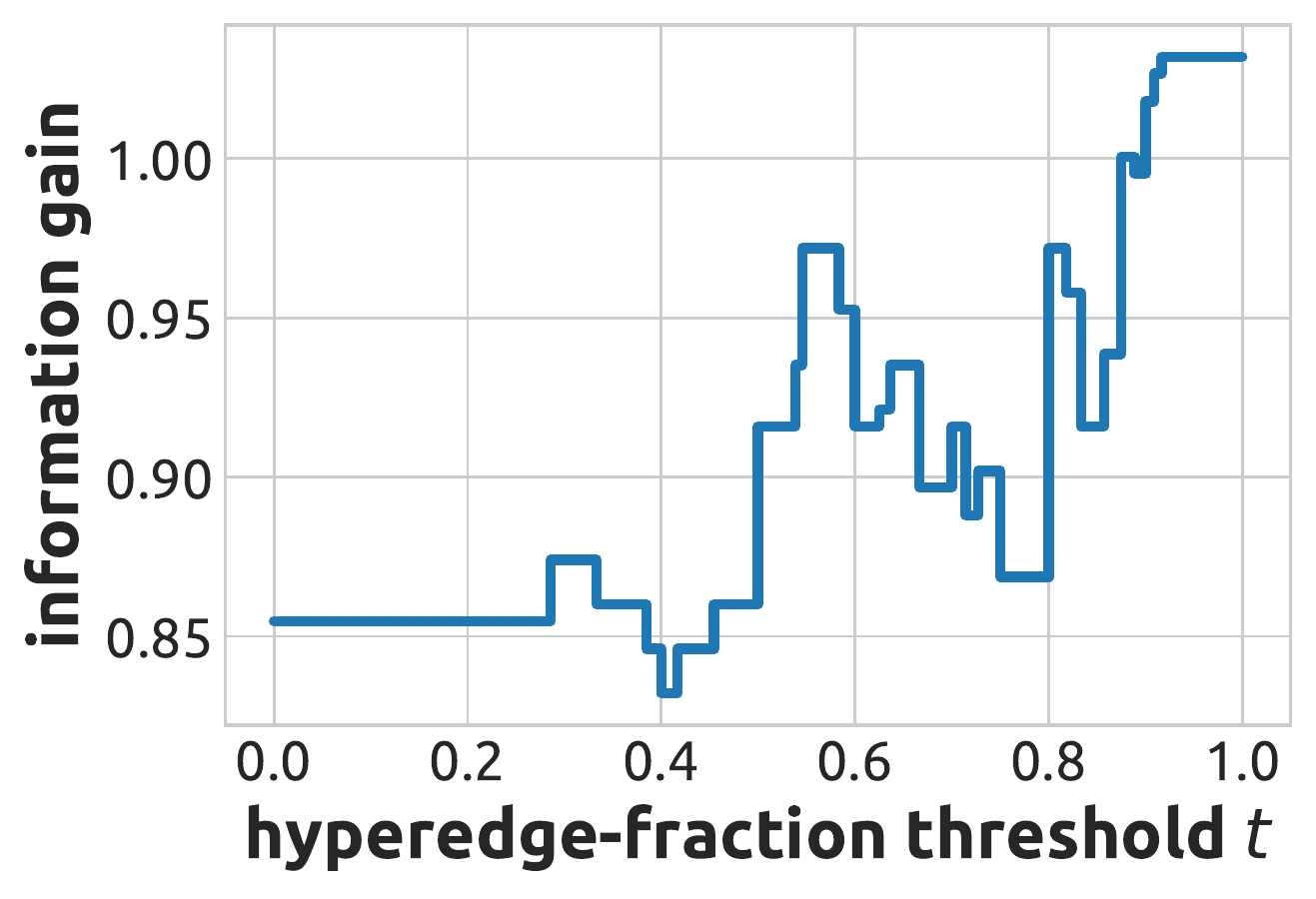}\\
		\caption{\Nu}
	\end{subfigure}
	\begin{subfigure}[b]{0.32\textwidth}
		\centering
		\includegraphics[scale=0.38]{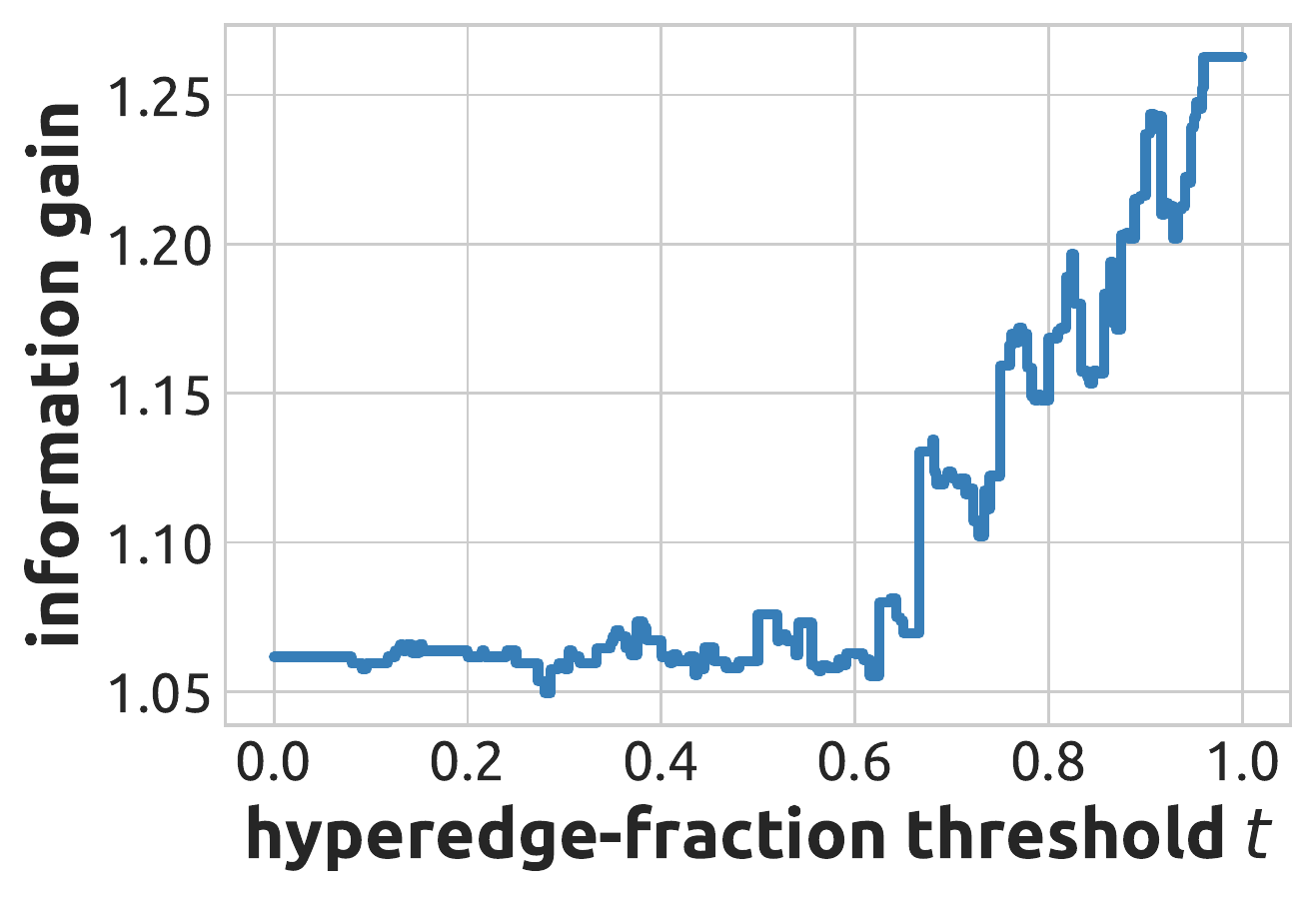}\\
		\caption{\Ni}
	\end{subfigure}\\
	\begin{subfigure}[b]{0.32\textwidth}
		\centering
		\includegraphics[scale=0.38]{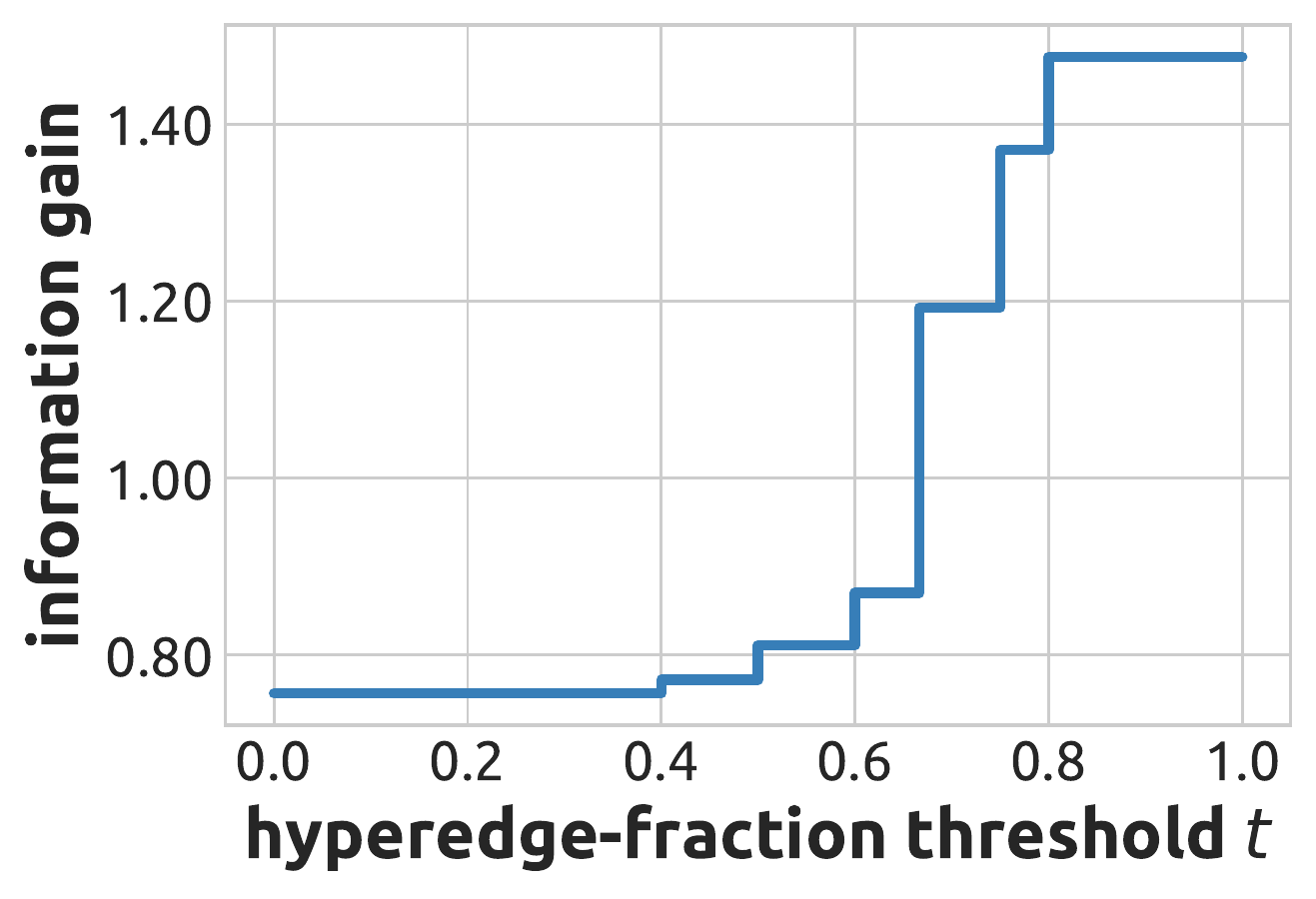}\\
		\caption{\No}
	\end{subfigure}
	\begin{subfigure}[b]{0.32\textwidth}
		\centering
		\includegraphics[scale=0.38]{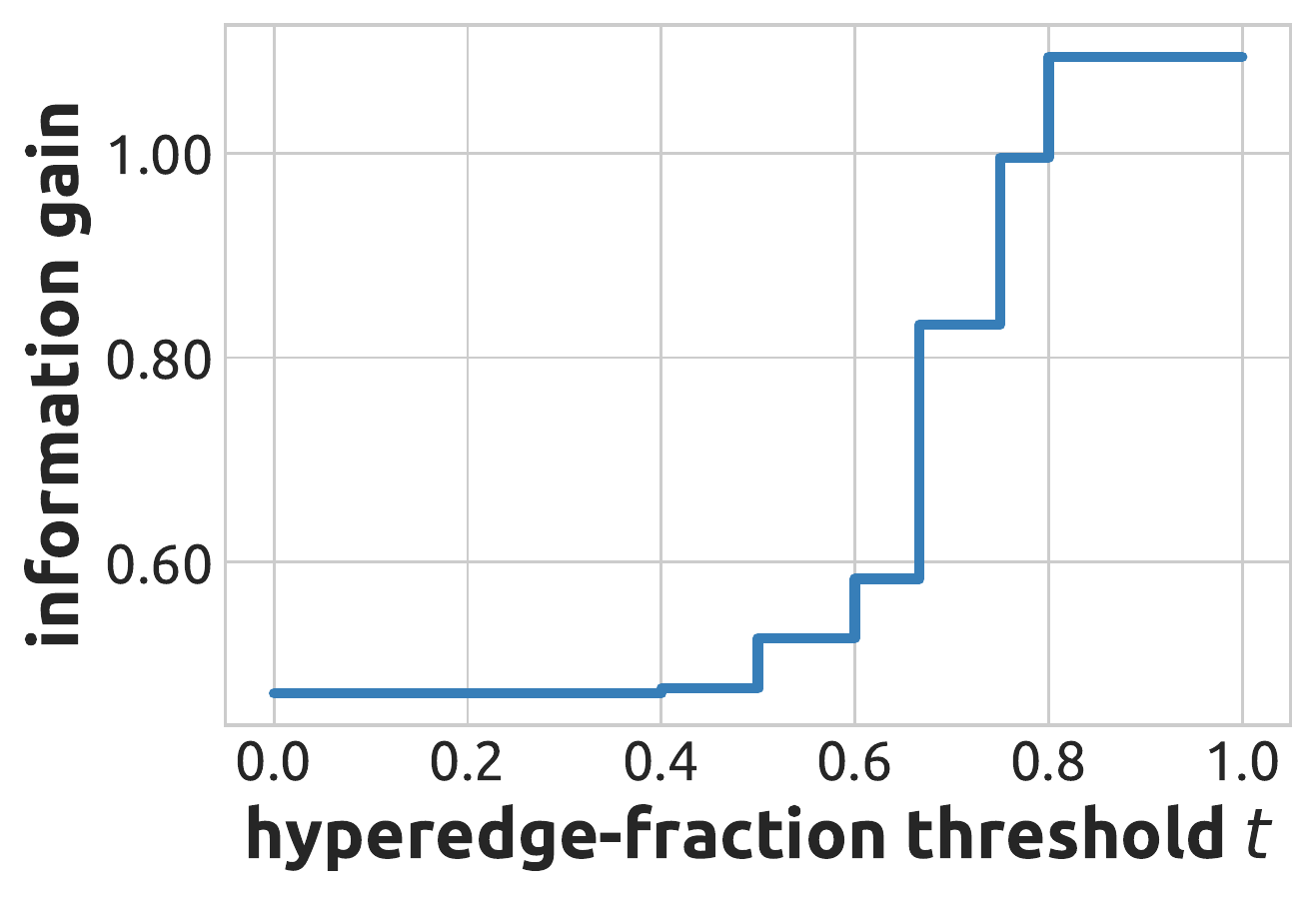}\\
		\caption{\Np}
	\end{subfigure}
	\begin{subfigure}[b]{0.32\textwidth}
		\centering
		\includegraphics[scale=0.38]{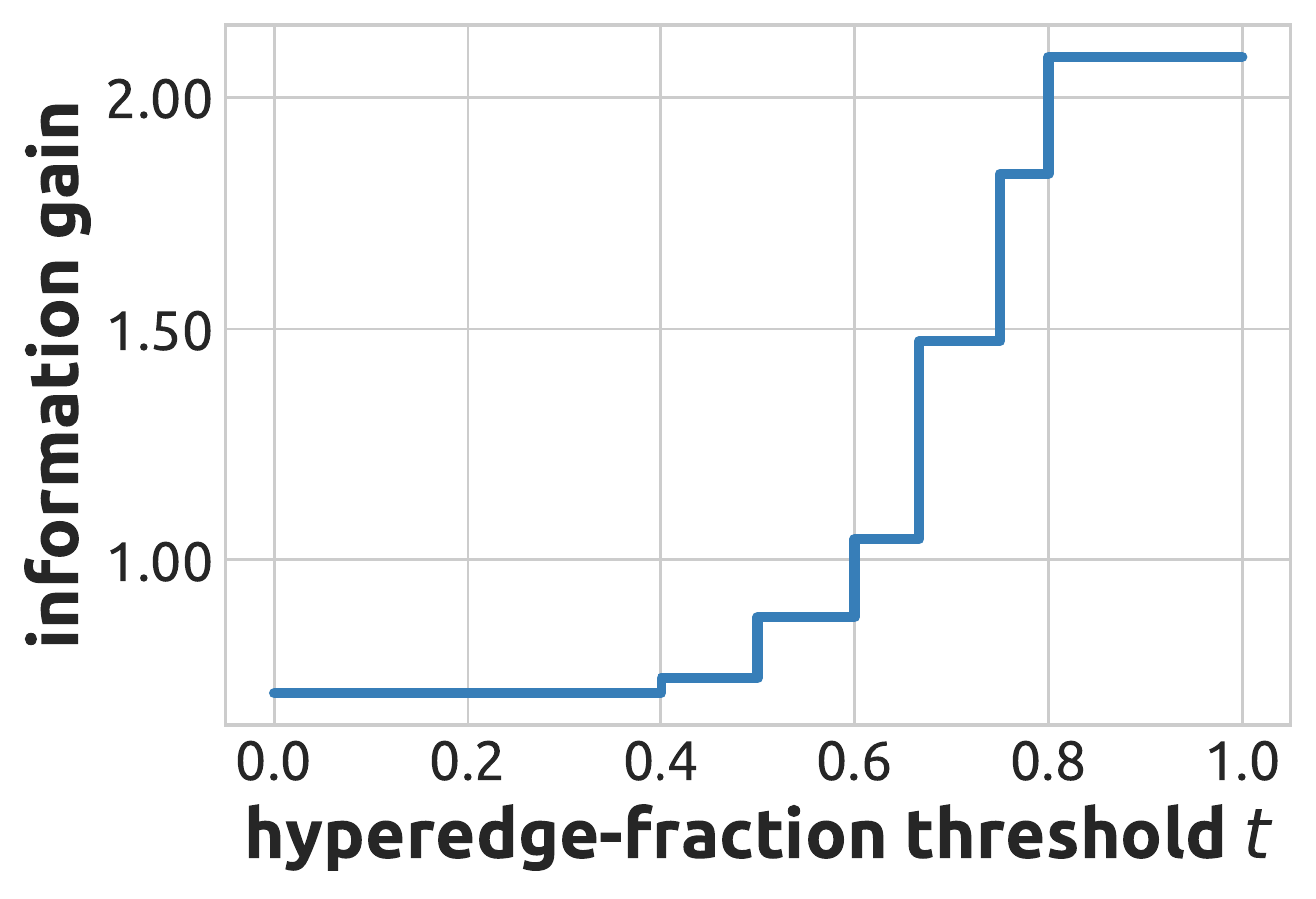}\\
		\caption{\Na}
	\end{subfigure}
	\begin{subfigure}[b]{0.32\textwidth}
		\centering
		\includegraphics[scale=0.38]{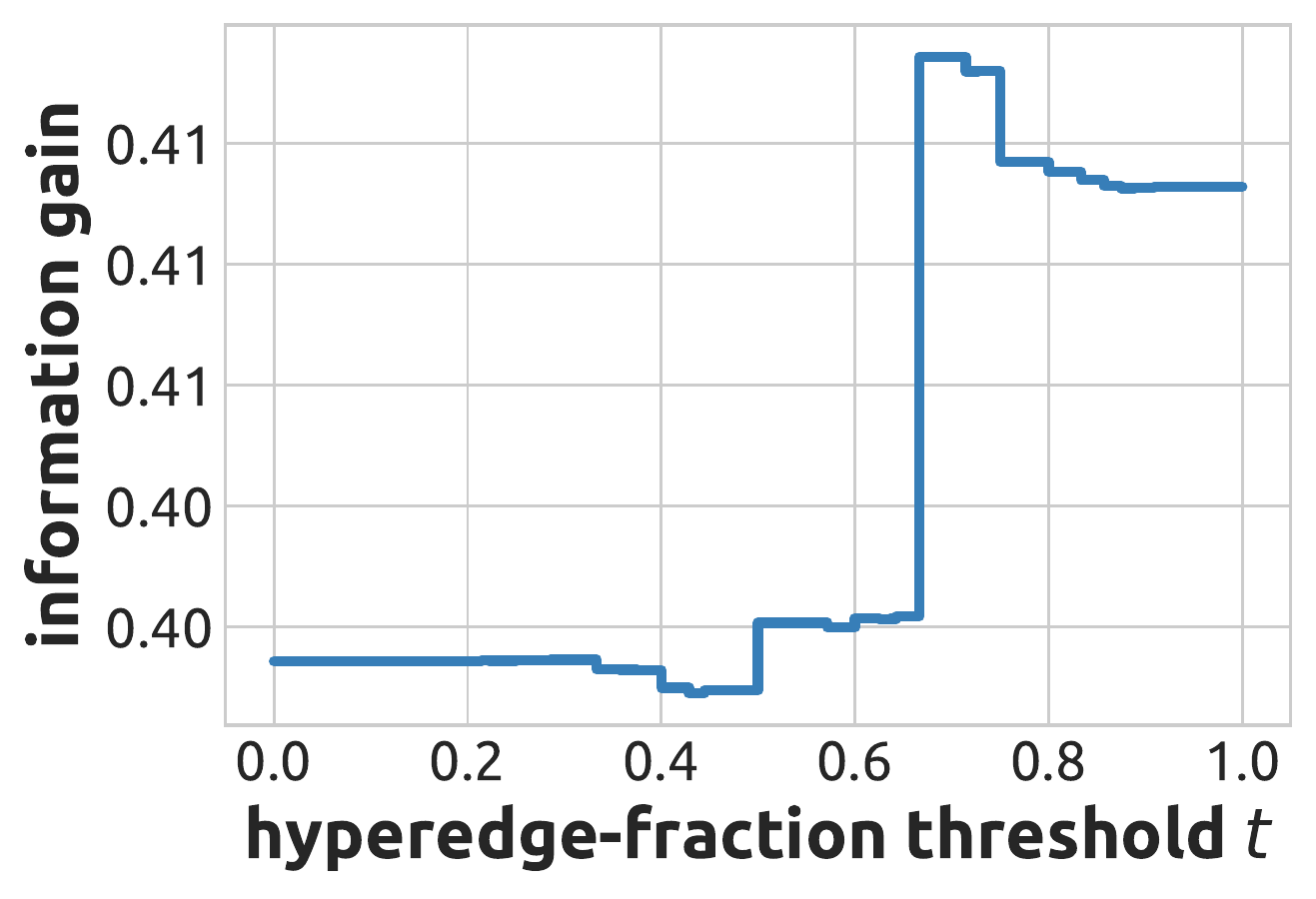}\\
		\caption{\Ns}
	\end{subfigure}
	\begin{subfigure}[b]{0.32\textwidth}
		\centering
		\includegraphics[scale=0.38]{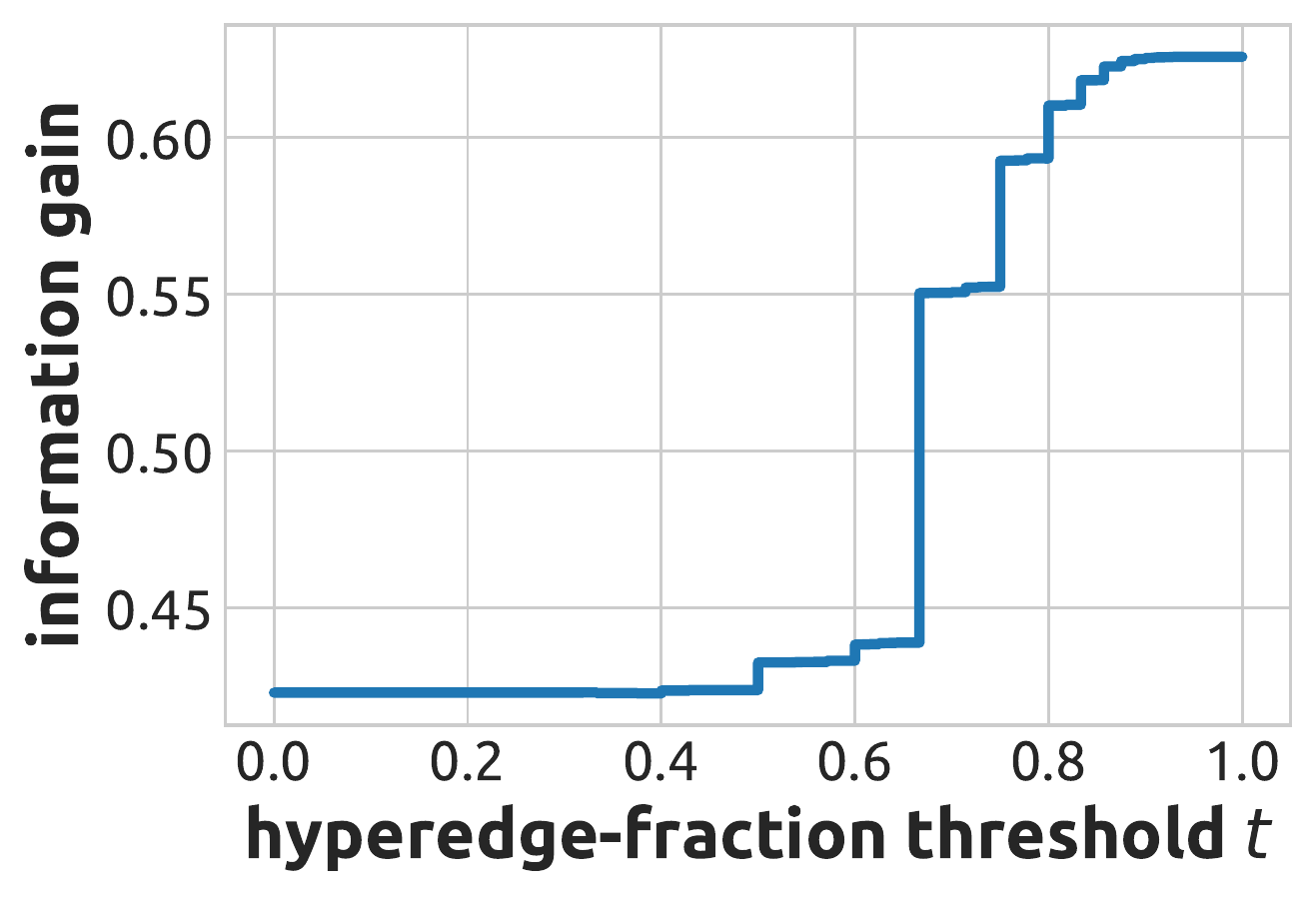}\\
		\caption{\Nd}
	\end{subfigure}
	\begin{subfigure}[b]{0.32\textwidth}
		\centering
		\includegraphics[scale=0.38]{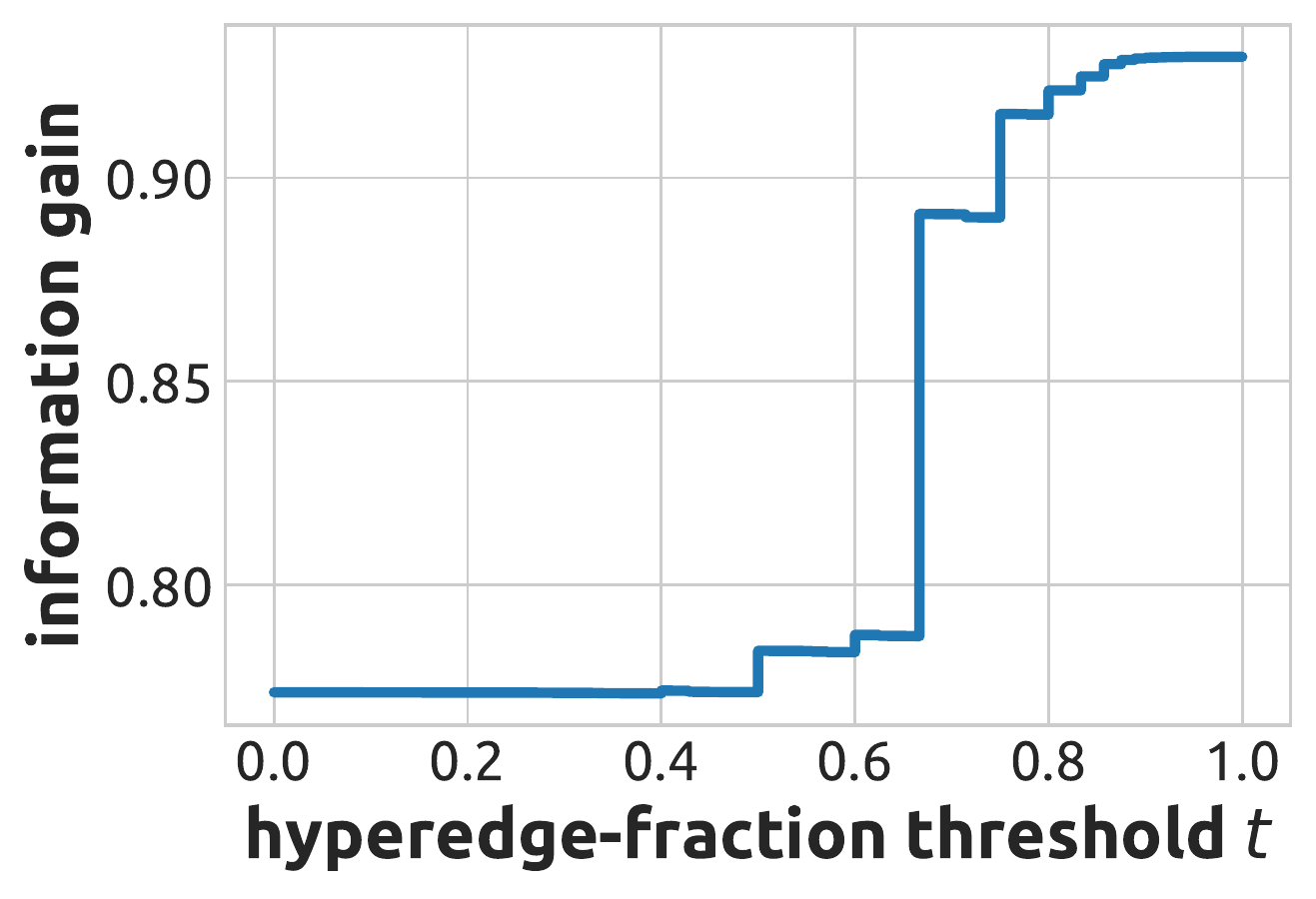}\\
		\caption{\Nf}
	\end{subfigure}
	\vspace{-0.3mm}
	\caption{Supplementary results for Fig.~\ref{fig:info_gain}.}
	\label{fig:info_gain_sr}
\end{figure*}

\begin{figure*}[b]
	\centering	
	\begin{subfigure}[b]{\textwidth}
		\centering
		\includegraphics[scale=0.5]{fig8_influence_legend.pdf}\\        
	\end{subfigure}    
	\vspace{1mm}
	\begin{subfigure}[b]{0.32\textwidth}
		\centering
		\includegraphics[scale=0.4]{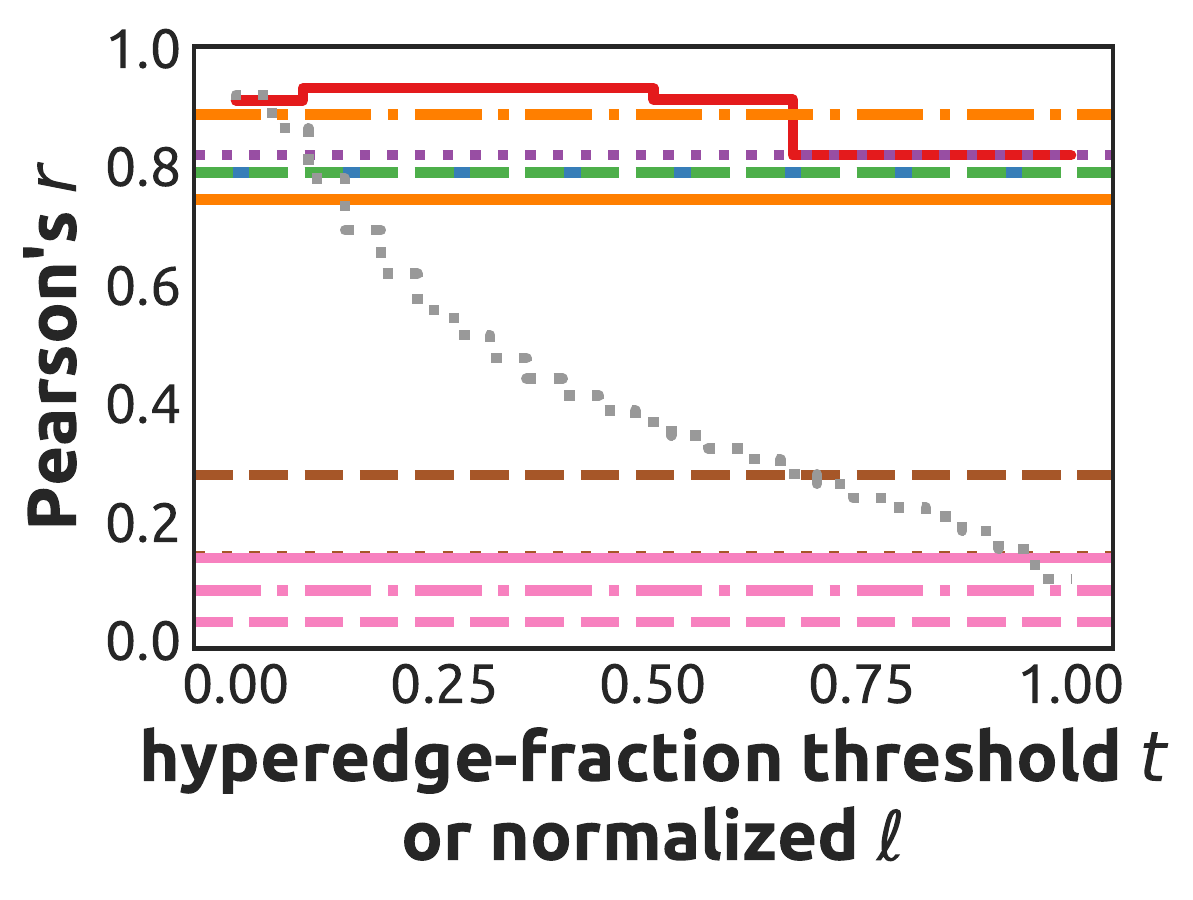}
		\caption{\Nw}
	\end{subfigure}
	\begin{subfigure}[b]{0.32\textwidth}
		\centering
		\includegraphics[scale=0.4]{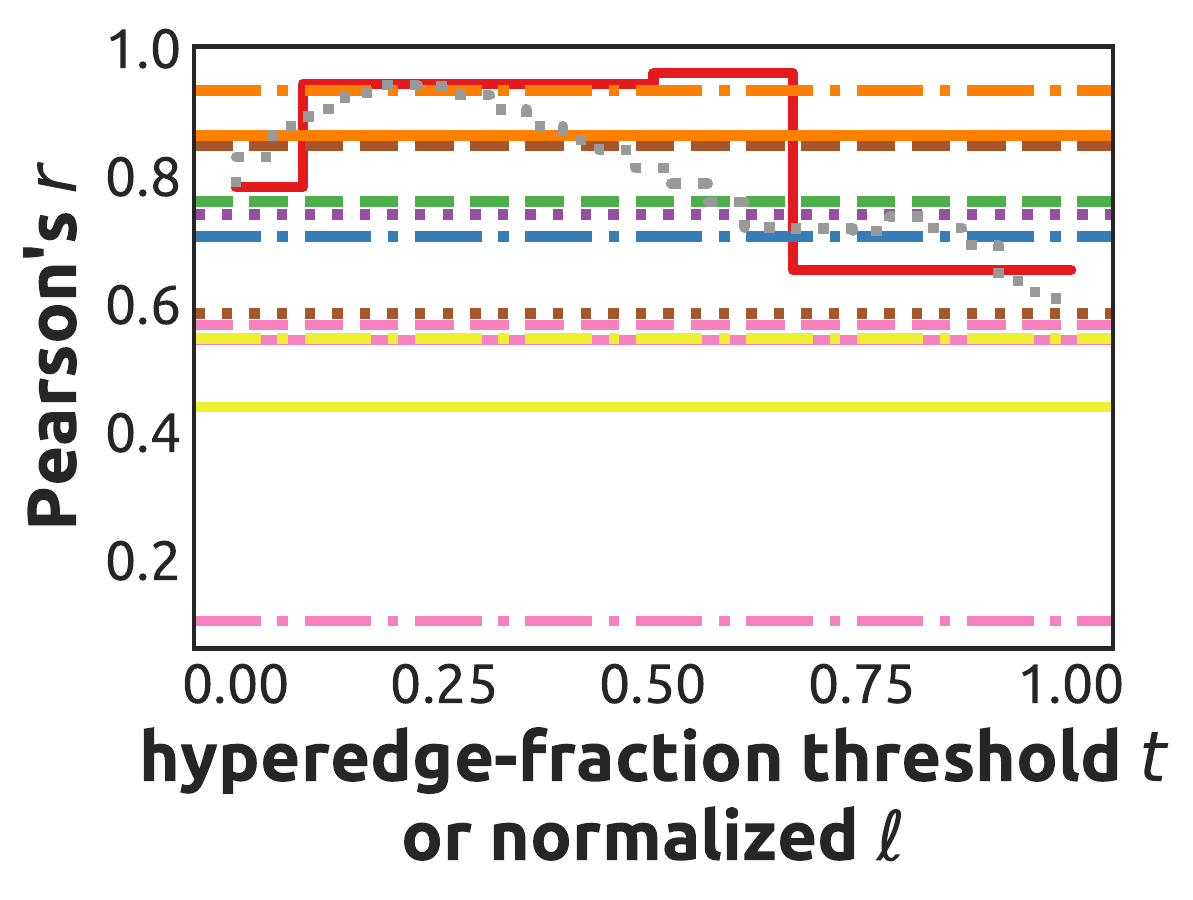}
		\caption{\Nr}
	\end{subfigure}
	\begin{subfigure}[b]{0.32\textwidth}
		\centering
		\includegraphics[scale=0.4]{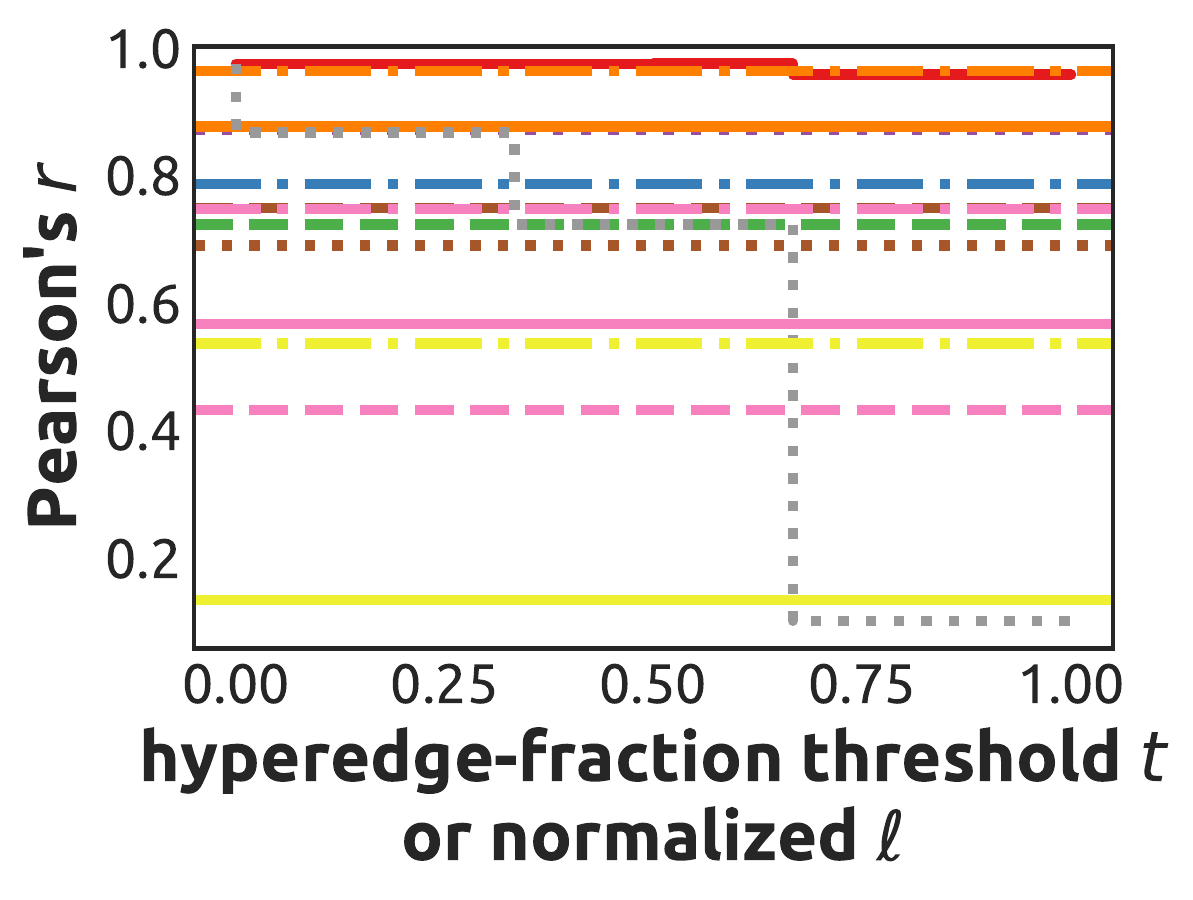}
		\caption{\Ny}
	\end{subfigure}
	\begin{subfigure}[b]{0.32\textwidth}
		\centering
		\includegraphics[scale=0.4]{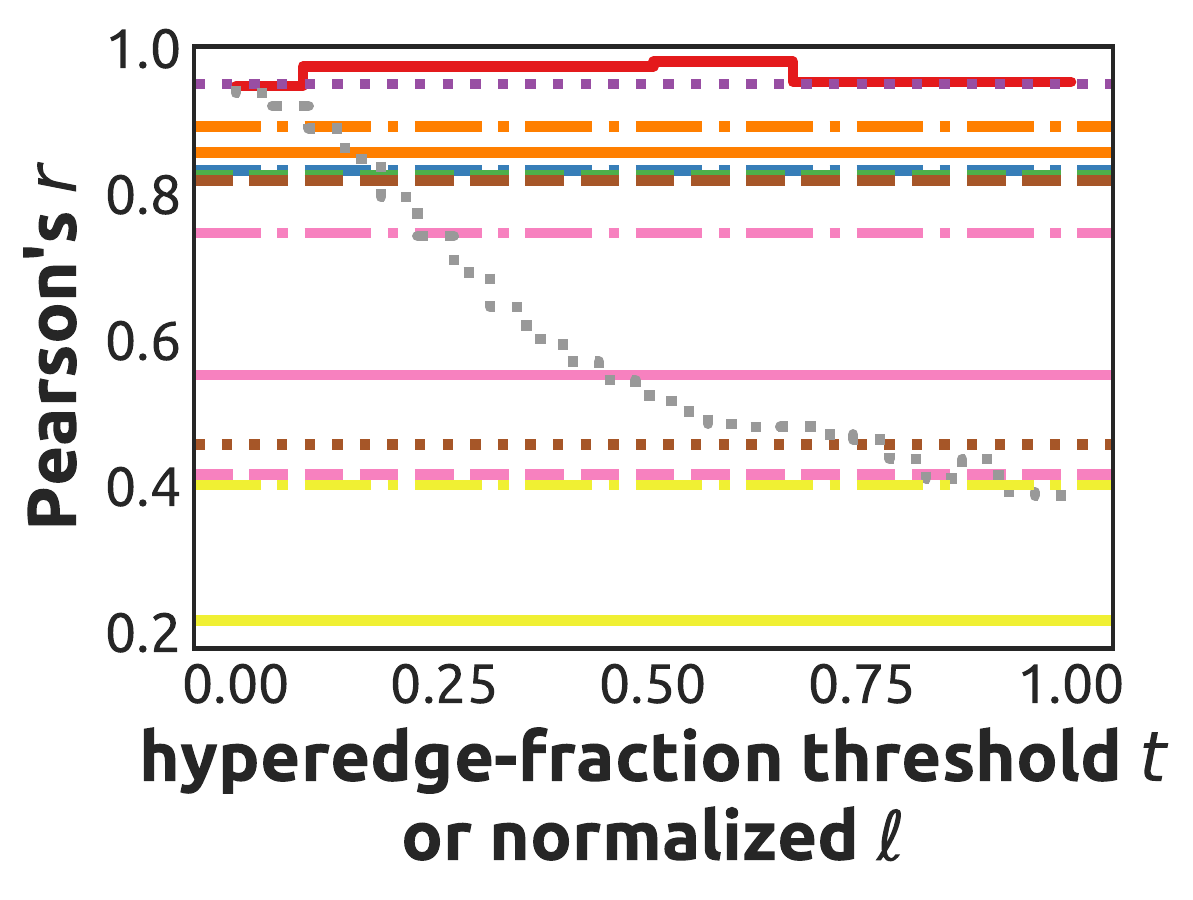}
		\caption{\Ni}
	\end{subfigure}
	\begin{subfigure}[b]{0.32\textwidth}
		\centering
		\includegraphics[scale=0.4]{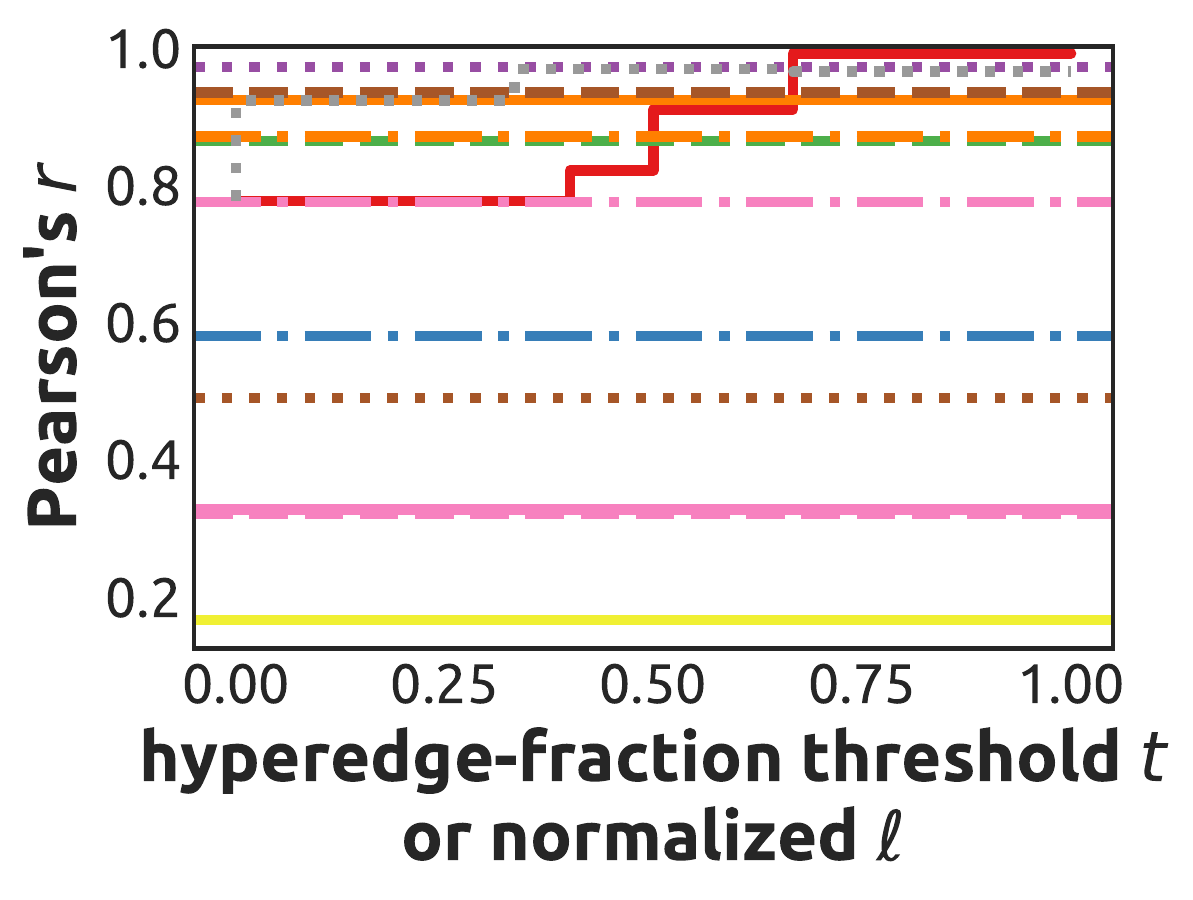}
		\caption{tags-math}
	\end{subfigure}
	\begin{subfigure}[b]{0.32\textwidth}
		\centering
		\includegraphics[scale=0.4]{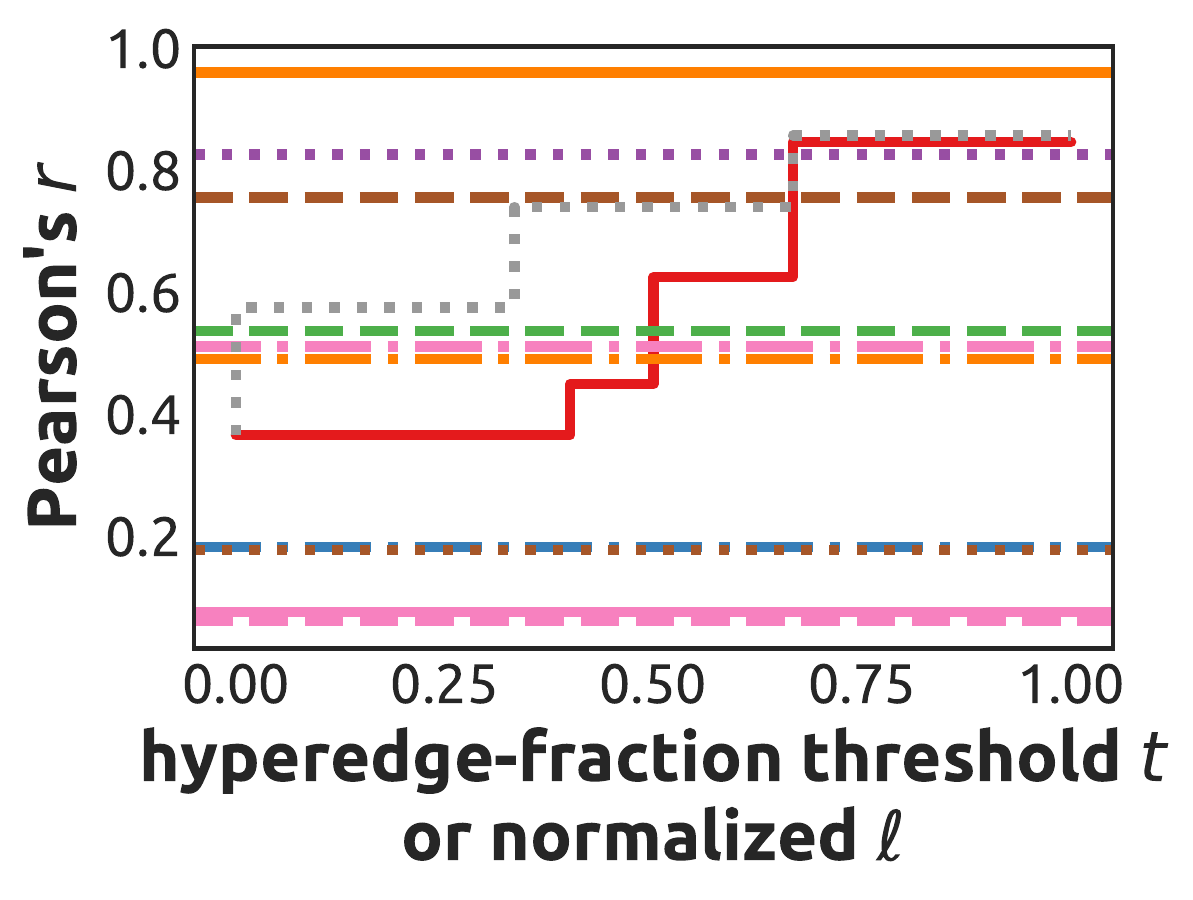}
		\caption{tags-SO}
	\end{subfigure}
	\begin{subfigure}[b]{0.32\textwidth}
		\centering
		\includegraphics[scale=0.4]{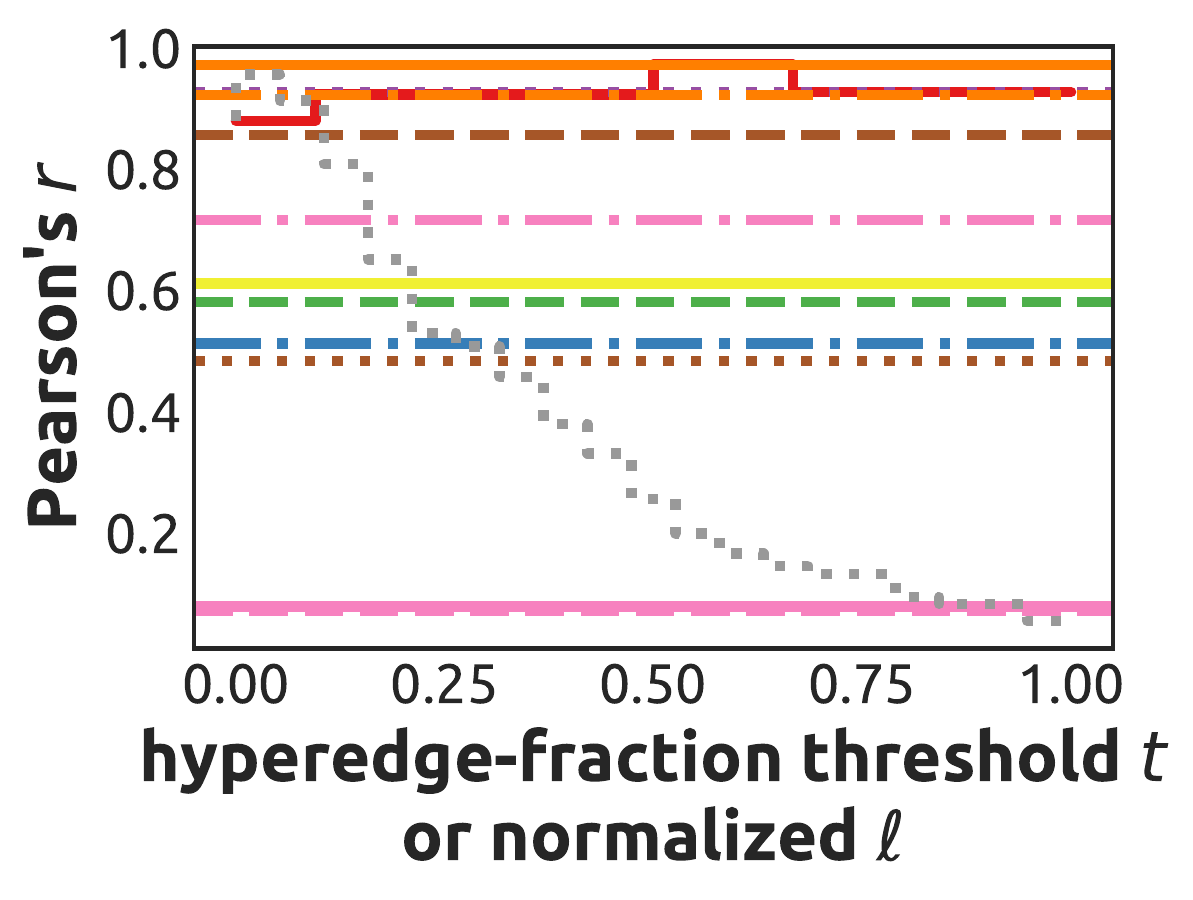}
		\caption{threads-math}
	\end{subfigure}
	\begin{subfigure}[b]{0.32\textwidth}
		\centering
		\includegraphics[scale=0.4]{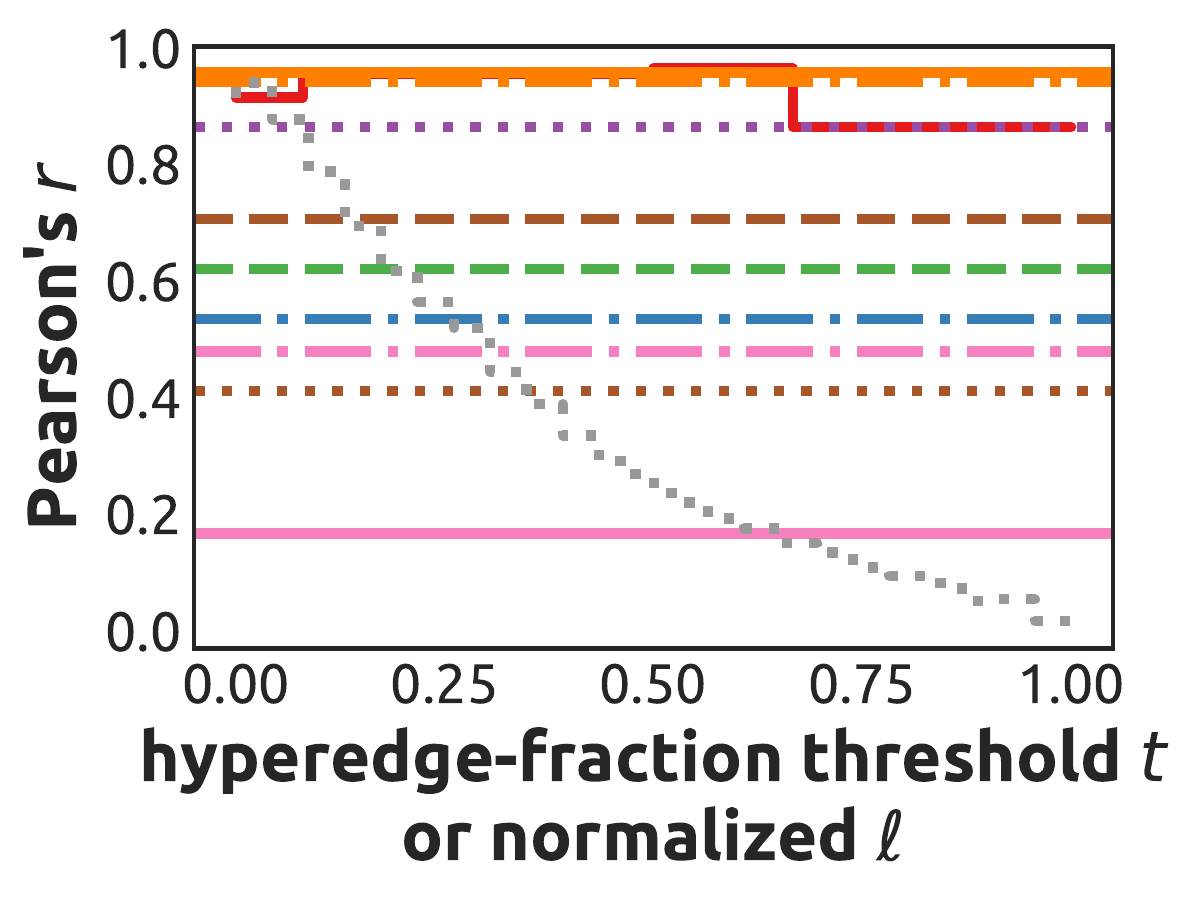}
		\caption{threads-SO}
	\end{subfigure}  
	\caption{Supplementary results for Fig.~\ref{fig:inf_summary}.}
	\label{fig:inf_summary_sr}
\end{figure*}

\begin{table}[t!]
	\caption{Full results of the collapsed $(k, t)$-hypercore problem ($b = 100$). Time: running time (in seconds). Red.: reduction in the hypercore size.} \label{tab:min_res}
	\centering
	\resizebox{0.76\columnwidth}{!}{%
		\begin{tabular}{lll rrrrr rrrrr}
			\toprule			
			&&& \multicolumn{2}{c}{\textsc{hyperCKC}} & \multicolumn{2}{c}{\textsc{HyCoM}-1} & \multicolumn{2}{c}{\textsc{HyCoM}-10} & \multicolumn{2}{c}{\textsc{HyCoM}-100} & \multicolumn{2}{c}{\textsc{HyCoM}+} \\
			dataset & $k$ & $t$ & time & red. & time & red. & time & red. & time & red. & time & red. \\
			\hline
   
		\multirow{15}{*}{coauth-DBLP} & $5$ & $0$   	& 63.01 & 1,176  & 15.80 & 680   & 16.44 & 729   & 17.50 & 812   & 5.37 & 1,088  \\
			 & $5$ & $0.2$ 	& 63.38 & 1,185  & 15.87 & 691   & 16.47 & 738   & 17.54 & 819   & 5.41 & 1,091  \\
			 & $5$ & $0.4$ 	& 70.43 & 1,378  & 18.21 & 831   & 18.87 & 897   & 19.91 & 986   & 6.28 & 1,250  \\
			 & $5$ & $0.6$ 	& 113.38 & 2,536 & 34.17 & 1,675 & 35.00 & 1,795 & 36.18 & 1,969 & 11.83 & 2,259 \\
			 & $5$ & $0.8$ 	& 130.19 & 6,681 & 43.21 & 4,487 & 43.90 & 5,075 & 45.86 & 5,785 & 13.88 & 5,943 \\
			 & $5$ & $1$ 	    & 87.75 & 7,523  & 25.40 & 5,914 & 26.10 & 6,108 & 27.93 & 6,752 & 8.70 & 6,731 \\
			 & $10$ & $0$ 	& 50.42 & 1,181  & 9.74 & 661    & 10.08 & 699   & 11.37 & 832   & 2.51 & 1,008 \\
			 & $10$ & $0.2$ 	& 50.82 & 1,195  & 9.74 & 676    & 10.21 & 716   & 11.46 & 852   & 2.55 & 1,019 \\
			& $10$& $0.4$ 	& 56.94 & 1,462  & 11.38 & 844   & 11.78 & 877   & 13.13 & 1,001 & 3.08 & 1,206 \\
			& $10$& $0.6$ 	& 75.03 & 2,798  & 16.05 & 1,621 & 16.45 & 1,711 & 18.14 & 2,050 & 4.30 & 2,328 \\
			& $10$& $0.8$ 	& 0.24 & 2,699   & 0.034 & 2,742 & 0.051 & 2,699 & 0.13 & 2,699  & 0.13 & 2,740 \\
			& $20$& $0$ 	& 33.32 & 1,251  & 5.21 & 590    & 5.45 & 632    & 6.72 & 788    & 1.02 & 945 \\
			& $20$& $0.2$ 	& 33.56 & 1,286  & 5.24 & 588    & 5.50 & 617    & 6.73 & 802    & 1.03 & 1,000 \\
			& $20$& $0.4$ 	& 34.72 & 1,558  & 5.14 & 749    & 5.41 & 784    & 6.71 & 1,003  & 1.11 & 1,206 \\
			& $20$& $0.6$ 	& 13.67 & 2,817  & 2.06 & 1,813  & 2.19 & 1,975  & 3.17 & 2,288  & 0.51 & 2,218 \\
			\hline		
			\multirow{15}{*}{coauth-Geology} & $5$& $0$   	& 16.65 & 859    & 3.83 & 343   & 4.03 & 352   & 4.44 & 423   & 1.35 & 708 \\
			& $5$& $0.2$ 	& 16.79 & 877    & 3.84 & 344   & 4.05 & 352   & 4.46 & 429   & 1.36 & 720 \\
			& $5$& $0.4$ 	& 20.55 & 1,120  & 4.60 & 386   & 4.86 & 403   & 5.26 & 518   & 1.77 & 865 \\
			& $5$& $0.6$ 	& 36.58 & 2,004  & 9.22 & 660   & 9.58 & 686   & 10.07 & 914  & 3.29 & 1,578 \\
			& $5$& $0.8$ 	& 32.32 & 5,334  & 7.58 & 2,992 & 7.92 & 3,164 & 8.45 & 3,986 & 2.53 & 3,891 \\
			& $5$& $1$ 		& 4.67 & 7,294   & 0.94 & 9,498 & 0.94 & 9,576 & 1.24 & 9,536 & 0.47 & 6,891 \\
			& $10$& $0$ 	& 12.81 & 856    & 2.65 & 316   & 2.77 & 328   & 3.14 & 423   & 0.64 & 720 \\
			& $10$& $0.2$ 	& 12.97 & 857    & 2.68 & 331   & 2.80 & 349   & 3.17 & 441   & 0.65 & 715 \\
			& $10$& $0.4$ 	& 16.81 & 1,132  & 3.34 & 374   & 3.49 & 386   & 3.89 & 524   & 0.85 & 926 \\
			& $10$& $0.6$ 	& 21.51 & 2,652  & 3.92 & 914   & 4.03 & 966   & 4.61 & 1,311 & 1.03 & 1,867 \\
			& $10$& $0.8$ 	& 0.0027 & 299   & 0.0023 & 388 & 0.0027 & 299 & 0.0027 & 299 & 0.03 & 388 \\
			& $20$& $0$ 	& 9.08 & 940     & 1.73 & 319   & 1.81 & 330   & 2.14 & 462   & 0.29 & 689 \\
			& $20$& $0.2$ 	& 9.18 & 977 	 & 1.74 & 334 	& 1.81 & 343   & 2.22 & 476   & 0.30 & 713 \\
			& $20$& $0.4$ 	& 10.13 & 1,470  & 1.78 & 528   & 1.85 & 550   & 2.23 & 742   & 0.32 & 1,105 \\
			& $20$& $0.6$ 	& 0.031 & 459    & 0.0093 & 529 & 0.014 & 459  & 0.030 & 459  & 0.033 & 527 \\
			\hline		
			\multirow{12}{*}{tags-SO} & $5$& $0$ 	& 191.84 & 1,042  & 42.03 & 930     & 55.66 & 962     & 108.93 & 1,012  & 7.87 & 1,042 \\
			& $5$& $0.6$ 	& 361.12 & 2,103  & 111.72 & 1,924  & 125.07 & 1,968  & 184.97 & 2,056  & 18.83 & 2,102 \\
			& $5$& $0.8$ 	& 605.38 & 9,669  & 185.49 & 9,232  & 208.27 & 9,306  & 296.96 & 9,550  & 25.39 & 9,663 \\
			& $5$& $1$ 	& 488.02 & 15,560 & 152.06 & 15,057 & 177.59 & 15,150 & 252.22 & 15,517 & 20.05 & 15,552 \\
			& $10$& $0$ 	& 202.64 & 1,258  & 42.03 & 1,126   & 55.65 & 1,160   & 107.56 & 1,203  & 8.26 & 1,256 \\
			& $10$& $0.6$ 	& 380.46 & 2,638  & 110.93 & 2,423  & 124.09 & 2,450  & 185.16 & 2,609  & 18.99 & 2,640 \\
			& $10$& $0.8$ 	& 609.48 & 11,253 & 179.73 & 10,919 & 203.65 & 11,035 & 291.75 & 11,227 & 24.31 & 11,244 \\
			& $10$& $1$ 	& 476.07 & 16,950 & 145.35 & 16,687 & 169.36 & 16,750 & 240.98 & 16,947 & 19.44 & 16,951 \\				
			& $20$& $0.4$ 	& 213.17 & 1,328  & 41.38 & 1,198   & 54.94 & 1,224   & 107.45 & 1,313  & 8.49 & 1,327 \\
			& $20$& $0.6$ 	& 394.77 & 2,905  & 108.75 & 2,743  & 122.387 & 2,796 & 183.53 & 2,885  & 18.99 & 2,900 \\
			& $20$& $0.8$ 	& 586.81 & 11,431 & 166.97 & 11,147 & 190.06 & 11,254 & 275.55 & 11,439 & 22.98 & 11,438 \\
			& $20$& $1$ 	& 445.40 & 16,427 & 131.03 & 16,302 & 153.92 & 16,355 & 219.95 & 16,427 & 18.12 & 16,437 \\
			\hline		
			\multirow{15}{*}{threads-math} & $5$& $0$   	& 18.65 & 4,549   & 2.89 & 4,319    & 3.61 & 4,395    & 7.04 & 4,543    & 0.70 & 4,552 \\
			& $5$& $0.4$ 	& 18.71 & 4,564   & 2.92 & 4,345    & 3.62 & 4,419    & 7.12 & 4,557    & 0.70 & 4,567 \\
			& $5$& $0.6$ 	& 21.23 & 5,339   & 3.93 & 5,209    & 4.62 & 5,259    & 8.35 & 5,339    & 0.91 & 5,336 \\
			& $5$& $0.8$ 	& 32.89 & 9,366   & 8.31 & 9,262    & 9.52 & 9,296    & 14.53 & 9,366   & 1.72 & 9,359 \\
			& $5$& $1$ 	& 32.61 & 9,741   & 8.59 & 9,619    & 9.74 & 9,670    & 14.54 & 9,741   & 1.72 & 9,731 \\
			& $10$& $0$ 	& 15.93 & 2,708   & 2.40 & 2,545    & 3.02 & 2,595    & 6.33 & 2,697    & 0.50 & 2,708 \\
			& $10$& $0.4$ & 16.08 & 2,722   & 2.45 & 2,565    & 3.07 & 2,615    & 6.39 & 2,713    & 0.51 & 2,708 \\
			& $10$& $0.6$ & 18.93 & 3,318   & 3.26 & 3,228    & 3.96 & 3,241    & 7.65 & 3.316    & 0.67 & 3,320 \\
			& $10$& $0.8$ & 22.86 & 5,618   & 4.78 & 5,522    & 5.71 & 5,582    & 9.51 & 5,618    & 0.94 & 5,608 \\
			& $10$& $1$ 	& 21.49 & 5,773   & 4.53 & 5,644    & 5.32 & 5,707    & 8.96 & 5,768    & 0.87 & 5,760 \\				
			& $20$& $0$ 	& 13.06 & 1,658   & 1.95 & 1,573    & 2.49 & 1,598    & 5.59 & 1,661    & 0.37 & 1,673 \\
			& $20$& $0.4$ & 13.25 & 1,691   & 1.99 & 1,594    & 2.55 & 1,623    & 5.64 & 1,690    & 0.38 & 1,683 \\
			& $20$& $0.6$ & 16.12 & 2,157   & 2.62 & 2,099    & 3.29 & 2,143    & 6.69 & 2,160    & 0.49 & 2,159 \\
			& $20$& $0.8$ & 12.80 & 3,693   & 2.09 & 3,627    & 2.63 & 3,667    & 5.03 & 3,692    & 0.43 & 3,701 \\
			& $20$& $1$ 	& 10.83 & 3,625   & 1.68 & 3,589    & 2.12 & 3,616    & 4.25 & 3,635    & 0.36 & 3,622 \\				
			\hline		
			\multirow{15}{*}{threads-SO} & $5$& $0$   	& 634.38 & 13,739  & 183.04 & 11,927 & 186.79 & 12,209 & 205.62 & 13,083 & 41.64 & 13,736 \\
			& $5$& $0.4$ 	& 642.64 & 13,748  & 182.99 & 11,935 & 187.02 & 12,216 & 205.41 & 13,093 & 41.60 & 13,745 \\
			& $5$& $0.6$ 	& 653.52 & 15,050  & 224.97 & 13,462 & 227.72 & 13,659 & 240.00 & 14,536 & 49.05 & 15,046 \\
			& $5$& $0.8$ 	& 1409.66 & 31,383 & 522.59 & 29,317 & 529.38 & 30,198 & 565.72 & 31,133 & 112.59 & 31,336 \\
			& $5$& $1$ 		& 1592.06 & 36,082 & 572.42 & 34,044 & 580.25 & 34,663 & 638.44 & 35,913 & 133.73 & 36,013 \\
			& $10$& $0$ 	& 562.91 & 8,425   & 148.98 & 7,432  & 154.16 & 7,532  & 170.28 & 8,100  & 29.87 & 8,417 \\
			& $10$& $0.4$ 	& 564.14 & 8,449   & 149.99 & 7,456  & 153.27 & 7,558  & 170.84 & 8,124  & 30.08 & 8,433 \\
			& $10$& $0.6$ 	& 641.40 & 10,100  & 204.75 & 9,210  & 205.99 & 9,331  & 219.53 & 9,879  & 38.35 & 10,065 \\
			& $10$& $0.8$ 	& 1105.59 & 21,202 & 340.25 & 19,490 & 346.77 & 19,819 & 375.53 & 21,091 & 67.12 & 21,099 \\
			& $10$& $1$ 	& 1131.47 & 23,234 & 342.87 & 22,056 & 349.27 & 22,311 & 382.53 & 23,058 & 66.88 & 23,115 \\				
			& $20$& $0$ 	& 477.38 & 5,155   & 120.34 & 4,502  & 123.52 & 4,621  & 140.99 & 4,943  & 19.31 & 5,124 \\
			& $20$& $0.4$ 	& 482.70 & 5,177   & 123.59 & 4,539  & 128.87 & 4,645  & 145.46 & 4,967  & 19.66 & 5,154 \\
			& $20$& $0.6$ 	& 610.01 & 7,060   & 165.86 & 6,457  & 170.97 & 6,591  & 187.87 & 6,889  & 28.21 & 7,019 \\
			& $20$& $0.8$ 	& 690.80 & 14,129  & 158.07 & 13,622 & 161.77 & 13,725 & 185.95 & 14,029 & 27.34 & 14,009 \\
			& $20$& $1$ 	& 533.83 & 14,734  & 109.77 & 14,259 & 113.39 & 14,358 & 134.04 & 14,614 & 20.65 & 14,461 \\				
			\hline		
		\end{tabular}%
	}
\end{table}

\end{document}